\documentclass[11pt]{imsart}
\RequirePackage{amsthm,amsmath,amsfonts,amssymb}
\RequirePackage[numbers]{natbib}

\startlocaldefs

\usepackage[textsize=footnotesize]{todonotes}
\usepackage{hyperref}

\usepackage{mathrsfs}
\usepackage{graphicx}
\usepackage{amsmath}        
\usepackage{amssymb}
\usepackage{amsthm}
\usepackage{bm}
\usepackage[mathscr]{eucal}
\usepackage[countmax]{subfloat}  
\usepackage{fourier}
\usepackage{xcolor}
\usepackage{mathtools}
\usepackage{enumerate}
\usepackage{enumitem}
\usepackage{chngcntr}
\counterwithin{figure}{section}
\usepackage{natbib}
\usepackage{url}
\usepackage{verbatim}
\usepackage{booktabs}
\usepackage{placeins}
\usepackage{multirow}
\usepackage{caption}
\graphicspath{{plots/}}
\usepackage[title]{appendix}
\usepackage{subcaption}

\newtheorem{prop}{Proposition}[section]     
\newtheorem{thm}{Theorem}[section]
\newtheorem{lemma}{Lemma}[section]

\newtheorem{Remark}{Remark}[section]

\newtheorem{assumption}{Assumption}[section]

\theoremstyle{definition}

\newcommand{\im}{\mathrm{i}}

\newcommand{\HD}{\textcolor{red}}

\newcommand{\znum}{\mathbb{Z}}

\newcommand{\mynegspace}{\hspace{-0.12em}}

\newcommand{\Bigsnorm}[1]{\Bigg\rvert\mynegspace\Bigg\rvert\mynegspace\Bigg\rvert\mynegspace {#1} \Bigg\rvert\mynegspace \Bigg\rvert\mynegspace \Bigg\rvert}
\newcommand{\bigsnorm}[1]{\Big\rvert\mynegspace\Big\rvert\mynegspace\Big\rvert\mynegspace {#1} \Big\rvert\mynegspace \Big\rvert\mynegspace \Big\rvert}
\newcommand{\snorm}[1]{\rvert\mynegspace\rvert\mynegspace\rvert\mynegspace {#1} \rvert\mynegspace \rvert\mynegspace \rvert}

\newcommand{\norm}[1]{{\|}{#1}{\|}}
\newcommand{\bignorm}[1]{\Big{\|}{#1}\Big{\|}}
\newcommand{\Bignorm}[1]{\Bigg{\|}{#1}\Bigg{\|}}
\newcommand{\inprod}[2]{\langle #1, #2 \rangle}
\newcommand{\biginprod}[2]{\Big\langle #1, #2 \Big\rangle}
\newcommand\tageq{\addtocounter{equation}{1}\tag{\theequation}}

\newcommand{\lamx}[2]{\lambda^{(#1)}_{X,#2}}
\newcommand{\lamy}[2]{\lambda^{(#1)}_{Y,#2}}
\newcommand{\phix}[1]{{\Pi}^{(\omega)}_{X,#1}}
\newcommand{\phiy}[1]{{\Pi}^{(\omega)}_{Y,#1}}
\newcommand{\phiT}[1]{\hat{\Pi}^{(\omega)}_{#1}}
\newcommand{\phixT}[1]{\hat{\Pi}^{(\omega)}_{X,#1}}
\newcommand{\phiyT}[1]{\hat{\Pi}^{(\omega)}_{Y,#1}}
\newcommand{\lamxT}[1]{\hat{\lambda}^{(\omega)}_{X,#1}}
\newcommand{\lamyT}[1]{\hat{\lambda}^{(\omega)}_{Y,#1}}
\newcommand{\cnum}{\mathbb{C}}
\newcommand{\rnum}{\mathbb{R}}
\newcommand{\E}{\mathbb{E}}
\newcommand{\Cov}{\text{Cov}}

\newcommand{\F}{\mathcal{F}}
\newcommand{\G}{\mathcal{G}}
\newcommand{\Hi}{\mathcal{H}}
\newcommand{\hi}{\mathbb{H}}

\newcommand{\op}{\mathcal{L}}
\newcommand{\flo}[1]{\lfloor #1 \rfloor}
\newcommand{\ldm}[2]{\mathcal{M}^{{#1}(\omega)}_{#2,m}}

\newcommand{\dmp}[2]{D^{(\omega)}_{#2,m,#1}}
\newcommand{\dmpi}[3]{D^{(\omega_{#2})}_{#3,m,#1}}

\newcommand{\jpr}{j^\prime}

\newcommand{\X}{\bm{{X}}}
\newcommand{\Hd}{\mathcal{H}}
\newcommand{\Hs}{\mathscr{H}}
\allowdisplaybreaks
\providecommand{\keywords}[1]{\textbf{\textit{keywords: }} #1}
\providecommand{\AMS}[1]{\textbf{\textit{AMS subject classification: }} #1}

\endlocaldefs
\usepackage{xr}
\begin{document}

\begin{frontmatter}
\title{Pivotal tests for relevant differences in the second order dynamics  of functional time series}
%\title{A sample article title with some additional note\thanksref{T1}}
\runtitle{Pivotal tests for relevant differences  in the second order dynamics  of functional time series}
%\thankstext{T1}{A sample of additional note to the title.}

\begin{aug}
\author[A]{\fnms{Anne} \snm{van Delft}\ead[label=e1]{anne.vandelft@columbia.edu}}
\and
\author[B]{\fnms{Holger} \snm{Dette}\ead[label=e2,mark]{holger.dette@rub.de}}
%%%%%%%%%%%%%%%%%%%%%%%%%%%%%%%%%%%%%%%%%%%%%%
%% Addresses                                %%
%%%%%%%%%%%%%%%%%%%%%%%%%%%%%%%%%%%%%%%%%%%%%%
\address[A]{Department of Statistics, Columbia University, 1255 Amsterdam Avenue, New York, NY 10027, USA. \printead{e1}}

\address[B]{Ruhr-Universit\"at Bochum, Fakult\"at f\"ur Mathematik, 44780 Bochum, Germany \printead{e2}}
\end{aug}

\begin{abstract}
Motivated by the need to statistically quantify differences between modern (complex) data-sets which commonly result as high-resolution measurements of stochastic processes varying over a continuum, we propose novel testing procedures to detect relevant differences between the second order dynamics of two functional time series. In order to take the between-function dynamics into account that characterize this type of functional data, a frequency domain approach is taken. Test statistics are developed to compare differences in the spectral density operators and in the primary modes of variation as encoded in the associated eigenelements. Under mild moment conditions, we show convergence of the underlying statistics to Brownian motions and construct  pivotal test statistics. The latter is essential because the  nuisance parameters can be unwieldy and their robust estimation infeasible, especially if the two functional time series are dependent. In addition to these novel features, the properties of the tests are robust to any choice of frequency band enabling also to compare energy contents at a single frequency. The finite sample performance of the tests are verified through a simulation study 
and are  illustrated with an application to fMRI data.
\end{abstract}


\end{frontmatter}
 \AMS{Primary: 62M15, 60G10; Secondary 62M10.}

%%%%%%%%%%%%%%%%%%%%%%%%%%%%%%%%%%%%%%%%%%%%%%
%%%% Main text entry area:
\section{Introduction}\label{sec:intro}
\def\theequation{1.\arabic{equation}}
\setcounter{equation}{0}

Functional time series analysis is concerned with the development of inference methods to model and  analyze data measurements from processes that take values over some continuum like a curve, a surface or a sphere and which exhibit a natural dependency between the  observations, each considered as a point in the function space $\Hi$. 
In the current day and age of technological advances where measurements of a process can be taken over its entire domain of definition at a high precision, it is not surprising that functional time series analysis is of increased applicability in numerous 
research areas. Examples can be found in molecular biophysics \citep[][]{TavPan2016}, brain imaging \citep{aston2012}, climatology  \citep[][]{zhang2011,Zhang2015}, environmental data \citep{HorKokNis17} or yet economics \citep{AnPapSap06,KMR2019}.  Naturally, this has led to an upsurge in the available literature on statistical  methodology  for the analysis of functional time series.

The main purpose of this paper is to develop frequency domain based inference methods which allow to quantify differences in the second order characteristics of two weakly stationary (possibly dependent)
 functional time series, say, $\{X_t\}_{t\in \znum}$ and $\{Y_t\}_{t\in \znum}$. Comparison of the second order characteristics of two functional time series is of interest in various applications and controlled experiments. The motivation in most cases is to know whether two series are similar or that a joint analysis on the pooled data is relevant to consider. Inherent to this type of sequentially collected functional data is the presence of temporal dependence. The second order structure is therefore more involved than for independent functional data, yet the development of appropriate inference methods are of the same eminent importance; the second order dynamics play a key role in providing information on the smoothness properties of the random functions and optimal dimension reduction techniques. 

For  independent functional data, statistical inference tools for comparing covariance operators 
 have  been developed by \cite{Panaretos2010}, \cite{fremdt2013},  \cite{Guo2016} and  \cite{Paparoditis2016}. \cite{BHK2009} and \cite{PoStGh16} investigated in how far the distribution of two random samples of independent functional data coincide by means of their Karhunen-Lo{\`e}ve expansion, and developed tests to compare the functional principal components, i.e., the eigenvalues and eigenfunctions of the autocovariance operator. In the context of temporally dependent functional data, methods in this direction have also been considered. Motivated by climate downscaling studies, \cite{Zhang2015} proposed testing for equality of the 0-lag covariance operators of two functional time series and of their associated eigenvalues and eigenfunctions. More recently \cite{Pilavakis2019}, proposed a test for the equality  of the $0$-lag covariance operators of several independent functional time series. 

Time domain methods as considered in aforementioned literature however suffer from important shortcomings when one wants to infer on the second order dynamics of temporally dependent functional data. 
The autocovariance operator only captures \textit{static} features and the long-run covariance operator, being a sum of the sequence of $h$-lag covariance operators, only captures crude features of the dynamics. In addition, functional principal component analysis (FPCA) does not provide an optimal dimension reduction since it ignores any temporal dynamics present in the collection of functional observations.
\\
To analyze or compare second order dynamics of functional time series, a frequency domain approach might in fact  be more appropriate. %Assuming no points of discontinuity in the spectral measure $\mathscr{F}_X(\cdot)$ on $[-\pi,\pi]$, 
Under regularity conditions, the \textit{spectral density operator} $\{\F_X^{(\omega)}\}_{\omega \in [-\pi,\pi]}$ characterizes the full second order dynamics of the process and the corresponding sequences of eigenelements $\{\phi^{(\omega)}_{X,k},\lambda_{X,k}^{(\omega)} \}_{k \ge 1}$ provide a starting point for an optimal lower dimensional representation that also captures the temporal dynamics of the process. More specifically, the \textit{Cram{\'e}r-Karhunen-Lo{\`e}ve decomposition} \citep[][]{PanTav2013b} of a zero-mean $\Hi$-valued stochastic process $\{X_t\}_{t \in \znum}$ is (formally) given by\[ 
X_t = \int_{-\pi}^\pi e^{\im \omega t} \big(\sum_{j=1}^{\infty} \phi^{\omega}_j \otimes \phi^{\omega}_j\big) Z_{X,\omega}  =
\int_{-\pi}^\pi e^{\im \omega t} \sum_{j=1}^{\infty} \inprod{d Z_{X,\omega}}{\phi^{\omega}_j} \phi^{\omega}_j~.
\]
This representation essentially decomposes the process into uncorrelated function-valued frequency components and --in analogy with the classical Karhunen-Lo{\`e}ve representation-- it separates the functional and stochastic components. Truncation of the series inside the integral at a finite level then provide a \textit{harmonic} FPCA, an optimal lower dimensional representation of the functional time series \citep[see also][on this topic]{PanTav2013b,Hormann2015,vDE19}. Indeed, the covariance operator of the infinitesimal increment $d Z_{X,\omega}$ is $\F_X^{(\omega)}$, and thus the eigenfunctions $\{\phi^{(\omega)}_{X,k}\}_{k \ge 1}$ form the optimal basis to expand $d Z_{X,\omega}$, whereas the corresponding eigenvalues $\{\lambda_{X,k}^{(\omega)}\}_{k \ge 1}$ provide insight on the relative contribution of each frequency component to the total variation in the process, as well as on the dimensionality of each component. In order to compare second order characteristics of functional time series, it is therefore of interest to be able to compare the spectral density operators as well as to compare the primary modes of variation as given by the respective eigenprojectors $\phix{k} = \phi_{X,k}^{(\omega)} \otimes \phi_{X,k}^{(\omega)}$ and eigenvalues $\lamx{\omega}{k}$ ($k\ge 1$).
\\
In this paper, our goal is to develop  pivotal tests statistics to detect relevant  differences in the second order structure between two functional time series 
based on the spectral density operators  and their associated characteristics as given by the eigensystems (eigenprojectors and eigenvalues). 
\\ 

 The novelty of our approach lies in four different aspects.
\begin{itemize} 
\item[(i)] 
 Firstly, while methods to test for equality of spectral density operators of two functional time series are available
  \citep[see e.g.,][]{TavPan2016,LPapSap2018},  tests to compare the  eigenelements  of spectral density operators have, to the best of our knowledge, not yet been considered in existing (function-valued) time series literature. Due to their central role in dimension reduction techniques, these tests are extremely relevant but far from trivial to construct. 
\item[(ii)] 
Secondly, our approach is in terms of a relevant testing framework, which means that we are only  
interested in deviations  that  surpass a certain threshold.  For example, in the context of  comparing spectral density operators  we 
do \textit{not} consider the problem of testing  for exact equality
of the spectral density operators  $\F_X^{(\cdot)} $ and $\F_Y^{(\cdot)} $, but instead propose to investigate hypotheses of the form 
\[ \label{h0} 
H_{0}:  \int_a^b\snorm{\F_X^{(\omega)}    - \F_Y^{(\omega)}}^{2}  d\omega  \leq \Delta \quad a\le b \in [0,\pi]
\]
  of no relevant deviation between $\F_X^{(\cdot)} $ and $\F_{Y}^{(\cdot)} $ over a given frequency band. 
Here $\snorm{\, \cdot}$ denotes an appropriate  norm and    $\Delta > 0$ is a pre-specified threshold. 
Note that classical hypotheses as considered in \cite{TavPan2016} and \cite{LPapSap2018}  are  obtained  with the threshold  set to zero but the case $\Delta=0$ is not considered in this paper.
Our motivation for considering relevant hypotheses (i.e., $\Delta >0$)   stems from  the observation that in many applications, it is clear from the scientific background  that exact equality of the second order structures  of  functional  time series  $\{X_t\}_{t\in \znum}$ and $\{Y_t\}_{t\in \znum}$  does not hold. However, one might be interested in working  under this assumption if the deviation between the two operators is small. In this case,  the two series could be merged in the statistical analysis  because  the 
difference between $\F_X^{(\cdot)}$ and $\F_Y^{(\cdot)}$  is small.  A similar comment applies   to  the   eigenfunctions and eigenprojectors of a spectral density operator, for which relevant  hypotheses can  be defined similarly  (see  Section \ref{sec22} for more  details).
We would like to emphasize that testing relevant hypotheses avoids the consistency problem as mentioned in   \cite{berkson1938}:   any consistent test will detect any arbitrary small change in the parameters if the sample size is sufficiently large.
\item[(iii)] 
Thirdly,  tests for hypotheses involving quantities derived from the spectral density operators are of a very complicated nature. The asymptotic distributions of corresponding test statistics oftentimes depend on the unknown objects of interest or on the higher order dynamics of the functional time series. 
For example, \cite{LPapSap2018} consider classical testing problems and use  the bootstrap to avoid estimation of a functional of the spectral density operator. However, if relevant hypotheses of the form  \eqref{h0} have to be tested then the construction of a bootstrap procedure is highly non-trivial as one has to mimic the distribution of
a  test statistic under a null hypothesis,  which differs only  in  a quantitative but not in  a qualitative way  from the alternative.
The situation becomes even more difficult in the construction of testing procedures for relevant hypotheses involving the eigenfunctions and eigenprojectors. 
In this paper we solve this problem; we develop  tests  that  are pivotal and do neither  require the estimation of such nuisance parameters nor a bootstrap approach. 
\item[(iv)] 
Fourthly,  we derive our results under extremely mild moment conditions, which
are  much weaker than those available in the existing literature on (functional) time series (see \autoref{sec3} for details). 
\end{itemize} 
The structure of this article is as follows. First, we introduce the precise form of our hypotheses, relate this to existing literature, and highlight the importance of considering pivotal test statistics. In \autoref{sec2}, we introduce our testing frameworks. All proposed test statistics can be expressed as a functional of a  `building block'  process.  In  \autoref{sec3}, we introduce and investigate this process, and establish its limiting distribution. These results are then used to develop new tests and to investigate  their statistical  properties.  In \autoref{sec:sec5}, we study  the finite sample properties of the proposed tests in a simulation study. Finally, in \autoref{sec:sec4}, we provide the main argument to establish the weak convergence of the  `building block'  process. Further technical  details are deferred to \autoref{sec:mainstat} and \autoref{sec:proofs} of the supplementary file  \cite{vDD20}. An application of the proposed methodology to resting state fMRI data may be found in  \autoref{sec:data}.

\section{Relevant hypotheses for second order dynamics} \label{sec2}
\def\theequation{2.\arabic{equation}}
\setcounter{equation}{0}
 \subsection{Notation} \label{sec21}

We start by introducing some required terminology. Let $\Hi$ be a separable Hilbert space with inner product $\inprod{\cdot}{\cdot}$ and induced norm $\|\cdot\|_{\mathcal{H}}$. We denote the Hilbert tensor product between two Hilbert spaces $(\Hi_j, \inprod{\cdot}{\cdot}_{\Hi_j})_{j=1,2}$ by $\Hi^{\otimes_2}=\Hi_1 \otimes \Hi_2$, whose elements are linear combinations of the simple tensors $h_1\otimes h_2$, $h_j \in \Hi_j, j =1,2$. This is a Hilbert space formed from the algebraic tensor product together with a bilinear map $\psi: \Hi_1 \times \Hi_2 \to \Hi_1 \otimes_{\mathrm{alg}} \Hi_2$ satisfying $\inprod{\psi(h_1,h_2)}{\psi(g_1,g_2)} = \inprod{h_1}{g_1}_{\Hi_1}\inprod{h_2}{g_2}_{\Hi_2}$ for $h_1, g_1 \in \Hi_1$ and $h_2,g_2 \in H_2$ and then taking the completion with respect to the induced norm. We denote the direct sum of two Hilbert spaces by $\Hi^{\oplus_2} =\Hi_1 \oplus \Hi_2$, of which elements are of the form  $h=(h_1, h_2)^\top$, where $(\cdot)^\top$ denotes the transpose operation. Observe that this is again a Hilbert space with inner product $\inprod{g}{h}=\sum_{j=1}^2 \inprod{g_j}{h_j}_{\Hi_j}$, for any $g, h \in \Hi^{\oplus_2}$. For more details on these facts we refer to \cite{KadRing97}. Let $\{\chi_j\}_{j \ge 1}$ be an orthonormal basis of $\Hi_1$. For a bounded linear operator $A:\Hi_1 \to \Hi_2$ we define, respectively, the operator norm by $\snorm{A}_{\infty}=\sup_{\|g\|_{\Hi_1}\le 1}\|A(g)\|_{\Hi_2}$, $g\in \Hi_1$, the Hilbert-Schmidt norm by $\snorm{A}_2 = \sum_{j\geq 1}  \big(\|A(\chi_j)\|^2_{\Hi_2}\big)^{1/2}$, which is induced by the inner product $\inprod{A_1}{A_2}_{S_2} = \sum_{j \ge 1} \inprod{A_1 \chi_j}{A_2 \chi_j}_{\Hi_2},$ $A_1$, $A_2: \Hi_1 \to \Hi_2$, and for $A: \Hi_1 \to \Hi_1$ the trace class norm by $\snorm{A}_1 = \sum_j \inprod{ (AA^{\dagger})^{1/2}(\chi_j)}{\chi_j}_{\Hi_1}$, where $A^\dagger$ denotes the adjoint of $A$. We write $A \in S_2(\Hi_1, \Hi_2)$ if it has finite Hilbert-Schmidt norm and abbreviate $S_2(\Hi) =S_2(\Hi, \Hi)$. For a bounded linear operator $A: \Hi \to \Hi$, with $\snorm{A}_1<\infty$ we write $A \in S_1(\Hi)$. For $f, g, v \in \Hi$, we define the tensor product operator $f \otimes g\colon \Hi \to \Hi$ as the bounded linear operator $(f \otimes g)v =\inprod{v}{g}f$.
 %, where we remark that  the mapping $\mathcal{T}\colon \Hi \otimes \Hi \to S_2(\Hi)$ defined by the linear extension of $\mathcal{T}(f \otimes g) = f \otimes \overline{g}$ is an isometric isomorphism, where $\overline{g}$ denotes the complex conjugate (function) of $g$.
We additionally define the Kronecker tensor product as $(A \widetilde{\otimes} B)C = ACB^{\dagger}$ for $A, B, C \in S_{2}(H)$ .  

Next, for a $\Hi$-valued random element $X$ over a probability space $(\Omega,\mathcal{A},\mathbb{P})$, we shall denote $X \in \op^p_\Hi$ if $\|X\|_{\hi,p}=(\E\|X\|_{\Hi}^p)^{1/p} <\infty$.  Observe that $\op^2_\Hi$ is a Hilbert space  consisting of $\Hi$-valued random elements with finite second order moment. We note moreover that for any $X, Y  \in \op^2_\Hi$ with zero mean, the cross-covariance operator is given by $\text{Cov}(X,Y) = \E(X\otimes Y)$ and belongs to $S_1(\Hi)$. For a zero mean element $X=(X_1,X_2)^\top \in \op^2_{\Hi \oplus \Hi}$, we note that $\Cov(X,X)=\E X^{\otimes_2}$
 consists of the components $\E (X_i \otimes X_j)$, $i,j =1,2$ which are elements of $S_1(\Hi)$. Furthermore, we denote the imaginary unit by $\im$ and $\overline{g}$ denotes the complex conjugate (function) of $g$. We use $\Re(\cdot)$ and $\Im(\cdot)$ to denote  the real and imaginary part, respectively, of a complex-valued object. We write $a_T \sim b_T$ if $\lim_{T \to \infty} \frac{a_T}{b_T} =1$. Weak convergence in $D[0,1]$ {--  the space of right-continuous functions with left-hand limits --} with respect to the Skorokhod topology will be denoted by  ${\rightsquigarrow}$, while convergence in distribution as $T\to \infty$ will be denoted by $\underset{T\to\infty}{\Rightarrow}$. Finally, we reserve $\mathbb{B}$ to denote standard Brownian motion on the interval $[0,1]$ and remark that $\flo{\cdot}$ denotes the floor function.

 \subsection{Relevant hypotheses} \label{sec22}
 
%\textcolor{blue}{In this paper, we consider $\Hi$-valued time series $\{X_t\}_{t \in \znum} \in \op^2_\Hi$ and $\{Y_t\}_{t \in \znum} \in \op^2_\Hi$, which are assumed to be weakly stationary. This implies in particular that $\Cov(X_{t+h}, X_{t}) = \Cov(X_{h}, X_{0})$ and $\Cov(Y_{t+h}, Y_{t}) = \Cov(Y_h,Y_0)$ $ \forall t, h \in \znum$, where the equality holds in trace class norm. }\\ \todo{replace red part with blue}
In this paper, we consider weakly stationary processes $\{\X_t \colon t \in \znum\}$ with $\X_t=(X_t, Y_t)^\top \in \op^2_{\Hi \oplus \Hi}$. This implies in particular that the $h$-lag covariance operator of $\{\X_t\}_{t\in \znum}$, satisfies 
\begin{align*} 
\Cov(\X_{t+h}, \X_{t}) &=\begin{pmatrix}\Cov(X_{t+h}, X_{t}) &  \Cov(X_{t+h}, Y_{t}) \\ \Cov(Y_{t+h}, X_{t}) & \Cov(Y_{t+h}, Y_{t}) \end{pmatrix}  
= \Cov(\X_h, \X_0) 
\end{align*}
for all $t, h \in \znum.$ Under mild regularity conditions, which we shall explain in more detail in \autoref{sec3}, the full second order dynamics of the component processes $\{X_t\}_{t\in \znum}$ and $\{Y_t\}_{t\in \znum}$ are respectively captured by the spectral density operators 
\begin{align*}
\F_X^{(\omega)} = \frac{1}{2\pi}\sum_{h \in \znum} \Cov(X_{h}, X_{0}) e^{-\im \omega h} \quad \text{ and } \quad \F_Y^{(\omega)} = \frac{1}{2\pi}\sum_{h \in \znum} \Cov(Y_{h}, Y_{0}) e^{-\im \omega h}, \quad \omega \in [-\pi,\pi].
\end{align*}
\\
 In the following, we introduce the three testing frameworks to test for relevant differences in the second order characteristics of 
 the component processes $\{X_t\}_{t \in \znum}$ and $\{Y_t\}_{t\in \znum}$. 
 As a first option, this can be framed as the following hypothesis testing problem on the spectral density operators;
\begin{align} \label{eq:testHypOp} 
H_0: \int_{a}^{b}\bigsnorm{\F^{(\omega)}_X-\F^{(\omega)}_Y}_2^2 d\omega \le \Delta  \quad \text{ versus } \quad H_A: \int_{a}^{b}\bigsnorm{\F^{(\omega)}_X-\F^{(\omega)}_Y}_2^2 d\omega >  \Delta, 
\end{align}
where $[a, b ] \subseteq  [0,\pi]$ and $\Delta >0$ is a pre-specified constant that represents the maximal value for which the distance
$ \int_{a}^{b}\snorm{\F^{(\omega)}_X-\F^{(\omega)}_Y}_2^2 d\omega $ is considered as not relevant. 
Note that by specifying the choice of $a$ and $b$, one can compare the spectral density operators  within a certain narrow frequency band or even at a single frequency, which is of interest in certain applications. For instance, activities of certain areas of the brain, such as the Nucleas Accumbens, are usually located within a small frequency band around frequency zero \citep[see e.g.,][]{FiecOmb16} and the characteristics of resting-state fMRI data tend to have rather frequency-specific biological interpretations \citep[see e.g.,][and references therein]{yuen2019}.

Besides \eqref{eq:testHypOp}, the main focus in this paper is on two more refined hypotheses testing problems that allow to infer relevant differences in the primary modes of variation. To make these more precise, we shall make use of the fact that the spectral density operators $\F_X^{(\omega)}$ and $\F_Y^{(\omega)}$ admit real-valued discrete spectra which are, respectively, given by
\[\F^{(\omega)}_X = \sum_{k=1}^{\infty}\lamx{\omega}{k}  \phix{k} \quad \text{ and } \quad \F^{(\omega)}_Y = \sum_{k=1}^{\infty}\lamy{\omega}{k}  \phiy{k}, \]
where $\{\lamx{\omega}{k}\}_{k\ge 1}$ is the sequence of eigenvalues of $\F^{(\omega)}_{X}$ arranged in descending order and where $\phix{k} = \phi_{X,k}^{(\omega)} \otimes \phi_{X,k}^{(\omega)}$ with $\big\{\phi_{X,k}^{(\omega)}\big\}_{k\ge 1}$ denoting the corresponding sequence of eigenfunctions. 
The operator $\phix{k}$ is a self-adjoint rank-one operator and will be referred to as the $k$-th \textit{eigenprojector} (at frequency $\omega$) since it projects onto the eigenspace of $\F^{(\omega)}_X$ corresponding to the $k$-th largest eigenvalue $\lamx{\omega}{k}$. The eigenelements of $\F^{(\omega)}_Y$ are defined in a similar manner.  
\\
To consider the relevant differences at component $k$ for some $k \in \mathbb{N}$, we are in particularly interested in providing a meaningful test for the hypotheses of  no relevant difference between the $k$-th eigenprojectors, that is
\begin{align} \label{eq:testHyp}
H_0: \int_{a}^{b}\bigsnorm{\phix{k}-\phiy{k}}_2^2 d\omega \le \Delta_{\Pi,k}  \quad \text{ versus } \quad H_A: \int_{a}^{b}\bigsnorm{\phix{k}-\phiy{k}}_2^2 d\omega >  \Delta_{\Pi,k}, 
\end{align}
where  $[a,b] \subseteq  [0,\pi]$ and where $\Delta_{\Pi,k}>0$ denotes, similarly to $\Delta$, a pre-specified constant. It is worth mentioning that the eigenfunctions $\{ \phi_{X,k}^{(\omega)}\}_{k \ge 1}$ are complex elements of $\Hi$ (except at $\omega=0,\pi$). Due to this, a test statistic based upon the difference of the empirical eigenfunctions is not feasible because these are only identifiable up to a rotation on the unit circle. The testing framework in \eqref{eq:testHyp} is therefore formulated in terms of the eigenprojectors since their empirical counterparts are rotationally invariant. We come back to this in \autoref{sec3}. Finally, we also consider the hypotheses of  no relevant difference between the $k$-th eigenvalues, that is 
\begin{align} \label{eq:testHypEv}
H_0: \int_{a}^{b}\big \vert\lamx{\omega}{k}-\lamy{\omega}{k}\big\vert^2 d\omega \le \Delta_{\lambda,k}  \quad \text{ versus } \quad H_A: \int_{a}^{b}\big\vert\lamx{\omega}{k}-\lamy{\omega}{k}\big\vert^2 d\omega >  \Delta_{\lambda,k},
\end{align}
where  $[a,b] \subseteq  [0,\pi]$ and where $\Delta_{\lambda,k}  >0$ is again a pre-specified constant that represents the maximal value for which the difference between the $k$-th eigenvalues is deemed not relevant.

In this article, we develop pivotal tests for the hypotheses in \eqref{eq:testHypOp}, \eqref{eq:testHyp} and \eqref{eq:testHypEv}, where the threshold  is positive (thus we do {\bf not} consider testing of classical hypotheses).
To elaborate on its relevance and to motivate that
this is a very challenging problem, observe that a natural approach to test hypotheses of the form \eqref{eq:testHypOp} is to construct an empirical distance measure 
\begin{equation}   \label{hd22}
\widehat{M}^{2}= \int_{a}^{b} \bigsnorm{ \hat{\F}^{(\omega)}_X-\hat{\F}^{(\omega)}_Y }^2_2 d\omega,
\end{equation} 
of the distance $M^{2}= \int_{a}^{b} \snorm{\F^{(\omega)}_X-\F^{(\omega)}_Y}_2^2 d\omega  $, 
where $\hat{\F}^{(\omega)}_X$ and $\hat{\F}^{(\omega)}_Y$ are suitable estimators of the spectral density operators ${\F}^{(\omega)}_X$ and ${\F}^{(\omega)}_Y$, respectively, and to reject the null hypothesis for large values of \eqref{hd22}. For classical hypotheses, i.e., where $H_0: M^2 = \int_{a}^{b} \snorm{\F^{(\omega)}_X-\F^{(\omega)}_Y}_2^2 d\omega  =0$, one then requires the (asymptotic) distribution of the statistic at $M^2 = 0$ in order to determine the critical values, which involves the estimation of certain nuisance parameters. The latter was for example considered by \cite{TavPan2016}, who construct a test for equality of spectral density operators based upon this distance restricted to a finite-dimensional subspace \citep[see also][who considered  this approach for covariance operators]{Panaretos2010}. A drawback is that the method can be sensitive to the specific choice of several regularization parameters, including an appropriate truncation level for the dimension  of which the optimal value is frequency-dependent. Another approach was taken in \cite{vDD18}, who introduced a fully functional similarity measure for (time-varying) spectral density operators of possibly nonstationary functional time series where the distance measure is estimated based upon integrated functionals of (localized) periodogram operators. While this avoids sensitivity to certain regularization parameters, the expressions of the asymptotic variance can still become quite involved when certain assumptions, such as independence of the two series, are relaxed. Alternatively, one could consider a bootstrap method to obtain the critical values of the test statistic, an approach taken by \cite{LPapSap2018}. However, even for classical hypotheses such an approach is computationally expensive. \\
For testing relevant hypotheses of the form \eqref{eq:testHypOp}, \eqref{eq:testHyp} and \eqref{eq:testHypEv} the problem is substantially more intricate:  In particular, the determination of critical values for the relevant hypotheses in  \eqref{eq:testHypOp} requires the (asymptotic) distribution of the statistic $\widehat{M}^{2}$ at \textit{any} point $M^{2} \geq 0$ of the alternative. As will be demonstrated in this paper, an appropriately normalized version of $\widehat{M}^{2}-M^2$ is in general asymptotically normally distributed but,  compared to the classical hypothesesis $H_0: M^2=0$, the variance of the limiting distribution now depends in a much more complicated way on the spectral density operators  ${\F}^{(\cdot)}_X$ and ${\F}^{(\cdot)}_Y$ and is therefore extremely difficult to estimate. Moreover, for the same reason it is unclear whether a bootstrap method for relevant hypotheses can be developed since one basically has to mimic the (asymptotic) distribution of the test statistic for any  pair of time series  $\{X_t\}_{t\in \znum}$ and $\{Y_t\}_{t\in \znum}$ such that their spectral density operators satisfy the null hypotheses in  \eqref{eq:testHypOp}.

The above approaches become even more problematic, if not infeasible, if either classical or relevant tests of the form  \eqref{eq:testHyp} and   \eqref{eq:testHypEv} for the eigenelements have to be constructed. As will become clear in the subsequent sections, the distributional properties of the corresponding empirical distance measures depend in a highly complicated manner on the dependence structure of the underlying processes (see, for example,  \autoref{thm:hatZbigPi} below). This to the extent that the estimation of nuisance parameters becomes close to impossible and such an approach highly unstable.  To circumvent this problem, we propose tests based on self-normalized or ratio statistics which are constructed via appropriate standardized estimators of the distance measures in \eqref{eq:testHypOp}, \eqref{eq:testHyp} and \eqref{eq:testHypEv}, and 
 have a limiting distribution which does not  depend on the dependence structure of the underlying processes. 
%The concept of  self-normalization has been used by numerous authors in the context of testing classical hypotheses, that is 
%$\Delta = \Delta_{\Pi,k}  = \Delta_{\lambda,k} =0$  \citep[see][among others]{shazha2010,shao2010,zhang2011,shao2015,Zhang2015}. 
The concept of  self-normalization has been used in other settings by numerous authors in the context of testing classical hypotheses %, that is 
%$\Delta = \Delta_{\Pi,k}  = \Delta_{\lambda,k} =0$  
\citep[see][among others]{shazha2010,shao2010,zhang2011,shao2015,Zhang2015}. 
Recently,
a new concept of self-normalization for testing relevant hypotheses regarding the mean and covariance operator of 
functional  time series has also been developed by  \cite{DKV2018}.
However, the development of frequency domain based tests  for relevant hypotheses and hence that allow to infer on the (full) second order dynamics of these processes is far from trivial.
As a further matter, we derive our results under very mild moment conditions which improve upon $L^{p}_m$-approximability assumptions and do not require summability of functional cumulant-mixing conditions. 
For the hypothesis in \eqref{eq:testHypOp}, our current work therefore not only provides a stable alternative to existing work but also relaxes upon underlying moment assumptions. %is also derived under much weaker moment assumptions on the processes.
 Because the construction and the distributional properties of the statistics are highly technical, we start the next section by providing the framework and assumptions and the main ingredient to our method. We then develop the test statistics in full detail for all three hypotheses.

\begin{Remark} \label{remarkh1}
\rm 
We conclude this section with a brief discussion of the choice of the threshold in the formulation of relevant hypotheses.
The classical approach avoids this choice by simply putting $\Delta =0$, but --as pointed out in the introduction-- in many applications   exact equality is unlikely to hold, but the distance between the parameters might be small. 
It is therefore recommended to define the size of the deviation one is really interested in by taking into consideration the scientific background of the testing problem such that the threshold reflects in every particular application the specific scientific context.
As an example we refer to  \cite{fogarty2014}, who use 
relevant hypotheses  to analyze data from a comparison study between two devices for pulmonary function testing. 
These authors actually interchange the null and alternative hypothesis, which is called  {\it bio-equivalence} problem.
We also emphasize that for univariate parameters bio-equivalence problems have been considered by numerous authors, and for specific problems  thresholds have been developed  by regulators (see   \cite{wellek2010} for a recent review). 
 \\
Finally, in cases where the choice of the threshold $\Delta$ is difficult, the methodology presented in the following section can also be used to provide (asymptotic) confidence intervals for the quantities
\begin{equation}
    \label{h40}
\int_{a}^{b}\bigsnorm{\F^{(\omega)}_X-\F^{(\omega)}_Y}_2^2 d\omega ~~,~~~
 \int_{a}^{b}\bigsnorm{\phix{k}-\phiy{k}}_2^2 d\omega ~~~\text{ and } ~~
 \int_{a}^{b}\big \vert\lamx{\omega}{k}-\lamy{\omega}{k}\big\vert^2 d\omega 
\end{equation}
which provide information about the size of the deviation with statistical guarantees.
We refer to Remark \ref{remarkh2}, where we also propose  to test relevant hypotheses for 
several thresholds simultaneously and to identify a maximal threshold $\Delta^*$ for which the 
relevant null hypothesis is not rejected at a controlled type I error.
\end{Remark}

\section{Methodology}\label{sec3}
\def\theequation{3.\arabic{equation}}
\setcounter{equation}{0}

Suppose that we observe a sample of length $T_1$ from component process $\{X_t\}_{t \in \znum}$ and of length $T_2$ from component process $\{Y_t\}_{t \in \znum}$. Central in the construction of the pivotal test statistics and the corresponding asymptotic  level $\alpha$ tests for the hypotheses \eqref{eq:testHypOp}, \eqref{eq:testHyp} and \eqref{eq:testHypEv} are processes of the form 
\begin{align*}
&\Big\{\eta \sqrt{b_1 T_1+b_2 T_2}\Big( \frac{r_1}{\sqrt{b_1 T_1}}\int_a^b\Re\Big(\biginprod{ \mathcal{Z}^{X,\omega}_{T,\eta} }{\mathcal{U}^{(\omega)}_{XY}}_{S_2} \Big)d\omega+ \frac{r_2}{\sqrt{b_2 T_2}}\int_a^b\Re\Big(\biginprod{ \mathcal{Z}^{Y,\omega}_{T,\eta} }{\mathcal{U}^{(\omega)}_{YX}}_{S_2} \Big)d\omega \Big\}_{\eta \in [0,1]}~, 
\tageq \label{eq:limZXsec1}%\\& \Big\{\eta \sqrt{b_1 T_1+b_2 T_2}\Big( \frac{1}{\sqrt{b_2 T_2}}\int_a^b\Re\Big(\biginprod{ \mathcal{Z}^{Y,\omega}_{T,\eta} }{\mathcal{U}^{(\omega)}_{YX}}_{S_2} \Big)d\omega\Big\}_{\eta \in [0,1]}~,%\tageq \label{eq:limZXsec1}
\end{align*}
where $r_1, r_2 \in \rnum$. 
Here, the operators $\mathcal{U}^{(\omega)}_{XY}, \mathcal{U}^{(\omega)}_{YX}   \in S_2(\Hi )$ are such that $\int_{-\pi}^{\pi}\snorm{\mathcal{U}^{(\omega)}_{\cdot, \cdot}}_{2} d\omega <\infty$, and are of varying nature depending on  the specific hypotheses under consideration. %(\eqref{eq:testHypOp}, \eqref{eq:testHyp} or   \eqref{eq:testHypEv})and where 
The $S_{2}(\Hi )$-valued partial
sum processes $\eta \mapsto \mathcal{Z}^{X,\omega }_{T,\eta }$ and
$\eta \mapsto  \mathcal{Z}^{Y,\omega }_{T,\eta }$ are defined by 
\begin{align} 
\mathcal{Z}^{X,\omega}_{T,\eta} & =\frac{\sqrt{b_1}}{\sqrt{T_1}}\sum_{s=1}^{\flo{\eta T_1}} \Big(\sum_{t=1}^{\flo{\eta T_1}}  \tilde{w}^{(\omega)}_{b_1,s,t}(X_s \otimes X_t)  - \F^{(\omega)}_X\Big), \tageq \label{eq:Zx}
\\
\mathcal{Z}^{Y,\omega}_{T,\eta} & =\frac{\sqrt{b_2}}{\sqrt{T_2}}\sum_{s=1}^{\flo{\eta T_2}} \Big(\sum_{t=1}^{\flo{\eta T_2}} \tilde{w}^{(\omega)}_{b_2,s,t}  (Y_s \otimes Y_t) - \F^{(\omega)}_Y\Big), \tageq \label{eq:Zy}
\end{align}
where 
\begin{align}  \label{weight}
\tilde{w}^{(\omega)}_{b_{i},s,t}=(2\pi)^{-1}w(b_{i}(t-s)) e^{\im \omega (t-s)} 
\end{align}
for some window function  $w(\cdot)$ and where $b_{i} = b({T_i})$, $i = 1,2$, are bandwidth parameters which are functions of the corresponding sample lengths $T_i$. Intuitively, the operators \eqref{eq:Zx} and \eqref{eq:Zy} can be interpreted as scaled and centered sequential estimators of the spectral density operators $\F^{(\omega)}_X$ and $\F^{(\omega)}_Y$.
While perhaps not immediately obvious, we shall demonstrate in the following three sections that the distributional properties of empirical versions of the three distance measures in \eqref{eq:testHypOp}, \eqref{eq:testHyp} and  \eqref{eq:testHypEv}, respectively, can --after centering around the population distance measure-- be derived from those of processes of the form \eqref{eq:limZXsec1}. For example, we will show in the next section that $\hat{M}^2-M^2$, with $\hat{M}^2$ as in \eqref{hd22}, can be expressed in terms of such a process that is evaluated at $\eta=1$. 

In order to make this more precise and to derive the distributional properties of the process defined in  \eqref{eq:limZXsec1}, we require the following technical assumptions. Firstly, we specify the dependence structure of $\{X_t\}_{t \in \znum} \in \op^2_\Hi$ and $\{Y_t\}_{t \in \znum} \in \op^2_\Hi$ jointly in terms of the bivariate functional time series $\{\X_t\}_{t \in \znum} = \{ (X_{t},Y_{t})^{\top}\}_{t \in \znum} $. For this, we consider conditions as given in \citet{vD19}, who studied limiting distributions of quadratic form statistics of functional time series under mild moment conditions and provided generalizations of the physical dependence measure \citep{Wu05} to Hilbert-valued processes. 
A functional time series $\{V_t\}_{t \in \znum}$ taking values in a separable Hilbert space $\Hi$ 
is said to have a \textit{physical dependence structure} for some $p >0$ if:
%Let  $\{V_t\}_{t \in \znum}$  be a functional time series that takes its values in a separable Hilbert space $\Hi$. Then, we assume that for some $p>0$:
\begin{enumerate}[label=\textbf{A.\arabic*}]
\item \label{I} The series admits a representation of the form
$V_t = g(\epsilon_t, \epsilon_{t-1},\ldots,) $
where $\{\epsilon_t: t\in\mathbb{Z}\}$ is an i.i.d. sequence of elements in some measurable space $S$ and $g: S^{\infty} \to \Hd$ is a measurable function. 
\item \label{II} The series' dependence structure is of the following nature. Define the measure
\[
\nu_{\hi,p}(V_t)=\|V_t -\E[V_t|\G_{t,\{0\}}] \|_{\hi,p},
\] 
where $\G_{t,\{0\}}$ is the filtration up to time $t$ but with the element at time 0 replaced with an independent copy, i.e., $\G_{t,\{0\}} = \sigma(\epsilon_t, \epsilon_{t-1},\ldots, \epsilon^\prime_0, \epsilon_{-1}, \ldots)$, for some independent copy $\epsilon^\prime_0$ of $\epsilon_0$ and where the conditional expectation is to be understood in the sense of a Bochner integral. The dependence structure of the process satisfies
\[
\sum_{j=0}^{\infty}\nu_{\hi,p}(V_j)< \infty \tageq \label{eq:depstruc2}
\]
for some $p \ge 0$.
\end{enumerate} 
The summability condition in \eqref{eq:depstruc2} is generally a weaker assumption to make than $L^p_m$- approximability as introduced by \cite{HorKok10}, or than summability  of the $p$-th order cumulant tensor for $p>2$ (see also \cite{vD19}).

Throughout this paper, we assume the following conditions on the function-valued time series.
\begin{assumption} \label{as:depstruc} 
The process $\{\X_t \}_{ t\in\mathbb{Z}} = \{ (X_t, Y_t)^\top \}_{ t\in\mathbb{Z}}$ is a centered weakly stationary bivariate functional time series in $\op^p_{\Hi\oplus \Hi}$ of which the component processes satisfy \ref{I}-\ref{II} with $p=4+\varepsilon$, for some small $\varepsilon >0$. 
\end{assumption}
Elementary calculations show  that \autoref{as:depstruc} implies the process $\{\X_t\}_{ t\in\mathbb{Z}}$ satisfies \ref{I}-\ref{II} with $p=4+\varepsilon$.
Observe furthermore that \autoref{as:depstruc} allows for the scenario of independence between the two component processes. It is worth mentioning that the zero-mean assumption simplifies notation but, in practice, the data can be centered without affecting the results of this paper. Processes which satisfy  conditions \ref{I}-\ref{II} for some $p\ge2$ have a well-defined spectral density operator. In particular, for processes that satisfy \autoref{as:depstruc}, the second order structure arises as elements of $(\Hi^{\oplus_2})^{\otimes_2}$ and the full second order dynamics are therefore described via the vector of spectral density operators given by 
\[
\F_{\X}^{(\omega)} = \begin{pmatrix} \F_X^{(\omega)},  \F_{XY}^{(\omega)}, \F_{YX}^{(\omega)},  \F_{Y}^{(\omega)}  \end{pmatrix}^\top \quad \omega \in [-\pi,\pi],
\]
where the operators $\F_{XY}^{(\omega)}=\frac{1}{2\pi}\sum_{h \in \znum} \E(X_h \otimes Y_0) e^{-\im \omega h}$ and $\F_{YX}^{(\omega)}=\frac{1}{2\pi}\sum_{h \in \znum} \E(Y_h \otimes X_0) e^{-\im \omega h}$ define the cross-spectral density operators. 
It can be shown that the convergence of the series is with respect to $\snorm{\cdot}_1$, uniformly in $\omega \in \rnum$. 

As a starting point for our test statistics, consider the following estimators of $\F_X^{(\omega)}$ and $\F_Y^{(\omega)}$%and in particular for the expression of $\mathcal{Z}^{X,\omega}_{T,\eta}$,   
\[\hat{\F}_{X}^{\omega}= \frac{1}{ T_1}\sum_{s,t=1}^{T_1} \tilde{w}^{(\omega)}_{b_1,s,t} \big( X_s \otimes X_t \big) ;  \tageq \label{eq:relFQ}\]
\[\hat{\F}_{Y}^{\omega}= \frac{1}{ T_2}\sum_{s,t=1}^{T_2}  \tilde{w}^{(\omega)}_{b_2,s,t}\big( Y_s \otimes Y_t \big) , \tageq \label{eq:relFQy}\]
%s  estimators of $\F_X^{(\omega)}$ and $\F_Y^{(\omega)}$, respectively,
 where the weights $\tilde{w}^{(\omega)}_{b_{T_i},s,t}$ are given by \eqref{weight}. For the construction of pivotal test statistics, we require sequential versions of the lag window estimators in \eqref{eq:relFQ} and \eqref{eq:relFQy}, which are respectively given by
\[
\hat{\F}_X^{(\omega)}(\eta)=\frac{1}{\flo{\eta T_1}}\sum_{s=1}^{\flo{\eta T_1}}\Big(\sum_{t=1}^{\flo{\eta T_1}}  \tilde{w}^{(\omega)}_{b_1,s,t} (X_s \otimes X_t)\Big), \tageq \label{eq:seqFQ}
\]
and 
\[
\hat{\F}_Y^{(\omega)}(\eta)=\frac{1}{\flo{\eta T_2}}\sum_{s=1}^{\flo{\eta T_2}}\Big(\sum_{t=1}^{\flo{\eta T_2}} \tilde{w}^{(\omega)}_{b_2,s,t} (Y_s \otimes Y_t)\Big), \tageq \label{eq:seqFQy}
\]
where $\eta \in [0,1]$. We denote the eigenvalues and eigenprojectors of \eqref{eq:seqFQ} and \eqref{eq:seqFQy},  by $\big\{ \lamxT{k}(\eta)\big\}_{k \ge 1}$, $\big\{  \phixT{k}(\eta) \big\}_{k \ge 1}$ and $\big\{ \lamyT{k}(\eta)\big\}_{k \ge 1}$, $\big\{  \phiyT{k}(\eta) \big\}_{k \ge 1}$, respectively. Empirical versions of the distance measures in \eqref{eq:testHypOp}, \eqref{eq:testHyp} and \eqref{eq:testHypEv} can then be expressed in terms of the sequential estimators evaluated at $\eta =1$, i.e.,
\[
\begin{array}{l c r}\int_a^b \snorm{\,\hat{\F}^{(\omega)}_X(1)-\hat{\F}^{(\omega)}_Y(1)}^2_2 d\omega, & \int_a^b \snorm{\,\phixT{k}(1)-\phiyT{k}(1)}^2_2 d\omega,& \int_a^b \big(\lamxT{k}(1)-\lamyT{k}(1)\big)^2 d\omega.
\end{array}\]
We assume the following mild requirements on the lag window function $w$.
\begin{assumption}\label{as:Weights}
Let $w$ be an even, bounded function on $\mathbb{R}$ with $\lim_{x \to 0} w(x) =1$ that is continuous except at a finite number of points. Suppose that $\lim_{b \to 0}b \sum_{h \in \znum} w^2(b h) = \kappa$ where $\kappa=\int^{\infty}_{-\infty} w^2(x) dx < \infty$ such that $\sup_{0 \le b \le 1} b \sum_{|h| \ge L/b} w^2(b h) \to 0$ as $L \to \infty$. Furthermore, we assume $w(x) - 1=O(x)$ as $x \to 0$.
\end{assumption}
Under these conditions, the following consistency result on the lag window estimators can be obtained.

\begin{prop}\label{prop:consF}
Suppose $\{V_t\}_{ t\in\mathbb{Z}}$ is a centered weakly stationary process in $\mathcal{L}^p_\Hi$ that satisfies conditions \ref{I}-\ref{II} with $p=2q, q \ge 1$. Furthermore, let \autoref{as:Weights} be satisfied and assume
\[\sum_{h \in \znum} h^\ell \nu_{\hi,2}(V_h) < \infty \tageq \label{eq:sumnul}\]
for some $\ell \ge 1$. Let $\hat{\F}_{V}^{\omega}= \frac{1}{ T}\sum_{s,t=1}^{T} \tilde{w}^{(\omega)}_{b_T,s,t}\big( V_s \otimes V_t \big)  $. Then, for $q \ge 2$, 
\begin{align*}
\norm{\hat{\F}_{V}^{\omega} -\E\hat{\F}_{V}^{\omega}}^q_{S_2,q} = O\big( {(b_T T)^{-q/2}}\big) ~~\mbox{and} ~~
 \norm{\E\hat{\F}_{V}^{\omega} -{\F}_V^{\omega} }^q_{S_2,q} =O\big(b^{\ell q}_T+T^{-q}\big) 
\end{align*}
uniformly in $\omega \in [-\pi,\pi]$, where $\F_V^\omega$ is the spectral density operator of process $\{V_t\}$.  In particular $ \norm{\hat{\F}_{V}^{\omega}  -{\F}_V^{\omega} }^q_{S_2,q} =O\big( {(b_T T)^{-q/2}}\big) +O\big(b^{\ell q}_T\big) $ uniformly in $\omega \in [-\pi,\pi]$.
\end{prop}

See section 4-5 of \cite{vD19} for a comparison of this estimator and the underlying assumptions with those considered in  \cite{PanTav2013a}, who derived a consistent estimator of a smoothed periodogram operator under cumulant mixing conditions. Note that the value of $\ell$ in \eqref{eq:sumnul} only affects the order of the bias, which decreases faster for processes with shorter memory. It is also worth mentioning that the estimator remains consistent for $\ell=0$. To ensure consistency in $q$-th mean, Proposition \autoref{prop:consF} gives rise to the following conditions on the rate of the bandwidth. 
{\begin{assumption}\label{as:bwrates}
Given $\sum_{h \in \znum} h^\ell \nu_{\hi,2}(V_h) < \infty$ for some $\ell \ge 1$, we require that 
%For $i=\{1,2\}$, $b_{T_i} \to 0$ such that $b_{T_i} T_i \to \infty$ as $T_i, \to \infty$ 
$b_T \to 0$ such that $b_T T/\log(T) \to \infty$ and such that $b_T^{1+2\ell} T \to 0$ as $T \to \infty$.
\end{assumption}}
Observe that the last part of the assumption simply means that larger bandwidths are allowed for processes with a "smoother'' spectral distribution. % (i.e., the higher the order of differentiability of the map  $\omega \mapsto \F^{(\omega)}$). 

Under Assumptions \ref{as:depstruc}--\ref{as:bwrates}, the sequential estimators  \eqref{eq:seqFQ} and \eqref{eq:seqFQy} provide us with  consistent estimators for $\F_X^{(\omega)}$ and $\F_Y^{(\omega)}$. Furthermore, the elements of their respective eigensystems $\big\{ \lamxT{k}(\eta),\phixT{k}(\eta) \big\}_{k \ge 1}$, $\big\{ \lamyT{k}(\eta),\phiyT{k}(\eta) \big\}_{k \ge 1}$, can then be shown to be  consistent estimators for their population counterparts for each $\eta \in [0,1], \omega \in [-\pi,\pi]$ (see \autoref{lem:diffboundEig}). Additionally, we obtain under these conditions a useful bound on the maximum of partial sum of the estimators of the spectral density operators (see \autoref{lem:boundseqSDO}).

The last assumption concerns the `balance' of the convergence rates.
\begin{assumption}\label{as:ratiorates}
%For $i=\{1,2\}$, $b_{T_i} \to 0$ such that $b_{T_i} T_i \to \infty$ as $T_i, \to \infty$ and
Let $b_{i}$, $i \in \{1,2\}$ satisfy \autoref{as:bwrates} for some $\ell \ge 1$. If the component processes $\{X_t\}$ and $\{Y_t\}$ of $\{\X_t\}$ are independent, we assume there exists a constant $\theta \in (0,1)$ such that
\[\lim_{T_1, T_2 \to \infty}\frac{b_1 {T_1} }{b_1 {T_1} +b_2 {T_2} } = \theta \in (0,1).\]
If the processes are dependent, we assume ${T_1}={T_2}$ and $b_1\sim b_2$. 
\end{assumption}

We can now state the main technical result of this paper which is crucial for the construction of pivotal tests for the hypotheses
\eqref{eq:testHypOp}, \eqref{eq:testHyp}
and \eqref{eq:testHypEv}  of no relevant difference in the spectral density operators, eigenprojectors or eigenvalues, respectively (see Sections \ref{sec31} - \ref{sec33}
for details).

\begin{thm}\label{thm:BuildingBlock}
Suppose  Assumptions \ref{as:depstruc}-\ref{as:ratiorates} are satisfied. Then
\begin{align*}
& \Big\{\eta \sqrt{b_1 T_1+b_2 T_2}\Big( \frac{1}{\sqrt{b_1 T_1}}\int_a^b\Re\Big(\biginprod{ \mathcal{Z}^{X,\omega}_{T,\eta} }{\mathcal{U}^{(\omega)}_{XY}}_{S_2} \Big)d\omega
- \frac{1}{\sqrt{b_2 T_2}}\int_a^b\Re\Big(\biginprod{ Z^{Y, \omega}_{T,\eta} }{\mathcal{U}^{(\omega)}_{YX}}_{S_2} \Big)d\omega\Big)
 \Big\}_{\eta \in [0,1]}
\\&\\&
~~~~~~~~~~~~~~~~~~~~~~~~~~~~~~~~~~~~~~~~~~~~~~~~~~~~~~~~~~
\,{\rightsquigarrow}\,  \big \{\tau_{XY} \eta \, \mathbb{B} (\eta) \big  \}_{\eta \in [0,1]}, \quad \text{ as } T_1, T_2 \to \infty, 
\end{align*}
where $\tau_{XY}$ is a constant, $\mathbb{B}$ is a Brownian motion
 and where $\mathcal{Z}^{X,\omega}_{T,\eta}$
 and $\mathcal{Z}^{Y,\omega}_{T,\eta}$ 
 are defined in \eqref{eq:Zx} and \eqref{eq:Zy}, respectively.
\end{thm}
 The proof of this statement 
 %is far from trivial and 
 relies on approximating martingale theory and  is postponed to \autoref{sec:proofBB}. 
We believe that the following  results  for  the dependent case can also be obtained for unequal sample sizes
satisfying 
$\lim_{T_1, T_2 \to \infty} {b_1 {T_1} }/{(b_1 {T_1} +b_2 {T_2})} = \theta \in (0,1).$
However, this considerably complicates the line of proof and is therefore not considered. Note that the scaling factor $\tau_{XY}$ depends  in a rather  complicated way 
 on the properties of $\{\F_X^{(\omega)}\}_{\omega \in [a,  b]}$ and $\{\F_Y^{(\omega)}\}_{\omega \in [a, b]}$  (see \autoref{sec:proofBB} for details)
 and is therefore very difficult to estimate. In the next sections, we  develop tests for the hypotheses of relevant differences
 between  the spectral density operators  $\{\F_X^{(\omega)}\}_{\omega \in [a,  b]}$ and $\{\F_Y^{(\omega)}\}_{\omega \in [a, b]}$ and the associated eigenelements, which do not require estimation of $\tau_{XY}$  and are in this sense pivotal. 
 
\subsection{No relevant difference in the spectral density operators: hypothesis \eqref{eq:testHypOp}}
\label{sec31}

We start with the construction of a pivotal test for  hypothesis \eqref{eq:testHypOp}
of no relevant difference between  the spectral density operators. Proofs of the statements can be found in Section  \ref{sec:proofshyps} of the supplement.
For fixed $\eta \in [0,1]$ and fixed $\omega$, denote the (pointwise) population distances and empirical distances of the spectral density operators by
\[
{M}_{\F}(\omega)=\F^{(\omega)}_X-\F^{(\omega)}_Y\text{~~ and ~ }\widehat{M}_{\hat\F}(\eta,\omega)  =\eta \big(\hat{\F}^{(\omega)}_X(\eta)-\hat{\F}^{(\omega)}_Y(\eta)\big),
\]
and observe that  under \autoref{as:depstruc}, these are both well-defined elements of $S_1(\mathcal{H})$ for any $\eta \in [0,1]$. The next step is to define a process which quantifies the difference between the empirical and population measures over a given frequency band, i.e., 
\begin{align*}
\hat{\mathcal{Z}}^{[a,b]}_{\F,T_1,T_2} (\eta)&
=\int_{a}^{b} \bigsnorm{\eta \big(\hat{\F}^{(\omega)}_X(\eta)-\hat{\F}^{(\omega)}_Y(\eta)\big)}^2_2 -\eta^2\bigsnorm{\F^{(\omega)}_X-\F^{(\omega)}_Y}_2^2 d\omega. % {\rightsquigarrow} \tau \{\eta  \mathbb{B}(\eta)\}_{\eta \in [0,1]}
\tageq \label{eq:hatZF}
\end{align*}
Elementary calculations show that we can write \eqref{eq:hatZF} as
\begin{align*}
\hat{\mathcal{Z}}^{[a,b]}_{\F,T_1,T_2} (\eta) &=\int_{a}^{b} \bigsnorm{\widehat{M}_{\hat\F}(\eta,\omega) }^2_2 -\eta^2\bigsnorm{{M}_{\F}(\omega)}_2^2 d\omega 
\\&=\int_{a}^{b}\Big\{\bigsnorm{\widehat{M}_{\hat\F}(\eta,\omega)-\eta {M}_{\F}(\omega) }^2_2 +\biginprod{\widehat{M}_{\hat\F}(\eta,\omega)  -\eta {M}_{\F}(\omega)}{\eta {M}_{\F}(\omega)}_{S_2}\\&
\phantom{\int_{a}^{b}}+ \overline{\biginprod{\widehat{M}_{\hat\F} (\eta,\omega)-\eta {M}_{\F}(\omega)}{\eta {M}_{\F}(\omega)}_{S_2}}\Big\}d\omega. \tageq \label{eq:hatZFapprox}
\end{align*}
Moreover, notice that
\begin{align*}
\widehat{M}_{\hat\F}(\eta,\omega)  -\eta {M}_{\F}(\omega) 
=\eta\big(\hat{\F}_X^{(\omega)}(\eta)-\F_X^{(\omega)}\big) -\eta\big(\hat{\F}^{(\omega)}_Y(\eta)-{\F}^{(\omega)}_Y\big)~.\tageq \label{eq:MdiffF}
\end{align*}
The following result, which requires to control the maximum of partial sums of \eqref{eq:seqFQ} and \eqref{eq:seqFQy}, shows that the first term of \eqref{eq:hatZFapprox} is of smaller order than the two other terms.

\begin{lemma} \label{lem:Fbound}
Suppose Assumptions \ref{as:depstruc}-\ref{as:ratiorates} are satisfied. Then 
\[
 \sup_{\eta \in [0,1]}\int_{a}^{b}\bigsnorm{\widehat{M}_{\hat\F}(\eta,\omega)-\eta {M}_{\F}(\omega) }^2_2 d\omega =o_P\Big(\frac{1}{\sqrt{b_{1} T_1+ b_{2}T_2}}\Big)~. \tageq \label{eq:lemFbound1}
\]
\end{lemma}
The next statement in turn then shows that we can approximate  the process in \eqref{eq:hatZFapprox} as a linear combination of functionals of processes of the form in \eqref{eq:Zx} and \eqref{eq:Zy}. 
\begin{thm}\label{thm:hatZFhyp1}
Suppose Assumptions \ref{as:depstruc}-\ref{as:ratiorates} are satisfied. Then 
\begin{align*}
 \sqrt{b_1 T_1+b_2 T_2}\hat{\mathcal{Z}}^{[a,b]}_{\F,T_1,T_2} (\eta)& =  \sqrt{b_1 T_1+b_2 T_2} \int_{a}^{b}
\frac{2}{\sqrt{b_1 T_1} }
 \Re\biginprod{\mathcal{Z}^{X,\omega}_{T,\eta}}{\eta {M}_{\F}(\omega)}_{S_2}  d\omega 
 \\&-\sqrt{b_1 T_1+b_2 T_2} \int_{a}^{b}\frac{2}{{\sqrt{b_2 T_2} }}\Re \biginprod{\mathcal{Z}^{Y,\omega}_{T,\eta}}{\eta {M}_{\F}(\omega)}_{S_2}d\omega +o_P(1).
\end{align*}
\end{thm}
We can now use \autoref{thm:BuildingBlock} with $ \mathcal{U}^{(\omega)}_{XY} = \mathcal{U}^{(\omega)}_{YX} = 2{M}_{\F}(\omega)$ and \autoref{thm:hatZFhyp1} to find the limiting distribution of the process in \eqref{eq:hatZF}, that is 
\begin{align}
\big\{\sqrt{b_{1} T_1+ b_{2}T_2}\big( \hat{\mathcal{Z}}^{[a,b]}_{\F,T_1,T_2} (\eta)\big)\big \}_{\eta \in [0,1]} {\rightsquigarrow} \tau_{\F} \{\eta  \mathbb{B}(\eta)\}_{\eta \in [0,1]}, \quad \text{ as } T_1, T_2 \to \infty  \tageq \label{eq:ZBF}~,
\end{align}
where $ \tau_{\F}$ is a    constant. To make the test independent of $\tau_{\F}$, we consider the following self-normalizing approach, which is similar in nature to \cite{DKV2018}. To be precise, define  the statistic 
 \[
\hat{V}^{[a,b]}_{\F,T_1, T_2}= \Big(\int^{1}_0 \Big( \int_a^b\bigsnorm{\widehat{M}_{\hat\F}(\eta,\omega)}^2_2 - \eta^2 \bigsnorm{\widehat{M}_{\hat\F}(1,\omega)}^2_2 d\omega \Big)^2 \nu(d\eta) \Big)^{1/2} ~,
\tageq \label{eq:hatVF}
\]
 where $\nu$ is a probability measure on the interval $(0,1)$. Then it is easy to see that 
 \[
 \hat{V}^{[a,b]}_{\F,T_1, T_2}=
\Big(\int^{1}_0 \big(  \hat{\mathcal{Z}}^{[a,b]}_{\F,T_1,T_2} (\eta)- \eta^2 \hat{\mathcal{Z}}^{[a,b]}_{\F,T_1,T_2} (1)  \big)^2 \nu(d\eta) \Big)^{1/2}
\]
 and the continuous mapping theorem  and the weak convergence  \eqref{eq:ZBF} imply 
 \begin{align} \label{eq:distZVF}
\sqrt{b_{1} T_1+ b_{2}T_2}& \Big( \hat{\mathcal{Z}}^{[a,b]}_{\F,T_1, T_2}(1), \hat{V}^{[a,b]}_{\F,T_1, T_2}\Big)\, \underset{T_1,T_2 \to \infty}{\Rightarrow} \,
\Bigg( \tau_{\F}\mathbb{B}(1), \Big(\int^{1}_0\tau_{\F}^2 \eta^2  \Big(  \mathbb{B}(\eta)- \eta \mathbb{B}(1)\Big)^2  \nu(d\eta)\Big)^{1/2} \Bigg)~.
\end{align}
Consequently, a further application of the continuous mapping theorem yields
 \begin{align*}
 \frac{ \hat{\mathcal{Z}}^{[a,b]}_{\F,T_1, T_2}(1)}{ \hat{V}^{[a,b]}_{\F,T_1, T_2}}
 \,\underset{T_1,T_2 \to \infty}{\Rightarrow}\, 
\mathbb{D} = \frac{\mathbb{B}(1)}{\big(\int^{1}_0 \eta^2  \big(  \mathbb{B}(\eta)- \eta \mathbb{B}(1)\big)^2  \nu(d\eta)\big)^{1/2} } \tageq \label{eq:matbbD}
\end{align*} 
whenever $\tau_{\F} \neq 0$. From this, we can obtain a pivotal test statistic for the hypothesis \eqref{eq:testHypOp} of no relevant  difference between the 
spectral density operators given by 
\[
\widehat{\mathbb{D}}^{[a,b]}_{T_1,T_2} = \frac{\int_{a}^{b} \snorm{\widehat{M}_{\hat\F}(1,\omega)}^2_2d\omega-\Delta}{\hat{V}^{[a,b]}_{\F,T_1, T_2}}, \tageq\label{eq:hatDF}
\]
and a  natural decision rule is then to reject the null hypothesis in \eqref{eq:testHypOp} whenever 
 \begin{align}
 \label{testOp}
%\hat{\mathbb{D}}_{T_1,T_2} > \Delta+ \hat{V}^{[a,b]}_{\F,T_1, T_2}\cdot q_{1-\alpha}(\mathbb{D}) 
\hat{\mathbb{D}}^{[a,b]}_{T_1,T_2} > q_{1-\alpha}(\mathbb{D})  \quad \Leftrightarrow \quad  \int_{a}^{b} \snorm{\widehat{M}_{\hat\F}(1,\omega)}^2_2d\omega> \Delta+ \hat{V}^{[a,b]}_{\F,T_1, T_2}\cdot q_{1-\alpha}(\mathbb{D}) 
 \end{align}
where $q_{1-\alpha}(\mathbb{D})$ denotes the $(1-\alpha)$-th quantile of the distribution of the  random variable $
\mathbb{D}$ defined in \eqref{eq:matbbD}.
 Consequently, the test no longer depends on the unknown nuisance parameter but only  on the measure $\nu$ used in the definition
 of the  self-normalizing factor $\hat{V}^{[a,b]}_{\F,T_1, T_2}$, which can be chosen by the 
 statistician and is therefore known.  Observe further that the quantiles $q_{1-\alpha}(\mathbb{D})$ are straightforward to simulate. The next result now shows that the test in \eqref{testOp} provides a consistent and asymptotic level $\alpha$ test. 
\begin{thm} \label{thm:leveltestZF}
Suppose Assumptions \ref{as:depstruc}-\ref{as:ratiorates} are satisfied. Then the decision rule  \eqref{testOp} provides an asymptotic level $\alpha$ test for the hypothesis 
 \eqref{eq:testHypOp} of no relevant difference between the spectral density operators $\F^{(\cdot)}_X$ and $\F^{(\omega)}_Y$, i.e., 
\[
\lim_{T_1,T_2 \to \infty} \mathbb{P}\Big(\hat{\mathbb{D}}^{[a,b]}_{T_1,T_2} >q_{1-\alpha}(\mathbb{D}) \Big) =
\begin{cases}
0 & \mbox{if } \Delta > \int_a^b \snorm{M^{(\omega)}_\F}^2 _2 d\omega; \\ 
\alpha & \mbox{if } \Delta = \int_a^b \snorm{M^{(\omega)}_\F}^2 _2 d\omega { \mbox{  and } \tau_{\F} \neq 0 };\\ 
1 & \mbox{if } \Delta < \int_a^b \snorm{M^{(\omega)}_\F}^2 _2 d\omega. \
\end{cases}
\]
\end{thm}
\begin{proof}
 Suppose first that $\int_a^b \snorm{M^{(\omega)}_\F}^2 _2 d\omega =0$. Then \eqref{eq:hatZF} becomes $\hat{Z}^{[a,b]}_{\F,T_1, T_2}(1) = \int_{a}^{b} \snorm{\widehat{M}_{\hat\F}(1,\omega)}^2_2d\omega$. By \eqref{eq:distZVF}, we have $\hat{Z}^{[a,b]}_{\F,T_1, T_2}(1) =o_P(1)$ and $ \hat{V}^{[a,b]}_{\F,T_1, T_2}= o_P(1)$  as $ T_1,T_2 \to \infty$. Consequently, we obtain $\lim_{T_1, T_2 \to \infty}\mathbb{P}\Big(\hat{\mathbb{D}}^{[a,b]}_{T_1,T_2} >q_{1-\alpha}(\mathbb{D}) \Big) =0$. 
Next, suppose $ \int_a^b \snorm{M^{(\omega)}_\F}^2 _2 d\omega >0$. In this case, we can write
\begin{align} \label{hol25}
& \mathbb{P}\Big(\hat{\mathbb{D}}^{[a,b]}_{T_1,T_2} >q_{1-\alpha}(\mathbb{D}) \Big) \\&
 = \mathbb{P}\Big({\int_{a}^{b} \snorm{\widehat{M}_{\hat\F}(1,\omega)}^2_2d\omega -\int_{a}^{b} \snorm{M_\F(\omega)}^2_2 d\omega}> {\Delta -\int_{a}^{b} \snorm{M_\F(\omega)}^2_2 d\omega} +q_{1-\alpha}(\mathbb{D})  {\hat{V}^{[a,b]}_{\F,T_1, T_2}}\Big). \nonumber 
\end{align}
From \eqref{eq:distZVF} it follows that 
 \begin{align*}
  \hat{\mathcal{Z}}^{[a,b]}_{\F,T_1, T_2}(1) &=  {\int_{a}^{b} \snorm{\widehat{M}_{\hat\F}(1,\omega)}^2_2d\omega -\int_{a}^{b} \snorm{M_\F(\omega)}^2_2 d\omega}=  o_{P}(1) ~,~~
\hat{V}^{[a,b]}_{\F,T_1, T_2}  = o_{P}(1),
\end{align*}
and consequently the assertion in   the cases  $ \Delta > \int_a^b \snorm{M^{(\omega)}_\F}^2 _2 d\omega$ and
  $ \Delta < \int_a^b \snorm{M^{(\omega)}_\F}^2 _2 d\omega$ follows easily.
Finally, if $ \tau_{\F} \neq 0$ we have  from \eqref{eq:matbbD} that 
$$
\frac{\int_{a}^{b} \snorm{\widehat{M}_{\hat\F}(1,\omega)}^2_2d\omega -\int_{a}^{b} \snorm{M_\F(\omega)}^2_2 d\omega}{\hat{V}^{[a,b]}_{\F,T_1, T_2}}
=  \frac{ \hat{\mathcal{Z}}^{[a,b]}_{\F,T_1, T_2}(1)}{ \hat{V}^{[a,b]}_{\F,T_1, T_2}}
\underset{T_1,T_2 \to \infty}{\Rightarrow} 
\mathbb{D},
$$
and 
%using again that $ \hat{V}^{[a,b]}_{\F,T_1, T_2}= o_P(1)$
 we obtain  the remaining case   $ \Delta = \int_a^b \snorm{M^{(\omega)}_\F}^2 _2 d\omega$  from \eqref{hol25}. 
\end{proof}

\begin{Remark}[sensitivity with respect to the measure $\nu$] \label{remarkh0}
{\rm  Note that the test \eqref{testOp} depends on the specification of the measure $\nu$, which has to be chosen in advance. We give a heuristic argument that the test is not very sensitive with respect to the choice of this measure. To see this, consider 
\begin{equation}
\label{hol51}
\mathbb{V} =\Big(\int^{1}_0 \eta^2  \big(  \mathbb{B}(\eta)- \eta \mathbb{B}(1)\big)^2  \nu(d\eta)\Big)^{1/2}  ,
\end{equation}
and $M^2 = \int_{a}^{b} \snorm{M_\F(\omega)}^2_2 d\omega $, 
and note  that it follows from \eqref{eq:distZVF} and \eqref{hol25}
\begin{align}  \nonumber 
\mathbb{P}\Big(\hat{\mathbb{D}}^{[a,b]}_{T_1,T_2} >q_{1-\alpha}(\mathbb{D}) \Big) 
& \approx 
 \mathbb{P}\big( \mathbb{B}(1) > \sqrt{b_{1} T_1+ b_{2}T_2} \cdot
  (\Delta - M^2)/\tau_{\F}  + \mathbb{V }\cdot q_{1-\alpha}(\mathbb{D})   \big) \\
  & =
 \mathbb{P}\big( \mathbb{B}(1) > \sqrt{b_{1} T_1+ b_{2}T_2} \cdot
  (\Delta - M^2) /\tau_{\F}  + \mathbb{V }\cdot q_{1-\alpha}( \mathbb{B}(1)/\mathbb{V})   \big)~.
  \label{h40a}
\end{align}
Observe that in the last expression  only the quantity $\mathbb{V }\cdot q_{1-\alpha}( \mathbb{B}(1)/\mathbb{V})$ depends 
on the measure $\nu$, which enters in the definition of the random variable $\mathbb{V}$.
However,  for fixed $v>0 $ we have $ {v }\cdot q_{1-\alpha}( \mathbb{B}(1)/v) = q_{1-\alpha}( \mathbb{B}(1) )$,
which gives a  heuristic explanation  why 
%. Moreover, it is easy to see that the random variables $\mathbb{B}(1)$ and $\mathbb{V}$ are independent. Therefore, integrating with respect
%to the distribution of $\mathbb{V}$ shows that
the 
probability in  \eqref{h40a} is not very sensitive with respect to the choice of  the measure $\nu$ (which was also observed empirically).
%equal to
%$ \mathbb{P}\big( \mathbb{B}(1) > \sqrt{b_{1} T_1+ b_{2}T_2} \cdot
%  (\Delta - M^2) /\tau_{\F}  +  q_{1-\alpha}( \mathbb{B}(1))   \big)$
%and does not depend  on the  measure $\nu$. 
Note that the same argument applies to the tests proposed in following Section \ref{sec32} 
and \ref{sec33}. 
}
\end{Remark}

\subsection{No relevant difference in the eigenprojectors: hypothesis \eqref{eq:testHyp}} 
\label{sec32}

In this section, we construct a pivotal test  for hypothesis \eqref{eq:testHyp}  of 
no relevant difference between the eigenprojectors $\phix{k}$ and $\phiy{k}$ of the functional time series  
$\{X_t\}_{t \in \znum}$ and  $\{Y_t\}_{t \in \znum}$.
%\begin{align}
%H_0: \int_{a}^{b}\bigsnorm{\phix{k}-\phiy{k}}_2^2 d\omega \le \Delta_k  \quad \text{ versus } \quad H_A: \int_{a}^{b}\bigsnorm{\phix{k}-\phiy{k}}_2^2 d\omega >  \Delta_k,
%\end{align}
The development is of a more intricate nature than for the spectral density operators, which we will elaborate upon.  Proofs of the statements provided in this section are relegated to \autoref{sec:mainstat}. To ease notation, denote the (pointwise) population distances and empirical distances of the $k$-th eigenprojectors at frequency $\omega$ by
\[
{M}_{\Pi, k}(\omega)=\phix{k}-\phiy{k} \text{ ~~ and ~~ }\widehat{M}_{\Pi,k}(\eta,\omega)  =\eta \big(\phixT{k}(\eta,\omega) -\phiyT{k}(\eta)\big), \quad \eta \in [0,1].\]
%where we recall that  $\Big\{  \phixT{k}(\eta) \Big\}_{k \ge 1}$ and $\Big\{  \phiyT{k}(\eta) \Big\}_{k \ge 1}$ are the eigenprojectors 
As already briefly mentioned in \autoref{sec2}, we construct a test based on the eigenprojectors rather than on the eigenfunctions because the latter are only defined up to some multiplicative factor $c$ on the unit circle. To understand the problem, suppose for simplicity that we would like to test for relevant differences in the $k$-th eigenspace of $\F^{(\omega)}_X$ and $\F^{(\omega)}_Y$. From the estimators in \eqref{eq:seqFQ} and \eqref{eq:seqFQy} we can obtain $c_1 \hat{\phi}_{X,k}^{(\omega)}$ and $c_2 \hat{\phi}_{Y,k}^{(\omega)}$ for some unknown $c_1, c_2 \in \cnum$ with $|c_1|^2 =|c_2|^2=1$. The empirical eigenfunctions might therefore not be comparable due the unknown rotation in different directions. Moreover, a consequence of this rotation is that a bound on the differences in norm between the population and empirical eigenfunctions does not follow from those of the corresponding operators. This is in contrast with eigenfunctions that strictly belong to the real-valued subspace of $\Hi$, of which only the sign is unknown for the empirical counterparts. We can however construct a test using the eigenprojectors because these are rotationally invariant since $\phixT{k} 
=c_1 \overline{c_1} \hat{\phi}_{X,k}^{(\omega)}\otimes  \hat{\phi}_{X,k}^{(\omega)}$. As a consequence, $\phixT{k}$ and $\phiyT{k}$ are directly comparable and a bound on the differences in norm between the population and empirical eigenprojectors can be derived. 
The derivation of the  theoretical properties of a test based upon eigenprojectors is however more involved than one based upon eigenfunctions {\citep[see also][who developed self-normalized tests for relevant changes  in the eigenfunctions of the covariance operator]{auedettrice2019}. }
% The augmentation in dimension, which results from considering a test based on eigenprojectors rather than on eigenfunctions, leads however to a more involved theoretical derivation of the test statistic. 

We start by defining the process
\begin{equation}
\hat{\mathcal{Z}}^{[a,b],(k)}_{\Pi,T_1,T_2} (\eta)=\int_{a}^{b} \bigsnorm{\widehat{M}_{\Pi,k}(\eta,\omega) }^2_2 -\eta^2\bigsnorm{{M}_{\Pi, k}(\omega) }_2^2 d\omega \tageq \label{eq:hatZbig} 
\end{equation}
%and again rewrite this as
%\begin{align*}
%\hat{\mathcal{Z}}^{[a,b],(k)}_{\Pi,T_1,T_2} (\eta)&=\int_{a}^{b} \bigsnorm{\widehat{M}_{\Pi,k}(\eta,\omega) }^2_2 -\eta^2\bigsnorm{{M}_{\Pi, k}(\omega)}_2^2 d\omega 
% = \int_{a}^{b}\bigsnorm{\widehat{M}_{\Pi,k}(\eta,\omega)  -\eta {M}_{\Pi, k}(\omega)+\eta {M}_{\Pi, k}(\omega) }^2_2 -\eta^2\bigsnorm{{M}_{\Pi, k}(\omega)}_2^2 d\omega 
%\\&= \int_{a}^{b}\bigsnorm{\widehat{M}_{\Pi,k}(\eta,\omega)  -\eta M_{\Pi,k}(\omega)}^2_2 +\biginprod{\widehat{M}_{\Pi,k}(\eta,\omega)  -\eta M_{\Pi,k}(\omega)}{\eta M_{\Pi,k}(\omega)}_{S_2}\\&
%\phantom{\int_{a}^{b}}+ \overline{\biginprod{\widehat{M}_{\Pi,k}(\eta,\omega)  -\eta M_{\Pi,k}(\omega)}{\eta M_{\Pi,k}(\omega)}}_{S_2}d\omega. 
%\\&= \sqrt{b_1 T_1+b_2 T_2} +\biginprod{\widetilde{M}_{\Pi,k}(\eta,\omega)}{\eta M_{\Pi,k}(\omega)}+ \overline{\biginprod{\widetilde{M}_{\Pi,k}(\eta,\omega) }{\eta M_{\Pi,k}(\omega)}}d\omega +o_P(1) 
%\tageq \label{eq:hatZbig}
%\end{align*}
and observe in this case that
\begin{align*}
\widehat{M}_{\Pi,k}(\eta,\omega)  - \eta M_{\Pi,k}(\omega) = \eta \Big(\phixT{k}(\eta)-\phix{k}\Big) - \eta \Big(\phiyT{k}(\eta)-\phiy{k}\Big).  
\end{align*}
Unlike the terms in \eqref{eq:MdiffF}, the properties of the two terms of the right-hand side are not obvious to  disentangle.  As a first step, denote $ \F^{\omega}_{\widetilde{\otimes}}= \F^{\omega}\widetilde{\otimes} \F^{\omega}$ and observe that this operator has decomposition
\[
\F^{\omega}_{\widetilde{\otimes}} = \sum_{j, \jpr} \lambda^{(\omega)}_{j}\lambda^{(\omega)}_{\jpr} \Big(\Pi^{(\omega)}_{j}{\widetilde{\otimes}} \Pi^{(\omega)}_{\jpr}\Big),
\]
from which it follows that $\F^{\omega}_{\widetilde{\otimes}} (\Pi^{(\omega)}_{k,l}) =  \lambda^{(\omega)}_{k} \lambda^{(\omega)}_{l} \Pi^{(\omega)}_{k,l}$. Thus, the operators $\big\{\Pi^{(\omega)}_{k,l}\big\}_{k,l \ge 1}$ may be viewed as the eigenfunctions of $\F^{\omega}_{\widetilde{\otimes}}$. By means of the formalism of perturbation theory, we exploit this in Section \ref{sec:Eigapprox} of the supplement to derive the following expansion for the sequential eigenprojectors.
 \begin{prop}\label{prop:bexp}
Let $\F^{\omega}$ be the spectral density operator of a weakly stationary functional time series with eigendecomposition $\sum_{i=1}^{\infty}\lambda^{(\omega)}_i \Pi^{\omega}_i$. Furthermore,  let $\big\{  \hat{\lambda}^{(\omega)}_i (\eta) \big\}_{i \ge 1}$, $\big\{  \hat{\Pi}^{(\omega)}_i (\eta) \big\}_{i \ge 1}$  be the sequence of eigenvalues and eigenprojectors, respectively, of the sequential estimators $\hat{\F}^{\omega}(\eta)$, $\eta \in [0,1]$, of $\F^{\omega}$ . Then 
\begin{align*}
\hat{\Pi}^{(\omega)}_{k}-{\Pi}^{(\omega)}_{k}&= \biginprod{\hat{\Pi}^{\omega}_{k}(\eta)-{\Pi}^{\omega}_{k}}{{\Pi}^{\omega}_{k}}_{S_2}\Pi^{(\omega)}_{k}\\&+ \sum_{\substack{j,  \jpr=1,\\ \{j,  \jpr = k\}^\complement}}^{\infty} \frac{1}{\lambda^{(\omega)}_{j}\lambda^{(\omega)}_{\jpr} -(\lambda^{(\omega)}_{k})^2}\Big[\biginprod{\big(\F^{\omega} \widetilde{\otimes} \F^{\omega} -\hat{\F}^{\omega}(\eta) \widetilde{\otimes} \hat{\F}^{\omega} (\eta)\big) \Pi^{(\omega)}_{k}}{\Pi^{(\omega)}_{j\jpr}}_{S_2} \\
& ~~~~~~~~~ ~~~~~~~~~ ~~~~~~~~~ ~~~~~~~~~  ~~~~~~~~~ ~~~~~~~~~  +
\biginprod{E^{(\omega)}_{k,b_T}(\eta)}{\Pi^{(\omega)}_{j\jpr}}_{S_2}\Big] \Pi^{(\omega)}_{j\jpr}
\end{align*}
where $\{\cdot\}^{\complement}$ denotes the complement set, $\Pi^{(\omega)}_{j\jpr}= \phi^{(\omega)}_j \otimes \phi^{(\omega)}_{\jpr}$,  and where 
\begin{align*}
E^{(\omega)}_{k,b_T}(\eta) =\Big[(\F^{\omega} \widetilde{\otimes} \F^{\omega} -\hat{\F}^{\omega}(\eta) \widetilde{\otimes} \hat{\F}^{\omega} (\eta))\Big]\Big[\phiT{k}(\eta) - \Pi^{(\omega)}_{k}\Big]+\Big[(\hat{\lambda}^{(\omega)}_k(\eta))^2 \ -( \lambda^{(\omega)}_k)^2\Big]\Big[\phiT{k}(\eta) - \Pi^{(\omega)}_{k}\Big].
\end{align*}
\end{prop}
%The proof can be found in Section \ref{sec:Eigapprox}. 
In order to make sure the above expansion is well-defined for the eigenprojectors   $\Pi_{X,k}^{(\omega )}$  and $\Pi_{Y,k}^{(\omega )}$  we require the following assumption on the eigenvalues.  Let $G_{\ell,\widetilde{\otimes},k} =  \inf_{j,j^\prime:\{j,  j^\prime = k\}^\complement}|\lambda^{(\omega)}_{\ell,j}\lambda^{(\omega)}_{\ell,j^\prime}-(\lambda^{(\omega)}_{\ell,k})^2|$ for  $\ell \in \{X,Y\}$. 
\begin{assumption}\label{as:eigsep}
The first $k_0 > 1$ eigenvalues of $\F^{\omega}_{X}$ and of $\F^{\omega}_{Y}$ are positive, and 
$G_{X,\widetilde{\otimes},k_0}, G_{Y,\widetilde{\otimes},k_0} >0$. 
\end{assumption}
This assumption guarantees separability of the eigenvalues and ensures that we can test for relevant differences in the first $k_0-1$ eigenprojectors. Even though the above expansion expression is quite involved, the next statement  shows that the properties are controlled by a functional of a stochastic process, $\widetilde{M}_{\Pi,k}(\eta,\omega)$ that we will be able to link again to a process of the form \eqref{eq:limZXsec1}. 

\begin{lemma} \label{lem:errortildeM}
Suppose Assumptions \ref{as:depstruc}-\ref{as:ratiorates} are satisfied and let $\mathcal{Z}^{X,\omega}_{T,\eta}$ and $\mathcal{Z}^{Y,\omega}_{T,\eta}$ be given by \eqref{eq:Zx} and \eqref{eq:Zy}, respectively. Suppose furthermore that Assumption \ref{as:eigsep} holds true. Then,
\begin{align*}
\mathrm{(i)} & \sup_{\eta \in [0,1]}\int_a^b \bigsnorm{\widehat{M}_{\Pi,k}(\eta,\omega) - \eta \Big(\phix{k}-\phiy{k}\Big)-\widetilde{M}_{\Pi,k}(\eta,\omega)  }_2 d\omega = o_P\Big(\frac{1}{\sqrt{b_1 T_1+b_2T_2}}\Big);\\
\mathrm{(ii)} & \sup_{\eta \in [0,1]}\int_a^b \bigsnorm{\widetilde{M}_{\Pi,k}(\eta,\omega)  }^2_2 d\omega = o_P\Big(\frac{1}{\sqrt{b_1 T_1+b_2T_2}}\Big),\end{align*}
where 
\begin{align*}
\widetilde{M}_{\Pi,k}(\eta,\omega) &=  \frac{1}{\sqrt{b_1 T_1}}\sum_{\substack{j,  \jpr=1,\\ \{j,  \jpr = k\}^\complement}}^{\infty}  \frac{1}{\lambda^{(\omega)}_{X,j}\lambda^{(\omega)}_{X,\jpr} -(\lambda^{(\omega)}_{X,k})^2} \biginprod{\F_X^{(\omega)} \widetilde{\otimes} \mathcal{Z}^{X,\omega}_{T,\eta}+ \mathcal{Z}^{X,\omega}_{T,\eta}\widetilde{\otimes} \F_X^{(\omega)} }{\phix{j\jpr}  \widetilde{\otimes} \phix{k}}_{S_2} \phix{j\jpr}\\& 
\phantom{\approx}-\frac{1}{\sqrt{b_2 T_2}}\sum_{\substack{j,  \jpr=1,\\ \{j,  \jpr = k\}^\complement}}^{\infty}  \frac{1}{\lambda^{(\omega)}_{Y,j}\lambda^{(\omega)}_{Y,\jpr} -(\lambda^{(\omega)}_{Y,k})^2} \biginprod{  \F_Y^{(\omega)} \widetilde{\otimes} \mathcal{Z}^{Y,\omega}_{T,\eta}+ \mathcal{Z}^{Y,\omega}_{T,\eta}\widetilde{\otimes} \F_Y^{(\omega)} }{\phiy{j\jpr}  \widetilde{\otimes} \phiy{k}}_{S_2} \phix{j\jpr}. \tageq \label{eq:tildeM}
\end{align*}
\end{lemma}
The proof of this result is involved and left to Section  \ref{sec:proofshyps} of the Appendix. 
The expression \eqref{eq:tildeM} follows from the definition of $\mathcal{Z}^{X,\omega}_{T,\eta}$ and $\mathcal{Z}^{Y,\omega}_{T,\eta}$ together with an application of \autoref{lem:floor} in the Appendix. 
This subsequently allows us to establish that the process $\hat{\mathcal{Z}}^{[a,b],(k)}_{\Pi,T_1,T_2} $ in \eqref{eq:hatZbig} admits a stochastic expansion
 of the form as given in \autoref{thm:BuildingBlock}, from which its distributional properties can be obtained.
 %The proof of the following result 
%is tedious and relegated to Section \ref{sec:proofshyps} of the Appendix.
\begin{thm}\label{thm:hatZbigPi}
Suppose Assumptions \ref{as:depstruc}-\ref{as:eigsep} are satisfied. Then 
\begin{align*}
\hat{\mathcal{Z}}^{[a,b],(k)}_{\Pi,T_1,T_2} (\eta)&=
\eta\int_a^b \Big(  \frac{1}{\sqrt{b_2 T_2}}\Re\biginprod{ \mathcal{Z}^{Y,\omega}_{T,\eta} }{\widetilde{\Pi}^{(\omega)}_{Y,X,k}}_{S_2}+\frac{1}{\sqrt{b_1 T_1}}\Re\biginprod{ \mathcal{Z}^{X,\omega}_{T,\eta} }{\widetilde{\Pi}^{(\omega)}_{X,Y,k}}_{S_2} \Big) d\omega+o_P\Big(\frac{1}{\sqrt{b_1 T_1+b_2T_2}}\Big),
\end{align*}
where
\begin{align*}
\widetilde{\Pi}^{(\omega)}_{X,Y,k}& =4\sum_{j \ne k} \frac{\lambda^{(\omega)}_{X,k}\Re\biginprod{\phix{k j}}{\phiy{k}}_{S_2}}{\lambda^{(\omega)}_{X,k}\lambda^{(\omega)}_{X,j} -(\lambda^{(\omega)}_{X,k})^2} \phix{jk}, %\tageq \label{eq:PibigX}
 %\\
\quad \text{and} \quad \widetilde{\Pi}^{(\omega)}_{Y,X,k} =4\sum_{j \ne k} \frac{\lambda^{(\omega)}_{Y,k}\Re\biginprod{\phiy{k j}}{\phix{k}}_{S_2}}{\lambda^{(\omega)}_{Y,k}\lambda^{(\omega)}_{Y,j} -(\lambda^{(\omega)}_{Y,k})^2} \phiy{jk}.% \tageq \label{eq:PibigY}
\end{align*}
\end{thm}

By   \autoref{thm:BuildingBlock}   with $ \mathcal{U}^{(\omega)}_{XY} =- \widetilde{\Pi}^{(\omega)}_{X,Y,k}$, $ \mathcal{U}^{(\omega)}_{YX} =- \widetilde{\Pi}^{(\omega)}_{Y,X,k}$ 
 and \autoref{thm:hatZbigPi}  we  obtain the weak convergence 
\begin{equation}\label{hd90}
\big\{\sqrt{b_{1} T_1+ b_{2}T_2}\big( \hat{\mathcal{Z}}^{[a,b],(k)}_{\Pi,T_1, T_2}(\eta)\big)\big \}_{\eta \in [0,1]} {\rightsquigarrow} \tau_\Pi \{\eta  \mathbb{B}(\eta)\}_{\eta \in [0,1]}, \quad \text{ as } T_1, T_2 \to \infty,
\end{equation}
for some  constant  $\tau_\Pi$ and an application of the continuous mapping theorem 
shows that 
\begin{align} \label{hol41} 
\frac{\hat{\mathcal{Z}}^{[a,b],(k)}_{\Pi,T_1, T_2}(1)}{\hat{V}^{[a,b],(k)}_{\Pi,T_1, T_2}}\underset{T_1,T_2 \to \infty}{\Rightarrow} 
\mathbb{D}~,
\end{align}
where the random variable $\mathbb{D}$ is defined in \eqref{eq:matbbD} and the normalizing factor $\hat{V}^{[a,b],(k)}_{\Pi,T_1, T_2}$ is given by  
\[
\hat{V}^{[a,b],(k)}_{\Pi,T_1, T_2}= \Big(\int^{1}_0 \Big( \int_a^b\bigsnorm{\widehat{M}_{\Pi,k}(\eta,\omega)}^2_2 - \eta^2 \bigsnorm{\widehat{M}_{\Pi, k}(1,\omega)}^2_2 d\omega \Big)^2 \nu(d\eta) \Big)^{1/2}. 
\]
Combining these findings with the arguments given in the proof of \autoref{thm:leveltestZF} yields a consistent and  asymptotic level $\alpha$ test for the hypothesis \eqref{eq:testHyp} of no relevant difference in the $k$-th eigenprojector.

\begin{thm}\label{thm:leveltestZPi}
Suppose Assumptions \ref{as:depstruc}-\ref{as:eigsep} satisfied.
Then the test which rejects  the  null hypothesis  in \eqref{eq:testHyp} of no relevant difference in the $k$-th eigenprojector whenever 
\begin{equation} \label{hol30}
\widehat{\mathbb{D}}^{[a,b],k}_{\Pi,T_1,T_2} = \frac{\int_{a}^{b} \snorm{\widehat{M}_{\Pi,k}(1,\omega)}^2_2d\omega-\Delta_{\Pi,k}}{\hat{V}^{[a,b],k}_{\Pi,T_1, T_2}}
> q_{1-\alpha}(\mathbb{D})
\end{equation}
is consistent and has  asymptotic level $\alpha$, i.e., 
\[
\lim_{T_1,T_2 \to \infty} \mathbb{P}\Big(\widehat{\mathbb{D}}^{[a,b],k}_{\Pi,T_1,T_2} >q_{1-\alpha}(\mathbb{D}\Big) =
\begin{cases}
0 & \mbox{if } \Delta_{\Pi,k} > \int_a^b \snorm{M_{\Pi,k}{(\omega)}}^2 _2 d\omega; \\ 
\alpha & \mbox{if } \Delta_{\Pi,k} = \int_a^b \snorm{M_{\Pi,k}{(\omega)}}^2 _2 d\omega {\mbox{  and } \tau_{\Pi} \neq 0}; \\ 
1 & \mbox{if } \Delta_{\Pi,k} < \int_a^b \snorm{M_{\Pi,k}{(\omega)}}^2 _2 d\omega. \
\end{cases}
\]
\end{thm}

\subsection{No relevant difference in the eigenvalues: hypothesis \eqref{eq:testHypEv}}
\label{sec33}

Finally, we briefly discuss the test for the hypothesis in \eqref{eq:testHypEv} of no relevant difference in the $k$-th eigenvalue. Denote the (pointwise) population distances and empirical distances of the $k$-th largest eigenvalues  at frequency $\omega$ by
\[
{M}_{\lambda, k}(\omega)=\lamx{\omega}{k}-\lamy{\omega}{k} \text{~~ and ~~ }\widehat{M}_{\lambda,k}(\eta,\omega) =\eta \big(\lamxT{k}(\eta)-\lamyT{k}(\eta)\big)\quad \eta \in [0,1],\]
and define 
\begin{align*}
\hat{\mathcal{Z}}^{[a,b],(k)}_{\lambda,T_1,T_2} (\eta)&=\sqrt{b_1 T_1+b_2 T_2}\int_{a}^{b} |\widehat{M}_{\lambda,k}(\eta,\omega)|^2 -\eta^2 |{M}_{\lambda, k}(\omega)|^2 d\omega.  \tageq \label{eq:Zlam}
\end{align*}
We make use of the following proposition, which is proved in Section \ref{sec:Eigapprox} of the Appendix.
\begin{prop}\label{prop:bexplam}
Let $\F^{\omega}$ have eigendecomposition $\sum_{k=1}^{\infty}\lambda^{(\omega)}_k \Pi^{\omega}_k$ and let  $\Big\{  \hat{\lambda}^{(\omega)}_k (\eta) \Big\}_{k \ge 1}$ $\Big\{  \hat{\Pi}^{(\omega)}_k (\eta) \Big\}_{k \ge 1}$ be the sequence of eigenvalues and eigenprojectors , respectively, of $\hat{\F}^{\omega}(\eta)$, $\eta \in [0,1]$. Then, 
\begin{align*}
\hat{\lambda}^{(\omega)}_{k}(\eta)-\lambda^{(\omega)}_{k}%&%= \biginprod{\Delta_\eta \F^{\omega}\Pi^{(\omega)}_{k}}{\Pi^{(\omega)}_{k}} -\biginprod{(\Delta_{\eta} \lambda^{(\omega)}_{k} -\Delta_\eta \F^{\omega} ) \Delta_{\eta}\Pi^{(\omega)}_{k}}{\Pi^{(\omega)}_{k}} \\& 
= \biginprod{\hat{\F}^{\omega}(\eta)-\F^{\omega}}{\Pi^{(\omega)}_{k}}_{S_2}+ \biginprod{E^{(\omega)}_{\lambda, k,b_T}(\eta)}{\Pi^{(\omega)}_{k}}_{S_2}, % -\biginprod{\big[ (\hat{\lambda}^{(\omega)}_{k}(\eta)-\lambda^{(\omega)}_{k})-(\hat{\F}^{\omega}(\eta)-\F^{\omega})\big] \big( \hat{\Pi}^{(\omega)}_{k}(\eta) -{\Pi}^{(\omega)}_{k}\eta)\big)}{\Pi^{(\omega)}_{k}} 
\end{align*}
where
\[
E^{(\omega)}_{\lambda, k,b_T}(\eta) = \big(\hat{\F}^{\omega}(\eta)-\F^{\omega}\big) \big( \hat{\Pi}^{(\omega)}_{k}(\eta) -{\Pi}^{(\omega)}_{k}\big)- \big(\hat{\lambda}^{(\omega)}_{k}(\eta)-\lambda^{(\omega)}_{k}\big) \big( \hat{\Pi}^{(\omega)}_{k}(\eta) -{\Pi}^{(\omega)}_{k}\big).
\]
\end{prop}
The following theorem is the counterpart of \autoref{thm:hatZFhyp1} and  \autoref{thm:hatZbigPi}, and shows that we can  express \eqref{eq:Zlam} into a process of the form \eqref{eq:limZXsec1}. 
\begin{thm}\label{thm:hatZbiglam}
Suppose Assumptions \ref{as:depstruc}--\ref{as:eigsep} are satisfied. Then 
\begin{align*}
\hat{\mathcal{Z}}^{[a,b],(k)}_{\lambda,T_1,T_2} (\eta)
%&=
%\int_a^b \biginprod{\widetilde{M}_{\lambda,k}(\eta,\omega)}{\eta M_{\lambda,k}(\omega)}d\omega+ \int_a^b \overline{\biginprod{\widetilde{M}_{\lambda,k}(\eta,\omega)}{\eta M_{\lambda,k}(\omega)}}d\omega + o_P\Big(\frac{1}{\sqrt{b_1 T_1+b_2T_2}}\Big)
%\\&
 &=\int_a^b \Big( \frac{2 \eta M_{\lambda,k}(\omega)}{\sqrt{b_1 T_1}}\Re\biginprod{ \mathcal{Z}^{X,\omega}_{T,\eta} }{\Pi^{(\omega)}_{X,k}}_{S_2} 
- \frac{2 \eta M_{\lambda,k}(\omega) }{\sqrt{b_2 T_2}}\Re\biginprod{ \mathcal{Z}^{Y,\omega}_{T,\eta} }{\Pi^{(\omega)}_{Y,k}}_{S_2} \Big)d\omega
\\&+o_P\Big(\frac{1}{\sqrt{b_1 T_1+b_2T_2}}\Big),
\end{align*}
\end{thm}
In this case, an application of \autoref{thm:BuildingBlock}   with $ \mathcal{U}^{(\omega)}_{XY} = 2M_{\lambda,k}(\omega) \Pi^{(\omega)}_{X,k} $ and $ \mathcal{U}^{(\omega)}_{YX} = 2M_{\lambda,k}(\omega) \Pi^{(\omega)}_{Y,k}  $ yields therefore
\begin{equation} \label{hd91}
\big\{\sqrt{b_{1} T_1+ b_{2}T_2}\big( \hat{\mathcal{Z}}^{[a,b],(k)}_{\lambda,T_1, T_2}(\eta)\big)\big \}_{\eta \in [0,1]} {\rightsquigarrow} \tau_\lambda \{\eta  \mathbb{B}(\eta)\}_{\eta \in [0,1]}, \quad \text{ as } T_1, T_2 \to \infty,
\end{equation}
for some constant  $\tau_\lambda$ (see Section \ref{sec:proofBB} for details). Now using the same arguments as in Section \ref{sec31}  and \ref{sec32} 
we obtain that the test, which rejects the null hypothesis of no relevant difference in the $k$-th eigenvector in \eqref{eq:testHypEv}, whenever
\begin{equation} \label{hol31}
\widehat{\mathbb{D}}^{[a,b],k}_{\lambda,T_1,T_2} = \frac{\int_{a}^{b} \snorm{\widehat{M}_{\lambda,k}(1,\omega)}^2_2d\omega-\Delta_{\lambda,k}}{\hat{V}^{[a,b],k}_{\lambda,T_1, T_2}}
 > q_{1-\alpha}(\mathbb{D})~,
\end{equation} 
is consistent and has asymptotic levek $\alpha$.  The proof is omitted for the sake of brevity. 
\begin{thm}\label{thm:leveltestZlam}
Suppose Assumptions \ref{as:depstruc}--\ref{as:eigsep} are satisfied.
Then, 
\[
\lim_{T_1,T_2 \to \infty} \mathbb{P}\Big(\widehat{\mathbb{D}}^{[a,b],k}_{\lambda,T_1,T_2} >q_{1-\alpha}(\mathbb{D) }\Big) =
\begin{cases}
0 & \mbox{if } \Delta_{\lambda,k} > \int_a^b \snorm{M^{(\omega)}_{\lambda,k}}^2 _2 d\omega; \\ 
\alpha & \mbox{if } \Delta_{\lambda,k} = \int_a^b \snorm{M^{(\omega)}_{\lambda,k}}^2 _2 d\omega {\mbox{  and } \tau_\lambda \neq 0;}\\ 
1 & \mbox{if } \Delta_{\lambda,k} < \int_a^b \snorm{M^{(\omega)}_{\lambda,k}}^2 _2 d\omega. \
\end{cases}
\]
\end{thm}

\begin{Remark} [further statistical applications] 
\label{remarkh2}  ~~
\rm
\begin{itemize}
\item[(1)]
{Our results can be used  to construct   (asymptotic) confidence intervals for the distance between the spectral operators, eigenprojectors and eigenvalues from the two series. We    recommend to use  these intervals  
if  it  is  difficult to specify the threshold  in the hypothesis of a relevant difference. 
To be precise, note  that it follows 
from \eqref{eq:matbbD} that the interval 
$$  
\Bigg [0,  \int_{a}^{b} \snorm{\widehat{M}_{\hat\F}(1,\omega)}^2_2d\omega  +q_{1-\alpha}(\mathbb{D} )  \hat{V}^{[a,b]}_{\F,T_1, T_2} \Bigg ].
$$
defines  
an asymptotic  $(1-\alpha)$ confidence interval for the distance 
$\int_{a}^{b}\snorm{\F^{(\omega)}_X-\F^{(\omega)}_Y}_2^2 d\omega$ between the two spectral operators. 
In the same way 
confidence intervals for the distance between the eigenprojectors and eigenvalues defined in \eqref{h40}
can be obtained from  \eqref{hol41} and the weak convergence 
 $ {\hat{\mathcal{Z}}^{[a,b],(k)}_{\lambda ,T_1, T_2}(1)}/{\hat{V}^{[a,b],(k)}_{\lambda,T_1, T_2}} {\Rightarrow}  
 \mathbb{D}$, respectively.}
\item[(2)]  
 Moreover, it is also possible to test for relevant differences for a finite number thresholds simultaneously. 
 To be precise, consider the problem of testing the hypotheses in \eqref{eq:testHypEv} for the  different
 thresholds $\Delta_{\lambda,k}^{(1)} < \ldots < \Delta_{\lambda,k}^{(L)}$.  In particular if the null hypothesis
  is accepted  with the threshold  $ \Delta_{\lambda,k}^{(L_0)} $, it is also accepted for all 
  thresholds $ \Delta_{\lambda,k}^{(L_0+1)} , \ldots  , \Delta_{\lambda,k}^{(L_)} $.
  Correspondingly, rejection for a $ \Delta_{\lambda,k}^{(L_0)} $  means rejection for all smaller thresholds. In this sense, evaluating the test for several thresholds is  logically consistent for the user, and it is possible to determine for fixed level $\alpha$   the largest threshold such that the null hypotheses is rejected.
\end{itemize} 
\end{Remark}

\begin{Remark}[local alternatives] \label{remarkh3}  
{\rm
The tests developed  in this section can detect local alternatives converging to the null hypothesis a rate $1/\sqrt{b_{1} T_1+ b_{2}T_2}$. Compared to the problem of testing classical hypotheses the formulation of this property is more complicated and for the sake of brevity we restrict ourselves to the case of comparing the $k$th eigenvalues (but similar statements can also be made for the other testing problems). To be precise, consider the testing problem  \eqref{eq:testHypEv} and assume a local alternative, such that 
\begin{equation} \label{hol50} 
\int_{a}^{b}\big \vert\lamx{\omega}{k}-\lamy{\omega}{k}\big\vert^2 d\omega  =   \Delta_{\lambda,k}  + {c \over  
\sqrt{b_{1} T_1+ b_{2}T_2} }  \big (1 +  o(1) \big )
\end{equation}
for some   constant $c >0$. Then it follows by similar  arguments as given in \eqref{hol25} 
that 
\[
\lim_{T_1,T_2 \to \infty} \mathbb{P}\Big(\widehat{\mathbb{D}}^{[a,b],k}_{\lambda,T_1,T_2} >q_{1-\alpha}(\mathbb{D) }\Big) =\mathbb{P} \Big( \mathbb{D} >  q_{1-\alpha}(\mathbb{D})  - {c \over \tau_{\lambda} 
\mathbb{V} }  \Big )
%=\mathbb{P} \Big( \mathbb{B}(1) >  q_{1-\alpha}(\mathbb{B}(1))  - {c \over  \tau_{\lambda} }  \Big )
> \alpha ~,
\]
where $\mathbb{V}$ is defined in \eqref{hol51}.
Note that the condition \eqref{hol50} is implied by  $| \lamx{\omega}{k}-\lamy{\omega}{k} |  =   \sqrt{\Delta_{\lambda,k} } + g(\omega) \big / \sqrt{b_{1} T_1+ b_{2}T_2} $ for some function $g$ such that 
$\int_a^b g(\omega) d\omega > 0$.
}
\end{Remark}

\section{Finite sample properties} \label{sec:sec5}
\def\theequation{4.\arabic{equation}}
\setcounter{equation}{0}

In order to assess the finite sample properties of the tests proposed in Section \ref{sec31}-\ref{sec33}, we conducted an extensive simulation study of which an overview is provided in this section. 
In all scenarios, the empirical rejection probabilities (ERP) are calculated over 1000 repetitions and the processes are generated on a grid of  1000 equi-spaced points in the interval $[0,1]$ and then converted into functional data objects using a Fourier basis, which we shall denote by $\{\psi_k\}_{k \ge 1}$. In order to define the self-normalization sequence we used the measure $\nu =\frac{1}{n-1}\sum_{i=1}^{n-1}\delta_{i/n}$, where $\delta_\eta$ denotes the Dirac measure at $\eta \in [0,1]$. Simulations reported below are conducted with $n=20$. Other values were also considered but we found comparable results for all other choices of $n$ for which the positive mass is sufficiently bounded away from the boundaries see also Remark \ref{remarkh0} for a heuristic  explanation of this observation). 
In the simulations reported below, we used a Daniell window with bandwidth $b_T = T^{-1/3}$ to estimate the sequential spectral density operators. Self-evidently, a more sophisticated calibration could provide better results but we found that moderate variations of the bandwidth did not fundamentally change the findings as presented in this paper. 
\\

 {\bf Setting A: Brownian bridges.} In the first setting, we generate a sequence $\{X_t\}^T_{t=1}$ of independent Brownian bridges with variance multiplied by a factor $2\pi$. Using the closed-form expression of the Karhunen-Lo{\`e}ve (KL) expansion of a Brownian Bridge \citep[see e.g.,][]{DeMa03}, it can be shown that the eigenvalues and eigenfunctions of the spectral density operator $\F^{(\omega)}_X$ are then respectively given by $\lamx{\omega}{k} = 1/(\pi k)^2$ and $\phi^{(\omega)}_{X,k}(\tau) = \sqrt{2} \sin (\pi k \tau)$, $\tau \in [0,1]$, $k \in \mathbb{N}$, for all $\omega \in \rnum$.  The number of basis functions is chosen to be $d=21$, which captures more than 95 percent of variation. We report the results for $[a,b]=[0,\pi]$. 
\begin{itemize}[leftmargin=*]
\item {\bf Scenario 1: shift in the eigenfunctions}. We generate the alternative processes $\{Y_{t}\}^T_{t=1}$ similar to $\{X_t\}^T_{t=1}$, i.e., as sequences of independent Brownian bridges with variance multiplied by a factor $2\pi$. However, we shift the first eigenfunction in the KL expansion of the  $\{Y_{t}\}^T_{t=1}$  to $ \sqrt{2} \sin (\pi k (\tau+\iota))$ with $\iota$ varying  between 0 and 0.15.  The corresponding  shifts cause a change in various eigenfunctions of the spectral density operator. In \autoref{fig:scenario1}(a), we provide the ERP corresponding to a true value $\iota=0.05$ of the test for the hypothesis of no relevant differences between the spectral density operators \eqref{testOp}, while \autoref{fig:scenario1}(b)--(c) depict the ERP of the test corresponding to the hypothesis of no relevant differences between the eigenprojectors (\eqref{hol30}) for $k=1,2$, respectively. This particular shift induces corresponding threshold values $\Delta \approx 0.00047$, $\Delta_{\Pi,1}\approx 0.0474$ and  $\Delta_{\Pi,2}\approx 0.040$, respectively. The behavior visible in the three plots clearly corroborates with the theoretical findings stated in  \autoref{thm:leveltestZF} and \autoref{thm:leveltestZPi}, respectively.
 For values of the shift belonging to the interior of the null hypothesis, i.e., $\iota<0.05$, we observe that the ERP are below the nominal level and are getting closer to zero as the value of $\iota$ gets  close to zero. For those values that belong to the interior of the alternative, i.e., $\iota >0.05$, we observe ERP strictly larger than the nominal level and which increase to 1 as the size of the shift increases. At the boundary of the null hypothesis, i.e., where $\iota=0.05$, the test is close to the nominal level of $\alpha=0.05$. As expected, one observes that estimation precision improves as the sample size $T$ increases. 
\begin{figure}[t!]
\vspace*{-10pt}
\centering
\begin{subfigure}[b]{0.33\linewidth}
\includegraphics[width=\linewidth]{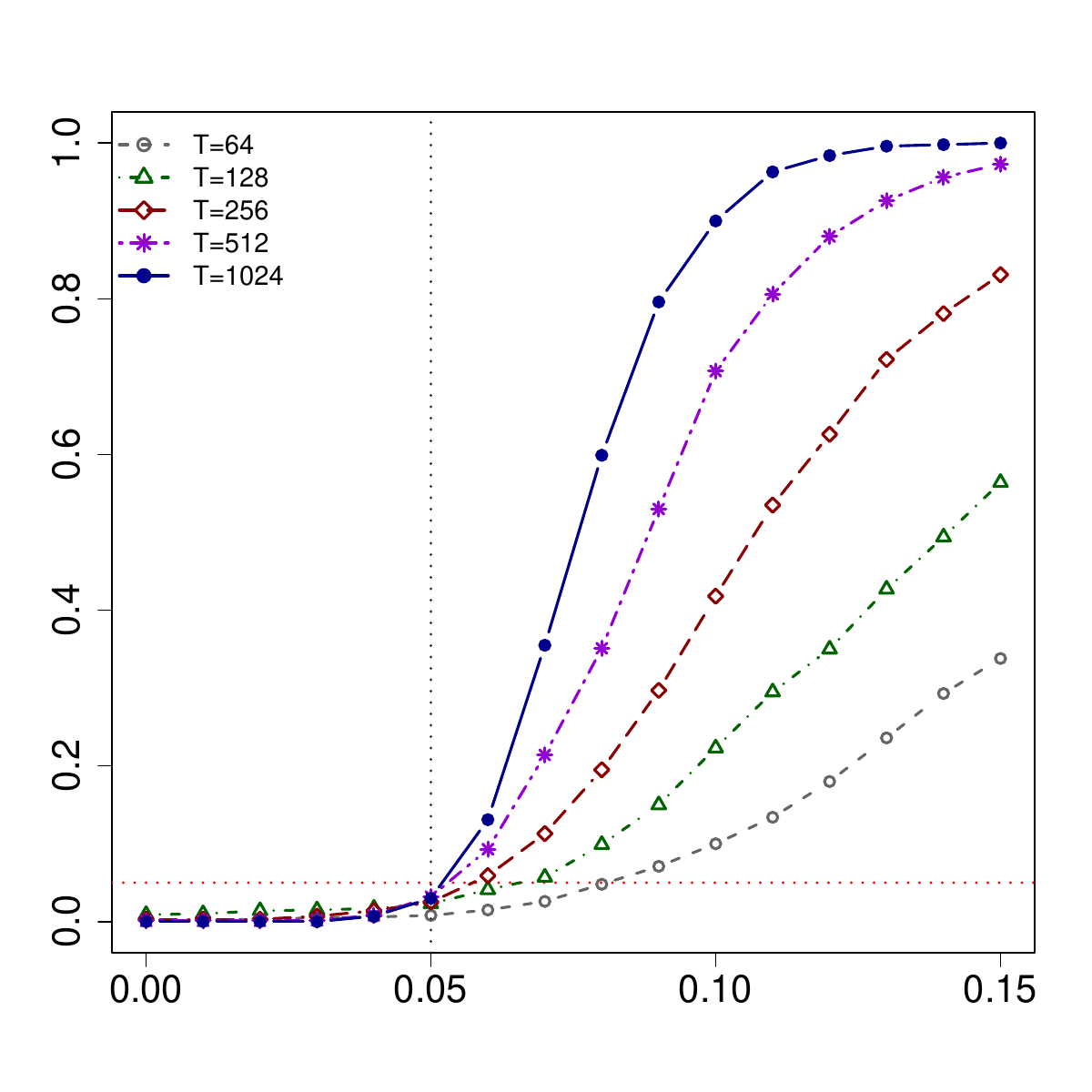}
\setlength{\abovecaptionskip}{-10pt}
\setlength{\belowcaptionskip}{-8pt} 
\caption{}
\end{subfigure}\hfil
\begin{subfigure}[b]{0.33\linewidth}
\includegraphics[width=\linewidth]{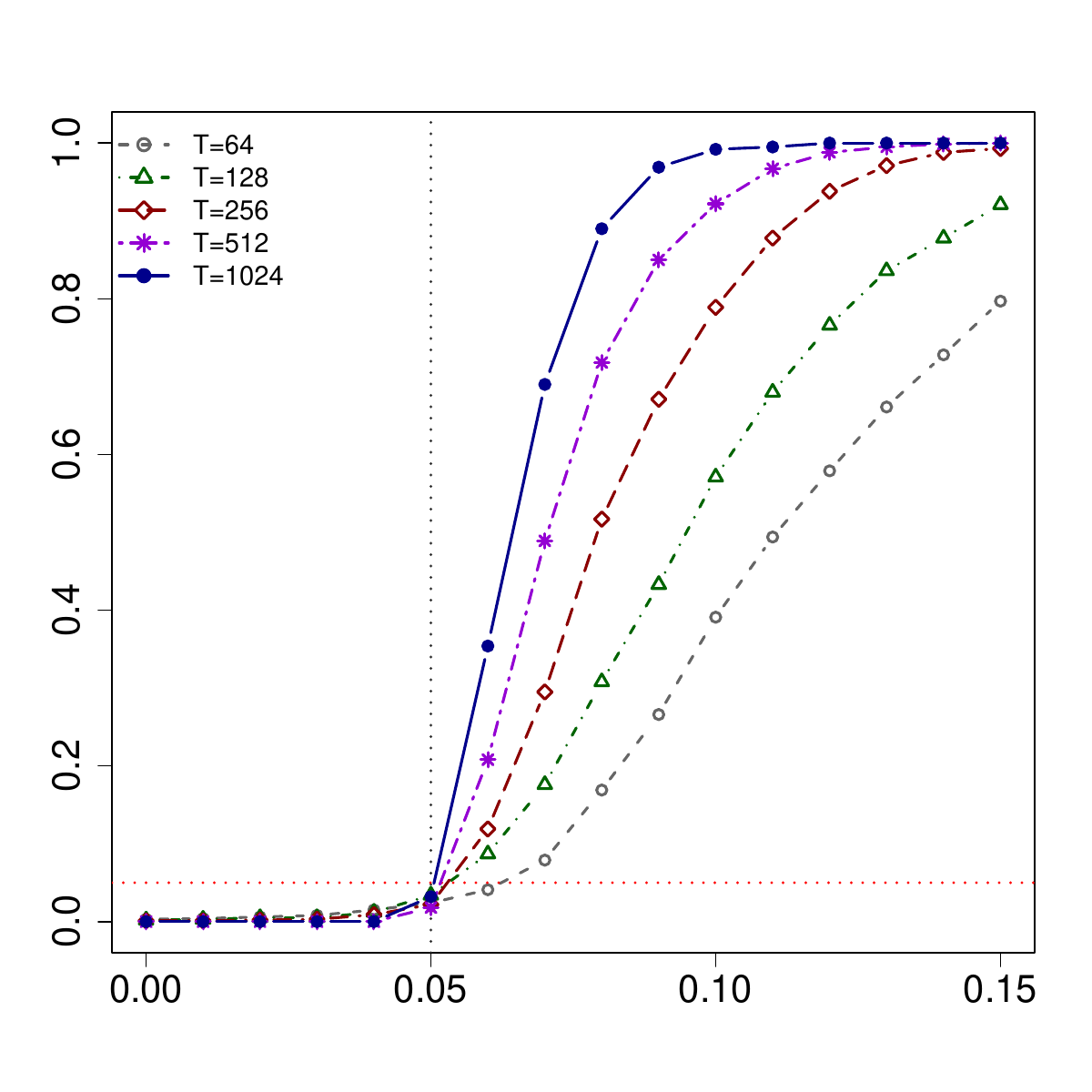}
\setlength{\abovecaptionskip}{-10pt}
\setlength{\belowcaptionskip}{-8pt} 
\caption{}
 \end{subfigure}\hfil
\begin{subfigure}[b]{0.33\linewidth}
\includegraphics[width=\linewidth]{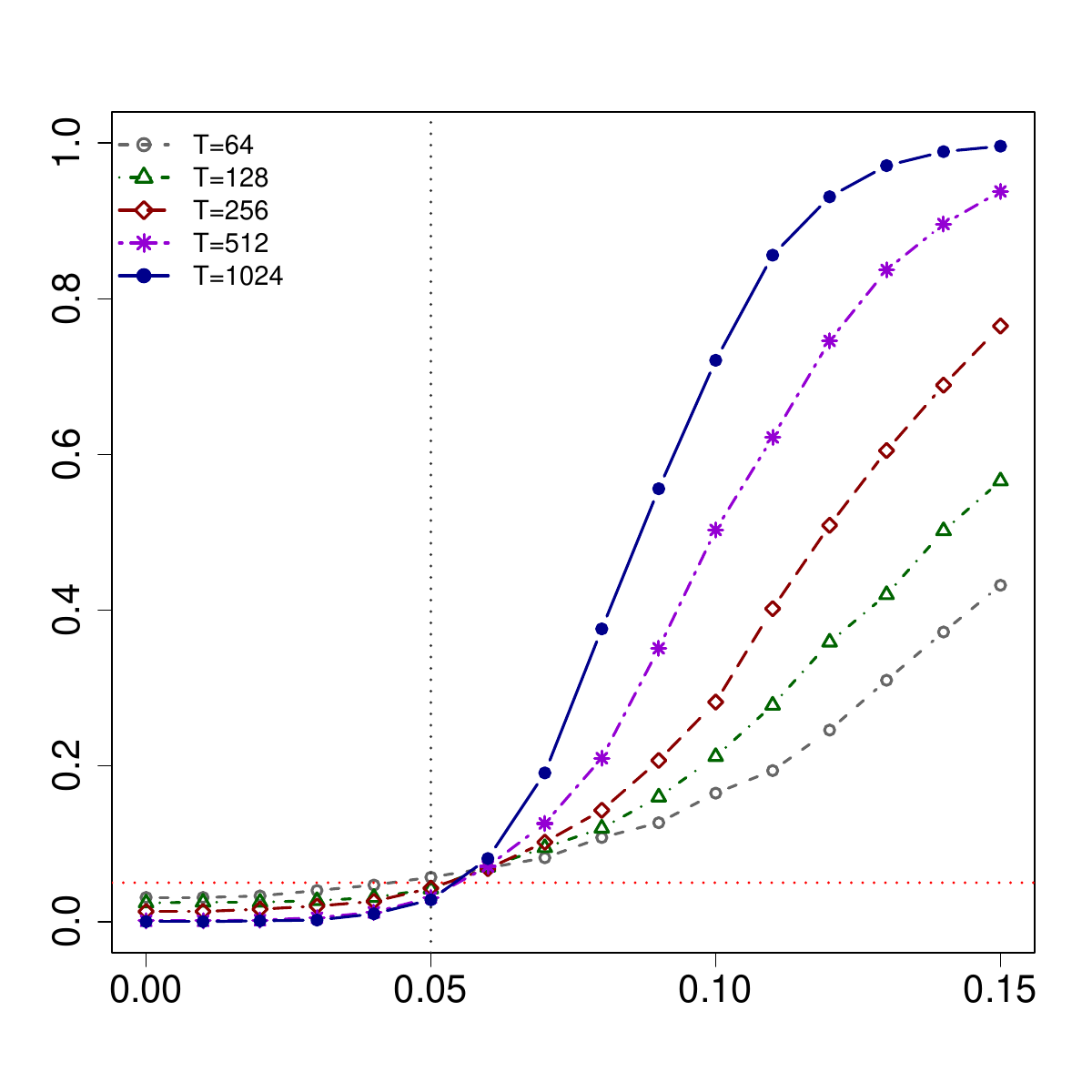}
\setlength{\abovecaptionskip}{-10pt}
\setlength{\belowcaptionskip}{-8pt} 
\caption{}
\end{subfigure}\hfil
\setlength{\belowcaptionskip}{-8pt}  
\caption{\it ERP under scenario 1 of the relevant hypotheses tests \eqref{testOp} (panel (a)) and \eqref{hol30} for $k=1,2$ (panel (b)--(c), resp.) plotted as a function of $\iota$ at the nominal level 0.05 (horizontal dotted line). The true shift $\iota=0.05$ is marked in the three panels by the vertical dotted line and corresponds to induced threshold values $\Delta \approx 0.00047
$, $\Delta_{\Pi,1}\approx 0.047$ and  $\Delta_{\Pi,2} \approx 0.040$, respectively.
}  
\label{fig:scenario1}
\end{figure}
\begin{figure}[!h]
\vspace*{-10pt}
\centering
\begin{subfigure}[b]{0.33\linewidth}
\includegraphics[width=\linewidth]{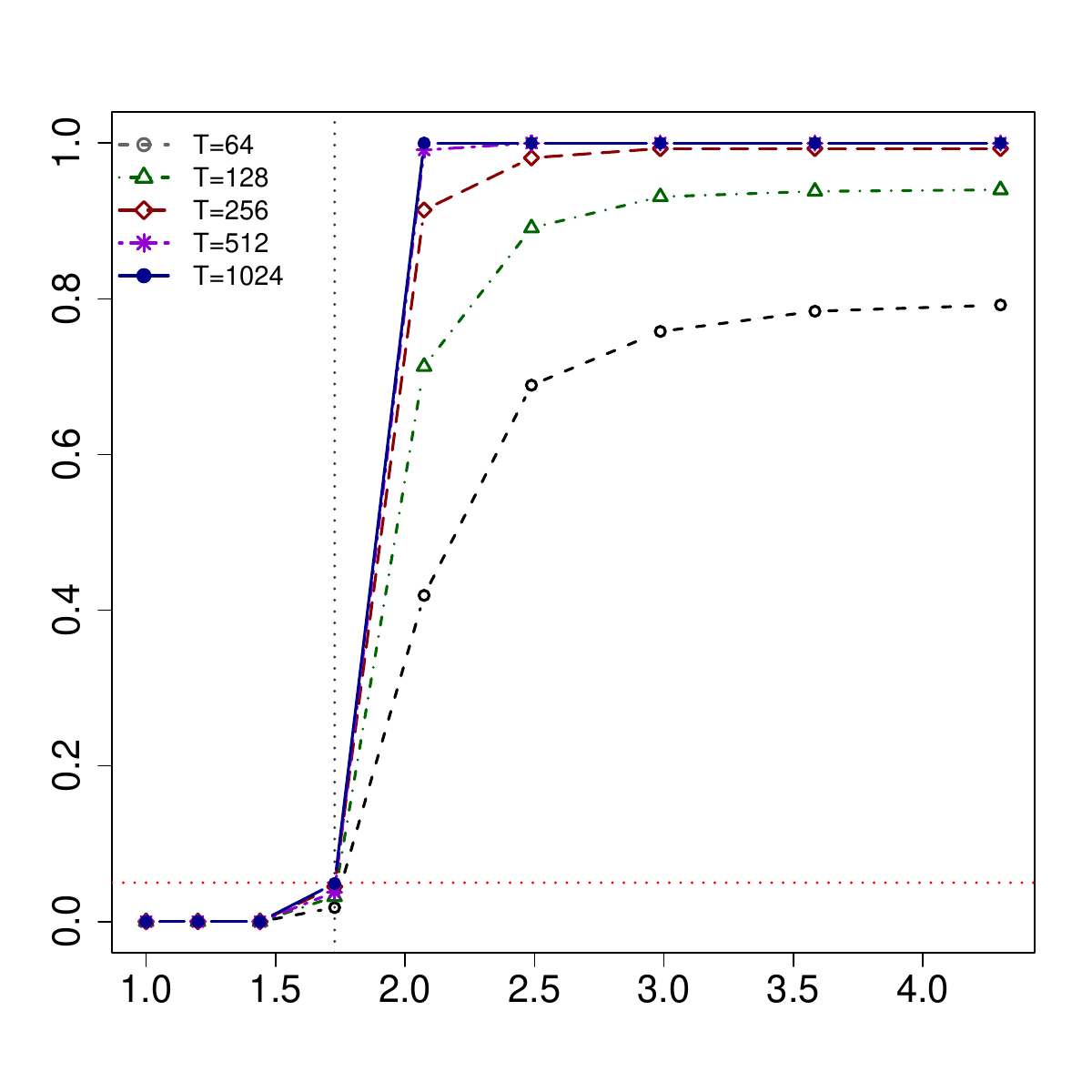}
\setlength{\abovecaptionskip}{-10pt}
\setlength{\belowcaptionskip}{-8pt} 
\caption{}
\end{subfigure}\hfil
\begin{subfigure}[b]{0.33\linewidth}
\includegraphics[width=\linewidth]{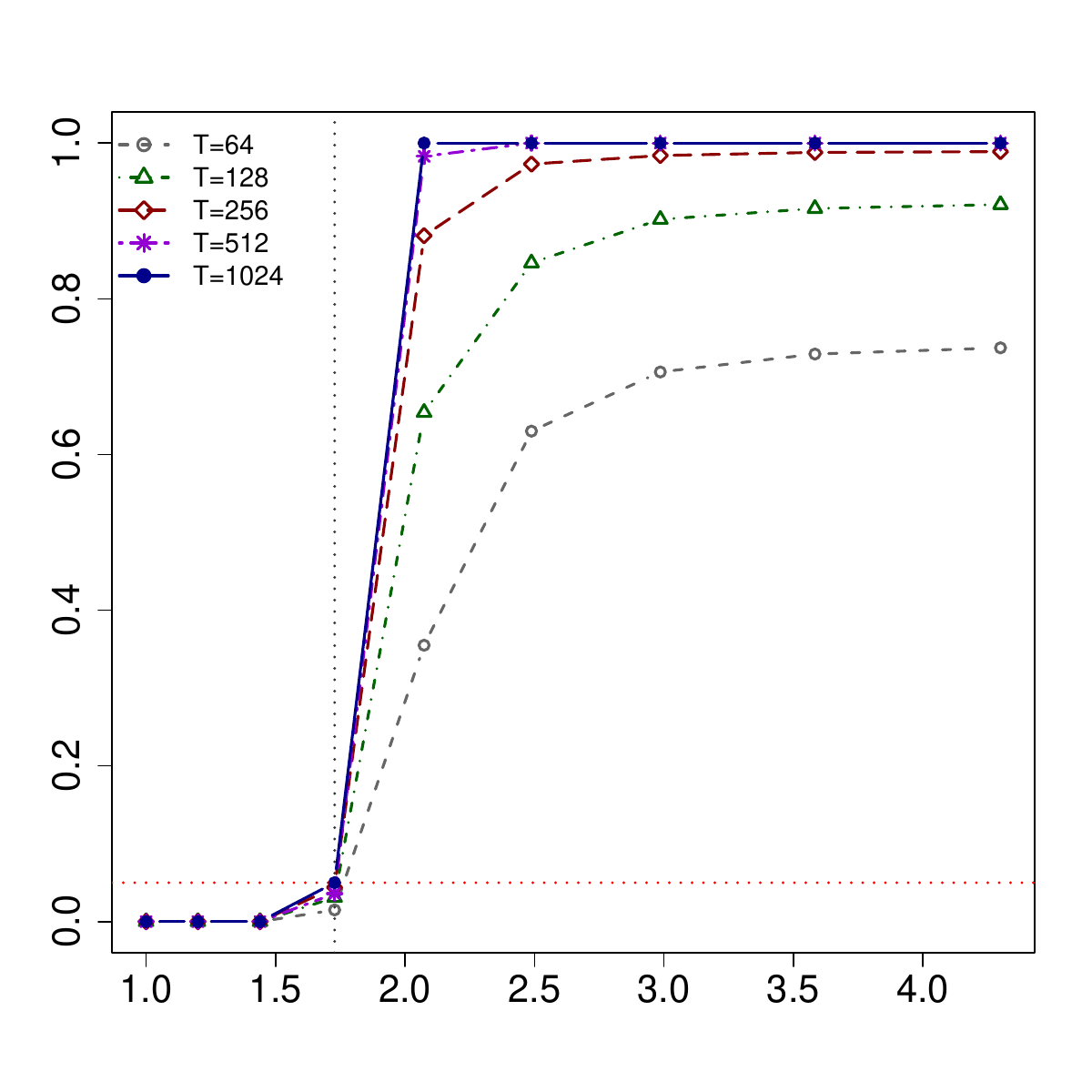}
\setlength{\abovecaptionskip}{-10pt}
\setlength{\belowcaptionskip}{-8pt} 
\caption{}
 \end{subfigure}\hfil
\begin{subfigure}[b]{0.33\linewidth}
\includegraphics[width=\linewidth]{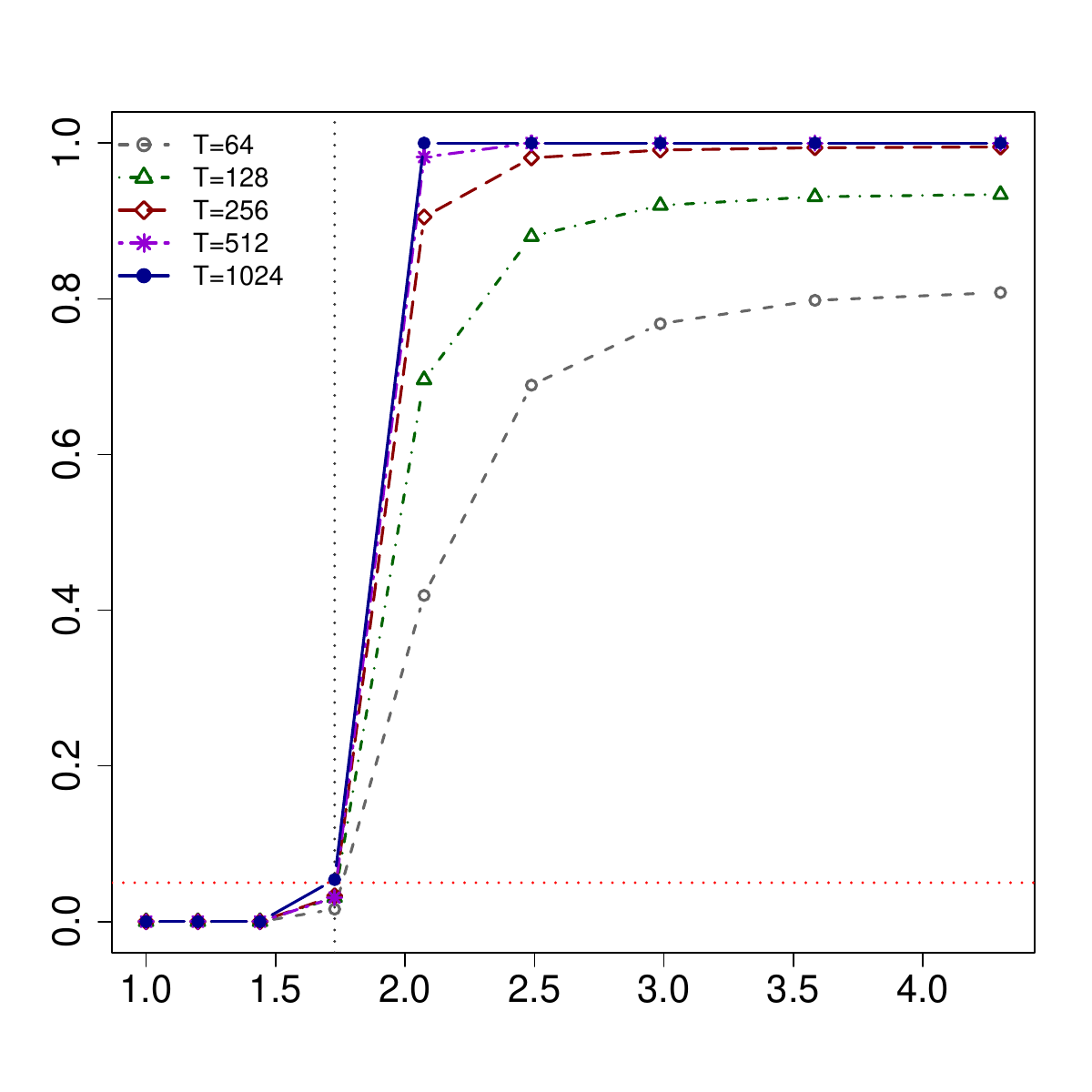}
\setlength{\abovecaptionskip}{-10pt}
\setlength{\belowcaptionskip}{-8pt} 
\caption{}
\end{subfigure}\hfil
\setlength{\belowcaptionskip}{-8pt}  
\caption{\it ERP under scenario 2 of the relevant hypotheses tests \eqref{testOp} (panel (a)) and \eqref{hol31} for $k=1,2$ (panel (b)--(c), resp.) plotted as a function of $\iota$ at the nominal level 0.05 (horizontal dotted line). The true amplitude factor $\iota=1.2^3$ is marked by the vertical dotted line and corresponds to induced threshold values $\Delta \approx 0.044$, $\Delta_{\lambda,1}\approx 0.040$ and  $\Delta_{\lambda,2}\approx 0.0025$, respectively.
}  
\label{fig:scenario2}
\end{figure}
\item {\bf Scenario 2: amplitude variation}. The alternative processes $\{Y_{t}\}^T_{t=1}$ are now generated as sequences from independent Brownian bridges where the standard deviation is multiplied by a factor $\iota=1.2^\ell \times \sqrt{2\pi}, \ell =0,1, \ldots 8$. We consider a true factor $\iota=1.2^3$.  \autoref{fig:scenario2} provides the corresponding ERP of the test \eqref{testOp} in Section \ref{sec31}  (difference between operators, panel (a)) and of the test \eqref{hol31} in Section \ref{sec33} (difference between eigenvalues for $k=1,2$, panel (b)--(c), resp.)  Similar observations as in the previous scenario  allow to conclude that the tests behave according to the derived theory, where the precision is quite accurate, even for the smaller sample sizes.
%\begin{figure}[!h]
%\centering
%\includegraphics[width=0.33\textwidth]{BBeigvaluesDaniel_2312_F_4_005_20.pdf}
%\includegraphics[width=0.33\textwidth]{BBeigvaluesDaniel_2312_4_005_1_20.pdf}
%\includegraphics[width=0.33\textwidth]{BBeigvaluesDaniel_2312_4_005_2_20.pdf} 
%\setlength{\abovecaptionskip}{-2pt}
%\setlength{\belowcaptionskip}{-8pt} 
%\caption{\it  ERP under scenario 2 for the relevant hypotheses in Section \ref{sec31} (left panel) and Section \ref{sec33} for $k=1,2$ (middle and right panel) as a function of $\iota$ at the nominal level 0.05 (horizontal dotted line). The true amplitude factor $\iota=1.2^3$ is marked by the vertical dotted line and corresponds to induced threshold values $\Delta \approx 0.044$, $\Delta_{\lambda,1}\approx 0.040$ and  $\Delta_{\lambda,2}\approx 0.0025$, respectively.}
%\label{fig:scenario2}
%\end{figure}
\end{itemize}

{\bf Setting B: Functional moving average}. Next, we consider processes of the form
\[
X_{t}= \sum_{s=0}^2 A_s (\epsilon_{t-s}) \tageq \label{eq:fma2}
\]
where $\{\epsilon_{t}\}$ is a collection of independent Brownian motions on $[0,1]$. It is well known that  $\epsilon_{t}$ can be represented using its KL expansion $\epsilon_t(\tau) =\sum_{k=1}^{\infty} \zeta_{k,t} e_{k}(\tau)$ where $\{\zeta_{k,t}\}_{k \ge 1}$ is a sequence of independent Gaussian random variables with variance $\text{var}(\zeta_{k,t}) =\frac{1}{(k-1/2)^2 \pi^2}$, and where the sequence $\{e_k\}_{k \ge 1}$ with  $e_k(\tau) = \sqrt{2}\sin\big( (k-1/2) \pi \tau \big)$, $\tau \in [0,1]$, forms an orthonormal system of $L^2([0,1])$. We represent the operators $A_s$ in the basis  $\{\psi_l \otimes e_k\}^{d,m}_{l=1,k=1}$, where we recall that $\{\psi_l\}$ denotes the Fourier basis on $[0,1]$. The vectors of coefficients $\tilde{X_t}=(\inprod{X_t}{\psi_1}, \ldots, \inprod{X_t}{\psi_d})^\top$ can then be represented as \[
\tilde{X_t}= \sum_{s=0}^2 \sum_{k=1}^m \tilde{A}_{s,(\cdot) k} \zeta_{k,t-s} 
= \sum_{s=0}^2 \tilde{{A}}_{s} \boldsymbol{\zeta}_{t-s}.
\]
We set $d=15, m=51$ and simulate the matrices $\tilde{{A}}_{s}=\{\tilde{A}_{s,lj}\}_{l,j}$ from a Gaussian distribution with independent entries such that $\tilde{A}_{0,lk} \sim \mathcal{N}(0,1)$,  $\tilde{A}_{1,lj} \sim \mathcal{N}(0, (l+j^{3/2})^{-1})$, and $\tilde{A}_{2,lj} \sim \mathcal{N}(0, l^{-2})$.
\setlength{\itemindent}{-2em}
\begin{itemize}[leftmargin=*]
\item {\begin{figure}[!h]
\vspace*{-10pt}
\centering
\begin{subfigure}[b]{0.33\linewidth}
\includegraphics[width=\linewidth]{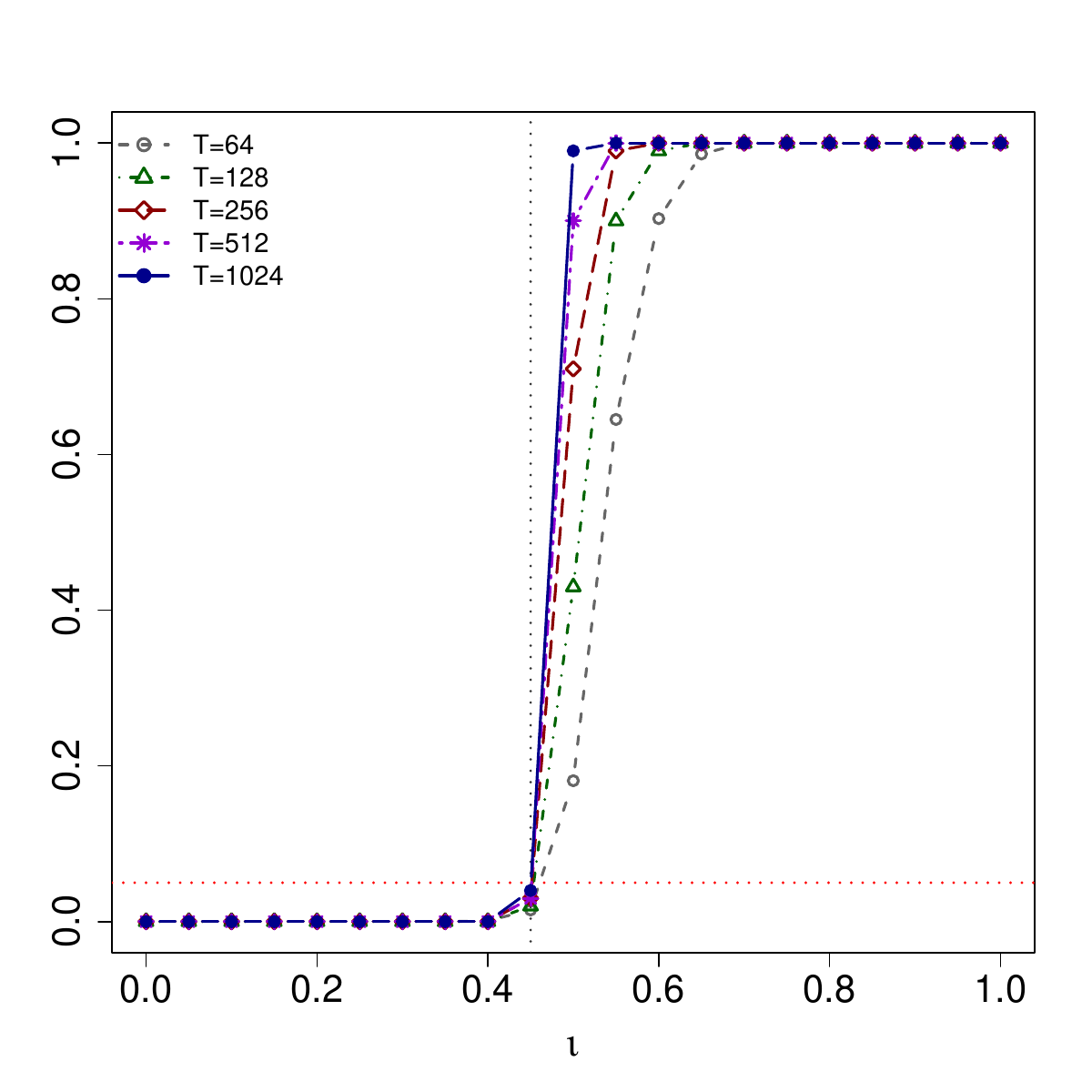}
\setlength{\abovecaptionskip}{-10pt}
\setlength{\belowcaptionskip}{-8pt} \caption{}
 \end{subfigure}\hfil
 \begin{subfigure}[b]{0.33\linewidth}
\includegraphics[width=\linewidth]{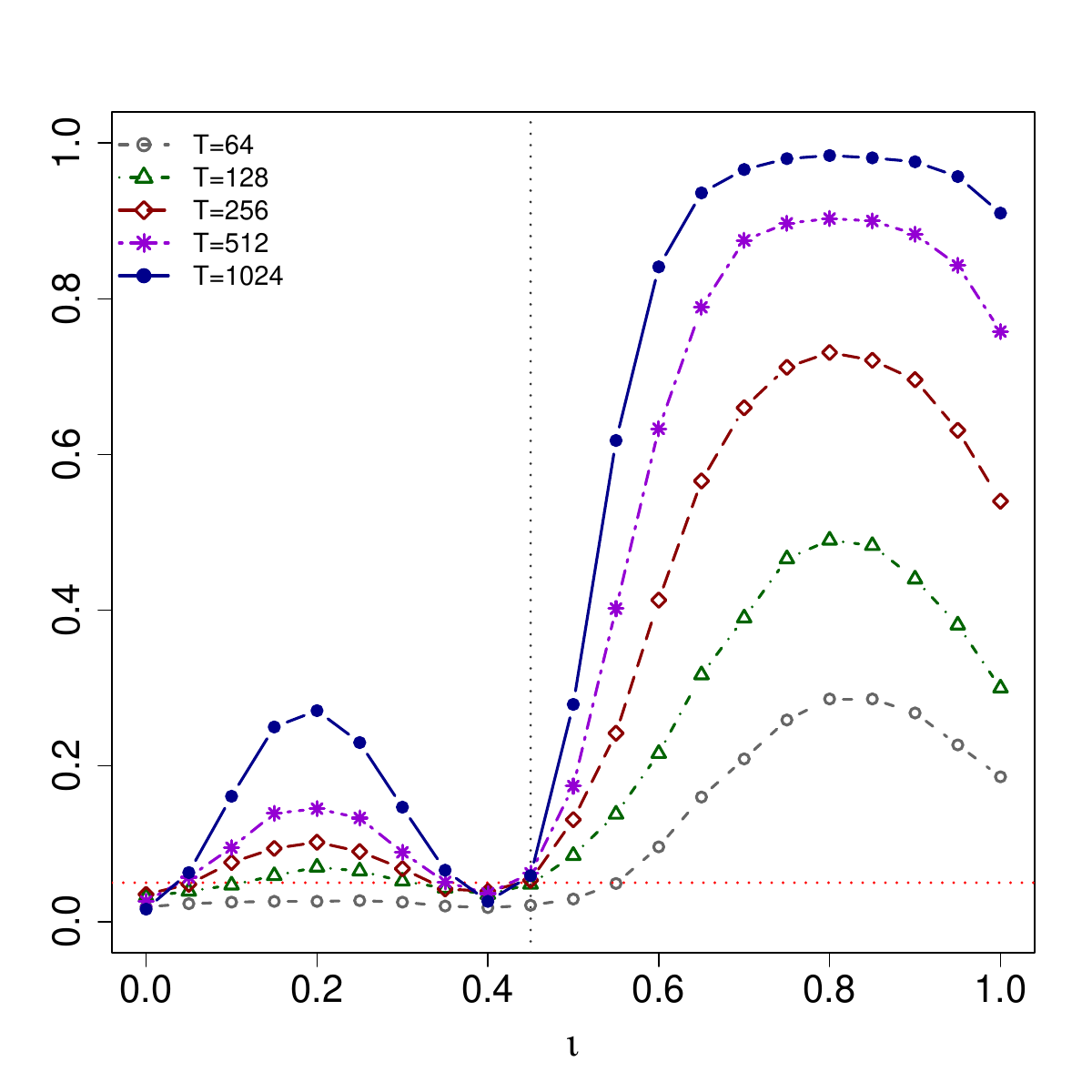}
\setlength{\abovecaptionskip}{-10pt}
\setlength{\belowcaptionskip}{-8pt}   \caption{}
  \end{subfigure}
\caption{\it  ERP under scenario 3 for the relevant hypotheses tests \eqref{hol30} (panel (a)) and \eqref{hol31} (panel (b))  for $k=1$ plotted as a function of the parameter $\iota$ at the nominal level $0.05$ (horizontal dotted line). The vertical dotted line illustrates the true shift $\iota=0.45$. The induced thresholds values for $[a,b]=[0,\pi]$ are $\Delta_{\Pi,1} \approx 0.24$ and $\Delta_{\lambda ,1} \approx 0.03$, respectively.
%, respectively. Panel  (c): differences between the  ``true''  distances $M^2$, $M^2_{\Pi ,1}$, $M^2_{\lambda ,1}$ and the thresholds  $\Delta$, $\Delta_{\Pi,1}$, $\Delta_{\lambda ,1}$, respectively.
}
\label{fig:FMA2_feigf}
\end{figure}
\begin{figure}[!h]
\vspace*{-10pt}
\centering
\begin{subfigure}[b]{0.33\linewidth}
\includegraphics[width=\linewidth]{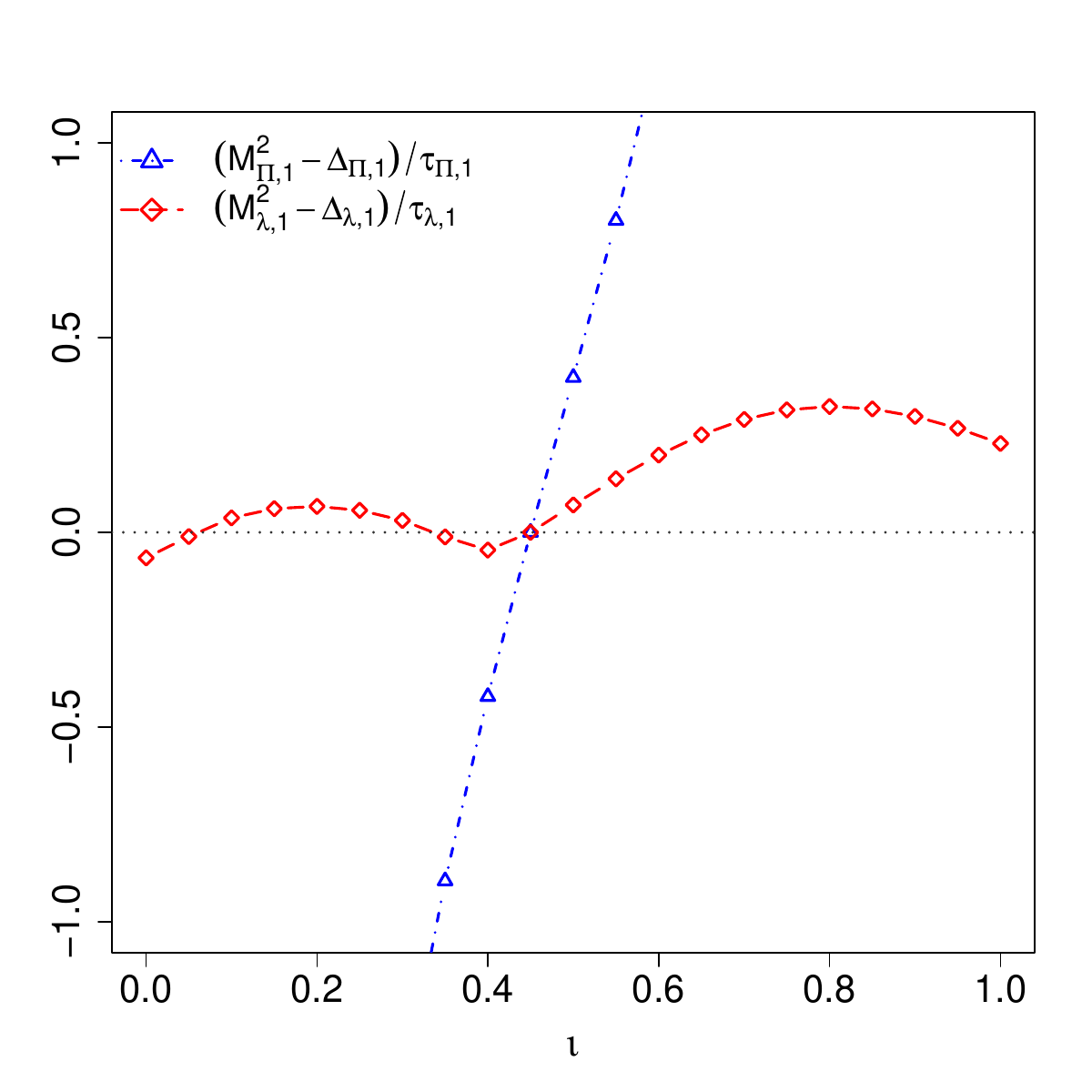}
\setlength{\abovecaptionskip}{-10pt}
\setlength{\belowcaptionskip}{-8pt} 
\caption{}
\end{subfigure}\hfil
\begin{subfigure}[b]{0.33\linewidth}
\includegraphics[width=\linewidth]{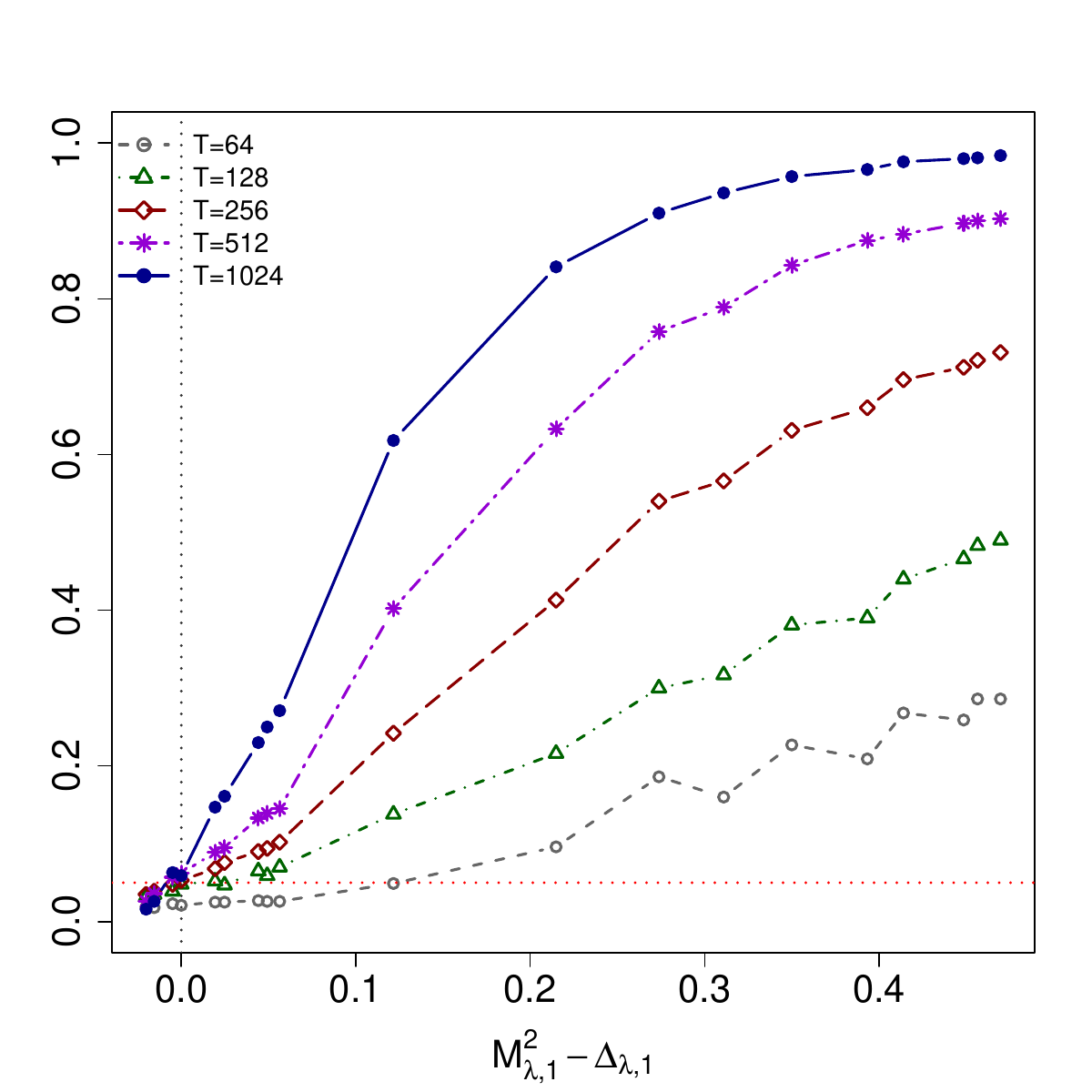}
\setlength{\abovecaptionskip}{-10pt}
\setlength{\belowcaptionskip}{-8pt} 
\caption{}
 \end{subfigure}
\setlength{\belowcaptionskip}{-8pt}  \caption{\it 
Panel (a): the ``relative'' differences  $(M^2_{\Pi ,1} - \Delta_{\Pi,1} )/ \tau_{\Pi}$ and $( M^2_{\lambda ,1}- \Delta_{\Pi,1}) / \tau_{\lambda}$ in scenario 3.
Panel (b):  ERP as a function of $M^2_{\lambda,1}-\Delta_{\lambda,1}$, where the vertical dotted lines corresponds to $M^2_{\lambda,1}=\Delta_{\lambda,1}$.}  
\label{fig:FMA2_ev}
\end{figure}
{\bf Scenario 3: shift in the eigenfunctions.} We simulate $\{X_t\}_{t=1}^T$ and the alternative processes $\{Y_t\}_{t=1}^T$ from \eqref{eq:fma2} , but for the alternative processes $\{Y_t\}_{t=1}^T$ we change the first eigenfunction of the innovations $\{\epsilon_{t}\}$ to ${e}_1(\tau) = \sqrt{2}\sin\big((1/2+\iota) \pi \tau \big), \tau \in [0,1]$, where $\iota$ varies from $0$ to $1$. This shift affects various aspects of the second order structure non-uniformly over the frequency band, and most prominently the first order marginal characteristics.
We therefore concentrate on testing relevant hypotheses for the largest  eigenvalues and corresponding  eigenprojectors  on $[a,b]=[0,\pi]$. We let $\iota=0.45$ be the true value. This  choice corresponds to the 
threshold $\Delta_{\Pi , 1} \approx 0.24$ in \eqref{eq:testHyp} and  $\Delta_{\lambda , 1}  \approx 0.02$ in \eqref{eq:testHypEv}. 
 In \autoref{fig:FMA2_feigf}(a), we depict the ERP of the test  
\eqref{hol30}  for the hypothesis 
  of no relevant differences between the first eigenprojectors
 (\autoref{thm:leveltestZPi}),  while the results corresponding  to  the test   \eqref{hol31}   for  the hypothesis  of no relevant differences between 
 the largest eigenvalues (\autoref{thm:leveltestZlam})  can be found   in \autoref{fig:FMA2_feigf}(b). 
 Both the test for the first eigenprojector  (a) and the test for the largest eigenvalue  (b) closely align with the derived theory, 
 but this statement needs a little more explanation.
 \\
 First, note that in the model under consideration it is not immediately clear  that the distances $M^2_{\lambda ,1} -\Delta_{\lambda ,1} $  and $M^2_{\Pi ,1} -\Delta_{\Pi ,1} $ are  monotone function of the parameter $\iota$. Therefore, we display in \autoref{fig:FMA2_ev}(a) 
 the ``relative''  differences $(M^2_{\lambda ,1} -\Delta_{\lambda ,1} ) / \tau_\lambda$  and $(M^2_{\Pi ,1} -\Delta_{\Pi ,1} )/  \tau_\Pi$ 
 as a function of $\iota$, where $ \tau_\lambda  $ and $\tau_\Pi$
are the standard deviations appearing in the limiting processes
\eqref{hd90} and \eqref{hd91}, respectively. Note again that the choice $\iota=0.45$
 corresponds to the boundary of the null hypothesis, that is  $M^2_{\Pi ,1}-\Delta_{\Pi  ,1}\approx 0$, $M^2_{\lambda ,1}-\Delta_{\lambda ,1}\approx 0$. We observe that  
  the  null hypothesis $H_0: M^2_{\Pi ,1}  \leq \Delta_{\Pi ,1} $ is equivalent to  $H_0: \iota \le 0.45$. 
 Moreover, a similar argument as in \eqref{hol25} shows that for large sample sizes the power of the test is  approximately given by 
\begin{align} 
 \mathbb{P}\Big( \mathbb{D}  > - {\sqrt{b_{1} T_1+ b_{2}T_2} 
 \over  \mathbb{V }} \cdot
  { M^2  - \Delta \over  \tau  } +  q_{1-\alpha}( \mathbb{D})   \Big)
% & =  \mathbb{P}\Big( \mathbb{B} (1)   > - \sqrt{b_{1} T_1+ b_{2}T_2} \cdot
%  {( M^2  - \Delta) /\tau  } +  q_{1-\alpha}(\mathbb{B (1) })   \Big) 
~, 
  %\nonumber \\
%  & = \Phi  \Big( \mathbb{B} (1)   > - \sqrt{b_{1} T_1+ b_{2}T_2} \cdot
%  {( M^2  - \Delta) } +  u_{1-\alpha} )   \Big)
  \label{hd50}
\end{align}
where 
%$\Phi$ is the cdf of the standard normal distribution, $u_{1-\alpha}$ its $(1-\alpha)$-quantile, 
$(M^2 -\Delta) / \tau  = ( M^2_{\Pi, 1}    - \Delta_{\Pi, 1} ) / \tau_\Pi $  for the test \eqref{hol30} (eigenprojectors) and
$( M^2 -\Delta ) /\tau   = ( M^2_{\lambda, 1} - \Delta_{\lambda, 1} ) /\tau_\lambda$   for the test \eqref{hol31}  (eigenvalues). Thus, the rejection probability depends  (approximately) on
the size of the relative difference $( M^2 - \Delta) /\tau $ and a larger  value means more power.
Consequently, for  the hypothesis of no relevant 
difference between the eigenprojectors, formula \eqref{hd50} indicates that for large sample sizes the power of the test \eqref{hol30} is a strictly increasing function of $\iota$ and this property is confirmed by  the ERP displayed in  \autoref{fig:FMA2_feigf}(a).
\\
On the other hand, from \autoref{fig:FMA2_ev}(a) we also observe that 
 the distance $( M^2_{\lambda ,1} -\Delta_{\lambda ,1} )/ \tau_{\lambda} $ is not a monotone function of the parameter $\iota$. 
 In fact, most values of  $\iota \in [0,1]$ represent the alternative
 $H_1: M^2_{\lambda ,1} > \Delta_{\lambda ,1} $  and only  small neighbourhoods at $0$ and $0.45$ correspond  to the  null  hypothesis.
Consequently, \autoref{fig:FMA2_feigf}(b) mainly displays ERP under the alternative, which explains the fact that the simulated rejection probabilities   exceed the nominal level for most $\iota \in [0,1]$ (note also that for $M^2_{\lambda ,1} -\Delta_{\lambda ,1} \approx 0 $
the level is well approximated).
The approximation for the power in formula \eqref{hd50} also provides an explanation for the non-monotonicity of the ERP in \autoref{fig:FMA2_feigf}(b). 
%Moreover,  the distances 
% $M^2_{\lambda ,1} -\Delta_{\lambda ,1} $  are smaller  for  $\iota < 0.45$  than for  $\iota >  0.45$, which gives some understanding 
% why the ERP for  $\iota >  0.45$ in Figure \autoref{fig:FMA2_feigf}(b) exceed those for  $\iota <  0.45$.  
Additionally, we can  use formula  \eqref{hd50}  for an heuristic  explanation why 
 the test \eqref{hol30} for the eigenprojectors has more power than the test \eqref{hol31} for the eigenvalues (compare \autoref{fig:FMA2_feigf})(a) and (b)).  Observe   in  \autoref{fig:FMA2_ev}(a) that  for  $\iota >  0.45$  the curve $(M^2_{\Pi, 1} - \Delta_{\Pi, 1} ) / \tau_\Pi$  exceeds    the curve   $(M^2_{\lambda ,1} -\Delta_{\lambda ,1} )/ \tau_\lambda$, which explains the (substantially) larger power of the test for the eigenprojectors.
 %and  note that the right hand side of the approximation \eqref{hd50}  is increasing with respect to the difference $M^2  -\Delta  $.
\\ 
 Finally, 
  we also plot in panel (b) of \autoref{fig:FMA2_ev}  the ERP of  the test \eqref{hol31} against the differences $M^2_{\lambda , 1}-\Delta_{\lambda , 1}$ and  observe that this test behaves according to the derived theory: i) the ERP converge to zero in the interior of the null hypothesis $M^2_{\lambda_1}< \Delta_{\lambda , 1}=0.02$ ii) the ERJP converge to the nominal level at the boundary of the null hypothesis, i.e., where  $M^2_{\lambda_1}\approx\Delta_{\lambda_1}$, and (iii) for those values belonging to the alternative hypothesis $M^2_{\lambda_1}>\Delta_{\lambda_1}$ power increases to $1$.
 } 
 \medskip 
 
\item {\bf Scenario 4: amplitude variation.} Consider again the FMA(2) process as specified in \eqref{eq:fma2}. We now simulate  $\{X_t\}_{t=1}^T$ and  $\{Y_t\}_{t=1}^T$ from this process, but the variance of the noise process of the alternative processes  $\{Y_t\}_{t=1}^T$ is multiplied  with a factor $\iota=1.05^{\ell}$ where $\ell =0,\ldots, 20$, that is, $\text{var}(\zeta_{k,t})=\frac{\iota}{ ((k-1/2)\pi)^{2}}$. We consider a true factor $\iota =1.05^{9}$ and $[a,b]=[0,\pi]$. 
The eigenprojectors are not affected by this change. Therefore, we display in 
\autoref{fig:FMA2_amp1}  the ERP of the test \eqref{testOp}   in Section \ref{sec31}  (difference between operators, panel (a)) and of the test  \eqref{hol31} in Section \ref{sec33} (difference between eigenvalues for $k=1,2$, panel  (b)--(c)). 
 The effect of this change in amplitude is in fact almost fully captured by the sequences of largest eigenvalues. This is also what one may observe in panel (a) and (b), where the ERP of the test for the spectral density operators and of the test for the   largest eigenvalues follow a similar pattern. 
\begin{figure}[!h]
\vspace*{-10pt}
\centering
\begin{subfigure}[b]{0.33\linewidth}
\includegraphics[width=\linewidth]{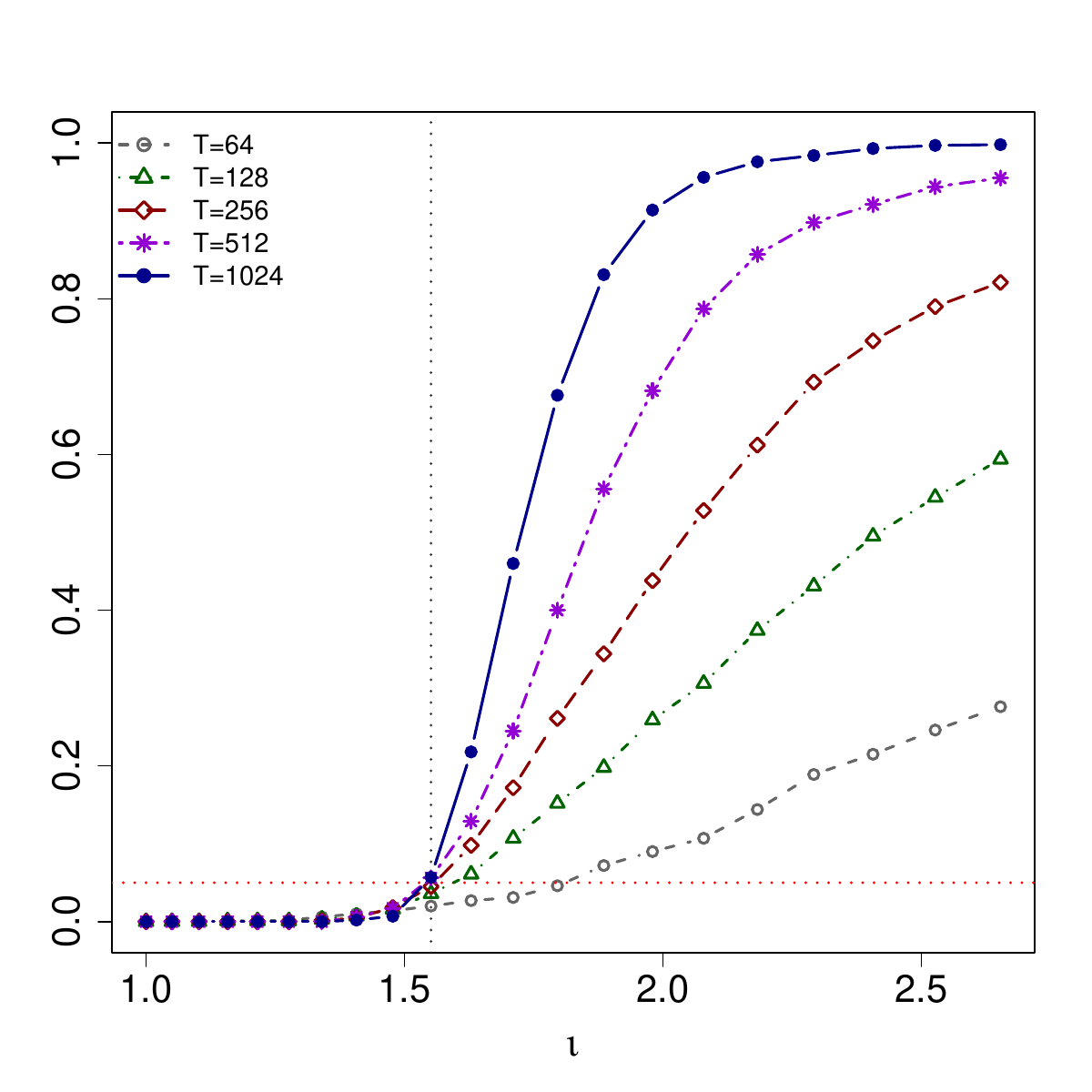}
\setlength{\abovecaptionskip}{-10pt}
\setlength{\belowcaptionskip}{-8pt} \caption{}
\end{subfigure}\hfil
 \begin{subfigure}[b]{0.33\linewidth}
 \includegraphics[width=\linewidth]{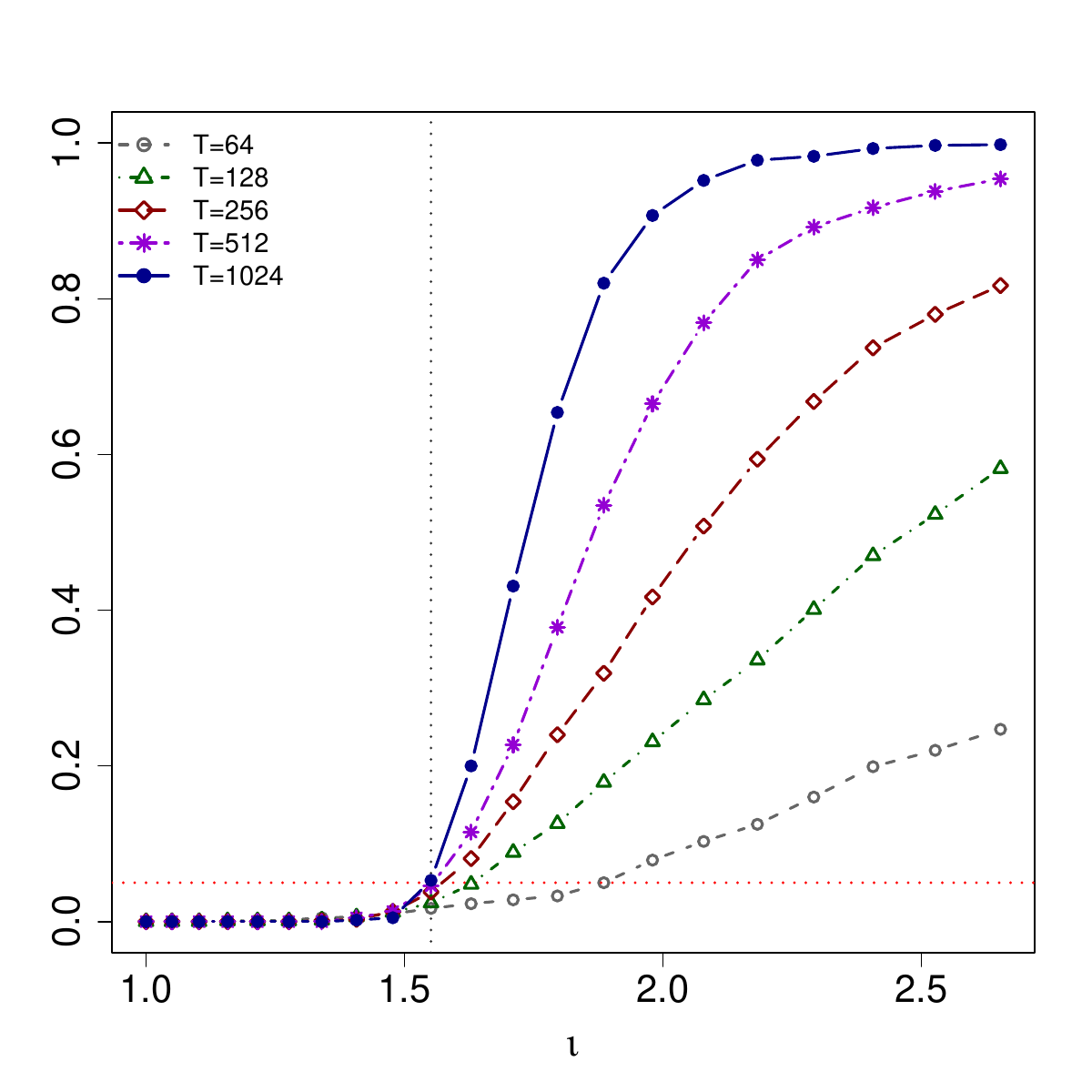}
\setlength{\abovecaptionskip}{-10pt}
\setlength{\belowcaptionskip}{-8pt}   \caption{}
  \end{subfigure}
   \begin{subfigure}[b]{0.33\linewidth}
 \includegraphics[width=\linewidth]{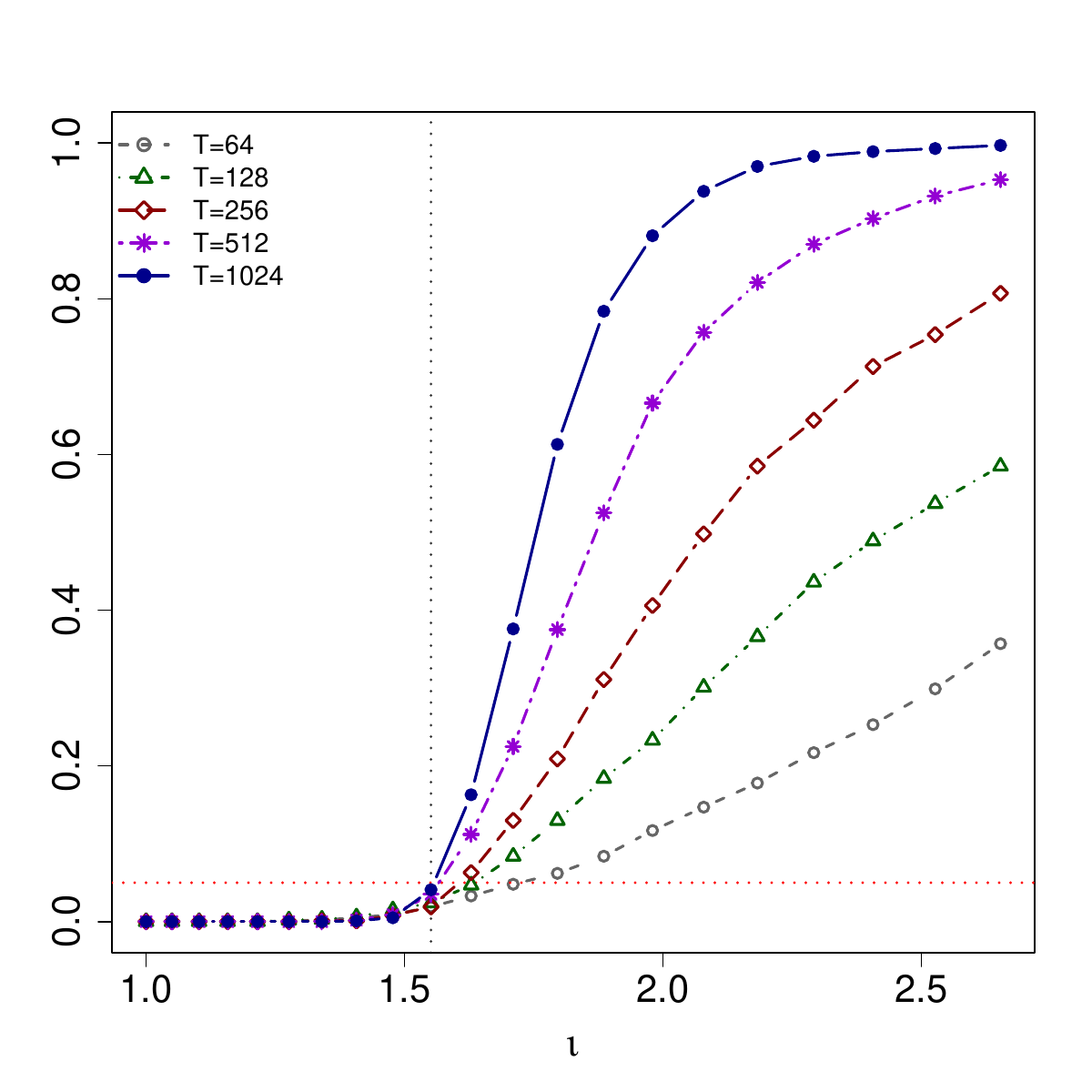}
\setlength{\abovecaptionskip}{-10pt}
\setlength{\belowcaptionskip}{-8pt}   \caption{}
  \end{subfigure}
\setlength{\belowcaptionskip}{-10pt}   \caption{\it  ERP under scenario 4 for the relevant hypotheses tests  \eqref{testOp} (panel (a)) and \eqref{hol31} for $k=1,2$ (panel (b)--(c), resp.) as a function of $\iota$ at the nominal level 0.05 (horizontal dotted line). The vertical dotted line illustrates the true factor $\iota=1.05^{9}$. The induced thresholds values for $[a,b]=[0,\pi]$ are $\Delta \approx 0.610$, $ \Delta_{\lambda , 1}  \approx 0.605$ and $ \Delta_{\lambda , 2}  \approx 0.004$, respectively.}
\label{fig:FMA2_amp1}
\end{figure}
\end{itemize}

{\bf Setting C: Functional autoregressive processes.}
Finally, we investigate finite sample performance in the context of functional AR(p) models 
\[X_{t} =\sum_{j=1}^{p} A_{j}(X_{t-j}) + \epsilon_t, \tageq \label{eq:FAR2}\]
 where $A_j \in S_{\infty}(\Hi)$ and $\{\epsilon_t\}_{t\in \mathbb{Z}} \in \mathcal{L}^2_\Hi$. A basis expansion yields that  the first $d$ coefficients $\widetilde{X}_t = \left(\langle X_t, \psi_1 \rangle, \dots, \langle X_t, \psi_{d} \rangle \right)^\top$ approximately satisfy the VAR(p) equation $\widetilde{{X}}_t = \sum_{j=1}^{p} \widetilde{A}_j(\widetilde{X}_{t-j}) + \widetilde{\boldsymbol{\epsilon}}_t$
where $\widetilde{\epsilon}_t = \left(\langle \epsilon_t, \psi_1 \rangle, \dots, \langle \epsilon_t, \psi_{d} \rangle \right)^T$ and where the $(l,l')$-th entry of $\tilde{{A}}_j$ is given by $\langle A_j(\psi_l),\psi_{l'}\rangle$. To ensure that the $A_j$ are bounded operators, we require that $\widetilde{A}_{j,l,l'} \to 0$ as $l,l' \to \infty$. To this end, the entries of the matrix $\widetilde{A}_j$ are generated as mutually independent $\mathcal{N}\big (0, \nu_{l,l'}\big )$ with  $\nu_{l,l'}=\exp{(-i-j)}$. We fix $p=2$, $d=15$ and set $\widetilde{A}_j=\kappa_j \widetilde{A}_j/\snorm{\widetilde{A}_j}_{\infty}, j=1,2$. The function-valued innovations $\epsilon_1, \dots, \epsilon_T$ are i.i.d. Gaussian with coefficient variances $\text{Var}(\langle \epsilon_t, \psi_l \rangle) = 2\pi\exp(-l), l=1,\ldots, d$. Results are again provided for $[a,b]=[0,\pi]$.
\begin{itemize}[leftmargin=*]
\item {\bf Scenario 5: amplitude variation.}
 The processes $\{X_t\}_{t=1}^{T}$ are generated from \eqref{eq:FAR2} with $\kappa_1=\kappa=0.7$ and $\kappa_2=0$. The alternative processes  $\{Y_t\}_{t=1}^{T}$ are generated similarly but with $\text{Var}(\langle \epsilon_t, \psi_l \rangle) = \iota \times 2\pi\exp(-l), l=1,\ldots, d$, where $\iota = 1.1^\ell, \ell =0,\ldots, 15$. In \autoref{fig:FAR1_lamp}, we present the results for a true factor $\iota=1.1^{7}$. Panel (a) displays the corresponding ERP  of the test \eqref{testOp} in Section \ref{sec31}  (difference between operators), whereas the ERP of the tests  \eqref{hol31} in Section \ref{sec33} (difference between eigenvalues for $k=1,2$) are given in panel (c) and (d), respectively. We again observe both good nominal levels and good power for sample sizes $T \ge 128$. 
  \begin{figure}[h!]
\centering
\vspace*{-10pt}
%\begin{subfigure}[b]{0.25\linewidth}
%\includegraphics[width=\linewidth]{FAR1Danielamp_1905_full_truedistances.pdf}
%\setlength{\abovecaptionskip}{-10pt}
%\setlength{\belowcaptionskip}{-8pt} \caption{}
%\end{subfigure}\hfil
\begin{subfigure}[b]{0.33\linewidth}
\includegraphics[width=\linewidth]{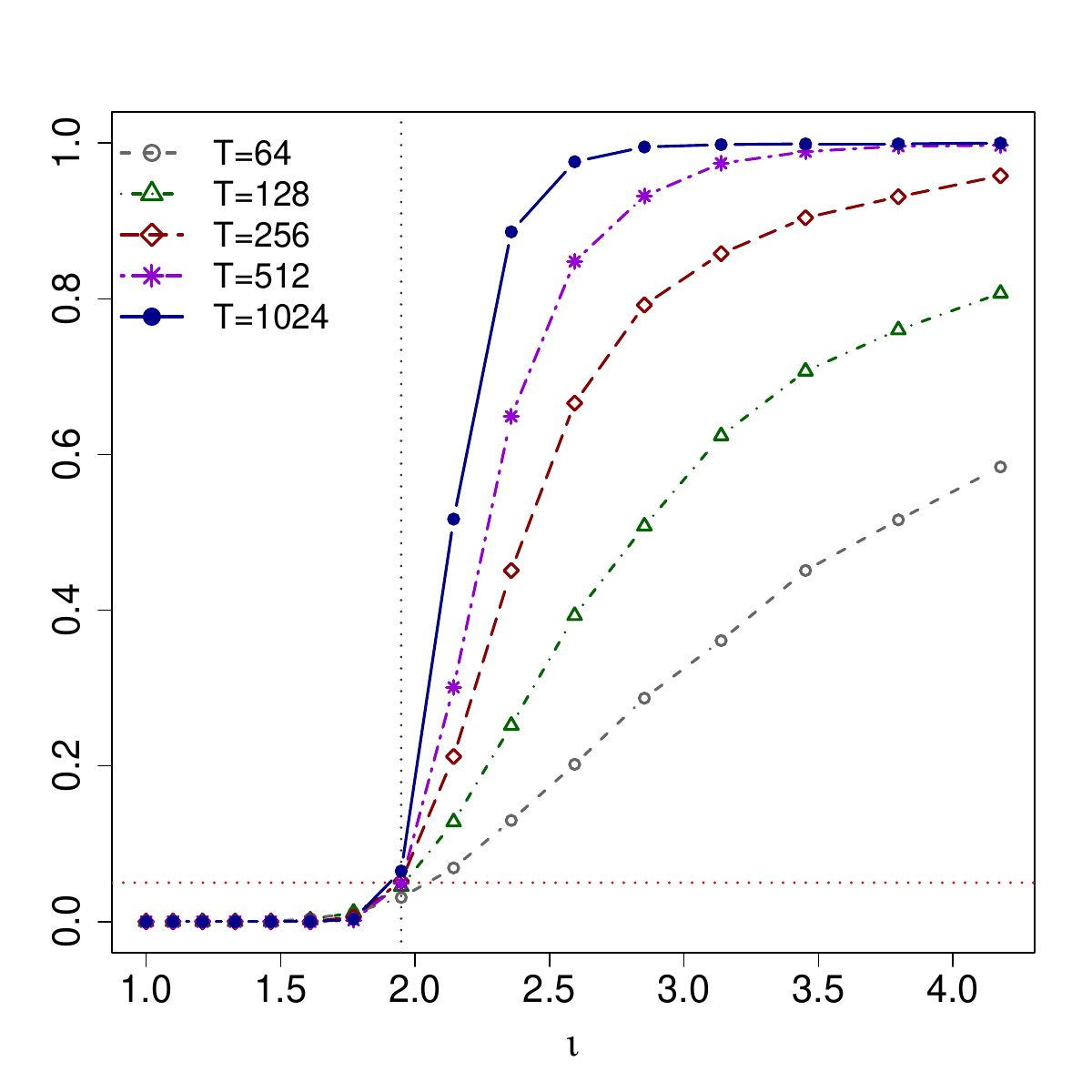}
\setlength{\abovecaptionskip}{-10pt}
\setlength{\belowcaptionskip}{-8pt} \caption{}
\end{subfigure}\hfil
\begin{subfigure}[b]{0.33\linewidth}
\includegraphics[width=\linewidth]{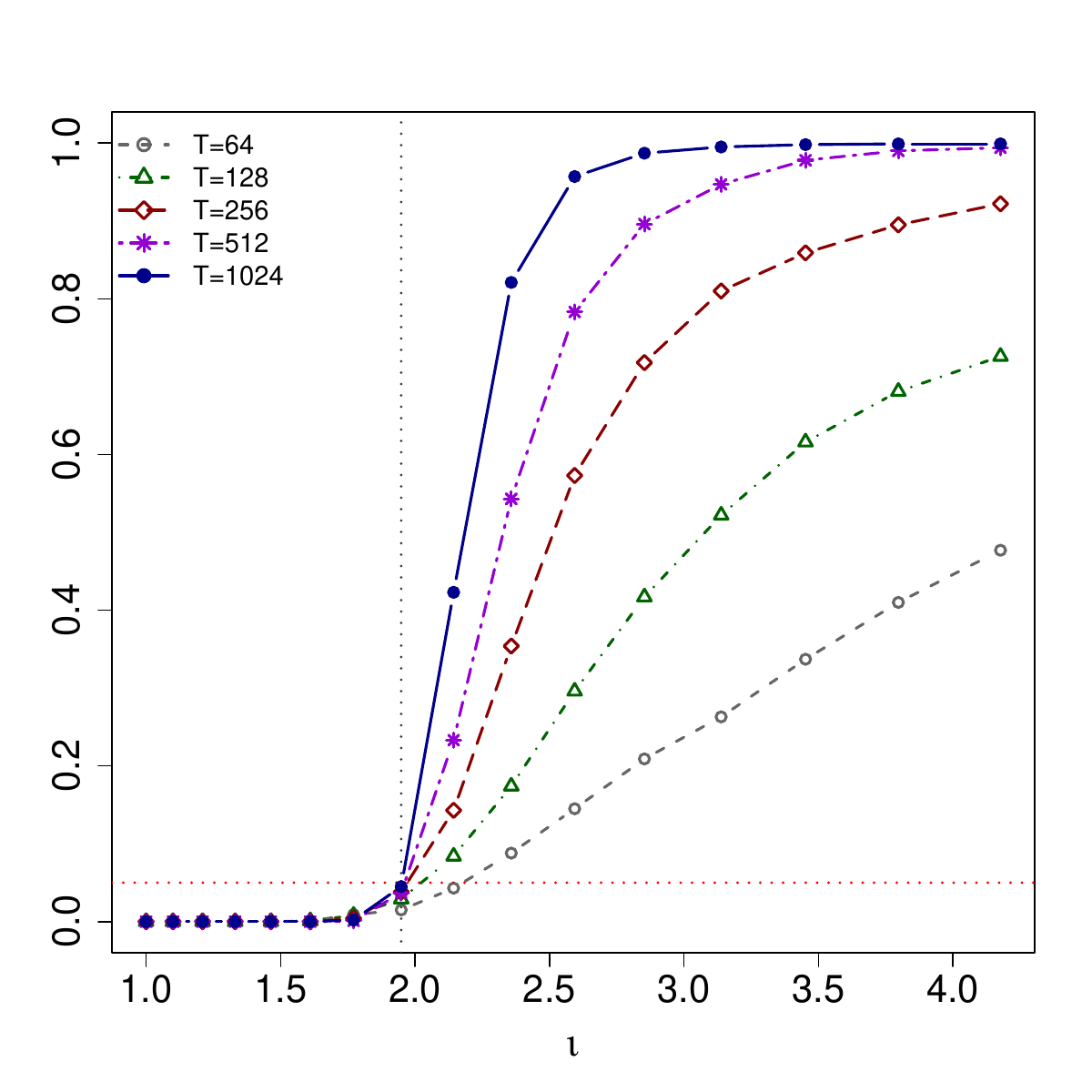}
\setlength{\abovecaptionskip}{-10pt}
\setlength{\belowcaptionskip}{-8pt} \caption{}
\end{subfigure}\hfil
\begin{subfigure}[b]{0.33\linewidth}
\includegraphics[width=\linewidth]{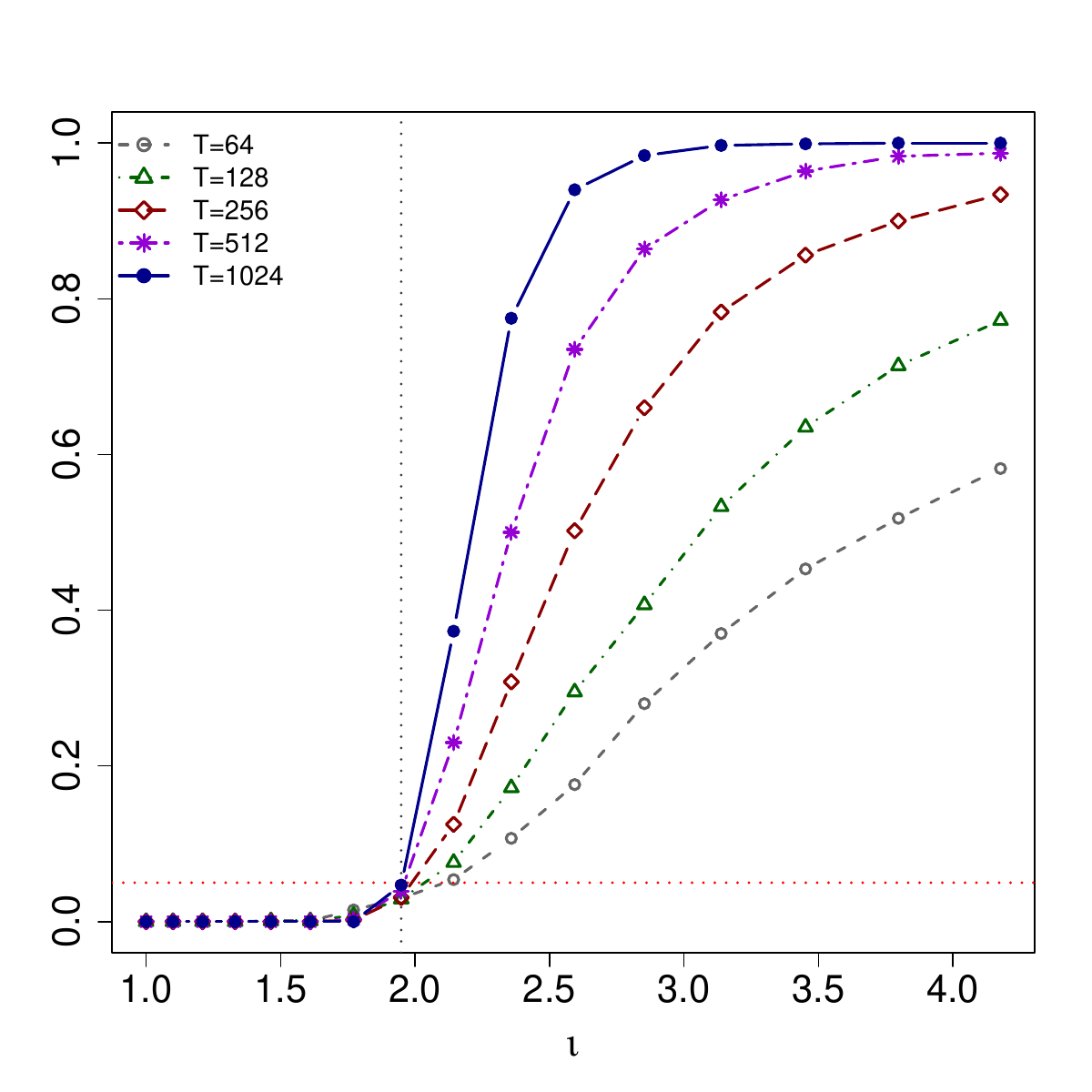}
\setlength{\abovecaptionskip}{-10pt}
\setlength{\belowcaptionskip}{-8pt} \caption{}
\end{subfigure}
\caption{\it ERP under scenario 5 for the relevant hypotheses tests introduced in Section \ref{sec31} (a) and in Section \ref{sec33} for $k=1,2$ ((b)--(c), resp.) as a function of $\iota$ at the nominal level 0.05 (horizontal dotted line). The vertical dotted line illustrates the  true factor $\iota=1.1^{7}$, which induces thresholds for $[a,b]=[0,\pi]$ of $\Delta\approx 0.237$, $\Delta_{\lambda,1} \approx 0.213$, $\Delta_{\lambda,2} \approx 0.022$, respectively. } 
\label{fig:FAR1_lamp}
\end{figure}
 \item {\bf Scenario 6: strength of dependence.}
 The processes $\{X_t\}_{t=1}^{T}$ are again simulated from \eqref{eq:FAR2} with the introduced specifications and with $\kappa_1=0.75$ and $\kappa_2=0$. The alternative processes $\{Y_t\}_{t=1}^{T}$ are generated similarly, except that we vary the value of $\kappa_2$ from $0$ to $-0.75$. Note that this means that the models have strong dependence with complex dynamics and satisfy the existence of a causal solution in the weak sense that an appropriate power of the matrix of autoregressive operators in the state space representation has operator norm less than 1 (see e.g., section 3 in \cite{vDE18}).  We let $\kappa_2=-0.2$ correspond to the boundary of the hypotheses. As illustrated in \autoref{fig:FAR2}(a), the effect on the relative  differences between the true distances and induced thresholds is most dominantly visible in the first two eigenprojectors, the spectral density operators, and to a lesser extent in the largest eigenvalues. The corresponding ERP of the proposed tests of no relevant difference over $[a,b]=[0,\pi]$ for a hypothesized value of $\kappa_2=-0.2$ are provided in panel (b)--(e). It can be seen that the tests hold the nominal level well at the boundary of the null hypothesis (dotted vertical line) and power increases steadily at the interior of the alternative, which corroborates with theory. The power  of the test for the  largest eigenvalues (panel (e)) is smaller compared to those of the eigenprojectors (panel (c)--(d), resp.), which we can again explain by looking at the relative distances in \autoref{fig:FAR2}(a) and using the approximation \eqref{hd50}.
 For the second eigenprojectors, the test is slightly oversized at the boundary which is most likely due to the fact that \autoref{as:eigsep} is not completely satisfied over the entire frequency band. The slower convergence of the test for the largest eigenvalues can be explained by the fact that there are extremely fast changing dynamics over the frequency band, which also affects the test for the spectral density operators.  
 \begin{figure}[h!]
\centering

\begin{subfigure}[b]{0.28\linewidth}
\includegraphics[width=\linewidth]{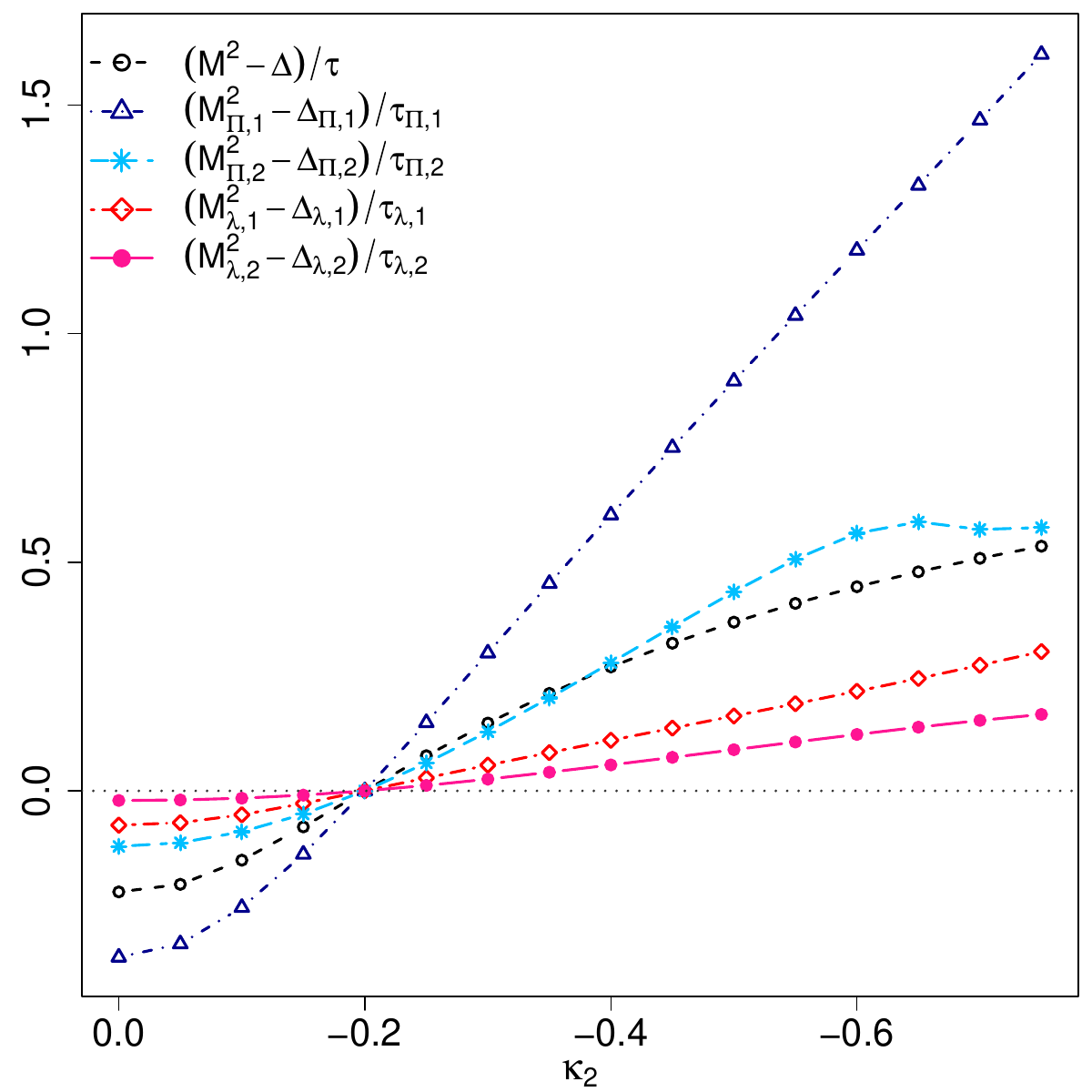}
\setlength{\abovecaptionskip}{-5pt}
\setlength{\belowcaptionskip}{0pt} \caption{}
\end{subfigure}\hfil
\begin{subfigure}[b]{0.28\linewidth}
\includegraphics[width=\linewidth]{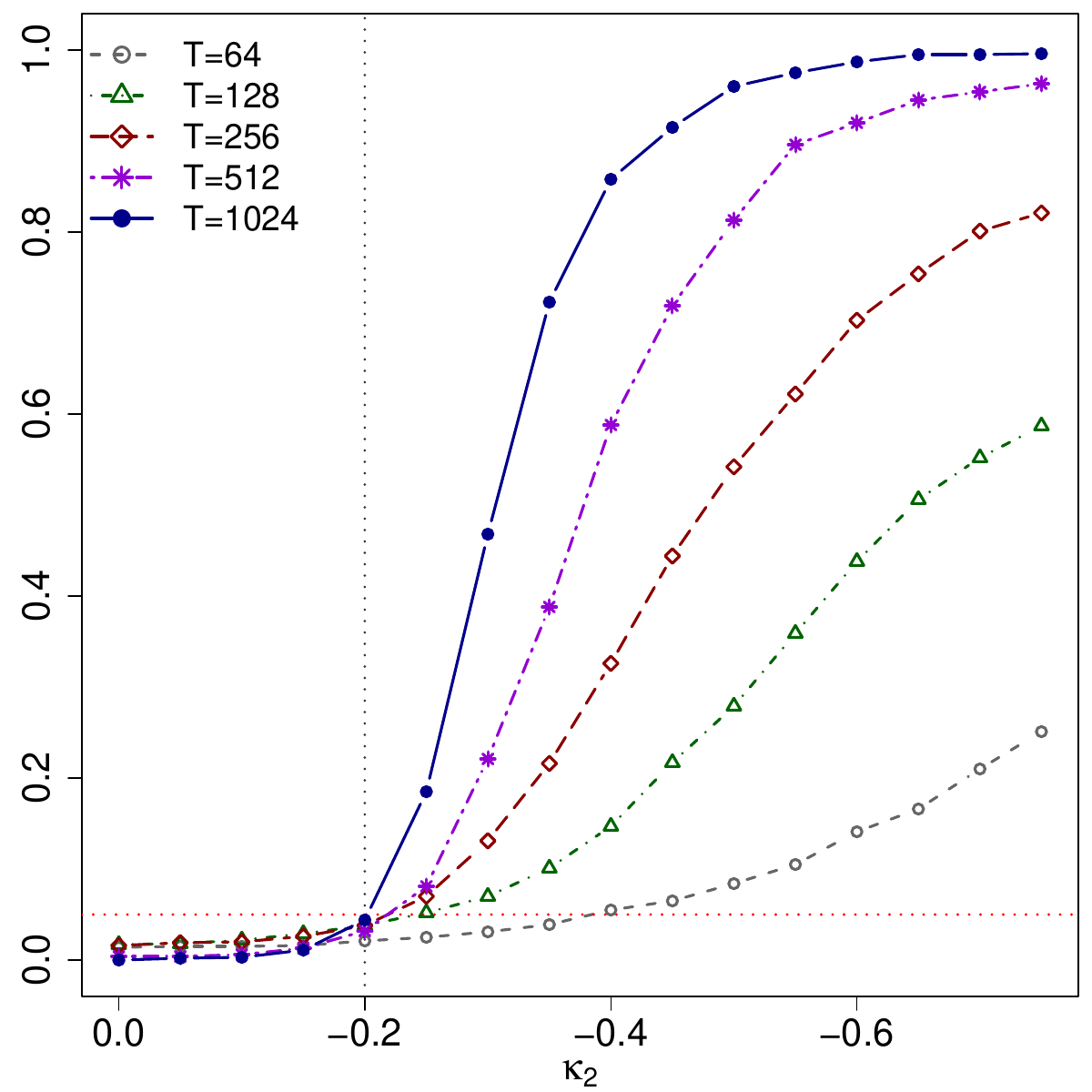}
\setlength{\abovecaptionskip}{-5pt}
\setlength{\belowcaptionskip}{0pt} \caption{}
\end{subfigure}\hfil
\begin{subfigure}[b]{0.28\linewidth}
\includegraphics[width=\linewidth]{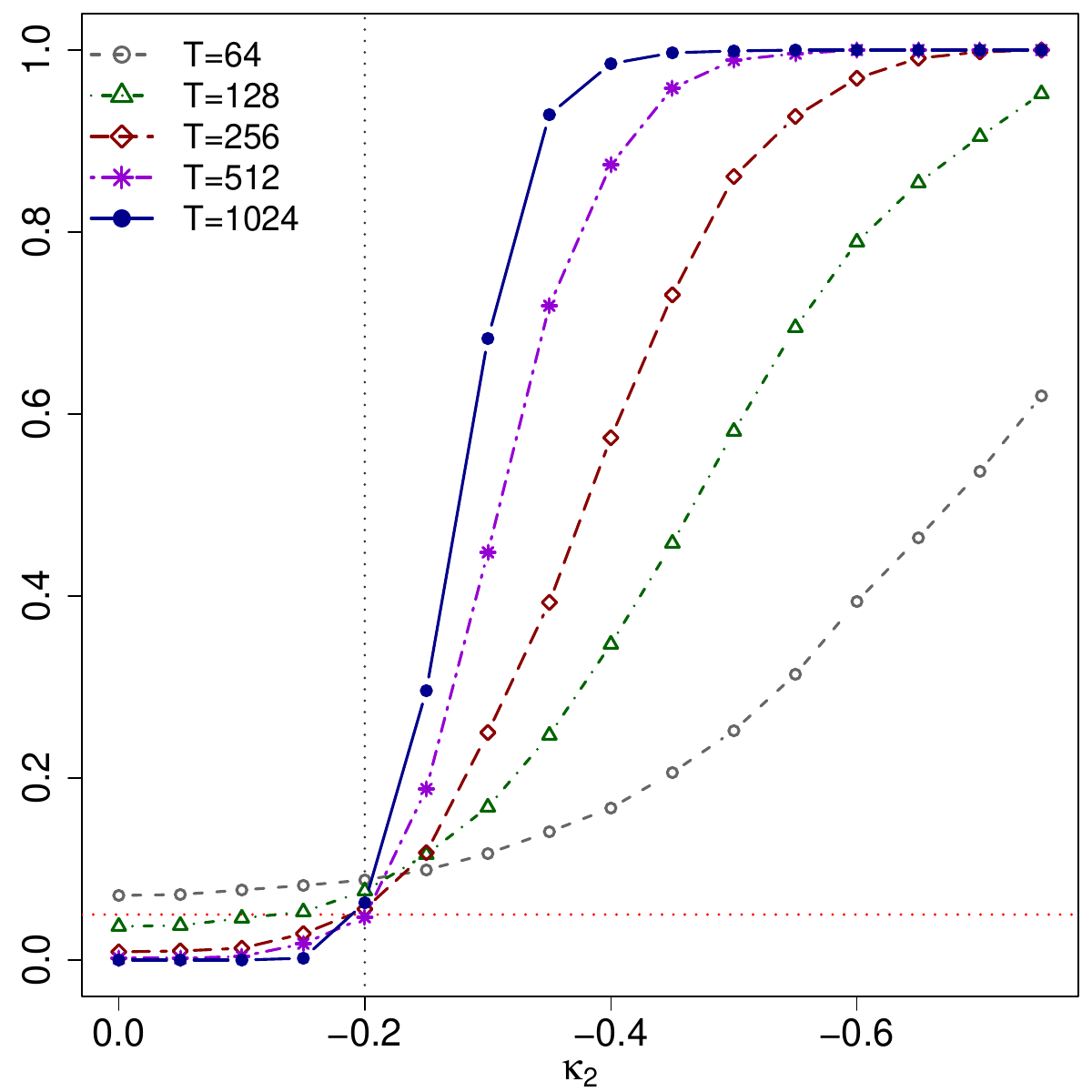}
\setlength{\abovecaptionskip}{-5pt}
\setlength{\belowcaptionskip}{0pt} \caption{}
\end{subfigure}\\
\begin{subfigure}[b]{0.28\linewidth}
\includegraphics[width=\linewidth]{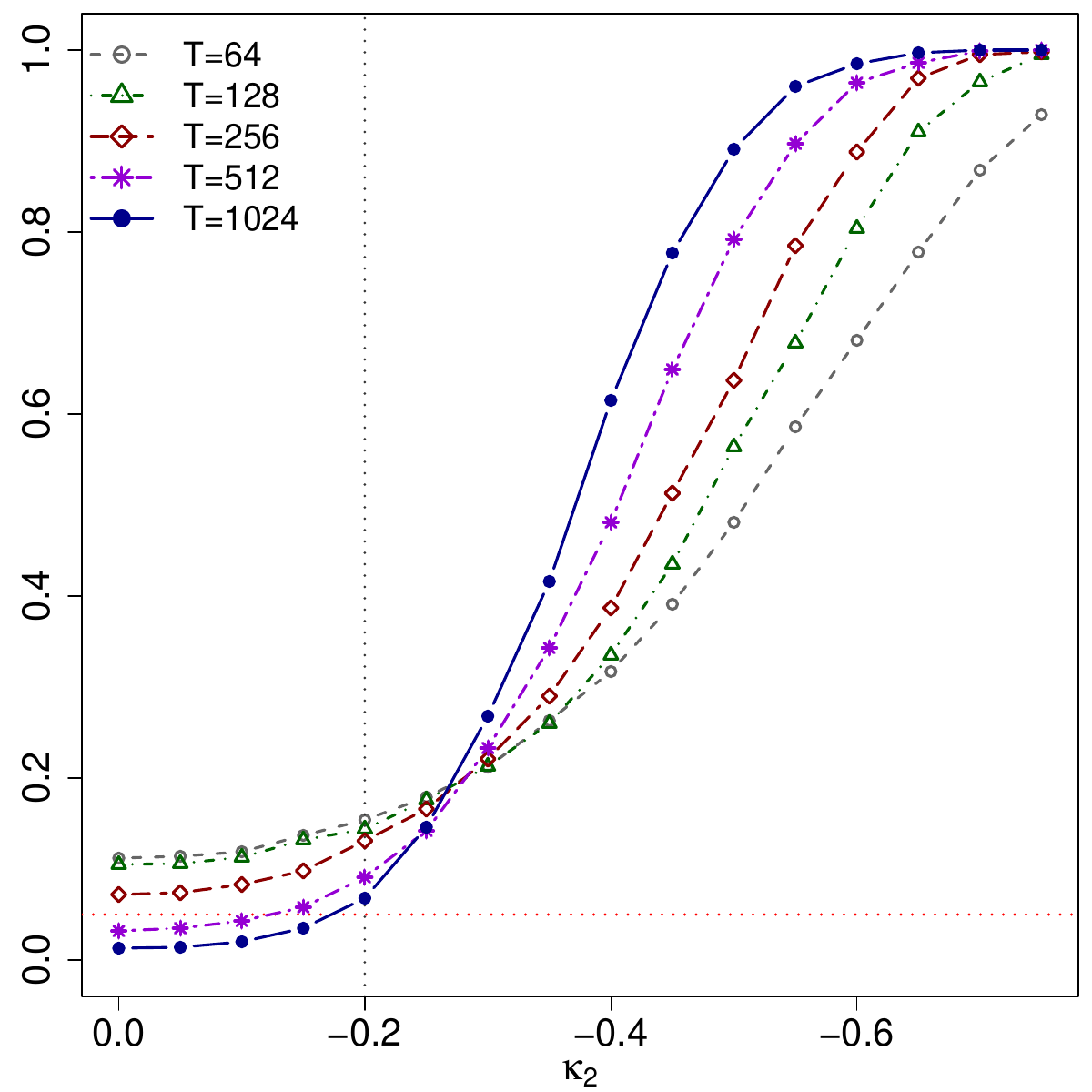}
\setlength{\abovecaptionskip}{-5pt}
\setlength{\belowcaptionskip}{-8pt} \caption{}
\end{subfigure}\hfil
\begin{subfigure}[b]{0.28\linewidth}
\includegraphics[width=\linewidth]{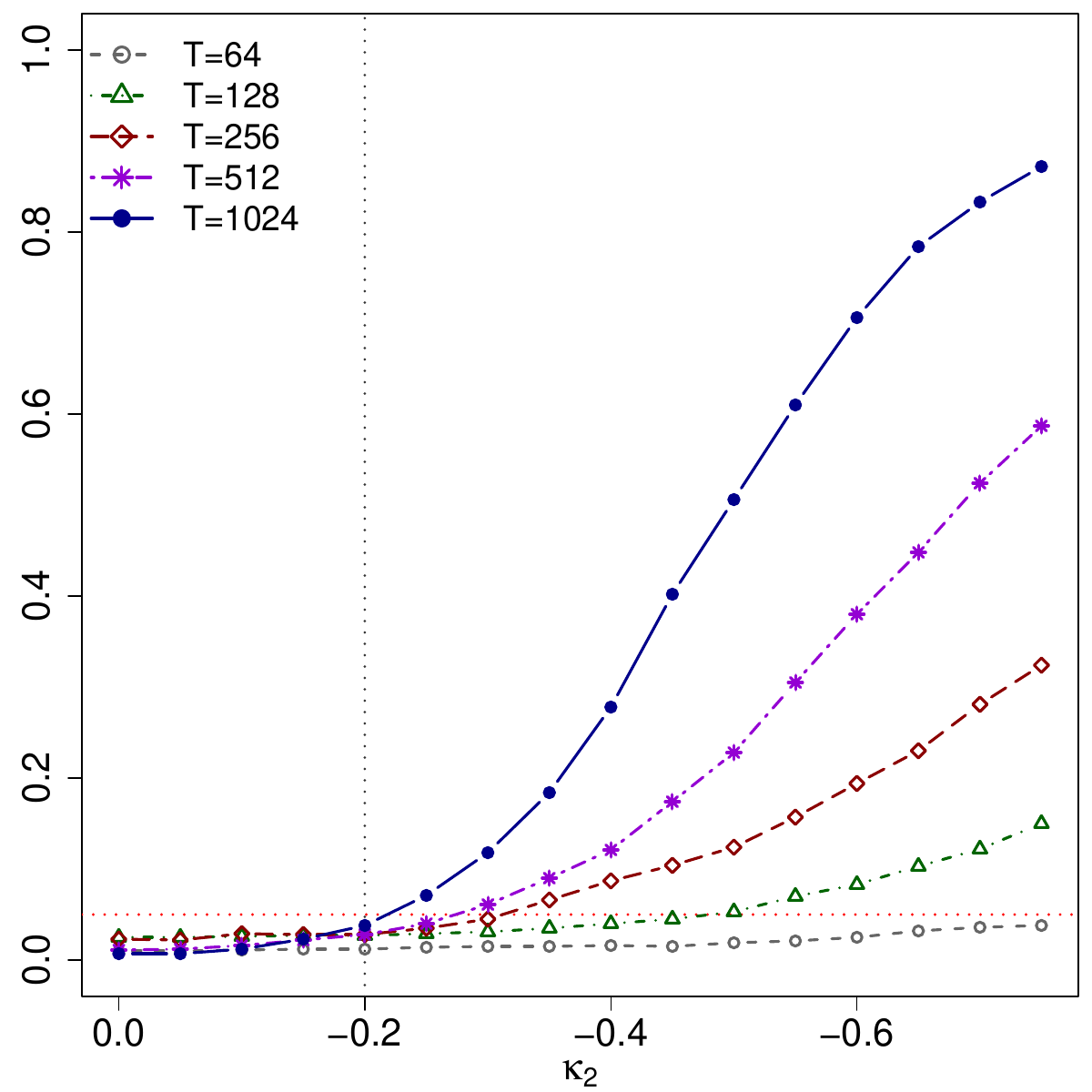}
\setlength{\abovecaptionskip}{-5pt}
\setlength{\belowcaptionskip}{-8pt} \caption{}
\end{subfigure}\hfil
\caption{\it Panel (a):  the relative differences $(M^2-\Delta )/\tau $ and  $( M^2_{\Pi,k}-\Delta_{\Pi,k})/ \tau_{\Pi}$ and $(M^2_{\lambda,k}-\Delta_{\lambda,k})/ \tau_{\lambda}$, for $k=1,2$  as a function of $\kappa_2$ under scenario 6. The  value of $\kappa_2=-0.2$
(vertical dotted line)
induces threshold values for $[a,b]=[0,\pi]$ of $\Delta\approx 0.024$,  $\Delta_{\Pi,1}=0.0314$, $\Delta_{\Pi,2}=0.0444$, $\Delta_{\lambda,1} \approx 0.008$, $\Delta_{\lambda,2} \approx 0$. Panel ((b)--(e)): ERP under scenario 6  as a function of $\kappa_2$ at the nominal level 0.05 (horizontal dotted line) of the test \eqref{testOp} (panel (b)), of the test \eqref{hol30} for $k=1,2$ (panel (c)--(d), resp.), and of the test \eqref{hol31} for $k=1$ (panel (e)).}
\label{fig:FAR2}
\end{figure}
\end{itemize}
\begin{itemize}[leftmargin=*]
\item {\bf Scenario 7: shift in the eigenfunctions}. For this scenario, we consider functional autoregressive processes of the form (see also \cite{Panaretos2010,Zhang2015,auedettrice2019})
\[
X_t(\tau) = \sum_{j=1}^2 \chi_{t,1,j} \sqrt{2} \sin(2\pi \tau j)+ \chi_{t,2,j} \sqrt{2} \cos(2\pi (\tau j +\iota_j)),  \quad \tau\in [0,1], \tageq \label{eq:AR}
\]
where the coefficients $\boldsymbol{\chi}_{t} =(\chi_{t,1,1},\chi_{t,2,1},\chi_{t,1,2},\chi_{t,2,2})$ are simulated from a vector autoregressive process, i.e., ${\boldsymbol{\chi}}_{t} =c \boldsymbol{\chi}_{t-1} +\sqrt{1-c^2} \epsilon_t$ with $\epsilon_t \in \rnum^4$. We generate $\{X_{t}\}^T_{t=1}$  from model \eqref{eq:AR} with $c=.7$, $\epsilon_t \in \rnum^4 \sim \mathcal{N}(\boldsymbol{0}, \mathrm{diag}(4,8,0.5,1.5))$, and $\iota_j=0, j =1,2$.  The alternative processes $\{Y_{t}\}^T_{t=1}$ are simulated in the same way, except that the value of $\iota_1$ is varied between 0 and 0.25.   In this model, higher degrees of persistence (i.e., larger values of $c$) lead to an increase of the relative contribution of the largest eigenvalues to the total variation, and to signal being concentrated in an increasingly narrow frequency band around zero, which results in violation of \autoref{as:eigsep}. Due to the strong persistence in this model for $c=0.7$, \autoref{as:eigsep} is not satisfied over the entire frequency band $[0,\pi]$. In order to ensure that \autoref{as:eigsep} is satisfied, we restrict the analysis to  the interval $[a,b]=[0,\pi/4]$.  We let $\iota=\iota_1 =0.025$ be the shift factor that corresponds to the boundary of the hypotheses. The induced threshold values are  $\Delta \approx 0.93$ $\Delta_{\Pi,1}=\Delta_{\Pi_2}\approx0.17$, respectively. In \autoref{fig:scenario3}, we display the ERP of the test \eqref{testOp} in Section \ref{sec31}  (difference between operators, panel (a)) and of the test \eqref{hol30} in Section \ref{sec32} (difference between projectors for $k=1,2$, panel ((b)--(c), resp.). The results closely align with theory. It is worth mentioning that for $\iota > 0.2$, the relative difference  $(M^2-\Delta)/\tau$ (as a function of $\iota$) exhibits a slight decline (similar to the phenomenon observed in scenario 3), the effect of which is discernible for the smaller sample sizes (see panel (a)).  \begin{figure}[h!]
\centering
\vspace*{-10pt}
\begin{subfigure}[b]{0.33\linewidth}
\includegraphics[width=\linewidth]{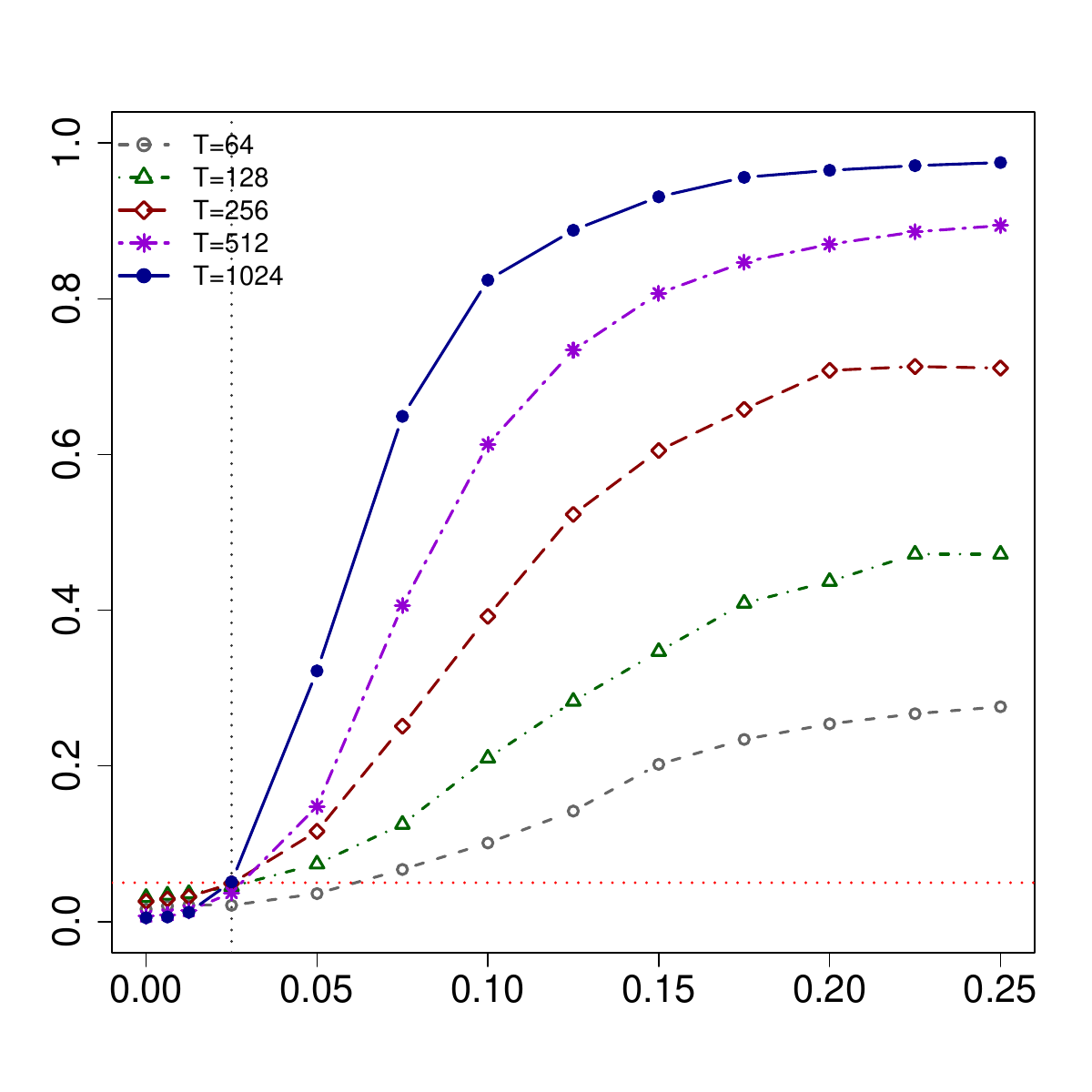}
\setlength{\abovecaptionskip}{-5pt}
\setlength{\belowcaptionskip}{-5pt} \caption{}
\end{subfigure}\hfil
\begin{subfigure}[b]{0.33\linewidth}
\includegraphics[width=\linewidth]{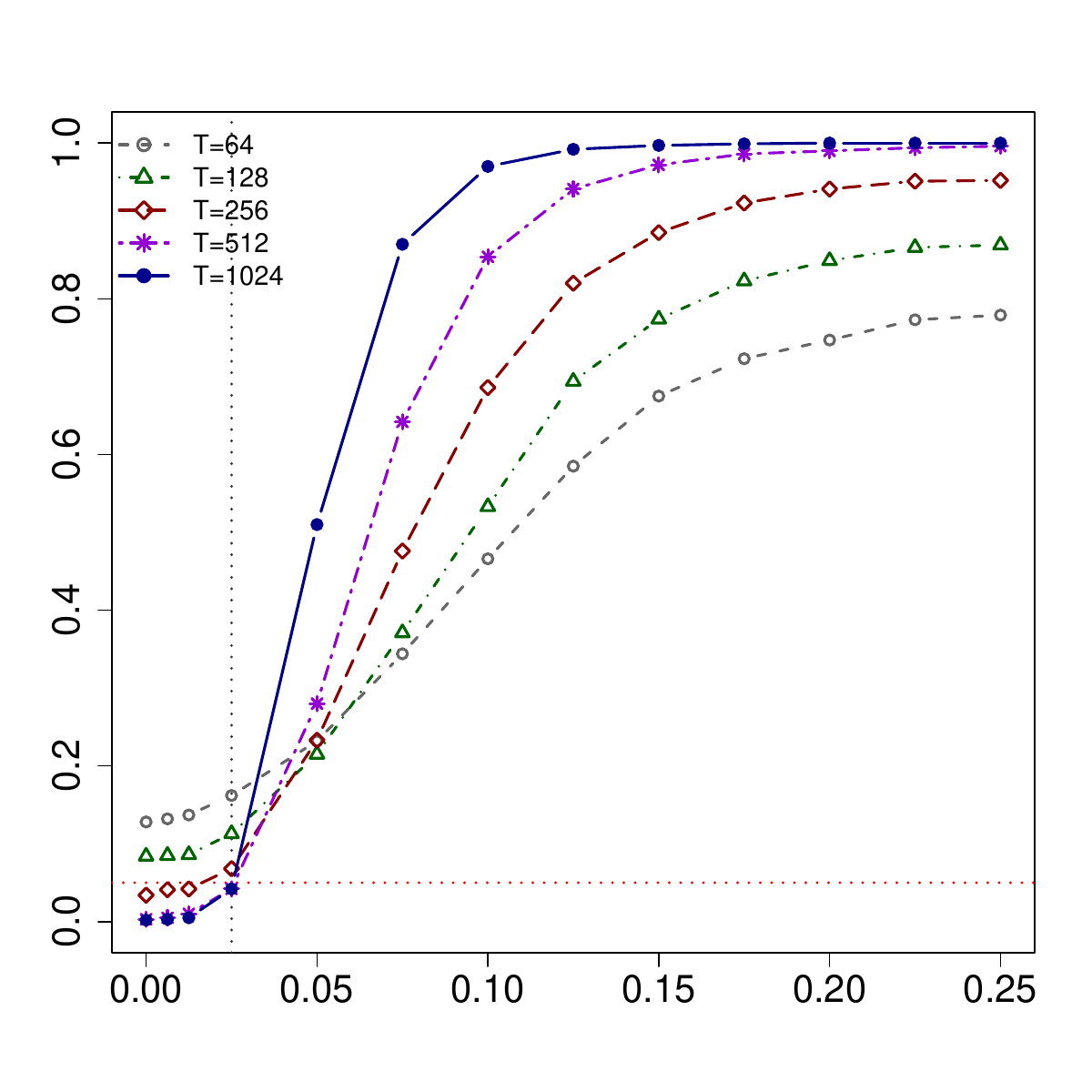}
\setlength{\abovecaptionskip}{-5pt}
\setlength{\belowcaptionskip}{-5pt} \caption{}
\end{subfigure}\hfil
\begin{subfigure}[b]{0.33\linewidth}
\includegraphics[width=\linewidth]{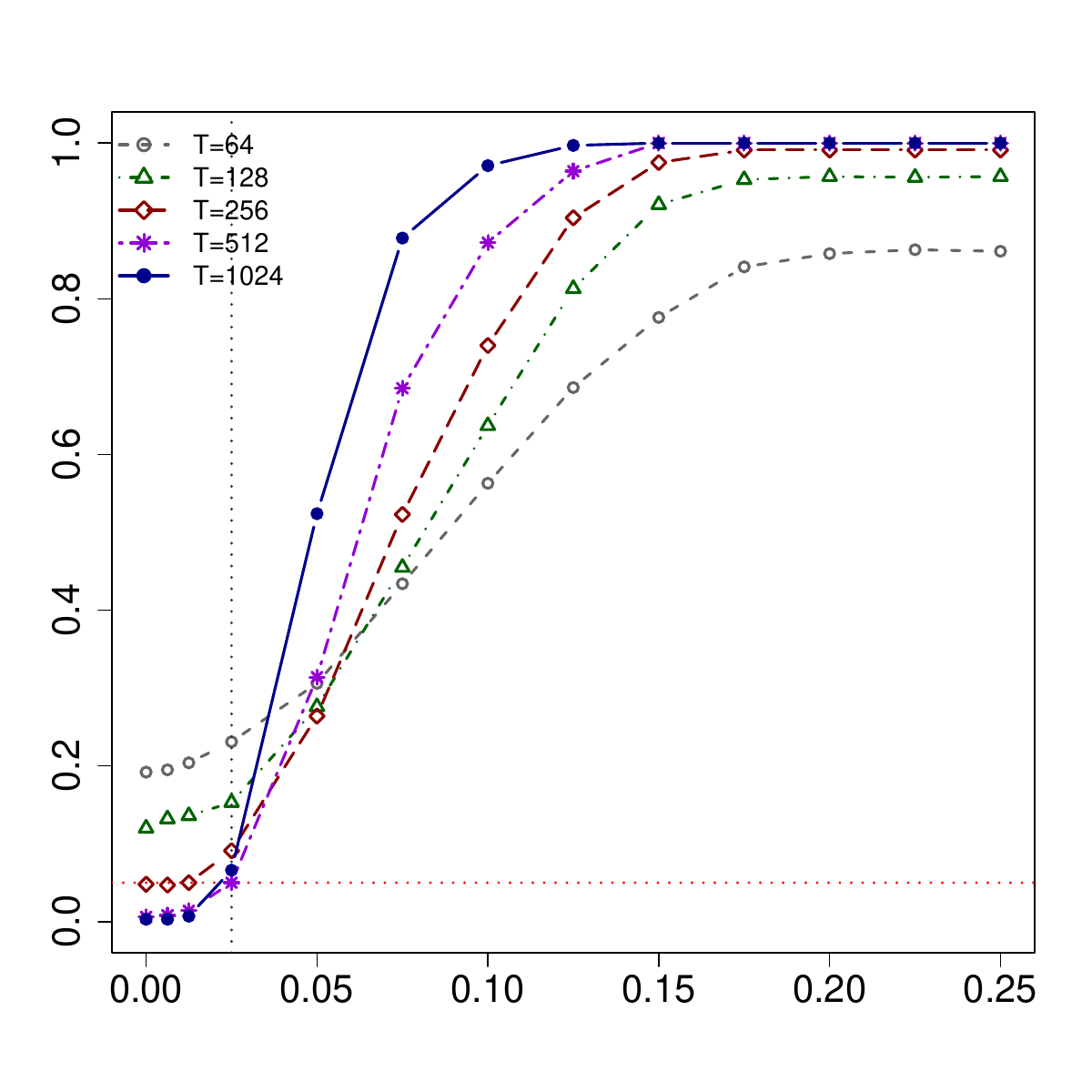}
\setlength{\abovecaptionskip}{-5pt}
\setlength{\belowcaptionskip}{-5pt} \caption{}
\end{subfigure}\hfill
%\begin{subfigure}[b]{0.25\linewidth}
%\includegraphics[width=\linewidth]{AReigfDaniel_11052021_1to25_4_5_4_20.pdf}
%\setlength{\abovecaptionskip}{-5pt}
%\setlength{\belowcaptionskip}{-5pt} \caption{}
%\end{subfigure}\hfil
\caption{\it ERP under scenario 7 as a function of $\iota$ at the nominal level 0.05 (horizontal dotted line) of the relevant hypotheses test \eqref{testOp} (panel (a)) and \eqref{hol30} for $k=1,2$ (panels (b)--(c), resp.). The vertical line illustrates the true shift factor $\iota =0.025$. The thresholds for $[a,b]=[0,\pi/4]$ induced by this shift are $\Delta \approx 0.93$, $\Delta_{\Pi,1}\approx \Delta_{\Pi,2}\approx 0.17$, respectively.}
\label{fig:scenario3}
\end{figure}
\end{itemize}

\section{Proof of \autoref{thm:BuildingBlock}} \label{sec:sec4} \label{sec:proofBB}
\def\theequation{5.\arabic{equation}}
\setcounter{equation}{0}
In this section, we prove \autoref{thm:BuildingBlock} and provide exact expressions of the constant $\tau_{XY}$ in terms of the spectral density operators and the factorizations of $\mathcal{U}^{(\omega)}_{XY}$ and $\mathcal{U}^{(\omega)}_{YX}$. We remark once more that the exact expression of $\mathcal{U}^{(\omega)}_{XY}$ depends on the specific hypothesis under consideration   and that both are allowed to depend on both component processes $X$ and $Y$.
In the following, consider 
\begin{equation} 
\mathcal{V}^{(\omega)}_{\X} = (\mathcal{V}^{(\omega)}_{X},\mathcal{V}^{(\omega)}_{Y})^{\top}
 \label{eq:mathV}
 \end{equation}
 where $\mathcal{V}^{(\omega)}_{X}, \mathcal{V}^{(\omega)}_{Y}$ are arbitrary elements of  $S_2(\Hi)$ and where the subscript only refers to the index of the component. Note that $\mathcal{V}^{(\omega)}_{\X}$ is an element of the Hilbert space $\Hs= S_2(\Hi) \oplus S_2(\Hi)$. 
In order to prove \autoref{thm:BuildingBlock}, we prove the following statement in detail. 
 \begin{thm}\label{thm:ZXdistint}
Suppose Assumptions \ref{as:depstruc}-\ref{as:ratiorates} are satisfied. Let $\mathcal{V}^{(\omega)}_{\X}\in \Hs$ be defined by \eqref{eq:mathV}, $\omega \in \rnum$,  
and define
 \[\mathcal{Z}^{\X,\omega}_{T,\eta}= \big( \mathcal{Z}^{X,\omega}_{T_1,\eta}/\sqrt{b_1 T_1},\mathcal{Z}^{Y,\omega}_{T_2,\eta}/\sqrt{b_2 T_2}\big)^\top \in \Hs\] where $ \mathcal{Z}^{X,\omega}_{T_1,\eta}$ and $\mathcal{Z}^{Y,\omega}_{T_2,\eta}$ are given by \eqref{eq:Zx} and \eqref{eq:Zy}, respectively. 
\begin{itemize} 
\item[i).] If the component processes of $\{\X_t\}$ are dependent, then
\begin{align*}
& \Big\{\eta \sqrt{b_1 T_1+b_2 T_2}\Big(\int_a^b\Re\Big(\biginprod{ \mathcal{Z}^{\X,\omega}_{T,\eta} }{\mathcal{V}^{(\omega)}_{\X}}_{\Hs} \Big)d\omega\Big\}_{\eta \in [0,1]}\,{\rightsquigarrow}\,  \Bigg\{\tau_{XY} \eta \, \mathbb{B}(\eta) \Bigg \}_{\eta \in [0,1]}, 
\end{align*}
where $\mathbb{B} $ is a Brownian motion,
 $\tau^2_{XY}=\frac{1}{4}\int^{b}_a \Re\Big[\Gamma^{\omega}_{\X}\big(\mathcal{V}^{(\omega)}_{\X}\big)+\Sigma^{\omega}_{\X}\big(\mathcal{V}^{(\omega)}_{\X}\big)\Big]d\omega$ 
 with 
\begin{align*}
 \Gamma^{\omega}_{\X}\big(\mathcal{V}^{(\omega)}_{\X}\big) 
 & =4\pi^2  \sum_{l,j \in \{X,Y\}} \Big( \biginprod{ \big(\F_{l,j}^{(\omega)} \widetilde{\otimes} \F_{j,l}^{\dagger(\omega)} \big)(\mathcal{V}^{(\omega)}_{j})}{\mathcal{V}^{(\omega)}_{l}}_{S_2}  +\mathrm{1}_{ \omega 
 \in \{0,\pi\}} \biginprod{ \big(\F_{l,j}^{(\omega)} \widetilde{\otimes} \F_{j,l}^{\dagger(\omega)} \big)(\mathcal{V}^{(\omega)}_{j})}{\overline{\mathcal{V}^{\dagger(\omega)}_{l}}}_{S_2}  \\&\phantom{ \frac{1}{\mathcal{W}^{2}_{b_T}}} 
 +\biginprod{ \big(\F_{j,l}^{(\omega)} \widetilde{\otimes} \F_{l,j}^{\dagger(\omega)} \big)(\mathcal{V}^{\dagger(\omega)}_{l})}{\mathcal{V}^{\dagger(\omega)}_{j}}_{S_2} +\mathrm{1}_{ \omega \in \{0,\pi\}} \biginprod{\big(\F_{j,l}^{(\omega)} \widetilde{\otimes} \F_{l,j}^{\dagger(\omega)} \big)({\mathcal{V}^{(\omega)}_{j}})}{\overline{(\mathcal{V}^{\dagger(\omega)}_{j}})}_{S_2}\Big)
 %&=4\pi^2\sum_{l,j \in \{X,Y\}}\Big(\inprod{\F^{(\omega)}_{lj}(u_j)}{u_{l,i}}  \inprod{\F^{(\omega)}_{jl}(v_{l,i})}{v_j} +\inprod{\F^{(\omega)}_{jl}(u_{l,i})}{u_j}  \inprod{\F^{(\omega)}_{lj}(v_j)}{v_{l,i}} \\& \qquad \qquad\qquad+\mathrm{2}_{\{0,\pi\}}\big(\inprod{\F^{(\omega)}_{lj}(u_j)}{u_{l,i}}  \inprod{\F^{(\omega)}_{lj}(v_j)}{v_{l,i}} \big)\Big), 
\shortintertext{and}
 \Sigma^{\omega}_{\X}\big(\mathcal{V}^{(\omega)}_{\X}\big) &=
 4\pi^2  \sum_{l,j \in \{X,Y\}}  \Big(\mathrm{1}_{\omega \in \{0,\pi\}}\Big(  \biginprod{ \big(\F_{l,j}^{(\omega)} \widetilde{\otimes} \F_{j,l}^{\dagger(\omega)} \big)(\overline{\mathcal{V}^{(\omega)}_{j}})}{\mathcal{V}^{(\omega)}_{l}}_{S_2} + \biginprod{ \big(\F_{l,j}^{(\omega)} \widetilde{\otimes} \F_{j,l}^{\dagger(\omega)} \big)(\mathcal{V}^{\dagger(\omega)}_{j})}{\mathcal{V}^{(\omega)}_{l}}_{S_2} 
 \\& +
\mathrm{1}_{\omega \in \{0,\pi\}}  \biginprod{ \big(\F_{j,l,m}^{(\omega)} \widetilde{\otimes} \F_{l,j}^{\dagger(\omega)} \big)({\mathcal{V}^{\dagger(\omega)}_{l}})}{\overline{\mathcal{V}^{\dagger (\omega)}_{j}}}_{S_2}   + \biginprod{ \big(\F_{j,l}^{(\omega)} \widetilde{\otimes} \F_{l,j}^{\dagger(\omega)} \big)({\mathcal{V}^{\dagger(\omega)}_{l}})}{{\mathcal{V}^{(\omega)}_{j}}}_{S_2}\Big)
 % =4\pi^2 \sum_{l,j \in \{X,Y\}}\Big(\inprod{\F^{(\omega)}_{jl}(u_{l,i})}{u_j}  \inprod{\F^{(\omega)}_{lj}(v_j)}{v_{l,i}} +\inprod{\F^{(\omega)}_{lj}(u_j)}{u_{l,i}}  \inprod{\F^{(\omega)}_{jl}(v_{l,i})}{v_j}\\&\qquad \qquad\qquad+\mathrm{2}_{\{0,\pi\}}\big(\inprod{\F^{(\omega)}_{lj}(u_j)}{u_{l,i}}  \inprod{\F^{(\omega)}_{lj}(v_j)}{v_{l,i}}\big)\Big)~;
 \end{align*}
 \item[ii).] If the component processes of $\{\X_t\}$ are independent, then
 \begin{align*}
& \Big\{\eta \sqrt{b_1 T_1+b_2 T_2}\Big(\int_a^b\Re\Big(\biginprod{ \mathcal{Z}^{\X,\omega}_{T,\eta} }{\mathcal{V}^{(\omega)}_{\X}}_{\Hs} \Big)d\omega\Big\}_{\eta \in [0,1]}\,{\rightsquigarrow}\, \Bigg\{ \eta  \Big(\frac{\upsilon_{X}}{\sqrt{\theta}}\, \mathbb{B}_X(\eta) +\frac{\upsilon_{Y}}{\sqrt{1-\theta}}\, \mathbb{B}_Y(\eta) \Big)\Bigg \}_{\eta \in [0,1]}%\overset{=}{dist} \Bigg\{\tau_{XY} \eta \, \mathbb{B}(\eta) \Bigg \}_{\eta \in [0,1]}, 
\end{align*}
where  $\mathbb{B}_X$ and $\mathbb{B}_Y$ are independent Brownian motions and 
$\upsilon^2_{l}=\frac{1}{4}\int^{b}_a \Re\Big[\Gamma^{\omega}_{l}\big(\mathcal{V}^{(\omega)}_{l}\big)+\Sigma^{\omega}_{l}\big(\mathcal{V}^{(\omega)}_{l}\big)\Big]d\omega$, $l \in \{X,Y\}$, with 
\begin{align*}\Gamma^{(\omega)} _{l}(\mathcal{V}^{(\omega)}_{l})&= { 4\pi^2}
 \Big( \biginprod{ \big(\F_{l}^{(\omega)} \widetilde{\otimes} \F_{l}^{(\omega)} \big)(\mathcal{V}^{(\omega)}_{l})}{\mathcal{V}^{(\omega)}_{l}}_{S_2}  +\mathrm{1}_{\omega \in \{0,\pi\}} \biginprod{ \big(\F_{l}^{(\omega)} \widetilde{\otimes} \F_{l}^{(\omega)} \big)(\mathcal{V}^{(\omega)}_{l})}{\overline{\mathcal{V}^{\dagger(\omega)}_{l}}}_{S_2}  \\& 
 \phantom{ 4\pi^2}
+\biginprod{ \big(\F_{l}^{(\omega)} \widetilde{\otimes} \F_{l}^{(\omega)} \big)(\mathcal{V}^{\dagger(\omega)}_{l})}{\mathcal{V}^{\dagger(\omega)}_{l}}_{S_2} +\mathrm{1}_{\omega \in \{0,\pi\}} \biginprod{\big(\F_{l}^{(\omega)} \widetilde{\otimes} \F_{l}^{(\omega)} \big)({\mathcal{V}^{(\omega)}_{l}})}{\overline{(\mathcal{V}^{\dagger(\omega)}_{l}})}_{S_2}\Big)
\shortintertext{and}\Sigma^{(\omega)} _{l}(\mathcal{V}^{(\omega)}_{l})= &
 4\pi^2  \Big(\mathrm{1}_{\omega \in \{0,\pi\}}\Big(  \biginprod{ \big(\F_{l}^{(\omega)} \widetilde{\otimes} \F_{l}^{(\omega)} \big)(\overline{\mathcal{V}^{(\omega)}_{l}})}{\mathcal{V}^{(\omega)}_{l}}_{S_2} + \biginprod{ \big(\F_{l}^{(\omega)} \widetilde{\otimes} \F_{l}^{(\omega)} \big)(\mathcal{V}^{\dagger(\omega)}_{l})}{\mathcal{V}^{(\omega)}_{l}}_{S_2} 
 \\&  \phantom{ 4\pi^2}+
\mathrm{1}_{\omega \in \{0,\pi\}}  \biginprod{ \big(\F_{l}^{(\omega)} \widetilde{\otimes} \F_{l}^{(\omega)} \big)({\mathcal{V}^{\dagger(\omega)}_{l}})}{\overline{\mathcal{V}^{\dagger (\omega)}_{l}}}_{S_2}   + \biginprod{ \big(\F_{l}^{(\omega)} \widetilde{\otimes} \F_{l}^{(\omega)} \big)({\mathcal{V}^{\dagger(\omega)}_{l}})}{{\mathcal{V}^{(\omega)}_{l}}}_{S_2}\Big).
 \end{align*}
 \end{itemize} 
\end{thm}
Note that the expression for the covariance and pseudo-covariance in both parts of the statement can be further simplified if we further assume that the $\mathcal{V}^{(\omega)}_{l}, l,j \in\{X,Y\}$ are self-adjoint, which is the case in all our statements. 
\begin{proof}[{Proof of \autoref{thm:ZXdistint}}]
The proof is involved and relies on several auxiliary results, which can be found in \autoref{sec:proofs}. We will only prove   part $i)$. The proof under independence follows similarly by verifying the steps for each component process separately and using the independence to conclude it for the linear combination. Using \autoref{as:ratiorates}, we can ease notation in the dependent scenario and write $T=T_1 =T_2$ throughout this section. Since we only assume very mild moment conditions, it is not obvious how to obtain the distributional properties directly. The principal idea is therefore to construct an approximating process of which the distributional properties can be established and then show that the process limiting distribution is the same as for the approximating process. Before we can introduce this process,
we require some necessary terminology. Let
\begin{align*}
\X^{(m)}_{t}= \E[\X_t| \sigma(\epsilon_t, \epsilon_{t-1},\ldots,\epsilon_{t-m})].
\end{align*}
and define the $\Hi^{\oplus_2}$-valued stochastic process
\[
D^{(\omega)}_{m,k} = (\dmp{k}{X}, \dmp{k}{Y})^\top =  \frac{1}{\sqrt{2\pi}} \sum_{t=0}^{\infty} \big(\E[\X^{(m)}_{t+k}|\G_k]-\E[\X^{(m)}_{t+k}|\G_{k-1}]\big)e^{-\im t \omega},
\tageq \label{eq:Dm}
\]
where  $\G_k = \sigma (\epsilon_k,  \ldots\epsilon_{k-1} , \ldots )$. %and where the conditional expectations should be understood in the sense of Bochner integrals. 
Under Conditions \ref{I}-\ref{II}, this process is an $m$-dependent stationary martingale difference sequence w.r.t. the filtration $\{\G_k  \} $  in $\op^p_{\Hi^{\oplus_2}}$ for each $\omega \in [-\pi,\pi]$.  Additionally, consider a process defined by
 \begin{align}
\ldm{}{T} = \sum_{t=2}^{T} \sum_{s=1}^{t-1}  \tilde{w}^{(\omega)}_{b_T,t,s} \mathcal{D}^{(\omega)}_{XY,s,t}, \tageq \label{eq:ldm}
\end{align}
where $\mathcal{D}^{(\omega)}_{XY,s,t} =\Big(\dmp{t}{X}\otimes \dmp{s}{X}, \dmp{s}{Y} \otimes \dmp{t}{Y}\Big)^\top $ and denote  $\mathcal{W}^2_{b_T}=\sum_{t,s=1}^{T} |\tilde{w}^{(\omega)}_{b_T,t,s}|^2$. Under \autoref{as:depstruc}, the process $\{\mathcal{W}^{-1}_{b_T} \ldm{}{T}\}$ is a martingale in $\op^{p/2}_{\Hs}$ with respect to the filtration $\{\G_T\}$ for $p \ge 4$.
The above claims on the properties of \eqref{eq:Dm} and \eqref{eq:ldm} can be verified similar to Proposition 3.2 and 3.3 of  \cite{vD19} by noting that for any $f =(f_1, f_2)^\top \in \Hs$, we have 
\[ \|f\|_{\Hs} = \snorm{f_1}_{S_2(\Hi)}+ \snorm{f_2}_{S_2(\Hi)}\tageq \label{eq:Hsnorm}\] 
To construct the required approximating process, consider the arrays
\[N^{(\omega)}_{m,T,t} = \sum_{s=1}^{t-1}\biginprod{  \tilde{w}^{(\omega)}_{b_T,t,s}\mathcal{D}^{(\omega)}_{XY,s,t}}{\mathcal{V}^{(\omega)}_{\X}}_{\Hs}+\sum_{s=1}^{t-1}\biginprod{\Big(  \tilde{w}^{(\omega)}_{b_T,t,s} \mathcal{D}^{(\omega)}_{XY,s,t}\Big)^{\dagger}}{\mathcal{V}^{(\omega)}_{\X}}_{\Hs}, \quad 2 \le t \le T, \tageq \label{eq:NmT}\]
and set $N^{(\omega)}_{m,T,1}=0$ for all $\omega \in [-\pi,\pi]$. The following theorem provides the distributional properties of the (scaled) partial sum of the real part of \eqref{eq:NmT} integrated over the frequency band $[a,b]$. The proof is involved and therefore postponed to Section \ref{proofthmNBMf} of the supplement. 
\begin{thm} \label{thm:NBMf}
Let  $N^{(\omega)}_{m,T,t}$ be defined as in \eqref{eq:NmT} and let $\mathcal{{V}}^{(\omega)}_{\X}$
be defined as in \eqref{eq:mathV}. Suppose that Assumptions \ref{as:depstruc}-\ref{as:bwrates} hold. Then, for fixed $m$,
\[
\frac{1}{\mathcal{W}_{b_T}}  \Big \{\sum_{t=1}^{\flo{\eta T}}\Re\Big(\int_a^b N^{(\omega)}_{m,T,t} d\omega\Big)  \Big\}_{\eta \in [0,1]}\, \rightsquigarrow \,\tau_{m,XY} \cdot \Big\{\eta \mathbb{B}(\eta)\Big \}_{\eta \in [0,1]}, \quad (T \to \infty)
\]
where $\tau^2_{m,XY}=\frac{1}{4}\int^{b}_a \Re\Big[\Gamma^{\omega}_{\X,m}\big(\mathcal{{V}}^{(\omega)}_{\X}\big)+\Sigma^{\omega}_{\X,m}\big(\mathcal{{V}}^{(\omega)}_{\X}\big)\Big]d\omega$ with 
\begin{align*}
\Gamma^{\omega}_{\X,m}\big(\mathcal{{V}}^{(\omega)}_{\X}\big) &=4\pi^2  \sum_{l,j \in \{X,Y\}}  \biginprod{ \big(\F_{l,j,m}^{(\omega)} \widetilde{\otimes} \F_{j,l,m}^{\dagger(\omega)} \big)(\mathcal{V}^{(\omega)}_{j})}{\mathcal{V}^{(\omega)}_{l}}_{S_2}  +\mathrm{1}_{\omega \in \{0,\pi\}} \biginprod{ \big(\F_{l,j,m}^{(\omega)} \widetilde{\otimes} \F_{j,l,m}^{\dagger(\omega)} \big)(\mathcal{V}^{(\omega)}_{j})}{\overline{\mathcal{V}^{\dagger(\omega)}_{l}}}_{S_2}  \\&\phantom{ \frac{1}{\mathcal{W}^{2}_{b_T}}} 
 +\biginprod{ \big(\F_{j,l,m}^{(\omega)} \widetilde{\otimes} \F_{l,j,m}^{\dagger(\omega)} \big)(\mathcal{V}^{\dagger(\omega)}_{l})}{\mathcal{V}^{\dagger(\omega)}_{j}}_{S_2} +\mathrm{1}_{\omega \in \{0,\pi\}} \biginprod{\big(\F_{j,l,m}^{(\omega)} \widetilde{\otimes} \F_{l,j,m}^{\dagger(\omega)} \big)({\mathcal{V}^{(\omega)}_{j}})}{\overline{(\mathcal{V}^{\dagger(\omega)}_{j}})}_{S_2}
%4\pi^2\sum_{l,j \in \{X,Y\}}\Big(\inprod{\F^{(\omega)}_{lj,m}(u_j)}{u_{l,i}}  \inprod{\F^{(\omega)}_{jl,m}(v_{l,i})}{v_j} +\inprod{\F^{(\omega)}_{jl,m}(u_{l,i})}{u_j}  \inprod{\F^{(\omega)}_{lj,m}(v_j)}{v_{l,i}} \\& \qquad \qquad\qquad+\mathrm{2}_{\{0,\pi\}}\big(\inprod{\F^{(\omega)}_{lj,m}(u_j)}{u_{l,i}}  \inprod{\F^{(\omega)}_{lj,m}(v_j)}{v_{l,i}} \big)\Big), 
\shortintertext{and}
\Sigma^{\omega}_{\X,m}\big(\mathcal{{V}}^{(\omega)}_{\X}\big) &=4\pi^2  \sum_{l,j \in \{X,Y\}}  \Big( \biginprod{ \big(\F_{l,j,m}^{(\omega)} \widetilde{\otimes} \F_{j,l,m}^{\dagger(\omega)} \big)(\mathcal{V}^{\dagger(\omega)}_{j})}{\mathcal{V}^{(\omega)}_{l}}_{S_2} + \biginprod{ \big(\F_{j,l,m}^{(\omega)} \widetilde{\otimes} \F_{l,j,m}^{\dagger(\omega)} \big)({\mathcal{V}^{\dagger(\omega)}_{l}})}{{\mathcal{V}^{(\omega)}_{j}}}_{S_2}
\\& \phantom{4\pi^2  \sum_{l,j \in \{X,Y\}} }\mathrm{1}_{\omega \in \{0,\pi\}}  \biginprod{ \big(\F_{l,j,m}^{(\omega)} \widetilde{\otimes} \F_{j,l,m}^{\dagger(\omega)} \big)(\overline{\mathcal{V}^{(\omega)}_{j}})}{\mathcal{V}^{(\omega)}_{l}}_{S_2} 
+ \mathrm{1}_{\omega \in \{0,\pi\}}  \biginprod{ \big(\F_{j,l,m}^{(\omega)} \widetilde{\otimes} \F_{l,j,m}^{\dagger(\omega)} \big)({\mathcal{V}^{\dagger(\omega)}_{l}})}{\overline{\mathcal{V}^{\dagger (\omega)}_{j}}}_{S_2} \Big),
 %\Sigma^{\omega}_{\X,m}\big(\mathcal{V}^{(\omega)}_{\X}\big) &=4\pi^2 \sum_{l,j \in \{X,Y\}}\Big(\inprod{\F^{(\omega)}_{jl,m}(u_{l,i})}{u_j}  \inprod{\F^{(\omega)}_{lj,m}(v_j)}{v_{l,i}} +\inprod{\F^{(\omega)}_{lj,m}(u_j)}{u_{l,i}}  \inprod{\F^{(\omega)}_{jl,m}(v_{l,i})}{v_j}\\&\qquad \qquad\qquad+\mathrm{2}_{\{0,\pi\}}\big(\inprod{\F^{(\omega)}_{lj,m}(u_j)}{u_{l,i}}  \inprod{\F^{(\omega)}_{lj,m}(v_j)}{v_{l,i}}\big)\Big), 
 \end{align*}
and where $\F^{(\omega)}_{XY,m}  =\E( \dmp{0}{X} \otimes \dmp{0}{Y})$.
\end{thm}
We shall use \autoref{thm:NBMf} in order to derive the distributional properties of the process given in \autoref{thm:ZXdistint}. We define
\[{\widetilde{\mathcal{Z}}}^{\X}_{\flo{\eta T}} = \eta \int^b_a {\widetilde{\mathcal{Z}}}^{\X,\omega}_{\flo{\eta T}}d\omega \quad \text{ and }\quad \widetilde{\mathcal{M}}^{\X}_{\flo{\eta T},m}=\frac{\eta}{\mathcal{W}_{b_T}}\sum_{t=1}^{\flo{\eta T}} \int_a^b N^{(\omega)}_{m,T,t} d\omega  \] 
where 
\begin{align*}
{\widetilde{\mathcal{Z}}}^{\X,\omega}_{\flo{\eta T}} &=\frac{1}{\mathcal{W}_{b_T}}\biginprod{ \sum_{s=1}^{\flo{\eta T}}\Big(  \sum_{t=1}^{\flo{\eta T}} \tilde{w}^{(\omega)}_{b_T,s,t}(X_s \otimes X_t, Y_s \otimes Y_t)^\top  - \E\big(\hat{\F}^{\omega}_{X}(\eta),\hat{\F}^{\omega}_{Y}(\eta)\big)^\top \Big)}{\mathcal{V}^{(\omega)}_{\X}}_{\Hs}.
\end{align*}
Let $d_S$ denote the Skorokhod metric on $D[0,1])$ and $d_U$ the uniform metric and recall that $(D[0,1], d_S)$ is a metric space. Let $F$ be a closed set of $D[0,1]$ and denote  $F_{d_S,\epsilon}=\big\{x: d_S\big(x,y\big) \le \epsilon, y \in F\big\}$
Since the Skorokhod metric is weaker than the uniform metric, we have
\begin{align*}
&\mathbb{P}\Big( \big\{\Re(\widetilde{\mathcal{Z}}^{\X}_{\flo{\eta T}})\big\}_{\eta \in [0,1]} \in F \Big)\\&
\le \mathbb{P}\Big( d_S\big(\big\{\Re(\widetilde{\mathcal{Z}}^{\X}_{\flo{\eta T}})\big\}_{\eta \in [0,1]},\big\{\Re(\widetilde{\mathcal{M}}^{\X}_{\flo{\eta T},m})\big\}_{\eta \in [0,1]}\big) \ge \epsilon \Big) + \mathbb{P}\Big(\big\{\Re(\widetilde{\mathcal{M}}^{\X}_{\flo{\eta T},m})\big\}_{\eta \in [0,1]} \in F_{d_S,\epsilon}\Big)\\&
\le \mathbb{P}\Big( d_U\big(\big\{\Re(\widetilde{\mathcal{Z}}^{\X}_{\flo{\eta T}})\big\}_{\eta \in [0,1]},\big\{\Re(\widetilde{\mathcal{M}}^{\X}_{\flo{\eta T},m})\big\}_{\eta \in [0,1]}\big) \ge \epsilon \Big) + \mathbb{P}\Big(\big\{\Re(\widetilde{\mathcal{M}}^{\X}_{\flo{\eta T},m})\big\}_{\eta \in [0,1]} \in F_{d_S,\epsilon}\Big) 
\end{align*}
We will first prove that 
\[
\lim_{m \to \infty} \lim_{T \to \infty}  \mathbb{P}\Big( d_U\big(\big\{\widetilde{\mathcal{Z}}^{\X}_{\flo{\eta T}}\big\}_{\eta \in [0,1]},\big\{\widetilde{\mathcal{M}}^{\X}_{\flo{\eta T},m}\big\}_{\eta \in [0,1]}\big) \ge \epsilon \Big) =0.\tageq \label{eq:limmT}
\]
By Markov's inequality,
\begin{align*}
  \mathbb{P}\Big( d_U\big(\big\{\widetilde{\mathcal{Z}}^{\X}_{\flo{\eta T}}\big\}_{\eta \in [0,1]},\big\{\widetilde{\mathcal{M}}^{\X}_{\flo{\eta T},m}\big\}_{\eta \in [0,1]}\big) \ge \epsilon \Big) 
  \le \epsilon^{-\gamma}\E\Big( \sup_{\eta \in [0,1]}\Big\vert\widetilde{\mathcal{Z}}^{\X}_{\flo{\eta T}}-\widetilde{\mathcal{M}}^{\X}_{\flo{\eta T},m}\Big\vert\Big)^{\gamma},
\end{align*}
where we take $\gamma >2$.
We find 
\begin{align*}
 \E\Big( \sup_{\eta \in [0,1]}\Big\vert\widetilde{\mathcal{Z}}^{\X}_{\flo{\eta T}}-\widetilde{\mathcal{M}}^{\X}_{\flo{\eta T},m}\Big\vert\Big)^{\gamma} &\le \E\Big( \sup_{\eta \in [0,1]} \eta  \int_a^b \Big\vert  \Big(\widetilde{\mathcal{Z}}^{\X,\omega}_{\flo{\eta T}}-\frac{1}{\mathcal{W}_{b_T}}\sum_{t=1}^{\flo{\eta T}} N^{(\omega)}_{m,T,t}
 \Big)\Big\vert d\omega\Big)^{\gamma}
\\& \le  \E\Big( \int_a^b \sup_{\eta \in [0,1]}  \eta \Big\vert  \Big(\widetilde{\mathcal{Z}}^{\X,\omega}_{\flo{\eta T}}-\frac{1}{\mathcal{W}_{b_T}}\sum_{t=1}^{\flo{\eta T}} N^{(\omega)}_{m,T,t}\Big)\Big\vert d\omega \Big)^{\gamma}
\\& \le  (b-a)^{\gamma-1}\int_a^b \E\Big(  \sup_{\eta \in [0,1]}  \eta \Big\vert  \Big(\widetilde{\mathcal{Z}}^{\X,\omega}_{\flo{\eta T}}-\frac{1}{\mathcal{W}_{b_T}}\sum_{t=1}^{\flo{\eta T}} N^{(\omega)}_{m,T,t}\Big)\Big\vert\Big)^{\gamma} d\omega, \tageq \label{eq:approxint2}
\end{align*}

where the last inequality follows from an application of Jensen's inequality to the integral since $\gamma >1$, and from Tonelli's theorem, which allows to interchange the expectation and integral.
Continuity of the Hilbert-Schmid inner product with respect to the product topology on $\mathcal{H} \otimes \mathcal{H}$ and the Cauchy-Schwarz inequality imply for the integrand of \eqref{eq:approxint2}
\begin{align*}
 &\E\Big(  \sup_{\eta \in [0,1]}  \eta \Big\vert  \Big(\widetilde{\mathcal{Z}}^{\X,\omega}_{\flo{\eta T}}-%\widetilde{\mathcal{M}}^{X,\omega}_{\flo{\eta T},m}
\frac{1}{\mathcal{W}_{b_T}}\sum_{t=1}^{\flo{\eta T}} N^{(\omega)}_{m,T,t}\Big)\Big\vert\Big)^{\gamma}
\\ & = \E\Big(  \sup_{\eta} \eta \Big \vert  \biginprod{\frac{1}{\mathcal{W}_{b_T}} \Big(\sum_{s=1}^{\flo{\eta T}}
\big(\sum_{t=1}^{\flo{\eta T}}\tilde{w}^{(\omega)}_{b_T,s,t}(X_s \otimes X_t, Y_s \otimes Y_t)^\top   - \E\big(\hat{\F}^{\omega}_{X}(\eta),\hat{\F}^{\omega}_{Y}(\eta)\big)^\top\big)-\ldm{}{\flo{ \eta T}}-\ldm{\dagger}{\flo{\eta T}}\Big)}{\mathcal{V}^{(\omega)}_{\X}}_{\Hs} \Big \vert \Big)^{\gamma}
 \\&
\le \big(\sup_{\omega}\snorm{\mathcal{V}^{(\omega)}_{\X}}^{\gamma}_2 \big) 
\frac{1}{\mathcal{W}^{\gamma}_{b_T}}\E \Big( \sup_{\eta}  \bigsnorm{ \sum_{s=1}^{\flo{\eta T}}\Big(\sum_{t=1}^{\flo{\eta T}}\tilde{w}^{(\omega)}_{b_T,s,t}(X_s \otimes X_t, Y_s \otimes Y_t)^\top- \E\big(\hat{\F}^{\omega}_{X}(\eta),\hat{\F}^{\omega}_{Y}(\eta)\big)^\top\Big) -\ldm{}{\flo{\eta T}}-\ldm{\dagger}{\flo{\eta T}}}_{\Hs} \Big)^{\gamma},
\end{align*}
where $\ldm{}{\flo{\eta T}}$ is defined in \eqref{eq:ldm}. \eqref{eq:limmT} then follows from \autoref{lem:approx} in the Appendix together with \eqref{eq:Hsnorm} and \eqref{eq:approxint2}.
Next, write the real part of a complex random variable as a linear combination with its conjugate and apply the triangle inequality to find%. The%\Re(x) =\frac{x+\overline{x}}{2}
%triangle inequality implies
\begin{align*}
\E\Big( \sup_{\eta \in [0,1]}\Big\vert\Re(\widetilde{\mathcal{Z}}^{\X,\omega}_{\flo{\eta T}})-\Re(\ldm{}{\flo{ \eta T}})\Big\vert\Big)^{\gamma} &
\le \E\Big( \frac{1}{2}\sup_{\eta \in [0,1]}\Big\vert\widetilde{\mathcal{Z}}^{\X,\omega}_{\flo{\eta T}}-\ldm{}{\flo{ \eta T}}\Big\vert+\frac{1}{2}\sup_{\eta \in [0,1]}\Big\vert\overline{\widetilde{\mathcal{Z}}^{\X,\omega}_{\flo{\eta T}}}-\overline{\ldm{}{\flo{ \eta T}}}\Big\vert\Big)^{\gamma}
\\& \le \frac{1}{2}\E\Big( \sup_{\eta \in [0,1]}\Big\vert\widetilde{\mathcal{Z}}^{\X,\omega}_{\flo{\eta T}}-\ldm{}{\flo{ \eta T}}\Big\vert\Big)^{\gamma}+ \frac{1}{2}\E\Big( \sup_{\eta \in [0,1]}\Big\vert\overline{\widetilde{\mathcal{Z}}^{\X,\omega}_{\flo{\eta T}}}-\overline{\ldm{}{\flo{ \eta T}}}\Big\vert\Big)^{\gamma}.
\end{align*}
Using \autoref{lem:approx} and \eqref{eq:approxint2} it then follows  that 
\begin{align*}
&\lim_{m \to \infty} \lim_{T \to \infty} \mathbb{P}\Big( d_U\big(\big\{\Re(\widetilde{\mathcal{Z}}^{\X}_{\flo{\eta T}})\big\}_{\eta \in [0,1]},\big\{\Re(\widetilde{\mathcal{M}}^{\X}_{\flo{\eta T},m})\big\}_{\eta \in [0,1]}\big) \ge \epsilon \Big)
\\&\le \lim_{m \to \infty} \lim_{T \to \infty}  \epsilon^{-\gamma}(b-a)^{\gamma-1}  \frac{1}{2} \int_a^b\Big\{\E\Big( \sup_{\eta \in [0,1]}\eta \Big\vert\widetilde{\mathcal{Z}}^{\X,\omega}_{\flo{\eta T}}-{\mathcal{M}}^{\X,(\omega)}_{\flo{\eta T},m}\Big\vert\Big)^{\gamma}+ \E\Big( \sup_{\eta \in [0,1]} \eta\Big\vert\overline{\widetilde{\mathcal{Z}}^{\X,\omega}_{\flo{\eta T}}}-\overline{{\mathcal{M}}^{\X,(\omega)}_{\flo{\eta T},m}}\Big\vert\Big)^{\gamma}\Big\}d\omega =0,
\end{align*}
which proves \eqref{eq:limmT}. Consequently, an application of \autoref{thm:NBMf} yields 
 \begin{align*}
\lim_{T \to \infty} \mathbb{P}\Big( \big\{\Re(\widetilde{\mathcal{Z}}^{X}_{\flo{\eta T}}) \big\}_{\eta \in [0,1]} \in F\Big)&
\le \lim_{m \to \infty} \lim_{T \to \infty}\mathbb{P}\Big(\big\{\Re(\widetilde{\mathcal{M}}^{\X}_{\flo{\eta T},m})\big\}_{\eta \in [0,1]} \in F_{d_S,\epsilon}\Big)\\&
\le \lim_{m \to \infty} \mathbb{P}\Bigg( \Big\{ \tau_{m,XY}  \cdot \mathbb{B}(\eta)\Big \}_{\eta \in [0,1]} \in F_{d_S,\epsilon} \Bigg)
\\&=  \mathbb{P}\Bigg( \Big\{\tau_{XY} \cdot \mathbb{B}(\eta)\Big \}_{\eta \in [0,1]} \in F_{d_S,\epsilon} \Bigg),
\end{align*}
where the last equality follows by taking the limit with respect to $m$ of $\Gamma^{\omega}_{\X,m}$ and $\Sigma^{\omega}_{\X,m}$ to obtain the limiting covariance structure \citep[see Proposition 3.2 of ][]{vD19}. Taking $\epsilon \downarrow 0$, we obtain
\begin{align*}
\Big\{\Re(\widetilde{\mathcal{Z}}^{\X}_{\flo{\eta T}})\Big\}_{\eta \in [0,1]}  \,{\rightsquigarrow}\,  \Bigg\{ \eta {\tau_{XY}} \mathbb{B}(\eta) \Bigg \}_{\eta \in [0,1]}
%\Big\{\Re(\eta \biginprod{\frac{\sqrt{b_{T_1}}}{ \sqrt{T_1 \kappa}} \sum_{s=1}^{\flo{\eta T_1}}\Big( \sum_{t=1}^{\flo{\eta T_1}} \frac{1}{2\pi}(X_s \otimes X_t)  \tilde{w}^{(\omega)}_{b_T,t,s} - \E\hat{\F}^{\omega}_X \Big)}{\widetilde{\Pi}^{(\omega)}_{X,Y,k}})  \Big\}_{\eta \in [0,1]}  \,{\rightsquigarrow}\,  \Bigg\{ \eta {\alpha^{(\omega)}_{X,Y}} \mathbb{B}_{X}(\eta) \Bigg \}_{\eta \in [0,1]}
\end{align*}
\autoref{thm:ZXdistint} now follows from \autoref{lem:BrMbias} in the Appendix % --which shows the bias converges to zero--
and from noting that
\begin{align*} \Big \vert \frac{\sqrt{b_{T}}}{\sqrt{\kappa T}}- \frac{1}{\mathcal{W}_{b_T}} \Big \vert &= \Big \vert  \frac{\sqrt{b_{T}}}{\sqrt{\kappa T}} \frac{\mathcal{W}_{b_T}}{\mathcal{W}_{b_T}}- \frac{1}{\mathcal{W}_{b_T}} \Big \vert 
=\frac{1}{\mathcal{W}_{b_T}} \Big \vert  \frac{\sqrt{b_{T}}}{\sqrt{\kappa T}}\Big( \frac{\sqrt{T\kappa}}{\sqrt{b_T }}(1+o(1))\Big)- 1 \Big \vert =o(\frac{1}{\mathcal{W}_{b_T}})=o(1),
\end{align*}
where we used that under \autoref{as:Weights}
$\mathcal{W}^2_{b_T} =  \sum_{|h|< T} (T-|h|)w(b_T h )^2 = \frac{T}{b_T}(\kappa+o(1))$.
\end{proof}

%%%%%%%%%%%%%%%%%%%%%%%%%%%%%%%%%%%%%%%%%%%%%%
%% Single Appendix:                         %%
%%%%%%%%%%%%%%%%%%%%%%%%%%%%%%%%%%%%%%%%%%%%%%
%\begin{appendix}
%\section*{???}%% if no title is needed, leave empty \section*{}.
%\end{appendix}
%%%%%%%%%%%%%%%%%%%%%%%%%%%%%%%%%%%%%%%%%%%%%%
%% Multiple Appendixes:                     %%
%%%%%%%%%%%%%%%%%%%%%%%%%%%%%%%%%%%%%%%%%%%%%%
%\begin{appendix}
%\section{???}
%
%\section{???}
%
%\end{appendix}

%%%%%%%%%%%%%%%%%%%%%%%%%%%%%%%%%%%%%%%%%%%%%%
%% Support information (funding), if any,   %%
%% should be provided in the                %%
%% Acknowledgements section.                %%
%%%%%%%%%%%%%%%%%%%%%%%%%%%%%%%%%%%%%%%%%%%%%%
 \section*{Acknowledgements}
 This work has been supported by the Collaborative Research Center ``Statistical modeling of nonlinear dynamic processes'' (SFB 823, Teilprojekt A1, C1) of the German Research Foundation
(DFG).  The authors would like to thank the referees and the editor for their constructive comments on the first version of this manuscript.
%
% The first author was supported by ...
%
% The second author was supported in part by ...

%%%%%%%%%%%%%%%%%%%%%%%%%%%%%%%%%%%%%%%%%%%%%%
%% Supplementary Material, if any, should   %%
%% be provided in {supplement} environment  %%
%% with title inside \textbf{} and short    %%
%% description below.                       %%
%%%%%%%%%%%%%%%%%%%%%%%%%%%%%%%%%%%%%%%%%%%%%%
\begin{supplement}
\textbf{Supplement to ``Pivotal tests for relevant differences in the second order dynamics of functional time series''}. 
The supplement contains the appendices with proofs of the statements presented in this paper as well as further auxiliary lemmas and technical results necessary to complete the proofs. The supplement furthermore contains an application of the proposed methodology to resting state fMRI data. 
\end{supplement}

%%  imsart-number.bst  will be used to                     %%
%%  MR numbers will be added by VTeX.                      %%
%%                                                         %%
%%  Use \cite{...} to cite references in text.             %%
%%                                                         %%
%%%%%%%%%%%%%%%%%%%%%%%%%%%%%%%%%%%%%%%%%%%%%%%%%%%%%%%%%%%%%

%% if your bibliography is in bibtex format, uncomment commands:
%\bibliographystyle{imsart-number} % Style BST file
%\bibliography{bibliography}       % Bibliography file (usually '*.bib')

%% or include bibliography directly:
% \begin{thebibliography}{}
% \bibitem{b1}
% \end{thebibliography}
\small{}

\appendix

%\bibliography{CovComparison}

\section{Proofs of main statements} \label{sec:mainstat}
\def\theequation{A.\arabic{equation}}
\setcounter{equation}{0}
\subsection{Proofs of statements from \autoref{sec3}}\label{sec:proofshyps}
\begin{proof}[Proof of Proposition \autoref{prop:consF}]
This follows from an adjustment of the proof of Theorem 4.1(ii) of \cite{vD19} for the value of $\ell$. Details are omitted. 
\end{proof}

\begin{proof}[Proof of \autoref{lem:Fbound}]
Using Tonelli's theorem,
\begin{align*}
\E\Big(\sup_{\eta\in [0,1]}\int_{a}^{b}\bigsnorm{\widehat{M}_{\hat\F}(\eta,\omega)-\eta {M}_{\F}(\omega) }^2_2  d\omega\Big)
%& \le 
%\E\Big(\int_{a}^{b}\sup_{\eta \in [0,1]}\bigsnorm{\widehat{M}_{\hat\F}(\eta,\omega)-\eta {M}_{\F}(\omega) }^2_2  d\omega\Big)
%\\&
 \le\int_{a}^{b} \E\Big(\sup_{\eta \in [0,1]}\bigsnorm{\widehat{M}_{\hat\F}(\eta,\omega)-\eta {M}_{\F}(\omega) }^2_2 \Big)  d\omega
\end{align*}
Observe that
\[
\bigsnorm{\widehat{M}_{\hat\F}(\eta,\omega)-\eta {M}_{\F}(\omega) }^2_2 \le 2   \snorm{\eta\big(\hat{\F}_X^{(\omega)}(\eta)-\F_X^{(\omega)}\big)}^2_2 + 2  \snorm{\eta\big(\hat{\F}^{(\omega)}_Y(\eta)-{\F}^{(\omega)}_Y\big)}^2_2,
\]
and that $\frac{\eta}{{\lfloor \eta T\rfloor}} =\frac{\eta T}{\eta T}\frac{\eta}{{\lfloor \eta T\rfloor}} = \frac{1}{T} \big(\frac{\eta T}{\flo{\eta T}} -1\big)+\frac{1}{T}$. Jensen's inequality and Minkowski's inequality yield
\begin{align*}
\E\sup_{\eta \in [0,1]} \bigsnorm{\eta\big(\hat{\F}_X^{(\omega)}(\eta)-\F_X^{(\omega)}\big)}^2_2 &
%{
%\le\E \Big(\sup_{\eta \in [0,1]} \bigsnorm{\eta\big(\hat{\F}_X^{(\omega)}(\eta)-\F_X^{(\omega)}\big)}_2\Big)^2}
%\\&
%{
%\le \Big(\E \big(\sup_{\eta \in [0,1]}\bigsnorm{\frac{\eta}{\flo{\eta T_1}}\sum_{s=1}^{\flo{\eta T_1}}\Big(\sum_{t=1}^{\flo{\eta T_1}} \tilde{w}^{(\omega)}_{b_1,s,t} (X_s \otimes X_t) - \F_X^{(\omega)}\Big)}_2 \big)^{(2\gamma/2)}\Big)^{2/\gamma},}
%
%\\&
{ \le \Big( \E \big( \sup_{\eta \in [0,1]}\bigsnorm{\frac{1}{T_1} \sum_{s=1}^{\flo{\eta T_1}}\Big(\sum_{t=1}^{\flo{\eta T_1}} \tilde{w}^{(\omega)}_{b_1,s,t} (X_s \otimes X_t) - \F_X^{(\omega)}\Big)}_2\big)^{\gamma}\Big)^{2/\gamma}}
\\
&{ +\sup_{\eta \in [0,1]} \Big \vert \frac{\eta T}{\flo{\eta T}}-1\Big \vert^2 \Big( \E  \big(\sup_{\eta \in [0,1]}\bigsnorm{\frac{1}{T_1} \sum_{s=1}^{\flo{\eta T_1}}\Big(\sum_{t=1}^{\flo{\eta T_1}} \tilde{w}^{(\omega)}_{b_1,s,t} (X_s \otimes X_t) - \F_X^{(\omega)}\Big)}_2\big)^{\gamma}\Big)^{2/\gamma}}
%\\& =O\Big(1+ T_1^{-1}\Big) \Big(O( b_1^{-1-\epsilon/2} T_1^{-2-\epsilon} T_1^{1+\epsilon/2}) + O(b_1^{\ell(2+\epsilon)} ) \Big)^{2/(2+\epsilon)}
\\&  = O( b^{-1}_1 T_1^{-1}) + O(b^{2 \ell}_1 ).
\end{align*}
where we used \eqref{eq:floor} and \eqref{eq:varB1} and \eqref{eq:biasB1} as in the proof of  \autoref{lem:boundseqSDO}. In complete analogy, we obtain 
\begin{align*}
\E\sup_{\eta \in [0,1]} \bigsnorm{\eta\big(\hat{\F}_Y^{(\omega)}(\eta)-\F_Y^{(\omega)}\big)}^2_2  = O( b^{-1}_2 T_2^{-1}) + O(b^{2 \ell}_2 ).
\end{align*}
%The same holds true for the cross-term by the Cauchy-Schwarz inequality and \autoref{as:ratiorates}. 
The statement now follows from \autoref{as:bwrates} and \autoref{as:ratiorates}. 
\end{proof}

\begin{proof}[Proof of \autoref{thm:hatZFhyp1}]
This follows from \autoref{lem:Fbound} and \autoref{lem:floor} and from using that $c \in \cnum$, $c+\overline{c} = 2 \Re(c)$.%, where $\Re(c)$ denotes the real part of $c$. 
\end{proof}

\begin{proof}[Proof of \autoref{lem:errortildeM}]
Denote the perturbation $\hat{\Delta}_\eta \F^{\omega} = \hat{\F}^{\omega} (\eta)- \F^{\omega}$.
We first consider $\mathrm{(i)}$. Observe that by Minkowski's inequality it suffices to show 
\begin{align*}
\sup_{\eta \in [0,1]}&\int_a^b \eta \Bigsnorm{\phixT{k}(\eta) - \phix{k}-\sum_{\substack{j,  \jpr=1,\\ \{j,  \jpr = k\}^\complement}}^{\infty} \frac{\biginprod{\F_X^{(\omega)} \widetilde{\otimes}\hat{\Delta}_\eta \F_X^{(\omega)}  +\hat{\Delta}_\eta \F_X^{(\omega)}  \widetilde{\otimes} \F_X^{(\omega)}}{\phix{j\jpr}  \widetilde{\otimes} \phix{k}}_{S_2} \phix{j\jpr}}{\lamx{\omega}{j} \lamx{\omega}{\jpr} -(\lamx{\omega}{k})^2}}_2 d\omega  
\\& = O_p\Big(\frac{\log^{4/\gamma}(T_1)}{b_{1} T_1}\Big) \tageq \label{eq:approxErrorX}
\end{align*}
and  
\begin{align*}
\sup_{\eta \in [0,1]}&\int_a^b \eta \Bigsnorm{\phiyT{k}(\eta) - \phiy{k}-\sum_{\substack{j,  \jpr=1,\\ \{j,  \jpr = k\}^\complement}}^{\infty} \frac{\biginprod{\F_Y^{(\omega)} \widetilde{\otimes}\hat{\Delta}_\eta \F_Y^{(\omega)}  +\hat{\Delta}_\eta \F_Y^{(\omega)}  \widetilde{\otimes} \F_Y^{(\omega)}}{\phiy{j\jpr}  \widetilde{\otimes} \phiy{k}}_{S_2} \phiy{j\jpr}}{\lamy{\omega}{j} \lamy{\omega}{\jpr} -(\lamy{\omega}{k})^2}}_2 d\omega  
\\& = O_p\Big(\frac{\log^{4/\gamma}(T_2)}{b_{2} T_2}\Big).
\end{align*}
As the proof for both processes is the same, we shall focus on \eqref{eq:approxErrorX} and drop the subscript $X$ in the following.
From Proposition \autoref{prop:bexp}, we have
\begin{align*}
\hat{\Pi}^{(\omega)}_{k}-{\Pi}^{(\omega)}_{k}&= \biginprod{\hat{\Pi}^{\omega}_{k}(\eta)-{\Pi}^{\omega}_{k}}{{\Pi}^{\omega}_{k}}\Pi^{(\omega)}_{k}\\&+ \sum_{\substack{j,  \jpr=1,\\ \{j,  \jpr = k\}^\complement}}^{\infty} \frac{1}{\lambda^{(\omega)}_{j}\lambda^{(\omega)}_{\jpr} -(\lambda^{(\omega)}_{k})^2}\Big[\biginprod{\big(\F^{\omega} \widetilde{\otimes} \F^{\omega} -\hat{\F}^{\omega}(\eta) \widetilde{\otimes} \hat{\F}^{\omega} (\eta)\big) \Pi^{(\omega)}_{k}}{\Pi^{(\omega)}_{j\jpr}}_{S_2}+
\biginprod{E^{(\omega)}_{k,b_T}(\eta)}{\Pi^{(\omega)}_{j\jpr}}_{S_2}\Big] \Pi^{(\omega)}_{j\jpr},
\end{align*}
where
\begin{align*}
E^{(\omega)}_{k,b_T}(\eta) =\Big[(\F^{\omega} \widetilde{\otimes} \F^{\omega} -\hat{\F}^{\omega}(\eta) \widetilde{\otimes} \hat{\F}^{\omega} (\eta))\Big]\Big[\hat{\Pi}^{(\omega)}_{k}-{\Pi}^{(\omega)}_{k}\Big]+\Big[\big(\lambda^{(\omega)}_k(\eta)\big)^2 \ -\big( \lambda^{(\omega)}_{k}\big)^2\Big]\Big[\hat{\Pi}^{(\omega)}_{k}-{\Pi}^{(\omega)}_{k}\Big]. \tageq \label{eq:Eerror}
\end{align*}
%Denote 
%\begin{align*}
%\Delta_{\eta} \F_{\widetilde{\otimes}}^{\omega}= \F^{\omega} \widetilde{\otimes} \F^{\omega} -\hat{\F}^{\omega}(\eta) \widetilde{\otimes} \hat{\F}^{\omega} (\eta)
%\end{align*}
Elementary calculations yield
\begin{align*}
\hat{\F}^{\omega}(\eta) \widetilde{\otimes} \hat{\F}^{\omega} (\eta)- \F^{\omega} \widetilde{\otimes} \F^{\omega} %&=  \hat{\F}^{\omega}(\eta)   \widetilde{\otimes} \big( \hat{\F}^{\omega} (\eta)- \F^{\omega} \big) +\big( \hat{\F}^{\omega} (\eta)- \F^{\omega} \big)   \widetilde{\otimes} {\F}^{\omega} \tageq \label{eq:pertF}
&=\F^{\omega} \widetilde{\otimes} \big( \hat{\Delta}_\eta \F^{\omega}\big) +\big( \hat{\Delta}_\eta \F^{\omega} \big)   \widetilde{\otimes} {\F}^{\omega} +\big( \hat{\Delta}_\eta \F^{\omega}\big) \widetilde{\otimes} \big( \hat{\Delta}_\eta \F^{\omega}\big).
\end{align*}
%\begin{align*}
%\hat{\F}^{\omega}(\eta) \widetilde{\otimes} \hat{\F}^{\omega} (\eta)- \F^{\omega} \widetilde{\otimes} \F^{\omega} %&=  \hat{\F}^{\omega}(\eta)   \widetilde{\otimes} \big( \hat{\F}^{\omega} (\eta)- \F^{\omega} \big) +\big( \hat{\F}^{\omega} (\eta)- \F^{\omega} \big)   \widetilde{\otimes} {\F}^{\omega} \tageq \label{eq:pertF}
%&=\F^{\omega} \widetilde{\otimes} \big( \hat{\F}^{\omega} (\eta)- \F^{\omega} \big) +\big( \hat{\F}^{\omega} (\eta)- \F^{\omega} \big)   \widetilde{\otimes} {\F}^{\omega}\\&+\big( \hat{\F}^{\omega}(\eta)  - \F^{\omega}\big) \widetilde{\otimes} \big( \hat{\F}^{\omega} (\eta)- \F^{\omega} \big). \tageq \label{eq:pertF}
%\end{align*}
Therefore, we will show that
\begin{align*}
&\sup_{\eta \in [0,1]}\eta \int_a^b\Bigsnorm{ \biginprod{\hat{\Pi}^{\omega}_{k}(\eta)-{\Pi}^{\omega}_{k}}{{\Pi}^{\omega}_{k}}_{S_2}\Pi^{(\omega)}_{k}+ \sum_{\substack{j,  \jpr=1,\\ \{j,  \jpr = k\}^\complement}}^{\infty} \frac{\Big[\biginprod{\hat{\Delta}_\eta \F^{\omega} \widetilde{\otimes} \hat{\Delta}_\eta \F^{\omega} \big) \Pi^{(\omega)}_{k} }{\Pi^{(\omega)}_{j\jpr}}_{S_2}+\biginprod{E^{(\omega)}_{k,b_T}(\eta)}{\Pi^{(\omega)}_{j\jpr}}_{S_2} \Big] \Pi^{(\omega)}_{j\jpr}}{\lambda^{(\omega)}_{j}\lambda^{(\omega)}_{\jpr} -(\lambda^{(\omega)}_{k})^2} }_2 d\omega 
\\&= O_p\Big(\frac{\log^{4/\gamma}(T_1)}{b_{T_1} T_1}\Big).
\end{align*}

By Minkowski's inequality, 
\begin{align*}
&\sup_{\eta \in [0,1]}\eta \int_a^b\Bigsnorm{ \biginprod{\hat{\Pi}^{\omega}_{k}(\eta)-{\Pi}^{\omega}_{k}}{{\Pi}^{\omega}_{k}}_{S_2}\Pi^{(\omega)}_{k}+ \sum_{\substack{j,  \jpr=1,\\ \{j,  \jpr = k\}^\complement}}^{\infty} \frac{\Big[\biginprod{\hat{\Delta}_\eta \F^{\omega} \widetilde{\otimes} \hat{\Delta}_\eta \F^{\omega} \big) \Pi^{(\omega)}_{k} }{\Pi^{(\omega)}_{j\jpr}}_{S_2}+\biginprod{E^{(\omega)}_{k,b_T}(\eta)}{\Pi^{(\omega)}_{j\jpr}}_{S_2} \Big] \Pi^{(\omega)}_{j\jpr}}{\lambda^{(\omega)}_{j}\lambda^{(\omega)}_{\jpr} -(\lambda^{(\omega)}_{k})^2} }_2 d\omega 
\\& \le \sup_{\eta \in [0,1]}\eta \int_a^b \Bigg\{\Bigsnorm{ \biginprod{\hat{\Pi}^{\omega}_{k}(\eta)-{\Pi}^{\omega}_{k}}{{\Pi}^{\omega}_{k}}_{S_2}\Pi^{(\omega)}_{k}}_2 + \Bigsnorm{\frac{\biginprod{\big( \hat{\F}^{\omega}(\eta)  - \F^{\omega}\big) \widetilde{\otimes} \big( \hat{\F}^{\omega} (\eta)- \F^{\omega} \big) \Pi^{(\omega)}_{k} }{\Pi^{(\omega)}_{j\jpr}}_{S_2} \Pi^{(\omega)}_{j\jpr}}{\lambda^{(\omega)}_{j}\lambda^{(\omega)}_{\jpr} -(\lambda^{(\omega)}_{k})^2} }_2 \\& \phantom{\sup_{\eta \in [0,1]}\eta  \int_a^b \quad}+\Bigsnorm{ \sum_{\substack{j,  \jpr=1,\\ \{j,  \jpr = k\}^\complement}}^{\infty} \frac{\biginprod{E^{(\omega)}_{k,b_T}(\eta)}{\Pi^{(\omega)}_{j\jpr}}_{S_2}}{\lambda^{(\omega)}_{j}\lambda^{(\omega)}_{\jpr} -(\lambda^{(\omega)}_{k})^2}\Pi^{(\omega)}_{j\jpr}  }_2\Bigg\} d\omega. \tageq \label{eq:bigerror}
\end{align*}
We treat these terms separately.  Firstly, observe that for any $A,B \in S_2(\mathcal{H})$, we can write
\begin{align*}
\snorm{A}^2_2+\snorm{B}^2_2 = \snorm{A-B}^2_2 + \inprod{A}{B}_{S_2} + \overline{\inprod{A}{B}}_{S_2}.
\end{align*}
Furthermore observe that $\snorm{\,{\Pi}^{\omega}_{k}(\eta)}_2 = \snorm{\,\hat{\Pi}^{\omega}_{k}(\eta)}_2 =1$ and that
\begin{align*}
 \biginprod{\hat{\Pi}^{\omega}_{k}(\eta)-{\Pi}^{\omega}_{k}}{{\Pi}^{\omega}_{k}}_{S_2} = |\inprod{\hat{\phi}^{(\omega)}_k}{{\phi}^{(\omega)}_k}|^2 - 1 =   |\overline{\inprod{\hat{\phi}^{(\omega)}_k}{{\phi}^{(\omega)}_k}}|^2-1 = \overline{\biginprod{\hat{\Pi}^{\omega}_{k}(\eta)-{\Pi}^{\omega}_{k}}{{\Pi}^{\omega}_{k}}}_{S_2}.
\end{align*}
Rearranging terms yields $ \biginprod{\hat{\Pi}^{\omega}_{k}(\eta)-{\Pi}^{\omega}_{k}}{{\Pi}^{\omega}_{k}}_{S_2} = -\frac{1}{2}\snorm{\hat{\Pi}^{\omega}_{k}(\eta)-{\Pi}^{\omega}_{k}}^2_2$. We obtain, using \autoref{lem:diffboundEig}, 
\begin{align*}
\sup_{\eta \in [0,1]}\eta \int_a^b \Bigsnorm{ \biginprod{\hat{\Pi}^{\omega}_{k}(\eta)-{\Pi}^{\omega}_{k}}{{\Pi}^{\omega}_{k}}_{S_2}\Pi^{(\omega)}_{k}}_2 d\omega & 
\le  \frac{1}{2} \int_a^b  \sup_{\eta \in [0,1]}\eta \snorm{\,\hat{\Pi}^{\omega}_{k}(\eta)-{\Pi}^{\omega}_{k}}^2_2 d\omega \tageq \label{eq:error1}
\\& \le C \sum_{i=2}^{4}\int_a^b  \Big(\sup_{\eta \in [0,1]} \big(\snorm{\eta^{1/i}\big(\hat{\F}_X^{(\omega)}(\eta)-\F_X^{(\omega)}\big)}_2\Big)^i d\omega
%\\& +O\Big(\sup_{\eta \in [0,1]}\eta \int_a^b \snorm{R^{(\omega)}_{k,T}(\eta) \mathrm{1}_{A_{k,T}^{\complement}(\omega)}}^2_2 d\omega \Big)
\end{align*}
for some constants $C>0$.  Consequently, \autoref{thm:maxPS} yields
\begin{align*}
\sup_{\eta \in [0,1]}\eta \int_a^b \Bigsnorm{ \biginprod{\hat{\Pi}^{\omega}_{k}(\eta)-{\Pi}^{\omega}_{k}}{{\Pi}^{\omega}_{k}}_{S_2}\Pi^{(\omega)}_{k}}_2 d\omega  =O_p\Big(\frac{\log^{4/\gamma}(T)}{b_T T}\Big) = o_p\Big(\frac{1}{\sqrt{b_{T_1} T_1+b_{T_2} T_2}}\Big).
\end{align*}
To treat the other two terms of \eqref{eq:bigerror}, we first observe that 
\begin{align*}
\sqrt{\frac{1}{\big(\lambda^{(\omega)}_{j}\lambda^{(\omega)}_{\jpr} -(\lambda^{(\omega)}_{k})^2\big)^2}}=\frac{1}{|\lambda^{(\omega)}_{j}\lambda^{(\omega)}_{l} -(\lambda^{(\omega)}_{k})^2|} \le  C
\end{align*}
for some bounded constant $C>0$.  Indeed, recall that $G_{\widetilde{\otimes},k} =  \inf_{j,l:\{j,  l = k\}^\complement}|\lambda^{(\omega)}_{j}\lambda^{(\omega)}_{l}-(\lambda^{(\omega)}_{k})^2|$, and hence, under \autoref{as:eigsep}, $C=1/G_{\widetilde{\otimes},k}<\infty$. 
By the Cauchy-Schwarz inequality and Holder's inequality for operators, we obtain
\begin{align*}
& \sup_{\eta \in [0,1]}\eta \int_a^b\Bigsnorm{\frac{\biginprod{\big( \hat{\F}^{\omega}(\eta)  - \F^{\omega}\big) \widetilde{\otimes} \big( \hat{\F}^{\omega} (\eta)- \F^{\omega} \big) \Pi^{(\omega)}_{k} }{\Pi^{(\omega)}_{j\jpr}}_{S_2} \Pi^{(\omega)}_{j\jpr}}{\lambda^{(\omega)}_{j}\lambda^{(\omega)}_{\jpr} -(\lambda^{(\omega)}_{k})^2} }_2 d\omega
%\\& \le C \sup_{\eta \in [0,1]}\eta \int_a^b\bigsnorm{\big( \hat{\F}^{\omega}(\eta)  - \F^{\omega}\big) \widetilde{\otimes} \big( \hat{\F}^{\omega} (\eta)- \F^{\omega} \big)}_2 \snorm{\Pi^{(\omega)}_{k} }_2 \snorm{\Pi^{(\omega)}_{j\jpr}}^2_2 d\omega
%\\&
\\& \le C \sup_{\eta \in [0,1]}\eta \int_a^b\bigsnorm{\hat{\F}^{\omega}(\eta)  - \F^{\omega}}^2_2 d\omega
  =O_p\big(\frac{\log^{2/\gamma}(T_1)}{b_{T_1} {T_1}}\big)=o_p\big(\frac{1}{\sqrt{b_{T_1} T_1+b_{T_2} {T_2}}}\big),
\end{align*}
where we used that the eigenprojectors are rank-one operators (and hence elements of $S_1(\mathcal{H})$) and where the order follows from \autoref{lem:boundseqSDO}. For the last term of \eqref{eq:bigerror}, Parseval's identity and orthogonality of the eigenprojectors yield
\begin{align*}
 &\sup_{\eta \in [0,1]}\eta \int_a^b\Bigsnorm{ \sum_{\substack{j,  \jpr=1,\\ \{j,  \jpr = k\}^\complement}}^{\infty} \frac{\biginprod{E^{(\omega)}_{k,b_T}(\eta)}{\Pi^{(\omega)}_{j\jpr}}_{S_2}}{\lambda^{(\omega)}_{j}\lambda^{(\omega)}_{\jpr} -(\lambda^{(\omega)}_{k})^2}\Pi^{(\omega)}_{j\jpr}  }_2 d\omega
%\le C\sup_{\eta \in [0,1]}\eta \int_a^b\Bigsnorm{ \sum_{\substack{j,  \jpr=1,\\ \{j,  \jpr = k\}^\complement}}^{\infty} \biginprod{E^{(\omega)}_{k,b_T}(\eta)}{\Pi^{(\omega)}_{j\jpr}}\Pi^{(\omega)}_{j\jpr}}_2 d\omega
\\& \le  C\sup_{\eta \in [0,1]}\eta \int_a^b\Big( \sum_{r,s=1}^{\infty} \Big \vert \sum_{\substack{j,  \jpr=1,\\ \{j,  \jpr = k\}^\complement}}^{\infty} \biginprod{E^{(\omega)}_{k,b_T}(\eta)}{\Pi^{(\omega)}_{j\jpr}}_{S_2}\biginprod{\Pi^{(\omega)}_{j\jpr}}{\Pi^{(\omega)}_{r,s}}_{S_2} \Big \vert^2 \Big)^{1/2}d\omega 
\\& \le
%  C\sup_{\eta \in [0,1]}\eta \int_a^b\Big( \sum_{r,s=1}^{\infty} \Big \vert\biginprod{E^{(\omega)}_{k,b_T}(\eta)}{\Pi^{(\omega)}_{rs}} \Big \vert^2 \Big)^{1/2}d\omega =
   C\sup_{\eta \in [0,1]}\eta \int_a^b\bigsnorm{E^{(\omega)}_{k,b_T}(\eta)}_2 d\omega. 
\end{align*}
We find using \autoref{lem:diffboundEig}
\begin{align*}
\bigsnorm{E^{(\omega)}_{k,b_T}(\eta)}_2 & \le \Bigg(2\snorm{\F_X^{(\omega)}}_2 \snorm{\hat{\F}_X^{(\omega)}(\eta)-\F_X^{(\omega)}}_2  + \snorm{\hat{\F}_X^{\omega}(\eta)  - \F_X^{\omega}}^2_2 
+C_1 \snorm{\hat{\F}_X^{(\omega)}(\eta)-\F_X^{(\omega)}}_2+ \snorm{\hat{\F}_X^{(\omega)}(\eta)-\F_X^{(\omega)}}^2_2\Bigg)
\\&\phantom{\Bigg\{2} \times C \Big( 2\snorm{\F_X^{(\omega)}}_2\snorm{\hat{\F}_X^{(\omega)}(\eta)-\F_X^{(\omega)}}_2+ \snorm{\hat{\F}_X^{(\omega)}(\eta)-\F_X^{(\omega)}}^2_2  \Big)
\end{align*}
for some bounded constants $C_1, C$. Similar to \eqref{eq:error1}, we thus obtain from \autoref{thm:maxPS} % we thus obtain from Tonelli's theorem and \autoref{lem:boundseqSDO} that
\begin{align*}
\tilde{C}\,  \sup_{\eta \in [0,1]}\eta \int_a^b\bigsnorm{E^{(\omega)}_{k,b_T}(\eta)}_2 d\omega 
&
\le \tilde{C}\, \sum_{i=2}^{4} \int_a^b  \Big(\sup_{\eta \in [0,1]} \snorm{\eta^{1/i}\big(\hat{\F}_X^{(\omega)}(\eta)-\F_X^{(\omega)}\big)}_2\Big)^i d\omega
\\& %= O_p\Big(\frac{\log^{4/\gamma}(T)}{b_T T}\Big)
 = o_p\Big(\frac{1}{\sqrt{b_{T_1} T_1+b_{T_2} T_2}}\Big).
\end{align*}
This proves $\mathrm{(i)}$. The proof of $\mathrm{(ii)}$ follows along the same lines and is therefore omitted.
\end{proof}

\begin{proof}[Proof of \autoref{thm:hatZbigPi}]
Using \autoref{lem:errortildeM}, it suffices to show that
\begin{align*}
&\int_a^b \biginprod{\widetilde{M}_{\Pi,k}(\eta,\omega)}{\eta M_{\Pi,k}(\omega)}_{S_2}d\omega+ \int_a^b \overline{\biginprod{\widetilde{M}_{\Pi,k}(\eta,\omega)}{\eta M_{\Pi,k}(\omega)}}_{S_2}d\omega + o_P\Big(\frac{1}{\sqrt{b_1 T_1+b_2T_2}}\Big)
\\& =\int_a^b \Big( \frac{1}{\sqrt{b_1 T_1}}\Re\Big(\biginprod{ \mathcal{Z}^{X,\omega}_{T,\eta} }{\widetilde{\Pi}^{(\omega)}_{X,Y,k}}_{S_2} \Big)
+ \frac{1}{\sqrt{b_2 T_2}}\Re\Big(\biginprod{ \mathcal{Z}^{Y,\omega}_{T,\eta} }{\widetilde{\Pi}^{(\omega)}_{X,Y,k}}_{S_2} \Big)\Big) d\omega+o_P\Big(\frac{1}{\sqrt{b_1 T_1+b_2T_2}}\Big).
\end{align*}
%with $\widetilde{\Pi}^{(\omega)}_{X,Y,k}$ and $\widetilde{\Pi}^{(\omega)}_{Y,X,k}$ \eqref{eq:PibigX} and \eqref{eq:PibigY}
To ease notation, denote
\[U^{X,\omega}_{T,\eta} =\F_X^{(\omega)} \widetilde{\otimes} Z^{X,\omega}_{T,\eta}+ Z^{X,\omega}_{T,\eta}\widetilde{\otimes} \F_X^{(\omega)} \quad \text{and} \quad U^{Y,\omega}_{T,\eta} =\F_Y^{(\omega)} \widetilde{\otimes} Z^{Y,\omega}_{T,\eta}+ Z^{Y,\omega}_{T,\eta}\widetilde{\otimes} \F_Y^{(\omega)}. \tageq \label{eq:Uxom}  \]
Using orthogonality of the eigenfunctions, we can write
\begin{align*}
&\sqrt{b_{T_1} T_1+b_{T_2} T_2}\int_{a}^{b}\biginprod{\widetilde{M}_{\Pi,k}(\eta,\omega)}{\eta M^{(k)}(\omega)}_{S_2}d\omega \\& 
 =\sqrt{b_{T_1} T_1+b_{T_2} T_2}\int_{a}^{b}\biginprod{\frac{1}{\sqrt{b_{T_1} T_1}}\sum_{\substack{j,  \jpr=1,\\ \{j,  \jpr = k\}^\complement}}^{\infty}  \frac{1}{\lambda^{(\omega)}_{X,j}\lambda^{(\omega)}_{X,\jpr} -(\lambda^{(\omega)}_{X,k})^2} \biginprod{U^{X,\omega}_{T,\eta} \phix{k}}{\phix{j\jpr}}_{S_2}  \phix{j\jpr}}{\eta \Big(\phix{k}-\phiy{k}\Big) }_{S_2}d\omega 
 \\& 
\phantom{\approx}
-\sqrt{b_{T_1} T_1+b_{T_2} T_2}\int_{a}^{b}\biginprod{\frac{1}{\sqrt{b_{T_2} T_2}}\sum_{\substack{j,  \jpr=1,\\ \{j,  \jpr = k\}^\complement}}^{\infty}  \frac{1}{\lambda^{(\omega)}_{Y,j}\lambda^{(\omega)}_{Y,\jpr} -(\lambda^{(\omega)}_{Y,k})^2}  \biginprod{U^{Y,\omega}_{T,\eta}  \phiy{k}}{\phiy{j\jpr}}_{S_2}  \phiy{j\jpr}}{\eta \Big(\phix{k}-\phiy{k}\Big) }_{S_2}d\omega 
\\& 
 =\eta\sqrt{b_{T_1} T_1+b_{T_2} T_2}\frac{-1}{\sqrt{b_{T_1} T_1}} \int_{a}^{b}\biginprod{\sum_{\substack{j,  \jpr=1,\\ \{j,  \jpr = k\}^\complement}}^{\infty}  \frac{\biginprod{U^{X,\omega}_{T,\eta} \phix{k}}{\phix{j\jpr}}_{S_2} }{\lambda^{(\omega)}_{X,j}\lambda^{(\omega)}_{X,\jpr} -(\lambda^{(\omega)}_{X,k})^2} \phix{j\jpr}}{\phiy{k} }_{S_2}d\omega  \tageq  \label{eq:MtilMX}
 \\& 
\phantom{\approx}
\eta\sqrt{b_{T_1} T_1+b_{T_2} T_2} \frac{-1}{\sqrt{b_{T_2} T_2}}\int_{a}^{b}\biginprod{\sum_{\substack{j,  \jpr=1,\\ \{j,  \jpr = k\}^\complement}}^{\infty}  \frac{\biginprod{U^{Y,\omega}_{T,\eta}  \phiy{k}}{\phiy{j\jpr}}_{S_2} }{\lambda^{(\omega)}_{Y,j}\lambda^{(\omega)}_{Y,\jpr} -(\lambda^{(\omega)}_{Y,k})^2}   \phiy{j\jpr}}{\phix{k} }_{S_2}d\omega.   \tageq  \label{eq:MtilMY}
\end{align*}
To simplify the expression, % recall the definition of $U^{X,\omega}_{T,\eta} $ in \eqref{eq:Uxom} and %further recall the definition note that
%\begin{align*}
%\biginprod{U^{X,\omega}_{T,\eta} \phix{k}}{\phix{j\jpr}}_{S_2}= \biginprod{  \F_X^{(\omega)} \widetilde{\otimes} Z^{X,\omega}_{T,\eta}+ Z^{X,\omega}_{T,\eta}\widetilde{\otimes} \F_X^{(\omega)} }{\phix{j\jpr}  \widetilde{\otimes} \phix{k}}_{S_2}. 
%\end{align*}
observe that the properties of the Kronecker product and the Hilbert-Schmidt inner product together with orthogonality of the eigenfunctions yield
\begin{align*}
 \biginprod{  (\F_X^{(\omega)} \widetilde{\otimes} Z^{X,\omega}_{T,\eta})\phix{k}}{\phix{j\jpr}}_{S_2}& = 
 \biginprod{  \F_X^{(\omega)}(\phi^{(\omega)}_{X,k}) {\otimes} Z^{X,\omega}_{T,\eta}(\phi^{(\omega)}_{X,k})}{\phix{j\jpr}}_{S_2}
% \\& =\lambda^{(\omega)}_{X,k}  \biginprod{ \phi^{(\omega)}_{X,k} {\otimes} Z^{X,\omega}_{T,\eta}(\phi^{(\omega)}_{X,k})}{\phix{j\jpr}}
% \\& =\lambda^{(\omega)}_{X,k}  \biginprod{ \phi^{(\omega)}_{X,k} {\otimes} Z^{X,\omega}_{T,\eta}(\phi^{(\omega)}_{X,k}) }{\phi^{(\omega)}_{X,j} \otimes \phi^{(\omega)}_{X,\jpr}}
 \\& = \lambda^{(\omega)}_{k}\biginprod{\phi^{(\omega)}_{X,k}}{\phi^{(\omega)}_{X,j}} \biginprod{\phi^{(\omega)}_{X,\jpr}}{ Z^{X,\omega}_{T,\eta}(\phi^{(\omega)}_{X,k})} = \lambda^{(\omega)}_{X,k} \biginprod{\phi^{(\omega)}_{X,\jpr}}{ Z^{X,\omega}_{T,\eta} (\phi^{(\omega)}_{X,k})}
\end{align*}
for $j=k$ and zero otherwise. Similarly,
\begin{align*}
 \biginprod{  (Z^{X,\omega}_{T,\eta}\ \widetilde{\otimes} \F_X^{(\omega)})\phix{k}}{\phix{j\jpr}}_{S_2}
 %& = 
% \lambda^{(\omega)}_{k} \biginprod{  Z^{X,\omega}_{T,\eta}(\phi^{(\omega)}_{X,k} {\otimes} \phi^{(\omega)}_{X,k}}{\phi^{(\omega)}_{X,j} \otimes \phi^{(\omega)}_{X,\jpr}}
% \\& = 
%  \lambda^{(\omega)}_{X,k} \inprod{  Z^{X,\omega}_{T,\eta}(\phi^{(\omega)}_{X,k}}{\phi^{(\omega)}_{X,j}} \inprod{ \phi^{(\omega)}_{X,\jpr}}{\phi^{(\omega)}_{X,k}} 
 % \\&
 =  \lambda^{(\omega)}_{X,k} \biginprod{  Z^{X,\omega}_{T,\eta}(\phi^{(\omega)}_{X,k})}{\phi^{(\omega)}_{X,j}}
\end{align*}
for $\jpr = k$ and zero otherwise. From this and  \eqref{eq:Uxom}, we obtain 
\begin{align*}
&\sum_{\substack{j,  \jpr=1,\\ \{j,  \jpr = k\}^\complement}}^{\infty}  \frac{1}{\lambda^{(\omega)}_{X,j}\lambda^{(\omega)}_{X,\jpr} -(\lambda^{(\omega)}_{X,k})^2} \biginprod{ U^{X,\omega}_{T,\eta} \phix{k} }{\phix{j\jpr}}_{S_2} \phix{j\jpr}\\& 
%=\sum_{\substack{j,  \jpr=1,\\ \{j,  \jpr = k\}^\complement}}^{\infty}  \frac{1}{\lambda^{(\omega)}_{X,j}\lambda^{(\omega)}_{X,\jpr} -(\lambda^{(\omega)}_{X,k})^2} \biginprod{  \F_X^{(\omega)} \widetilde{\otimes} Z^{X,\omega}_{T,\eta}+ Z^{X,\omega}_{T,\eta}\widetilde{\otimes} \F_X^{(\omega)} }{\phix{j\jpr}  \widetilde{\otimes} \phix{k}}_{S_2} \phix{j\jpr}\\& 
=\sum_{\jpr \ne k} \frac{\lambda^{(\omega)}_{k} \inprod{\phi^{(\omega)}_{\jpr}}{ Z^{X,\omega}_{T,\eta} \phi^{(\omega)}_k} }{\lambda^{(\omega)}_{X,k}\lambda^{(\omega)}_{X,\jpr} -(\lambda^{(\omega)}_{X,k})^2}  \phix{k\jpr}+\sum_{j \ne k} \frac{\lambda^{(\omega)}_{k} \inprod{  Z^{X,\omega}_{T,\eta}\phi^{(\omega)}_k}{\phi^{(\omega)}_j}}{\lambda^{(\omega)}_{X,j}\lambda^{(\omega)}_{X,k} -(\lambda^{(\omega)}_{X,k})^2}   \phix{j k}
\end{align*}
which means that the integrand of \eqref{eq:MtilMX} becomes
\begin{align*}
& \sum_{\jpr \ne k} \frac{\lambda^{(\omega)}_{k} \inprod{\phi^{(\omega)}_{\jpr}}{ Z^{X,\omega}_{T,\eta} (\phi^{(\omega)}_k}) }{\lambda^{(\omega)}_{X,k}\lambda^{(\omega)}_{X,\jpr} -(\lambda^{(\omega)}_{X,k})^2}  \biginprod{\phix{k\jpr}}{\phiy{k}}_{S_2}+\sum_{j \ne k} \frac{\lambda^{(\omega)}_{k} \inprod{  Z^{X,\omega}_{T,\eta}(\phi^{(\omega)}_k)}{\phi^{(\omega)}_j}}{\lambda^{(\omega)}_{X,j}\lambda^{(\omega)}_{X,k} -(\lambda^{(\omega)}_{X,k})^2} \biginprod{ \phix{j k}}{\phiy{k}}_{S_2}
=\\&\sum_{j \ne k} \frac{\lambda^{(\omega)}_{k}}{\lambda^{(\omega)}_{X,k}\lambda^{(\omega)}_{X,j} -(\lambda^{(\omega)}_{X,k})^2} \Big[\biginprod{\phi^{(\omega)}_{j}}{ Z^{X,\omega}_{T,\eta} (\phi^{(\omega)}_k)} \biginprod{\phix{k j}}{\phiy{k}}_{S_2}+ \biginprod{  Z^{X,\omega}_{T,\eta}(\phi^{(\omega)}_k)}{\phi^{(\omega)}_j} \biginprod{ \phix{j k}}{\phiy{k}}_{S_2} \Big] 
\end{align*}
whereas for its conjugate, which arises in $\overline{\biginprod{\widetilde{M}_{\Pi,k}(\eta,\omega) }{\eta M^{(k)}(\omega)}}_{S_2}$, we obtain
\begin{align*}
%& \biginprod{\phiy{k}}{\sum_{\jpr \ne k} \frac{\lambda^{(\omega)}_{k} \inprod{\phi^{(\omega)}_{\jpr}}{ Z^{X,\omega}_{T,\eta} \phi^{(\omega)}_k} }{\lambda^{(\omega)}_{X,k}\lambda^{(\omega)}_{X,\jpr} -(\lambda^{(\omega)}_{X,k})^2} \phix{k\jpr}}_{S_2}+\biginprod{\phiy{k}}{\sum_{j \ne k} \frac{\lambda^{(\omega)}_{k} \inprod{  Z^{X,\omega}_{T,\eta}\phi^{(\omega)}_k}{\phi^{(\omega)}_j}}{\lambda^{(\omega)}_{X,j}\lambda^{(\omega)}_{X,k} -(\lambda^{(\omega)}_{X,k})^2} \phix{j k}}_{S_2}
%\\&=
%&\sum_{\jpr \ne k} \frac{\lambda^{(\omega)}_{k} \inprod{\phi^{(\omega)}_{\jpr}}{ Z^{X,\omega}_{T,\eta}(\phi^{(\omega)}_k)} }{\lambda^{(\omega)}_{X,k}\lambda^{(\omega)}_{X,\jpr} -(\lambda^{(\omega)}_{X,k})^2}  \biginprod{\phiy{k}}{\phix{k\jpr}}_{S_2}+\sum_{j \ne k} \frac{\lambda^{(\omega)}_{k} \inprod{  Z^{X,\omega}_{T,\eta}(\phi^{(\omega)}_k)}{\phi^{(\omega)}_j}}{\lambda^{(\omega)}_{X,j}\lambda^{(\omega)}_{X,k} -(\lambda^{(\omega)}_{X,k})^2}\biginprod{\phiy{k}}{ \phix{j k}}_{S_2}
\sum_{j \ne k} \frac{\lambda^{(\omega)}_{k}}{\lambda^{(\omega)}_{X,k}\lambda^{(\omega)}_{X,j} -(\lambda^{(\omega)}_{X,k})^2} \Big[\overline{\biginprod{\phi^{(\omega)}_{j}}{ Z^{X,\omega}_{T,\eta} (\phi^{(\omega)}_k)} \biginprod{\phix{k j}}{\phiy{k}}}_{S_2} +  \overline{\biginprod{  Z^{X,\omega}_{T,\eta}(\phi^{(\omega)}_k)}{\phi^{(\omega)}_j}\biginprod{ \phix{j k}}{\phiy{k}}}_{S_2}\Big]. 
\end{align*}
Next, recall that for complex numbers $c_1, c_2 \in \mathbb{C}$, $c_1+\overline{c_1} = \Re(c_1)+\im \Im(c_1) +  \Re(c_1)-\im \Im(c_1)= 2 \Re(c_1)$ and $\Re(c_1 c_2) = \Re(c_1)\Re(c_2)-\Im(c_1)\Im(c_2)$ and $\Re(\overline{c_1} c_2) = \Re(c_1)\Re(c_2)+\Im(c_1)\Im(c_2)$. Therefore, summing the respective integrands of ${\biginprod{\widetilde{M}_{\Pi,k}(\eta,\omega) }{\eta M^{(k)}(\omega)}}_{S_2}$ and $\overline{\biginprod{\widetilde{M}_{\Pi,k}(\eta,\omega) }{\eta M^{(k)}(\omega)}}_{S_2}$ yields
\begin{align*}
&\sum_{j \ne k} \frac{\lambda^{(\omega)}_{X,k}}{\lambda^{(\omega)}_{X,k}\lambda^{(\omega)}_{X,j} -(\lambda^{(\omega)}_{X,k})^2} \Big[2 \Re\Big(\biginprod{\phi^{(\omega)}_{j}}{ Z^{X,\omega}_{T,\eta} \phi^{(\omega)}_k} \biginprod{\phix{k j}}{\phiy{k}}_{S_2}\Big)+ 2 \Re\Big(\biginprod{  Z^{X,\omega}_{T,\eta}\phi^{(\omega)}_k}{\phi^{(\omega)}_j} \biginprod{ \phix{j k}}{\phiy{k}}_{S_2}\Big) \Big] 
\\&=\sum_{j \ne k} \frac{\lambda^{(\omega)}_{X,k}}{\lambda^{(\omega)}_{X,k}\lambda^{(\omega)}_{X,j} -(\lambda^{(\omega)}_{X,k})^2} \Big[2 \Re\Big(\overline{\biginprod{  Z^{X,\omega}_{T,\eta}(\phi^{(\omega)}_k)}{\phi^{(\omega)}_j}}  \biginprod{\phix{k j}}{\phiy{k}}_{S_2}\Big)+ 2 \Re\Big(\biginprod{  Z^{X,\omega}_{T,\eta}(\phi^{(\omega)}_k)}{\phi^{(\omega)}_j} \biginprod{ \phix{j k}}{\phiy{k}}_{S_2}\Big) \Big] 
\\&=\sum_{j \ne k} \frac{\lambda^{(\omega)}_{X,k}}{\lambda^{(\omega)}_{X,k}\lambda^{(\omega)}_{X,j} -(\lambda^{(\omega)}_{X,k})^2} \Big[4 \Re\Big(\biginprod{  Z^{X,\omega}_{T,\eta}(\phi^{(\omega)}_k)}{\phi^{(\omega)}_j}\Big)\Re\Big(\biginprod{\phix{k j}}{\phiy{k}}_{S_2}\Big) \Big]. 
\end{align*}
The result now follows from applying the same argument to the integrand in \eqref{eq:MtilMY} and noting that $\inprod{  Z^{X,\omega}_{T,\eta}(\phi^{(\omega)}_k)} {\phi^{(\omega)}_j} =\biginprod{  Z^{X,\omega}_{T,\eta}}{\phix{jk}}_{S_2} $.% and doing the same for the integrand in \eqref{eq:MtilMY}.% together with  \autoref{lem:errortildeM}. \autoref{thm:BuildingBlock} then yields 
%\[
%\big\{\sqrt{b_{1} T_1+ b_{2}T_2}\big( \hat{\mathcal{Z}}^{[a,b],(k)}_{\Pi,T_1, T_2}(\eta)\big)\big \}_{\eta \in [0,1]} {\rightsquigarrow} \tau_\Pi \{\eta  \mathbb{B}(\eta)\}_{\eta \in [0,1]}
%\]
\end{proof}
\begin{proof}[Proof of \autoref{thm:hatZbiglam}]
Similar to the processes $\big\{\hat{Z}^{[a,b]}_{\F,T_1,T_2}(\eta)\big\}_{\eta \in [0,1]}$ and $\big\{\hat{Z}^{[a,b],(k)}_{\Pi,T_1,T_2}(\eta)\big\}_{\eta \in [0,1]}$, we may write
\begin{align*}
\hat{\mathcal{Z}}^{[a,b],(k)}_{\lambda,T_1,T_2} (\eta)&=\int_{a}^{b} \big\vert\widehat{M}_{\lambda,k}(\eta,\omega)\big\vert -\eta^2\big \vert M_{\lambda,k}(\omega)\big\vert^2 d\omega 
\\&= \int_{a}^{b}\Big\{\big \vert\widehat{M}_{\lambda,k}(\eta,\omega) -\eta M_{\lambda,k}(\omega)\big \vert^2 +\eta \overline{M_{\lambda,k}(\omega)}\big(\widehat{M}_{\lambda,k}(\eta,\omega) -\eta M_{\lambda,k}(\omega)\big)\\&
\phantom{\int_{a}^{b}}+\eta M_{\lambda,k}(\omega) \overline{\big(\widehat{M}_{\lambda,k}(\eta,\omega) -\eta M_{\lambda,k}(\omega)\big)}\Big\}d\omega.\tageq \label{eq:Zlamk}
%\\&=\sqrt{b_1 T_1+b_2 T_2} \int_{a}^{b} \eta M_{\lambda,k}(\omega)\widetilde{M}_{\lambda,k}(\eta,\omega)+\eta M_{\lambda,k}(\omega) \overline{(\widetilde{M}_{\lambda,k}(\eta,\omega)}d\omega +o_P(1)
\end{align*}
We can therefore apply Proposition  \ref{prop:bexplam} and \autoref{lem:errortildeMlam} below to find the result.
\end{proof}
\begin{lemma} \label{lem:errortildeMlam}
Suppose Assumptions \ref{as:depstruc}-\ref{as:ratiorates} are satisfied. Then,
\begin{align*}
\mathrm{(i)} & \sup_{\eta \in [0,1]}\int_a^b \Big \vert \widehat{M}_{\lambda,k}(\eta,\omega) - \eta  M_{\lambda,k}(\omega)
%\Big(\lamx{\omega}{k}-\lamy{\omega}{k}\Big)
-\widetilde{M}_{\lambda,k}(\eta,\omega)  \Big \vert d\omega = o_P\Big (\frac{1}{\sqrt{b_1 T_1+b_2T_2}} \Big );\\
\mathrm{(ii)} & \sup_{\eta \in [0,1]}\int_a^b\Big \vert \widetilde{M}_{\lambda,k}(\eta,\omega) \Big \vert ^2 d\omega = o_P \Big (\frac{1}{\sqrt{b_1 T_1+b_2T_2}} \Big ),\end{align*}
where 
\begin{align*}
&\widetilde{M}_{\lambda,k}(\eta,\omega) =
%= \eta \biginprod{\big(\hat{\F}_X^{(\omega)}(\eta)-\F_X^{(\omega)}\big)}{\Pi^{(\omega)}_{X,k}}-\eta \biginprod{\big(\hat{\F}_Y^{(\omega)}(\eta)-\F_Y^{(\omega)}\big)}{\Pi^{(\omega)}_{Y,k}}\\&
 \int_{a}^{b}\frac{1}{\sqrt{b_1 T_1}}\biginprod{\mathcal{Z}^{X,\omega}_{T,\eta}}{\Pi^{(\omega)}_{X,k}}_{S_2} -\frac{1}{\sqrt{b_2 T_2}}\biginprod{\mathcal{Z}^{Y,\omega}_{T,\eta}}{\Pi^{(\omega)}_{Y,k}}_{S_2}  d\omega
\end{align*}
with $\mathcal{Z}^{X,\omega}_{T,\eta}$ and $\mathcal{Z}^{Y,\omega}_{T,\eta}$ respectively given by \eqref{eq:Zx} and \eqref{eq:Zy}. 
\end{lemma}
\begin{proof}[Proof of \autoref{lem:errortildeMlam}]
Similar to the proof of \autoref{lem:errortildeM}, we can show (in this case using Proposition \autoref{prop:bexplam} and \autoref{lem:diffboundEig}) that \[\sqrt{b_1 T_1+b_2T_2}\sup_{\eta \in [0,1]}\int^b_a  \snorm{E^{(\omega)}_{\lambda, k,b_T}(\eta)}_2 d\omega\] converges to zero in probability using as $T_1, T_2 \to \infty$. The result then follows from \autoref{lem:floor}. The details are omitted for the sake of brevity. 
\end{proof}

%\begin{proof}[Proof of \autoref{thm:hatZbiglam}]
%This follows from \eqref{eq:Zlamk} and \autoref{lem:errortildeMlam}.
%\end{proof}
\subsection{Perturbations of eigenelements - proofs of Proposition \autoref{prop:bexp} and \ref{prop:bexplam} } \label{sec:Eigapprox}

%\subsection{eigenprojectors perturbation}
\begin{proof}[Proof of Proposition \autoref{prop:bexp}]
We write 
\begin{align*}
\hat{\F}^{\omega} (\eta)\widetilde{\otimes}\hat{\F}^{\omega} (\eta)= \F^{\omega}\widetilde{\otimes} \F^{\omega}+\underbrace{ (\hat{\F}^{\omega}(\eta) \widetilde{\otimes} \hat{\F}^{\omega} (\eta)-\F^{\omega} \widetilde{\otimes} \F^{\omega} )}_{=\Delta_{\eta} \F_{\widetilde{\otimes}}^{\omega}}
\end{align*}
and to ease notation we shall moreover denote $\hat{\F}_{\widetilde{\otimes}}^{\omega}(\eta)=\hat{\F}^{\omega} (\eta)\widetilde{\otimes}\hat{\F}^{\omega} (\eta)$ and $ \F^{\omega}_{\widetilde{\otimes}}= \F^{\omega}\widetilde{\otimes} \F^{\omega}$. Observe that we have the following decomposition
\[
\F^{\omega}_{\widetilde{\otimes}} = \sum_{j, \jpr} \lambda^{(\omega)}_{j}\lambda^{(\omega)}_{\jpr}  \Big(\Pi^{(\omega)}_{j}{\widetilde{\otimes}} \Pi^{(\omega)}_{\jpr}\Big),
\]
and note that 
\[
\F^{\omega}_{\widetilde{\otimes}} \Pi^{(\omega)}_{k,l} = \sum_{j, \jpr} \lambda^{(\omega)}_{j}\lambda^{(\omega)}_{\jpr} \Big( \Pi^{(\omega)}_{j}{\widetilde{\otimes}} \Pi^{(\omega)}_{\jpr} \Big)\Pi^{(\omega)}_{k,l} = \lambda^{(\omega)}_{k} \lambda^{(\omega)}_{l} \Pi^{(\omega)}_{k,l}
\]
Hence $\{\Pi^{(\omega)}_{k,l}\}_{k,l}$ are the eigenfunctions of $\F^{\omega}_{\widetilde{\otimes}}$.
We would like to obtain expressions for $\Delta_{\eta} \Pi^{(\omega)}_{k,l} $ and $\Delta_{\eta} \lambda^{(\omega)}_{k,l}$ by solving the equation
\begin{align*}
\big(\F_{\widetilde{\otimes}}^{\omega}+ \Delta_\eta \F_{\widetilde{\otimes}}^{\omega}\big)(\Pi^{(\omega)}_{k,l}+\Delta_{\eta} \Pi^{(\omega)}_{k,l})&=(\lambda^{(\omega)}_{k,l}+\Delta_{\eta} \lambda^{(\omega)}_{k,l}) (\Pi^{(\omega)}_{k,l}+\Delta_{\eta} \Pi^{(\omega)}_{k,l})
\\& \Leftrightarrow \\
\F^{\omega}_{\widetilde{\otimes}} \Pi^{(\omega)}_{k,l} + \Delta_\eta \F_{\widetilde{\otimes}}^{\omega}\Pi^{(\omega)}_{k,l} + \F_{\widetilde{\otimes}}^{\omega} \Delta_{\eta} \Pi^{(\omega)}_{k,l} + \Delta_\eta \F_{\widetilde{\otimes}}^{\omega} \Delta_{\eta} \Pi^{(\omega)}_{k,l} &=\lambda^{(\omega)}_{k,l} \Pi^{(\omega)}_{k,l} +\Delta_{\eta} \lambda^{(\omega)}_{k,l}\Pi^{(\omega)}_{k,l} + \lambda^{(\omega)}_{k,l} \Delta_{\eta} \Pi^{(\omega)}_{k,l}+\Delta_{\eta} \lambda^{(\omega)}_{k,l} \Delta_{\eta} \Pi^{(\omega)}_{k,l}
\\& \Leftrightarrow \\
\Delta_\eta \F_{\widetilde{\otimes}}^{\omega}\Pi^{(\omega)}_{k,l} + \F_{\widetilde{\otimes}}^{\omega} \Delta_{\eta} \Pi^{(\omega)}_{k,l}  &=\Delta_{\eta} \lambda^{(\omega)}_{k,l}\Pi^{(\omega)}_{k,l} + \lambda^{(\omega)}_{k,l} \Delta_{\eta} \Pi^{(\omega)}_{k,l}+(\Delta_{\eta} \lambda^{(\omega)}_{k,l} -\Delta_\eta \F_{\widetilde{\otimes}}^{\omega} ) \Delta_{\eta} \Pi^{(\omega)}_{k,l}
\end{align*}
since $\F_{\widetilde{\otimes}}^{\omega} \Pi^{(\omega)}_{k,l} =\lambda^{(\omega)}_{k,l} \Pi^{(\omega)}_{k,l} $. First note that $\Delta_{\eta} \Pi^{(\omega)}_{k,l} $ is a well-defined element of $S_2(\Hi)$. we can therefore use a basis expansion to write
 \[
\Delta_{\eta} \Pi^{(\omega)}_{k,l} =\phixT{k,l}(\eta) - \phix{k,l} \overset{S_2}{=}\sum_{n,m=1}^{\infty} a^{\eta}_{n,m}\Pi^{(\omega)}_{nm}
\]
where $a^{\eta}_{n,m}$ is a set of coefficients. Plugging this into the second term on the left and right hand side of the  above equation we have
\begin{align*}
\Delta_\eta \F_{\widetilde{\otimes}}^{\omega}\Pi^{(\omega)}_{k,l} + \F_{\widetilde{\otimes}}^{\omega}\sum_{n,m=1}^{\infty} a^{\eta}_{n,m}\Pi^{(\omega)}_{nm} &=\Delta_{\eta} \lambda^{(\omega)}_{k,l}\Pi^{(\omega)}_{k,l} + \lambda^{(\omega)}_{k,l} \sum_{n,m=1}^{\infty} a^{\eta}_{n,m}\Pi^{(\omega)}_{nm}+(\Delta_{\eta} \lambda^{(\omega)}_{k,l} -\Delta_\eta \F_{\widetilde{\otimes}}^{\omega} ) \Delta_{\eta} \Pi^{(\omega)}_{k,l}.
%\Delta_\eta \F^{\omega}\Pi^{(\omega)}_{k} + \F^{\omega} \sum_{n,m=1}^{\infty} a^{\eta}_{n,m}\Pi^{(\omega)}_{nm} &=\Delta_{\eta} \lambda^{(\omega)}_{k}\Pi^{(\omega)}_{k} + \lambda^{(\omega)}_k\sum_{n,m=1}^{\infty} a^{\eta}_{n,m}\Pi^{(\omega)}_{nm}\\&+(\Delta_{\eta} \lambda^{(\omega)}_{k} -\Delta_\eta \F^{\omega} ) \Delta_{\eta} \Pi^{(\omega)}_{k}.
 \tageq \label{eq:eq1}
\end{align*}
Observe that orthogonality of the eigenfunctions yields
\begin{align*}
  \F_{\widetilde{\otimes}}^{\omega}\sum_{n,m=1}^{\infty} a^{\eta}_{n,m}\Pi^{(\omega)}_{nm}= \sum_{r,s=1}^{\infty} \lambda^{(\omega)}_{r,s} \Pi^{(\omega)}_{r} \widetilde{\otimes} \Pi^{(\omega)}_{s}  \sum_{n,m=1}^{\infty} a^{\eta}_{n,m}\Pi^{(\omega)}_{nm} = \sum_{r,s=1}^{\infty}\lambda^{(\omega)}_{r,s}  a^{\eta}_{r,s}\Pi^{(\omega)}_{r,s}.
\end{align*}
and therefore \eqref{eq:eq1} becomes
\begin{align*}
\Delta_\eta \F_{\widetilde{\otimes}}^{\omega}\Pi^{(\omega)}_{k,l} + \sum_{r,s=1}^{\infty}\lambda^{(\omega)}_{r,s}  a^{\eta}_{r,s}\Pi^{(\omega)}_{r,s}&=\Delta_{\eta} \lambda^{(\omega)}_{k,l}\Pi^{(\omega)}_{k,l} + \lambda^{(\omega)}_{k,l} \sum_{n,m=1}^{\infty} a^{\eta}_{n,m}\Pi^{(\omega)}_{nm}+(\Delta_{\eta} \lambda^{(\omega)}_{k,l} -\Delta_\eta \F_{\widetilde{\otimes}}^{\omega} ) \Delta_{\eta} \Pi^{(\omega)}_{k,l}.
\end{align*}
Taking the Hilbert-Schmidt inner product with $\Pi^{(\omega)}_{j\jpr}, j \ne k, \jpr \ne l$
\begin{align*}
\biginprod{\Delta_\eta \F_{\widetilde{\otimes}}^{\omega}\Pi^{(\omega)}_{k,l}}{\Pi^{(\omega)}_{j\jpr}}_{S_2}+\biginprod{ \sum_{r,s=1}^{\infty}\lambda^{(\omega)}_{r,s}  a^{\eta}_{r,s}\Pi^{(\omega)}_{r,s}}{\Pi^{(\omega)}_{j\jpr}}_{S_2}&=\biginprod{\Delta_{\eta} \lambda^{(\omega)}_{k,l}\Pi^{(\omega)}_{k,l}}{\Pi^{(\omega)}_{j\jpr}}_{S_2} +\biginprod{ \lambda^{(\omega)}_{k,l} \sum_{n,m=1}^{\infty} a^{\eta}_{n,m}\Pi^{(\omega)}_{nm}}{\Pi^{(\omega)}_{j\jpr}}_{S_2}\\&+\biginprod{(\Delta_{\eta} \lambda^{(\omega)}_{k,l} -\Delta_\eta \F_{\widetilde{\otimes}}^{\omega} ) \Delta_{\eta} \Pi^{(\omega)}_{k,l}}{\Pi^{(\omega)}_{j\jpr}}_{S_2}
\end{align*}
which becomes
\begin{align*}
\biginprod{\Delta_\eta \F_{\widetilde{\otimes}}^{\omega}}{\Pi^{(\omega)}_{j\jpr} \widetilde{\otimes}\Pi^{(\omega)}_{k,l}}_{S_2}+\lambda^{(\omega)}_{j,\jpr}  a^{\eta}_{j,\jpr}&=0+ \lambda^{(\omega)}_{k,l}  a^{\eta}_{j,\jpr}+\biginprod{(\Delta_{\eta} \lambda^{(\omega)}_{k,l} -\Delta_\eta \F_{\widetilde{\otimes}}^{\omega} ) \Delta_{\eta} \Pi^{(\omega)}_{k,l}}{\Pi^{(\omega)}_{j\jpr}}_{S_2}
\end{align*}
rearranging, we find the coefficients are given by 
\begin{align*}
a^\eta_{j,\jpr}=\frac{1}{\lambda^{(\omega)}_{j,\jpr} -\lambda^{(\omega)}_{k,l}}\Big[-\biginprod{\Delta_\eta \F_{\widetilde{\otimes}}^{\omega}}{\Pi^{(\omega)}_{j\jpr} \widetilde{\otimes}\Pi^{(\omega)}_{k,l}}_{S_2}
+\biginprod{(\Delta_{\eta} \lambda^{(\omega)}_{k,l} -\Delta_\eta \F_{\widetilde{\otimes}}^{\omega} ) \Delta_{\eta} \Pi^{(\omega)}_{k,l}}{\Pi^{(\omega)}_{j\jpr}}_{S_2}\Big].
\end{align*}
If $j=k, l=\jpr$, we set $a^{\eta}_{k,l} = \biginprod{\hat{\Pi}^{\omega}_{k,l}-{\Pi}^{\omega}_{k,l}}{{\Pi}^{\omega}_{k,l}}_{S_2}$. The statement now follows.
%Hence, we obtain the following expansion.
%\begin{align*}
%\hat{\Pi}^{(\omega)}_{k}-{\Pi}^{(\omega)}_{k}= \biginprod{\hat{\Pi}^{\omega}_{k}(\eta)-{\Pi}^{\omega}_{k}}{{\Pi}^{\omega}_{k}}+\sum_{\substack{j,  \jpr=1,\\ \{j,  \jpr = k\}^\complement}}^{\infty} \frac{1}{\lambda^{(\omega)}_{j,\jpr} -(\lambda^{(\omega)}_{k})^2}&\Big[-\biginprod{\Delta_\eta \F_{\widetilde{\otimes}}^{\omega}\Pi^{(\omega)}_{k}}{\Pi^{(\omega)}_{j\jpr}}\\&+\biginprod{(\Delta_{\eta} \lambda^{(\omega)}_{k,l} -\Delta_\eta \F_{\widetilde{\otimes}}^{\omega} ) \Delta_{\eta} \Pi^{(\omega)}_{k}}{\Pi^{(\omega)}_{j\jpr}}\Big].
%\end{align*}
\end{proof}

%\subsection{eigenvalue perturbation}
\begin{proof}[Proof of Proposition \autoref{prop:bexplam}] 
In order to obtain an expression for $\Delta_{\eta} \lambda^{(\omega)}_{k}$, write
\[\hat{\F}^{\omega} (\eta)= \F^{\omega}+\underbrace{ (\hat{\F}^{\omega}(\eta)-\F^{\omega})}_{=\Delta_{\eta} \F^{\omega}}.\]
Since $\F^{\omega} \Pi^{(\omega)}_{k} =\lambda^{(\omega)}_k \Pi^{(\omega)}_{k}$, we obtain similar to the proof of Proposition \autoref{prop:bexp} that
\begin{align*}
\big(\F^{\omega}+ \Delta_\eta \F^{\omega}\big)(\Pi^{(\omega)}_{k}+\Delta_{\eta} \Pi^{(\omega)}_{k})&=(\lambda^{(\omega)}_k+\Delta_{\eta} \lambda^{(\omega)}_{k}) (\Pi^{(\omega)}_{k}+\Delta_{\eta} \Pi^{(\omega)}_{k})
%\\& \Leftrightarrow \\
%\F^{\omega} \Pi^{(\omega)}_{k} + \Delta_\eta \F^{\omega}\Pi^{(\omega)}_{k} + \F^{\omega} \Delta_{\eta} \Pi^{(\omega)}_{k} + \Delta_\eta \F^{\omega} \Delta_{\eta} \Pi^{(\omega)}_{k} &=\lambda^{(\omega)}_k \Pi^{(\omega)}_{k} +\Delta_{\eta} \lambda^{(\omega)}_{k}\Pi^{(\omega)}_{k} + \lambda^{(\omega)}_k \Delta_{\eta} \Pi^{(\omega)}_{k}+\Delta_{\eta} \lambda^{(\omega)}_{k} \Delta_{\eta} \Pi^{(\omega)}_{k}
\\& \Leftrightarrow \\
\Delta_\eta \F^{\omega}\Pi^{(\omega)}_{k} + \F^{\omega} \Delta_{\eta} \Pi^{(\omega)}_{k}  &=\Delta_{\eta} \lambda^{(\omega)}_{k}\Pi^{(\omega)}_{k} + \lambda^{(\omega)}_k \Delta_{\eta} \Pi^{(\omega)}_{k}+(\Delta_{\eta} \lambda^{(\omega)}_{k} -\Delta_\eta \F^{\omega} ) \Delta_{\eta} \Pi^{(\omega)}_{k}
\end{align*}
And using again that $\Delta_{\eta} \Pi^{(\omega)}_{k}=\sum_{n,m=1}^{\infty} a^{\eta}_{n,m}\Pi^{(\omega)}_{nm}$ and taking the inner product with $\Pi^{(\omega)}_{k}$
\begin{align*}
&\biginprod{\Delta_\eta \F^{\omega}\Pi^{(\omega)}_{k}}{\Pi^{(\omega)}_{k}}_{S_2} + \biginprod{\sum_{r,m=1}^{\infty}\lambda^{(\omega)}_{r} a^{\eta}_{r,m}\Pi^{(\omega)}_{r m}}{\Pi^{(\omega)}_{k}}_{S_2} \\ &=\biginprod{\Delta_{\eta} \lambda^{(\omega)}_{k}\Pi^{(\omega)}_{k} }{\Pi^{(\omega)}_{k}}_{S_2}+ \biginprod{\lambda^{(\omega)}_k \sum_{n,m=1}^{\infty} a^{\eta}_{n,m}\Pi^{(\omega)}_{nm}}{\Pi^{(\omega)}_{k}}_{S_2}+\biginprod{(\Delta_{\eta} \lambda^{(\omega)}_{k} -\Delta_\eta \F^{\omega} ) \Delta_{\eta}\Pi^{(\omega)}_{k}}{\Pi^{(\omega)}_{k}}_{S_2}
\end{align*}
which becomes, due to orthogonality
\begin{align*}
\biginprod{\Delta_\eta \F^{\omega}\Pi^{(\omega)}_{k}}{\Pi^{(\omega)}_{k}}_{S_2}& + \lambda^{(\omega)}_{k} a^{\eta}_{k,k} =\Delta_{\eta} \lambda^{(\omega)}_{k}+  \lambda^{(\omega)}_{k} a^{\eta}_{k,k}+\biginprod{(\Delta_{\eta} \lambda^{(\omega)}_{k} -\Delta_\eta \F^{\omega} ) \Delta_{\eta}\Pi^{(\omega)}_{k}}{\Pi^{(\omega)}_{k}}_{S_2}\end{align*}
rearranging yields
\begin{align*}
\Delta_{\eta} \lambda^{(\omega)}_{k}& = \biginprod{\Delta_\eta \F^{\omega}\Pi^{(\omega)}_{k}}{\Pi^{(\omega)}_{k}}_{S_2} -\biginprod{(\Delta_{\eta} \lambda^{(\omega)}_{k} -\Delta_\eta \F^{\omega} ) \Delta_{\eta}\Pi^{(\omega)}_{k}}{\Pi^{(\omega)}_{k}}_{S_2} 
%\\& = \biginprod{\Delta_\eta \F^{\omega}}{\Pi^{(\omega)}_{k}}_{S_2} -\biginprod{(\Delta_{\eta} \lambda^{(\omega)}_{k} -\Delta_\eta \F^{\omega} )}{\Pi^{(\omega)}_{k}( \Delta_{\eta}\Pi^{(\omega)}_{k})^{\dagger}}_{S_2}. 
\\& = \biginprod{\Delta_\eta \F^{\omega}}{\Pi^{(\omega)}_{k}}_{S_2} -\biginprod{(\Delta_{\eta} \lambda^{(\omega)}_{k} -\Delta_\eta \F^{\omega} ) \Delta_{\eta}\Pi^{(\omega)}_{k}}{\Pi^{(\omega)}_{k}}_{S_2}. 
\end{align*}
\end{proof}

\subsection{Proof of \autoref{thm:NBMf}}\label{proofthmNBMf}

\begin{proof}[\bf Proof of \autoref{thm:NBMf}]
First, observe that $\{ N^{(\omega)}_{m,T,t},1 \le t \le T\}$ forms a (square integrable)  complex-valued martingale difference sequence for each $\omega \in [-\pi,\pi]$. To ease notation, we shall sometimes denote its integral in frequency direction over $[a,b]$ by
\[N_{m,T,t} = \int_a^b N^{(\omega)}_{m,T,t} d\omega, \quad 2 \le t \le T. \tageq \label{eq:NmTf}\]
We derive the result by verifying the conditions of Corollary 3.8 of \cite{McL74}. Firstly, observe that by Cauchy's inequality and \autoref{lem:Burkh}
\begin{align*}
\int_a^b \E\big| N^{(\omega)}_{m,T,t}|d\omega&
\le 2 \sup_{\omega }\norm{\mathcal{V}^{(\omega)}_{\X}}_{\Hs}\int_a^b\E\bignorm{ \sum_{s=1}^{t-1}  \tilde{w}^{(\omega)}_{b_T,t,s}  \mathcal{D}^{(\omega)}_{XY,t,s}}_{\Hs}d\omega
\\&
\le 2 \sup_{\omega }\big(\snorm{\mathcal{U}^{(\omega)}_{X,Y}}_{2}+\snorm{\mathcal{U}^{(\omega)}_{Y,X}}_{2}\big)\\&\times  \int_a^b \E\Big(\bigsnorm{ \sum_{s=1}^{t-1}  \tilde{w}^{(\omega)}_{b_{T},t,s} (\dmp{t}{X} \otimes \dmp{s}{X})}_{2}+ \bigsnorm{ \sum_{s=1}^{t-1}  \tilde{w}^{(\omega)}_{b_T,t,s} (\dmp{t}{Y} \otimes \dmp{s}{Y})}_{2}\Big)d\omega 
\\& \le2 \sup_{\omega }\big(\snorm{\mathcal{U}^{(\omega)}_{X,Y}}_{2}+\snorm{\mathcal{U}^{(\omega)}_{Y,X}}_{2}\big) \times
\\& \int_a^b \|\dmp{t}{X}\|_{\hi,2} \big(\E \bignorm{ \sum_{s=1}^{t-1}  \tilde{w}^{(\omega)}_{b_T,t,s} \dmp{s}{X})}^2_\Hi\big)^{1/2} + \|\dmp{t}{Y}\|_{\hi,2} \big(\E \bignorm{ \sum_{s=1}^{t-1}  \tilde{w}^{(\omega)}_{b_T,t,s} \dmp{s}{Y})}^2_\Hi\big)^{1/2}  d\omega
\\& \le C  \max_{t}\big(\sum_{s=1}^{t-1} | {w}_{b_T,t,s} |^2 \big)^{1/2} = O(b^{-1/2}_T). \tageq \label{eq:FubN}
\end{align*}
for some constant $C$. Therefore, since $T$ is fixed, Fubini's theorem and the tower property imply that, for any 
$G \in  \G_{t-1}$,
\begin{align*}
\sum_{t=1}^{\flo{\eta T}}\Big \vert\E\big[\mathrm{1}_{G}\int_a^b N^{(\omega)}_{m,T,t} d\omega\big]\Big \vert 
& = \sum_{t=1}^{\flo{\eta T}}\Big \vert\E\big[\int_a^b \mathrm{1}_{G} N^{(\omega)}_{m,T,t} d\omega\big]\Big \vert 
%\\&= \sum_{t=1}^{\flo{\eta T}}\Big \vert \int_a^b \E\big[\mathrm{1}_{G} N^{(\omega)}_{m,T,t}\big] d\omega\Big \vert 
\\&= \sum_{t=1}^{\flo{\eta T}}\Big \vert\int_a^b \E\Big[\E[ \mathrm{1}_{G} N^{(\omega)}_{m,T,t}| \G_{t-1} \big]\Big]d\omega \Big \vert 
\\&= \sum_{t=1}^{\flo{\eta T}}\Big \vert\int_a^b \E \Big[  \mathrm{1}_{G}\E[N^{(\omega)}_{m,T,t}| \G_{t-1} \big]\Big]d\omega  \Big \vert=0 \quad \forall \eta \in [0,1],  
\end{align*}
where the last equality follows from the fact that $\{ N^{(\omega)}_{m,T,t},1 \le t \le T\}$ forms a martingale difference sequence with respect to the filtration $\{\G_t\}$. %Since % \todo{$\mathbb{P}(\{\G_{t-1}\})>0$}
%$\E[\mathrm{1}_{\G_{t-1}}N_{m,T,t} ]= \mathbb{P}(\{\G_{t-1}\}) \E [N_{m,T,t}|\G_{t-1}]$, 
We therefore obtain
\[
\sum_{t=1}^{\flo{\eta T}}\Big \vert\E\big[\int_a^b N^{(\omega)}_{m,T,t} d\omega \big \vert \G_{t-1}\big]\Big \vert =0  \quad \text{ a.s. } \forall T \in \mathbb{N}, \forall \eta \in [0,1],
\]
showing that condition (3.11) of \cite{McL74} is satisfied. Next, we verify conditions (3.9) and (3.10) of \cite{McL74}. These follow almost along the same lines as in the proof of Theorem 3.2 of \cite{vD19}. Therefore, we only give the main steps. From Jensen's inequality, Tonelli's theorem, Cauchy Schwarz's inequality  and  \autoref{lem:secM}, we obtain
\begin{align*}
&\E\Big(  \max_{1 \le k \le T}\Big \vert \int_a^{b} \biginprod{ \sum_{t=2}^{k} \sum_{s=t-4m+1 \vee 1}^{t-1}  \tilde{w}^{(\omega)}_{b_T,t,s} \mathcal{D}^{(\omega)}_{XY,t,s}}{\mathcal{V}^{(\omega)}_{\X}}_{\Hs} d\omega \Big \vert\Big)^\gamma
\\& \le(b-a)^{\gamma-1} \int_a^{b} \E\Big(  \max_{1 \le k \le T}\Big \vert  \biginprod{ \sum_{t=2}^{k} \sum_{s=t-4m+1 \vee 1}^{t-1}  \tilde{w}^{(\omega)}_{b_T,t,s} \mathcal{D}^{(\omega)}_{XY,t,s}}{\mathcal{V}^{(\omega)}_{\X}}_{\Hs} \Big \vert\Big)^\gamma d\omega 
\\& \le  (b-a)^{\gamma-1}  2\sup_{\omega}\norm{\mathcal{V}^{(\omega)}_{\X}}^\gamma_{\Hs}\int_a^{b} \E\Big(  \max_{1 \le k \le T} \bignorm{ \sum_{t=2}^{k}\sum_{s=t-4m+1 \vee 1}^{t-1}  \tilde{w}^{(\omega)}_{b_T,t,s} \mathcal{D}^{(\omega)}_{XY,t,s}}_{\Hs}\Big)^\gamma d\omega 
 =o(\mathcal{W}^{\gamma}_{b_T}),
\end{align*}
where we used in the last equation that $o(\frac{T}{b_T}) = o(\mathcal{W}^{2}_{b_T})$.
%%%%%%%%%%%%%%%%%%%%%%%%%%%%%%%%%%%%%%%%%5
\begin{comment}
\begin{align*}
  \sum_{t=2}^{\flo{\eta T}}&\E \Big\vert \int_a^{b} \biginprod{ \sum_{s=t-4m+1 \vee 1}^{t-1}  \tilde{w}^{(\omega)}_{b_T,t,s} \mathcal{D}^{(\omega)}_{XY,t,s}}{\mathcal{V}^{(\omega)}_{\X}}_{\Hs} d\omega \Big \vert^2
%\\ &\le   \sum_{t=2}^{\flo{\eta T_1}}2\E \Big\vert \int_a^{b} \biginprod{ \sum_{s=1 }^{t-4m}  \tilde{w}^{(\omega)}_{b_T,t,s} (\dm{t} \otimes \dm{s})}{\mathcal{U}^{(\omega)}_{X,Y}}_{S_2}d\omega \Big \vert^2+
%\\& +\sum_{t=2}^{\flo{\eta T_1}} 2\E  \Big \vert\int_a^b \biginprod{\sum_{s=t-4m+1 \vee 1}^{t-1}  \tilde{w}^{(\omega)}_{b_T,t,s} (\dm{t} \otimes \dm{s})}{\mathcal{U}^{(\omega)}_{X,Y}}_{S_2} d\omega  \Big \vert^2
\\& \le  (b-a) \int_a^{b} \E  \sum_{t=2}^{\flo{\eta T}}\Big\vert \biginprod{ \sum_{s=t-4m+1 \vee 1}^{t-1}  \tilde{w}^{(\omega)}_{b_T,t,s} \mathcal{D}^{(\omega)}_{XY,t,s}}{\mathcal{V}^{(\omega)}_{\X}}_{\Hs}\Big \vert^2 d\omega 
\\& \le  (b-a) \int_a^{b} 2\sup_{\omega}\norm{\mathcal{V}^{(\omega)}_{\X}}^2_{\Hs} \E  \sum_{t=2}^{\flo{\eta T}}\bignorm{ \sum_{s=t-4m+1 \vee 1}^{t-1}  \tilde{w}^{(\omega)}_{b_T,t,s} \mathcal{D}^{(\omega)}_{XY,t,s}}^2_{\Hs} d\omega 
 =o(\mathcal{W}^{2}_{b_T}),
\end{align*}
where we used in the last equation that $o(\frac{\flo{\eta T}}{b_T}) = o(\mathcal{W}^{2}_{b_T})$. 
\end{comment}
%%%%%%%%%%%%%%%%%%%%%%%%%%%%%
It is therefore sufficient to focus on 
\[
\mathcal{W}^{-1}_{b_T}\sum_{t=1+4m}^{\flo{\eta T}} \int_a^b \tilde{N}^{(\omega)}_{m,T,t} d\omega, \tageq \label{eq:Ntildeom}
\]
where 
$$\tilde{N}^{(\omega)}_{m,T,t}= \biginprod{ \sum_{s=1 }^{t-4m}  \tilde{w}^{(\omega)}_{b_T,t,s}\mathcal{D}^{(\omega)}_{XY,t,s}}{\mathcal{V}^{(\omega)}_{\X}}_{\Hs} +\biginprod{  \sum_{s=1}^{t-4m}\big(\tilde{w}^{(\omega)}_{b_T,s,t} \mathcal{D}^{(\omega)}_{XY,s,t}\big)^\dagger}{\mathcal{V}^{(\omega)}_{\X}}_{\Hs}~.$$
To verify the conditional Lindeberg condition, observe that
\begin{align*}
\sum_{t=1+4m}^{\flo{\eta T}}\E\Big\{|\tilde{N}_{m,T,t}\big|^2 \mathrm{1}_{|\tilde{N}_{m,T,t}|>\epsilon}\Big\}
&\le   2 \sum_{t=1+4m}^{\flo{\eta T}}\E\Big\{\Big\vert \int_a^b \biginprod{ \sum_{s=1}^{t-4m}  \tilde{w}^{(\omega)}_{b_T,t,s} \mathcal{D}^{(\omega)}_{XY,t,s}}{\mathcal{V}^{(\omega)}_{\X}}_{\Hs}d\omega \Big \vert^2\mathrm{1}_{|\tilde{N}_{m,T,t}|>\epsilon}\Big\}
\\&+2 \sum_{t=1+4m}^{\flo{\eta T}}\E\Big\{\Big\vert \int_a^b\biginprod{  \sum_{s=1}^{t-4m}\big(\tilde{w}^{(\omega)}_{b_T,s,t} \mathcal{D}^{(\omega)}_{XY,s,t}\big)^\dagger}{\mathcal{V}^{(\omega)}_{\X}}_{\Hs} d\omega \Big \vert^2\mathrm{1}_{|\tilde{N}_{m,T,t}|>\epsilon}\Big\}.
\end{align*}
We only verify the first term, because the second is of the same order. Jensen's inequality, 
\autoref{lem:Burkh}, and independence of $ \dmp{t}{\cdot}$ and  $\dmp{s}{\cdot}$, $|s-t|\ge 4m$, yield under \autoref{as:depstruc},
%\begin{align*}
%&\frac{1}{\mathcal{W}^{2}_{b_T}}  \sum_{t=1+4m}^{\flo{\eta T}}\E\Big\{\int_a^b\biginprod{ \sum_{s=1}^{t-1}  \tilde{w}^{(\omega)}_{b_T,t,s}\mathcal{D}^{(\omega)}_{XY,t,s}}{\mathcal{V}^{(\omega)}_{\X}}_{\Hs} d\omega \Big \vert^2\mathrm{1}_{|N_{m,T,t}|>\epsilon}\Big\}
%\\& 
%\le %\frac{(b-a)^3}
%\frac{C}{\mathcal{W}^{4}_{b_T}} \int_a^{b}\sum_{t=1+4m}^{\flo{\eta T}}\E\Big \vert\biginprod{ \sum_{s=1}^{t-1}  \tilde{w}^{(\omega)}_{b_T,t,s} \mathcal{D}^{(\omega)}_{XY,t,s}}{\mathcal{V}^{(\omega)}_{\X}}_{\Hs} \Big \vert^4 d\omega
% \\& =O\big(\frac{b^2_T}{ \kappa^2 \flo{\eta T}^2}\big) O\big(\frac{\flo{\eta T} \kappa^2}{b^2_T}\big)= o(1)
%\end{align*}
\begin{align*}
&\sum_{t=1+4m}^{\flo{\eta T}}\E\Big\{\Big \vert \mathcal{W}^{-1}_{b_T}\int_a^b\biginprod{ \sum_{s=1}^{t-4m}  \tilde{w}^{(\omega)}_{b_T,t,s}\mathcal{D}^{(\omega)}_{XY,t,s}}{\mathcal{V}^{(\omega)}_{\X}}_{\Hs} d\omega \Big \vert^2\mathrm{1}_{|\mathcal{W}^{-1}_{b_T} \tilde{N}_{m,T,t}|>\epsilon}\Big\}
\\& 
\le %\frac{(b-a)^3}
\frac{1}{\epsilon^2}\frac{C}{\mathcal{W}^{4}_{b_T}} \int_a^{b}\sum_{t=1+4m}^{\flo{\eta T}}\E\Big \vert\biginprod{ \sum_{s=1}^{t-4m}  \tilde{w}^{(\omega)}_{b_T,t,s} \mathcal{D}^{(\omega)}_{XY,t,s}}{\mathcal{V}^{(\omega)}_{\X}}_{\Hs} \Big \vert^4 d\omega
 \\& =O\big(\frac{b^2_T}{ \kappa^2 T^2}\big) O\big(\frac{\flo{\eta T} \kappa^2}{b^2_T}\big)= o(1),
\end{align*}
for some constant $C$ as $T \to \infty$. %Using that $|x+\im y |= \sqrt{|x+\im y |^2}= \sqrt{x^2+y^2}$ and that $\{|\Re(N_{m,T,t})|>\epsilon\} \subset \{|N_{m,T,t}|>\epsilon\}$, 
Consequently, for all $\epsilon>0$
\begin{align*}
\lim_{T \to \infty} \sum_{t=1+4m}^{\flo{\eta T}}\E\Big\{ \big(\Re( \mathcal{W}^{-1}_{b_T} \tilde{N}_{m,T,t})\big)^2 \mathrm{1}_{|\Re(\mathcal{W}^{-1}_{b_T} \tilde{N}_{m,T,t})|>\epsilon}\Big\} =0.
\end{align*}
Finally, we verify condition (3.10) of \cite{McL74}. Observe first that for the conditional variance, we obtain
\begin{align*}
\frac{1}{\mathcal{W}^{2}_{b_T}}\sum_{t=1+4m}^{\flo{\eta T}}  \E[ |\tilde{N}_{m,T,t}|^2 | \G_{t-1}] &=
\frac{1}{\mathcal{W}^{2}_{b_T}}\sum_{t=1+4m}^{\flo{\eta T}}  \E\Big[\int_a^b \tilde{N}^{(\omega)}_{m,T,t} d\omega \overline{\int_a^b \tilde{N}^{(\lambda)}_{m,T,t} d\lambda}\Big| \G_{t-1}\Big] \tageq \label{eq:convar1}
%\\&=\frac{1}{\mathcal{W}^{2}_{b_T}}\sum_{t=1}^{\flo{\eta T}}\int_a^b \int_a^b   \E\Big[N^{(\omega)}_{m,T,t} \overline{N^{(\lambda)}_{m,T,t}} \Big| \G_{t-1}\Big]d\omega d\lambda
\\&=\frac{1}{\mathcal{W}^{2}_{b_T}} \sum_{t=1+4m}^{\flo{\eta T}}\lim_{n\to \infty}\Big(  \sum_{i_1=1}^{n}\frac{(b-a)}{n}  \E\Big[ |\tilde{N}^{\omega_{i_1}}_{m,T,t}|^2 \Big| \G_{t-1}\Big]\\& \phantom{\frac{1}{\mathcal{W}^{2}_{b_T}} \sum_{t=1}^{\flo{\eta T}}\lim_{n\to \infty}\Big(}+ \frac{(b-a)^2}{n^2} \sum_{\substack{i_1,i_2=1 \\i_1 \ne i_2}}^{n}  \E\Big[\tilde{N}^{\omega_{i_1}}_{m,T,t} \overline{\tilde{N}^{\omega_{i_2}}_{m,T,t}} \Big| \G_{t-1}\Big]\Big),
  \end{align*}
where Fubini's theorem justifies, via an analoguous reasoning to the above, the interchange of integrals. The conditions for Fubini's theorem can be verified via a similar derivation as in \eqref{eq:FubN}. More specifically, using the Cauchy-Schwarz inequality  and \autoref{lem:Burkh}, we obtain% from the independence of $ \dmp{t}{\cdot}$ and  $\dmp{s}{\cdot}$ for $|s-t|>m$ that,
\begin{align*}
\int_a^b \int_a^b   \E\Big\vert \tilde{N}^{(\omega)}_{m,T,t} \overline{\tilde{N}^{(\lambda)}_{m,T,t}} \Big|d\omega d\lambda & \le \int_a^b   \sqrt{\E |\tilde{N}^{(\omega)}_{m,T,t}|^2}d\omega  \int_a^b \sqrt{\E|\tilde{N}^{(\lambda)}_{m,T,t}|^2} d\lambda
\\& \le \Big(C \pi  \max_{t}\big(\sum_{s=1}^{t-1} | {w}_{b_T,t,s} |^2 \big)^{1/2}\Big)^2 = O(b^{-1}_T).
\end{align*}
%By separability of the $\Hi_i$, observe that we can write $\mathcal{V}^{(\omega)}_{XY} =
%(\mathcal{U}^{(\omega)}_{XY},\mathcal{U}^{(\omega)}_{YX})^\top= (u^{(\omega)}_{XY,X}\otimes v^{(\omega)}_{XY,X},u^{(\omega)}_{XY,Y}\otimes v^{(\omega)}_{XY,Y})^\top$ for some $u^{(\omega)}_{XY,X}, v^{(\omega)}_{XY,X},u^{(\omega)}_{XY,Y}, v^{(\omega)}_{XY,Y} \in \Hi$. 
Let $\mathcal{V}^{(\omega)}_{X}$ and $\mathcal{V}^{(\omega)}_{Y}$ be arbitrary elements of $S_2(\Hi)$. These may be written in their canonical form, i.e., in the form $\sum_i s^{(\omega)}_{l,i} (u^{(\omega)}_{l,i} \otimes v^{(\omega)}_{l,i})$ where, for fixed $\omega$, $\{u^{(\omega)}_{l,i}\}$ and $\{v^{(\omega)}_{l,i}\}$ are orthonormal bases of $\Hi$ and $\{s^{(\omega)}_{l,i}\}$ is a non-decreasing sequence of non-negative numbers converging to zero. Hence,   
\begin{equation} 
\mathcal{{V}}^{(\omega)}_{\X} = (\mathcal{V}^{(\omega)}_{X},\mathcal{V}^{(\omega)}_{Y})^{\top} =
 (\sum_i s^{(\omega)}_{X,i} u^{(\omega)}_{X,i}\otimes v^{(\omega)}_{X,i}, \sum_i s^{(\omega)}_{Y,i} u^{(\omega)}_{Y,i}\otimes v^{(\omega)}_{Y,i})^\top.
 \end{equation}
  %To ease notation let $u_{l,i}=u^{(\omega)}_{XY,l}$ and $v_{l,i}=v^{(\omega)}_{XY,l}, l \in \{X,Y\}$. 
 Using the definition of $\inprod{\cdot}{\cdot}_{\Hs}$ and of $\inprod{\cdot}{\cdot}_{S_2}$, and the fact that the latter is continuous with respect to the $S_2(\Hi)$-norm topology, we can write \eqref{eq:Ntildeom}
 \begin{align*}
\tilde{N}^{(\omega)}_{m,T,t} &= \sum_{s=1}^{t-4m}  \tilde{w}^{(\omega)}_{b_T,t,s}\biginprod{\mathcal{D}^{(\omega)}_{XY,t,s}}{\mathcal{{V}}^{(\omega)}_{\X}}_{\Hs}+\biginprod{  \big(\tilde{w}^{(\omega)}_{b_T,s,t} \mathcal{D}^{(\omega)}_{XY,s,t}\big)^\dagger}{\mathcal{V}^{(\omega)}_{\X}}_{\Hs}
%\\& =\sum_{l \in \{X,Y\}}\sum_{s=1}^{t-1}  \sum_{i=0}^{\infty} s^{(\omega)}_{l,i} \tilde{w}^{(\omega)}_{b_T,t,s}\biginprod{ (\dmp{t}{l} \otimes \dmp{s}{l})}{ s^{(\omega)}_{l,i} (u^{(\omega)}_{l,i}\otimes v^{(\omega)}_{l,i})}_{S_2(\Hi_i)}+\tilde{w}^{-\omega}_{b_T,s,t}\biginprod{ \dmp{s}{l} \otimes \dmp{t}{l}}{ s^{(\omega)}_{l,i} (u^{(\omega)}_{l,i}\otimes v^{(\omega)}_{l,i})}_{S_2(\Hi_i)}
%\\& =\sum_{l \in \{X,Y\}} \sum_{s=1}^{t-1}  \sum_{i=1}^{\infty} s^{(\omega)}_{l,i}   \tilde{w}^{(\omega)}_{b_T,t,s}\inprod{\dmp{t}{l}}{u_{l,i}} \inprod{v_{l,i}}{\dmp{s}{l}} +\tilde{w}^{-\omega}_{b_T,s,t}\inprod{v_{l,i}}{\dmp{t}{l}} \inprod{\dmp{s}{l}}{u_{l,i}}
\\& = \sum_{l \in \{X,Y\}} \sum_{i=1}^{\infty} s^{(\omega)}_{l,i} \Big( \dmp{t}{l}(u_{l,i}) \overline{\widetilde{J}^{(\omega)}_{m,b_T,t}{(v_{l,i})}} + \overline{\dmp{t}{l}(v_{l,i})} {\widetilde{J}^{(\omega)}_{m,b_T,t}{(u_{l,i})}}\Big)
 \end{align*}
 where we abbreviated
 \begin{align}
 \dmp{t}{l}(u_{l,i}) : = \inprod{\dmp{t}{l}}{u_{l,i}} ~ \text{ and } ~
 \widetilde{J}^{(\omega)}_{m,b_T,t}(v_{l,i})=\sum_{s=1}^{t-4m} \widetilde{w}^{(\omega)}_{b_T,t,s}\inprod{\dmp{s}{l}}{v_{l,i}} \quad l \in \{X,Y\}. \label{eq:J4m}
\end{align} 
%\begin{align*}
%&\E\Big[N^{\omega_{r}}_{m,T,t} \overline{N^{\omega_{r}}_{m,T,t}} \Big| \G_{t-1}\Big] 
%\\&= \sum_{l,j \in \{X,Y\}} \sum_{i,k=0}^{\infty} s^{(\omega)}_{l,i} s^{(\omega)}_{j,k} \E\Big[\Big( \dmp{t}{l}(u_{l,i}) \overline{\widetilde{J}^{(\omega)}_{m,b_T,t}{(v_{l,i})}} + \overline{\dmp{t}{l}(v_{l,i})} {\widetilde{J}^{(\omega)}_{m,b_T,t}{(u_{l,i})}}\Big) \Big( \overline{\dmp{t}{j}(u_{j,k})} {\widetilde{J}^{(\omega)}_{m,b_T,t}{(v_{j,k})}} + {\dmp{t}{j}(v_{j,k})}\overline{\widetilde{J}^{(\omega)}_{m,b_T,t}{(u_{j,k})}}\Big)\Big \vert \G_{t-1}\Big]
%\\&= \sum_{l,j \in \{X,Y\}} \sum_{i,k=0}^{\infty} s^{(\omega)}_{l,i}s^{(\omega)}_{j,k}  \E\Big[  \dmp{t}{l}(u_{l,i})  \overline{\dmp{t}{j}(u_{j,k})}
% \overline{\widetilde{J}^{(\omega)}_{m,b_T,t}{(v_{l,i})}} {\widetilde{J}^{(\omega)}_{m,b_T,t}{(v_{j,k})}} + \dmp{t}{l}(u_{l,i}) {\dmp{t}{j}(v_{j,k})}
% \overline{\widetilde{J}^{(\omega)}_{m,b_T,t}{(v_{l,i})}} \overline{\widetilde{J}^{(\omega)}_{m,b_T,t}{(u_{j,k})}}  
% \\& +  \overline{\dmp{t}{l}(v_{l,i})}   \overline{\dmp{t}{j}(u_{j,k})}{\widetilde{J}^{(\omega)}_{m,b_T,t}{(u_{l,i})}} {\widetilde{J}^{(\omega)}_{m,b_T,t}{(v_{j,k})}}+ \overline{\dmp{t}{l}(v_{l,i})}{\dmp{t}{j}(v_{j,k})} {\widetilde{J}^{(\omega)}_{m,b_T,t}{(u_{l,i})}}\overline{\widetilde{J}^{(\omega)}_{m,b_T,t}{(u_{j,k})}}  \Big \vert \G_{t-1}\Big]
%\end{align*}
The rest of the proof now follows similar to theorem 3.2 of \cite{vD19}. We use that we can write the left-hand side of \eqref{eq:convar1} in the form $\E(\cdot|\G_{t-1}) = \sum_{k=1}^{m} P_{t-k}(\cdot)+\E(\cdot|\G_{t-m-1})$. Similar techniques as above will then allow one to show that the first term on the right-hand side of the latter is of  order $o_p(1)$. Then using the fact that $\dmp{t}{j}(\cdot)$ is $\G_{t,t-m}$-measurable and $ \widetilde{J}^{(\omega)}_{m,b_T,t}(\cdot)$ is $\G_{t-4m}$-measurable, we obtain from Lemma B.2 of \cite{vD19} that
  \begin{align*}
&\frac{1}{\mathcal{W}^{2}_{b_T}} \sum_{t=1+4m}^{\flo{\eta T}}\lim_{n\to \infty}\Big(  \sum_{r=1}^{n}\frac{(b-a)}{n}  \E\Big[\tilde{N}^{(\omega_{r})}_{m,T,t} \overline{\tilde{N}^{(\omega_{r})}_{m,T,t}} \Big| \G_{t-1}\Big]+ \frac{(b-a)^2}{n^2}  \sum_{\substack{r_1,r_2=1 \\r_1 \ne r_2}}^{n}  \E\Big[\tilde{N}^{(\omega_{r_1})}_{m,T,t} \overline{\tilde{N}^{(\omega_{r_2})}_{m,T,t}} \Big| \G_{t-1}\Big]\Big)
%\\= &\frac{1}{\mathcal{W}^{2}_{b_T}} \sum_{t=1+4m}^{\flo{\eta T}}\lim_{n\to \infty}\Big(  \sum_{r=1}^{n}\frac{(b-a)}{n}  \E\Big[\tilde{N}^{\omega_{r}}_{m,T,t} \overline{\tilde{N}^{\omega_{r}}_{m,T,t}} \Big| \G_{t-1}\Big]+ \frac{(b-a)^2}{n^2}  \sum_{\substack{r_1,r_2=1 \\r_1 \ne r_2}}^{n}  \E\Big[\tilde{N}^{\omega_{r_1}}_{m,T,t} \overline{\tilde{N}^{\omega_{r_2}}_{m,T,t}} \Big| \G_{t-m-1}\Big]\Big)+o_p(1)
\\&= \frac{1}{\mathcal{W}^{2}_{b_T}} \sum_{t=1+4m}^{\flo{\eta T}} \lim_{n\to \infty}\Big\{\frac{(b-a)}{n}\sum_{r=1}^n \sum_{j,l \in \{X,Y\}}  \sum_{k,i=0}^{\infty} s^{(\omega_r)}_{l,i} s^{(\omega_r)}_{j,k}\Big( \E  \dmpi{t}{r}{l}(u_{l,i})  \overline{\dmpi{t}{r}{j}(u_{j,k})}\E  \overline{\widetilde{J}^{(\omega_r)}_{m,b_T,t}{(v_{l,i})}} {\widetilde{J}^{(\omega_r)}_{m,b_T,t}{(v_{j,k})}}
\\&\phantom{  \frac{1}{\mathcal{W}^{2}_{b_T}} \sum_{t=1+4m}^{\flo{\eta T}} \lim_{n\to \infty}\Big\{} +\E\Big[\dmpi{t}{r}{l}(u_{l,i}) {\dmpi{t}{r}{j}(v_{j,k})}\Big]
 \E\Big[\overline{\widetilde{J}^{(\omega_r)}_{m,b_T,t}{(v_{l,i})}} \overline{\widetilde{J}^{(\omega_r)}_{m,b_T,t}{(u_{j,k})}}  \Big]
 \\&\phantom{ \frac{1}{\mathcal{W}^{2}_{b_T}} \sum_{t=1+4m}^{\flo{\eta T}} \lim_{n\to \infty}\Big\{} + \E \Big[\overline{\dmpi{t}{r}{l}(v_{l,i})}   \overline{\dmpi{t}{r}{j}(u_{j,k})}\Big]\E\Big[{\widetilde{J}^{(\omega_r)}_{m,b_T,t}{(u_{l,i})}} {\widetilde{J}^{(\omega_r)}_{m,b_T,t}{(v_{j,k})}}\Big]
\\&\phantom{\frac{1}{\mathcal{W}^{2}_{b_T}} \sum_{t=1+4m}^{\flo{\eta T}} \lim_{n\to \infty}\Big\{} + \E\Big[\overline{\dmpi{t}{r}{l}(v_{l,i})}{\dmpi{t}{r}{j}(v_{j,k})}\Big] \E \Big[{\widetilde{J}^{(\omega_r)}_{m,b_T,t}{(u_{l,i})}}\overline{\widetilde{J}^{(\omega_r)}_{m,b_T,t}{(u_{j,k})}} \Big]\Big) \Big\}
+o_p(1). \tageq \label{eq:someting}
\end{align*}
%Now let $x= u_{l,i}$, $y= u_{j,k}$, $w=v_{l,i}$ and $z=v_{j,k}$ then we can write the term in round brackets
%\begin{align*}
%&  \E \big[ \dmpi{t}{r}{l}(x)\overline{\dmpi{t}{r}{j}(y)}\big]\E\big[ \overline{\widetilde{J}^{(\omega_{r})}_{m,b_T,t}{(w)}}{\widetilde{J}^{(\omega_{r})}_{m,b_T,t}{(z)}}\big]  +\E\big[ \overline{\dmpi{t}{r}{l}(w)} \overline{\dmpi{t}{r}{j}(y)}\big]\E \big[{\widetilde{J}^{(\omega_{r})}_{m,b_T,t}{(x)}} {\widetilde{J}^{(\omega_{r})}_{m,b_T,t}{(z)}}\big]
%\\& + \E\big[\dmpi{t}{r}{l}(x)  \dmpi{t}{r}{j}(z)\E\overline{\widetilde{J}^{(\omega_{r})}_{m,b_T,t}{(w)}} \overline{\widetilde{J}^{(\omega_{r})}_{m,b_T,t}{(y)}}\big]
%+ \E\big[ \overline{\dmpi{t}{r}{l}(w)}{\dmpi{t}{r}{j}(z)}\big] \E \big[{\widetilde{J}^{(\omega_{r})}_{m,b_T,t}{(x)}}\overline{\widetilde{J}^{(\omega_{r})}_{m,b_T,t}{(y)}}\Big)\big]
%\end{align*}
Furthermore, using \autoref{lem:inpsZt} and 
\[\E \inprod{\dmp{t}{l}}{x}\overline{\inprod{\dmp{t}{j}}{z}} = \sum_{|k|\le m}  \inprod{C^{(m)}_{l j,k}(z)}{x}e^{-\im \lambda k} =2\pi \inprod{\F^{(\lambda)}_{l,j,m}(z)}{x},\]
where $C^{(m)}_{l j,k}=\E(l^{(m)}_{k}\otimes j_{0}^{(m)})$ for $l, j \in \{X,Y\}$ 
\citep[see proof of Proposition 3.2 of][]{vD19}, 
%\begin{align*}
%& 4\pi^2 \sum_{s=1}^{t-4m}  w(b_T(t-s))^2\Big(\biginprod{\F_{l,j,m}^{(\omega_r)}(y)}{{x}} \biginprod{\F_{j,l,m}^{(\omega_r)}(w)}{(z)} +\mathrm{1}_{\omega_r \in \{0,\pi\}} \biginprod{\F_{l,j,m}^{(\omega_r)}(y)}{\overline{w}} \biginprod{\F_{j,l,m}^{(\omega_r)}(\overline{x})}{(z)} 
%\\&
 %+\biginprod{\F_{j,l,m}^{(\omega_r)}(w)}{z} \biginprod{\F_{l,j,m}^{(\omega_r)}(y)}{(x)}  +\mathrm{1}_{\omega_r \in \{0,\pi\}} \biginprod{\F_{l,j,m}^{(\omega_r)}(\overline{x})}{z} \biginprod{\F_{j,l,m}^{(\omega_r)}({y})}{\overline{w} } \Big)
%\\&= 4\pi^2 \sum_{s=1}^{t-4m}  w(b_T(t-s))^2\Big(  \biginprod{ \big(\F_{l,j,m}^{(\omega_r)} \widetilde{\otimes} \F_{j,l,m}^{\dagger(\omega_r)} \big)(y \otimes z)}{x \otimes w}_{S_2} +\mathrm{1}_{\omega_r \in \{0,\pi\}} \biginprod{ \big(\F_{l,j,m}^{(\omega_r)} \widetilde{\otimes} \F_{j,l,m}^{\dagger(\omega_r)} \big)(y \otimes z)}{\overline{w} \otimes \overline{x}}_{S_2} 
% \\&+\biginprod{ \big(\F_{j,l,m}^{(\omega_r)} \widetilde{\otimes} \F_{l,j,m}^{\dagger(\omega_r)} \big)(w \otimes x)}{z  \otimes y}_{S_2}\mathrm{1}_{\omega_r \in \{0,\pi\}} \biginprod{ \big(\F_{j,l,m}^{(\omega_r)} \widetilde{\otimes} \F_{l,j,m}^{\dagger(\omega_r)} \big)(y \otimes z)}{\overline{w} \otimes \overline{x}}_{S_2} 
%\end{align*}
 \eqref{eq:someting} becomes
\begin{align*}
\frac{4\pi^2\sum_{t=1+4m}^{\flo{\eta T}} \sum_{s=1}^{t-4m}  w(b_T(t-s))^2}{\sum_{t=1}^{T} \sum_{s=1}^{T}  w(b_T(t-s))^2} \lim_{n\to \infty}\Big\{&\frac{(b-a)}{n}\sum_{r=1}^n \sum_{j,l \in \{X,Y\}}  \sum_{k,i=1}^{\infty} s^{(\omega)}_{l,i} s^{(\omega)}_{j,k}\Big( \biginprod{ \big(\F_{l,j,m}^{(\omega_r)} \widetilde{\otimes} \F_{j,l,m}^{\dagger(\omega_r)} \big)(u_{j,k} \otimes v_{j,k})}{u_{l,i}\otimes v_{l,i}}_{S_2} \\& +\mathrm{1}_{\omega_r \in \{0,\pi\}} \biginprod{ \big(\F_{l,j,m}^{(\omega_r)} \widetilde{\otimes} \F_{j,l,m}^{\dagger(\omega_r)} \big)(u_{j,k}\otimes v_{j,k})}{\overline{v_{l,i}} \otimes \overline{u_{l,i}}}_{S_2} \\&
 +\biginprod{ \big(\F_{l,j,m}^{(\omega_r)} \widetilde{\otimes} \F_{j,l,m}^{\dagger(\omega_r)} \big)(v_{l,i} \otimes u_{l,i})}{v_{j,k}  \otimes u_{j,k}}_{S_2} \\&+\mathrm{1}_{\omega_r \in \{0,\pi\}} \biginprod{ \big(\F_{l,j,m}^{(\omega_r)} \widetilde{\otimes} \F_{j,l,m}^{\dagger(\omega_r)} \big)(u_{j,k}\otimes v_{j,k})}{\overline{v_{l,i}} \otimes \overline{u_{l,i}}}_{S_2} \Big)\Big\}
+o_p(1)
\\& = \frac{4\pi^2\sum_{t=1+4m}^{\flo{\eta T}} \sum_{s=1}^{t-4m}  w(b_T(t-s))^2}{\sum_{t=1}^{T} \sum_{s=1}^{T}  w(b_T(t-s))^2} \int_{0}^{\pi} \Gamma^{\omega}_{m,\X}\big(\mathcal{{V}}^{(\omega)}_{\X}\big)+o_p(1)
\end{align*}
 where, by continuity of the inner product $\inprod{\cdot}{\cdot}_{S_2}$ with respect to the $S_2$-norm topology,
  \begin{align*}
 \Gamma^{\omega}_{\X,m}\big(\mathcal{{V}}^{(\omega)}_{\X}\big) &=4\pi^2  \sum_{l,j \in \{X,Y\}}  \biginprod{ \big(\F_{l,j,m}^{(\omega)} \widetilde{\otimes} \F_{j,l,m}^{\dagger(\omega)} \big)(\mathcal{V}^{(\omega)}_{j})}{\mathcal{V}^{(\omega)}_{l}}_{S_2}  +\mathrm{1}_{\omega \in \{0,\pi\}} \biginprod{ \big(\F_{l,j,m}^{(\omega)} \widetilde{\otimes} \F_{j,l,m}^{\dagger(\omega)} \big)(\mathcal{V}^{(\omega)}_{j})}{\overline{\mathcal{V}^{\dagger(\omega)}_{l}}}_{S_2}  \\&\phantom{ \frac{1}{\mathcal{W}^{2}_{b_T}}} 
 +\biginprod{ \big(\F_{j,l,m}^{(\omega)} \widetilde{\otimes} \F_{l,j,m}^{\dagger(\omega)} \big)(\mathcal{V}^{\dagger(\omega)}_{l})}{\mathcal{V}^{\dagger(\omega)}_{j}}_{S_2} +\mathrm{1}_{\omega  \in \{0,\pi\}} \biginprod{\big(\F_{j,l,m}^{(\omega)} \widetilde{\otimes} \F_{l,j,m}^{\dagger(\omega)} \big)({\mathcal{V}^{(\omega)}_{j}})}{\overline{(\mathcal{V}^{\dagger(\omega)}_{j}})}_{S_2}.
\end{align*}
Similar to the conditional covariance, we obtain for the conditional pseudo-covariance 
\begin{align*}
\frac{1}{\mathcal{W}^{2}_{b_T}}\sum_{t=1+4m}^{\flo{\eta T}}  \E[ (\tilde{N}_{m,T,t})^2 | \G_{t-1}] &= \frac{\sum_{t=1+4m}^{\flo{\eta T}} \sum_{s=1}^{t-4m}  w(b_T(t-s))^2}{\sum_{t=1}^{T} \sum_{s=1}^{T}  w(b_T(t-s))^2} \int_a^b \Sigma^{\omega}_{X,m}\big(\mathcal{{V}}^{(\omega)}_{\X}\big) d\omega+o_p(1)
\end{align*}
where now, 
\begin{align*}
 \Sigma^{\omega}_{\X,m}\big(\mathcal{{V}}^{(\omega)}_{\X}\big) &=4\pi^2  \sum_{l,j \in \{X,Y\}}  \Big( \biginprod{ \big(\F_{l,j,m}^{(\omega)} \widetilde{\otimes} \F_{j,l,m}^{\dagger(\omega)} \big)(\mathcal{V}^{\dagger(\omega)}_{j})}{\mathcal{V}^{(\omega)}_{l}}_{S_2} + \biginprod{ \big(\F_{j,l,m}^{(\omega)} \widetilde{\otimes} \F_{l,j,m}^{\dagger(\omega)} \big)({\mathcal{V}^{\dagger(\omega)}_{l}})}{{\mathcal{V}^{(\omega)}_{j}}}_{S_2}
\\& \phantom{4\pi^2  \sum_{l,j \in \{X,Y\}} }\mathrm{1}_{\omega \in \{0,\pi\}}  \biginprod{ \big(\F_{l,j,m}^{(\omega)} \widetilde{\otimes} \F_{j,l,m}^{\dagger(\omega)} \big)(\overline{\mathcal{V}^{(\omega)}_{j}})}{\mathcal{V}^{(\omega)}_{l}}_{S_2} 
+ \mathrm{1}_{\omega \in \{0,\pi\}}  \biginprod{ \big(\F_{j,l,m}^{(\omega)} \widetilde{\otimes} \F_{l,j,m}^{\dagger(\omega)} \big)({\mathcal{V}^{\dagger(\omega)}_{l}})}{\overline{\mathcal{V}^{\dagger (\omega)}_{j}}}_{S_2} \Big).
\end{align*}
By a change of variables;
\begin{align*}\frac{\sum_{t=1}^{\flo{\eta T}} \sum_{s=1}^{\flo{\eta T}}  w(b_T(t-s))^2}{\sum_{t=1}^{T} \sum_{s=1}^{T}  w(b_T(t-s))^2}
&=\frac{\sum_{|h|< \flo{\eta T}} (\flo{\eta T}-|h|)  w(b_T h )^2}{\sum_{|h|< T} (T-|h)  w(b_T h)^2}\\& = \frac{\frac{\sum_{|h|< \flo{\eta T}} (\flo{\eta T}-|h|)  w(b_T h )^2}{\flo{\eta T}b_T^{-1} \int |w(x)|^2 d x} \flo{\eta T}b_T^{-1} \int |w(x)|^2 d x}{\frac{\sum_{|h|< T} (T-|h|)  w(b_T h )^2}{ Tb_T^{-1} \int |w(x)|^2 d x} T b_T^{-1} \int |w(x)|^2 d x} 
\\& \to \eta, \end{align*}
as $T \to \infty$,
% where we used the dominated convergence theorem and that $\sum_{|h|\ge \frac{\flo{\eta T} b_T}{b_T}} w(b_T h)^2 =o(1/b_T)$, 
which follows from  \autoref{as:Weights} together with \autoref{as:ratiorates}%, where the latter implies  $\flo{\eta T} b_T \to \infty$ as $T \to \infty$ for any $\eta \in [0,1]$
. Therefore, we obtain for fixed $m$ that
\[\frac{\sum_{t=1+4m}^{\flo{\eta T}} \sum_{s=1}^{t-4m}  w(b_T(t-s))^2}{\sum_{t=1}^{T} \sum_{s=1}^{T}  w(b_T(t-s))^2}=\frac{\eta}{2} +o(1).\]
Observe then that  
%simply because \E Re^2(x) =\E (\frac{1}{2} x+\overline{x})^2 ) =\frac{1}{4}(  E x^2+ E \overline{x}^2+E x \overline{x}+E \overline{x} x)= \frac{1}{2} \Re(E x^2)+\frac{1}{2} \Re(\E x \overline{x})
\[
\sum_{t=1}^{\flo{\eta T}} \E[ \big(\Re(N_{m,T,t})\big)^2 | \G_{t-1}] = \frac{1}{2}\sum_{t=1}^{\flo{\eta T}} \Re\Big(\E[ |N_{m,T,t}|^2 | \G_{t-1}]\Big) + \frac{1}{2}\sum_{t=1}^{\flo{\eta T}} \Re\Big(\E[ (N_{m,T,t})^2 | \G_{t-1}]\Big)
\]
which, together with the above, yields
\[
\frac{1}{\mathcal{W}^{2}_{b_T}} \sum_{t=1}^{\flo{\eta T}} \E[ \big(\Re(N_{m,T,t})\big)^2 | \G_{t-1}] \to \eta \cdot \frac{1}{4} \Re\Big(\int_a^b \Gamma^{\omega}_{\X,m}\big(\mathcal{V}^{(\omega)}_{\X}\big)+\Sigma^{\omega}_{\X,m}\big(\mathcal{V}^{(\omega)}_{\X}\big) d\omega \Big)
\]
in probability as $T \to \infty$. The result now follows.
\end{proof}

\section{Some technical results and auxiliary statements}\label{sec:proofs}
\def\theequation{B.\arabic{equation}}
\setcounter{equation}{0}

\begin{lemma}[Maximum of partial sums] \label{lem:boundseqSDO}
Let $\{X_t\}_{t \in \znum}$ satisfy Conditions \ref{I}-\ref{II} with $p=4+\epsilon$ and let the complex-valued array 
$\big\{\tilde{w}^{(\omega)}_{b_T,s,t}\big\}_{s,t \in \mathbb{N}}$ be defined through \eqref{weight}. Furthermore, suppose Assumption \ref{as:Weights} and Assumption \ref{as:bwrates} for some $\ell \ge 1$ are satisfied.
Denote the spectral density operators of $\{X_t\}$ by $\F^{(\omega)}$ and consider the partial sum process
\[
S^{\omega}_{k} = \sum_{s=1}^{k}  \sum_{t=1}^{k} \tilde{w}^{(\omega)}_{b_T,s,t}\big( X_s \otimes X_t \big)  \quad 1 \le k \le T.
\]
Then
\begin{align*}
\max_{1\le  k\le T}\frac{b^{1-i}_T}{k^i} \snorm{S^{\omega}_{k}-k \F^{\omega}}_2 &=
O_p(\log^{1/\gamma}(T)),      \quad 1/2 \le i \le 1.
\end{align*}
where $\gamma =2+\epsilon/2$.
\end{lemma}

\begin{proof}
It follows from the proof of Theorem 4.1 of \cite{vD19}, that for $q>2$
\begin{align*}\norm{S^{\omega}_{k}-\E S^{\omega}_{k}}^{2}_{S_2,q}
&=   k\sum_{h=1}^{k-1} \big\vert \tilde{w}_{\bf}^{(\lambda)}(b_T h) \big\vert^2 O\Big(
  \big(\sum_{t=0}^{\infty} \nu_{\hi,2q }(X_t)\big)^{4}\Big)
\end{align*}
Thus, if we take $g(b,k)= k \sum_{h=b+1}^{b+k}\big\vert {w}_{\bf}(b_T h) \big\vert^2 $ we observe that for all $b \ge 0$ and $1\le k < k+l$,
\[ 
g(b,k) + g(b+k,l) \leq g(b,k+l),
 \tageq \label{eq:mor1}
\]
which implies condition (1.1) of Theorem 1 of \cite{Mor76} is satisfied. Therefore, we obtain
\[
\E (\max_{1 \le k \le T} \snorm{S^{\omega}_{k}-\E S^{\omega}_{k}}_2)^{\gamma} = O(b_T^{-1-\epsilon/2}T^{1+\epsilon/2}). \tageq \label{eq:varB1}
\]
Note moreover that, under the conditions of Proposition \autoref{prop:consF}, $\snorm{\,\E S^{\omega}_k-k \F^{\omega}}_2 = O(b^{\ell}_T k)$, and thus
\[
\norm{\,\E S^{\omega}_k-k \F^{\omega}}^{\gamma}_{S_2,\gamma}= O(b^{\ell \gamma}_T k^{\gamma}). \tageq \label{eq:biasB1}
\]
Hence,  
\[
\E \big( \max_{1\le k \le T} {\snorm{S^{\omega}_{k}-k\F^{\omega}}_2 } \big)^\gamma =O(b_T^{-\gamma/2}T^{\gamma/2})+ O(b^{\gamma \ell}_T T^\gamma).
\]

By standard arguments, we have for $1/2 \le i \le 1$
\begin{align*}
\max_{1\le  k\le T}\frac{b^{1-i}_T}{k^{i}} \snorm{S^{\omega}_{k}-k \F^{\omega}}_2 &
\le b^{1-i}_T \max_{1 \le j \le c \log(T)} \max_{\beta^{j-1} < k \le  \beta^j}\frac{1}{k^{i}} \snorm{S^{\omega}_{k}-k \F^{\omega}}_2
\\& \le b^{1-i}_T \max_{1 \le j \le c \log(T)}\beta^{-(j-1)i} \max_{\beta^{j-1} < k \le \beta^j} \snorm{S^{\omega}_{k}-k \F^{\omega}}_2, 
\end{align*}
since $x^{-i}$ is a decreasing function of $x$.  Hence, we obtain
\begin{align*}
&\mathbb{P}(\max_{1\le  k\le T}\frac{b^{1-i}_T}{{k}^i} \snorm{S^{\omega}_{k}-k \F^{\omega}}_2 > \varepsilon)
\\&
 \le b^{p(1-i)}_T  \sum_{j=1}^{c \log(T)} \beta^{-(j-1)pi} \varepsilon^{-p} \E \Big[\big( \max_{\beta^{j-1} < k \le \beta^{j}} {\snorm{S^{\omega}_{k}-k\F^{\omega}}_2 } \big)^p\Big]
\\&\le \beta^{pi}     \sum_{j=1}^{c \log(T)} \varepsilon^{-p}2^{p-1} C \Big[ \beta^{j p(1/2-i)} \min(\beta^{j p/2}, b^{-p(1/2)}_T)  b^{p(1-i)}_T   + \beta^{jp(1-i) } b^{p(1-i+\ell )}_{T}\Big]
\\&\le  \beta^{pi}     \sum_{j=1}^{c \log(T)} \varepsilon^{-p}2^{p-1} C \Big[ \min(b^{p(1-i)}_T \beta^{j p(1-i)}, \beta^{j p(1/2-i)} b^{p(1/2-i)}_T)    + \beta^{jp(1-i) } b^{p(1-i+\ell )}_{T}\Big].
\end{align*} 
Taking $b_{T}=T^{-\kappa}$ and noting that $\beta^{j} \le T$ for all $1 \le  j \le c \log(T)$, this is bounded by
\begin{align*}
 \beta^{pi} \tilde{C}     \sum_{j=1}^{c \log(T)} \varepsilon^{-p}\min(b^{p(1-i)}_T \beta^{j p(1-i)}, \beta^{j p(1/2-i)} b^{p(1/2-i)}_T) +\beta^{pi}   \tilde{C}     \sum_{j=1}^{c \log(T)} \varepsilon^{-p}  T^{p\big(1-i -\kappa(1-i+\ell)\big)}.\tageq \label{eq:upboundmax}
\end{align*}
Observe that if $i=1/2$, we can bound the first term by choosing $\varepsilon = c_2 \log^{1/p}(T) $ for some sufficiently large constant $c_2$. The second term can in this case be bounded by a constant if $b_T^{1+2\ell} T=O(1)$, i.e., if $\kappa = 1/(1+2\ell)$ and is of lower order if $\kappa > 1/(1+2\ell)$. %For example if $\ell=1$, the bandwidths must satisfy $b_T \le T^{-1/3}$, while if $\ell=2$, we find $b_T \le T^{-1/5}$.
 It follows therefore that under the conditions of \autoref{lem:boundseqSDO}.
\[
\max_{1\le  k\le T}\frac{\sqrt{b_T}}{\sqrt{k}} \snorm{S^{\omega}_{k}-k \F^{\omega}}_2  =O_p(\log^{1/\gamma}(T)).
\]
Consider then the case $i >1/2$. Let us first look at the second term of \eqref{eq:upboundmax}. Observe that $b_T^{1-i+\ell} T^{1-i}=O(1)$, i.e., if $\kappa = (1-i)/(1-i+\ell)$ and is of lower order if $\kappa > (1-i)/(1-i+\ell)$. Since $\frac{1/2}{ (1/2+\ell)}> \frac{1-i}{(1-i+\ell)}$ for any $1/2<i \le 1$, we obtain $(1-i)<\kappa(1-i+\ell)$ for any $1/2<i \le 1$ under the conditions of \autoref{lem:boundseqSDO} since these require $\kappa \ge \frac{1/2}{ (1/2+\ell)}$.
For the first term of \eqref{eq:upboundmax}, observe that if $\beta^j \le  b^{-1}_T$, the first term in the minimum is smallest, and vice versa if $\beta^j > b^{-1}_T$. Hence,
\begin{align*}
    \sum_{j=1}^{c \log(T)} \min(b^{p(1-i)}_T \beta^{j p(1-i)}, \beta^{j p(1/2-i)} b^{p(1/2-i)}_T) & \le \sum_{j:\beta^j \le  b^{-1}_T} b^{p(1-i)}_T b_{T}^{-p(1-i)} 
    \\&+ \sum_{j:\beta^j >  b^{-1}_T} \beta^{j p(1/2-i)} \beta^{-jp(1/2-i)} = O(\log(T)).
\end{align*} 
Hence, a similar derivation as in the above yields also in this case $ \max_{1 \le k \le T}\frac{b^{1-i}_T}{{k}^i} \snorm{S^{\omega}_{k}-k \F^{\omega}}_2 = O_p(\log^{1/\gamma}(T))$.

\end{proof}
\begin{thm} \label{thm:maxPS}
Suppose Assumptions \ref{as:depstruc}-\ref{as:bwrates} are satisfied. Then,
\begin{align*}
\sup_{\eta \in [0,1]}\int^b_a \bigsnorm{\eta^{1/r}(\hat{\F}_X^{(\omega)}(\eta)- \F^{(\omega)}_X)}^{r}_2 d\omega  =
O_p({\log^{r/\gamma}(T) b^{-1}_{T} T^{-1}}),      \quad r \ge 2,
\end{align*}
 for some $\gamma > 2$.
\end{thm}
\begin{proof}
We have
\begin{align*}
\sup_{\eta \in [0,1]}\int^b_a \bigsnorm{\eta^{1/r}(\hat{\F}_X^{(\omega)}(\eta)- \F^{(\omega)}_X)}^{r}_2 d\omega
&\le \int^b_a \sup_{\eta \in [0,1]} \bigsnorm{\eta^{1/r}(\hat{\F}_X^{(\omega)}(\eta)- \F^{(\omega)}_X)}^{r}_2 d\omega  
\\& \le \int^b_a \Big( \sup_{\eta \in [0,1]} \bigsnorm{\eta^{1/r}(\hat{\F}_X^{(\omega)}(\eta)- \F^{(\omega)}_X)}_2\Big)^{r} d\omega,  
\end{align*}
We remark that the integrand on the right-hand side is measurable. Observe then that 
\[\frac{{\eta}^{1/r}}{\lfloor \eta T\rfloor} =\frac{1}{b^{1/r}_{T} T^{1/r}}\frac{ {(\eta T)}^{1/r} }{ (\flo{\eta T})^{1/r}} \frac{b^{1/r}_T}{(\flo{\eta  T})^{1-1/r}}\]
where, similar to the proof of \autoref{lem:floor}, we can set $\frac{\eta T}{\flo{\eta T}}=0$ for $\eta \le 1/T$ as otherwise the term is zero. Then using \eqref{eq:floor} and a change of variables $i=1-1/r$, therefore yields for $\eta \ge 1/T$, 
%\[\frac{{\eta}^{1-i}}{\lfloor \eta T\rfloor} = \frac{1}{b^{1-i}_{T} T^{1-i}}\frac{ {(\eta T)}^{1-i} }{ (\flo{\eta T})^{1-i}} \frac{b^{1-i}_T}{(\flo{\eta  T})^{i}} +O(T^{-1+i}).
%\]\todo{blue line below to replace this line}

\[\frac{{\eta}^{1-i}}{\lfloor \eta T\rfloor} = \frac{1}{b^{1-i}_{T} T^{1-i}}\frac{ {(\eta T)}^{1-i} }{ (\eta T)^{1-i}} \frac{b^{1-i}_T}{(\flo{\eta  T})^{i}} +\frac{1}{b^{1-i}_{T} T^{1-i}}O(T^{-1+i}) \frac{b^{1-i}_T}{(\flo{\eta  T})^{i}} .
\]
The result then follows from {the proof of} \autoref{lem:boundseqSDO} and Minkowski's inequality. 
\end{proof}

\begin{lemma} \label{lem:floor}
Let $\mathcal{Z}^{X,\omega}_{T,\eta}$ and $\mathcal{Z}^{Y,\omega}_{T,\eta}$ be defined as in \autoref{thm:BuildingBlock} and let $\mathcal{U}^{(\omega)}_{XY} \in S_2(\Hi)$, then 
\begin{align*}
&\sup_{\eta \in [0,1]}\Big \vert\int_a^b \biginprod{\eta(\hat{\F}_X^{(\omega)}(\eta)- \F_X^{(\omega)})}{\eta \mathcal{U}^{(\omega)}_{XY}} d\omega -\int_a^b \biginprod{\frac{\mathcal{Z}^{X,\omega}_{T,\eta}}{\sqrt{b_1 T_1}} }{\eta \mathcal{U}^{(\omega)}_{XY}}  d\omega \Big \vert =o_P\Big(\frac{1}{\sqrt{b_{1} T_1+ b_{2}T_2}}\Big); \tageq \label{eq:lemFbound1}
\\&
\sup_{\eta \in [0,1]}\Big \vert\int_a^b \biginprod{\eta(\hat{\F}_Y^{(\omega)}(\eta)- \F_Y^{(\omega)})}{\eta \mathcal{U}^{(\omega)}_{XY}} d\omega -\int_a^b \biginprod{\frac{\mathcal{Z}^{Y,\omega}_{T,\eta}}{\sqrt{b_2 T_2}} }{\eta \mathcal{U}^{(\omega)}_{XY}}  d\omega \Big \vert=o_P\Big(\frac{1}{\sqrt{b_{1} T_1+ b_{2}T_2}}\Big). \tageq \label{eq:lemFbound2}
\end{align*}
\end{lemma}
\begin{proof}
We prove only \eqref{eq:lemFbound1} as \eqref{eq:lemFbound2} follows similarly. Note from the definition of $ \mathcal{Z}^{X,\omega}_{T,\eta}$ and the triangle inequality, it suffices to show that 
\begin{align*}
&  \sup_{\eta \in [0,1]} \int_a^b  \Big \vert \frac{\eta T_1}{\flo{\eta T_1}}-1\Big \vert \Big \vert  \biginprod{\frac{1}{\sqrt{b_1 T_1}} \mathcal{Z}^{X,\omega}_{T,\eta}}{\eta \mathcal{U}^{(\omega)}_{XY}} \Big \vert d\omega =o_P\Big(\frac{1}{\sqrt{b_{1} T_1+ b_{2}T_2}}\Big). \tageq \label{eq:lemFbound1}
%\\& \sup_{\eta \in [0,1]} \int_a^b  \Big \vert \frac{\eta T_2}{\flo{\eta T_2}}-1\Big \vert \Big \vert  \biginprod{\frac{1}{\sqrt{b_2 T_2}} \mathcal{Z}^{Y,\omega}_{T,\eta}}{\eta \mathcal{U}(\omega} \Big \vert d\omega =o_P\Big(\frac{1}{\sqrt{b_{1} T_1+ b_{2}T_2}}\Big). \tageq \label{eq:lemFbound2}
\end{align*}
The Cauchy-Schwarz inequality implies
\begin{align*}
& \E \sup_{\eta \in [0,1]} \int_a^b  \Big \vert \frac{\eta T}{\flo{\eta T}}-1\Big \vert \Big \vert  \biginprod{\frac{1}{T_1} \sum_{s=1}^{\flo{\eta T_1}}\Big(\sum_{t=1}^{\flo{\eta T_1}}  \tilde{w}^{(\omega)}_{b_1,s,t} (X_s \otimes X_t) - \F_X^{(\omega)}\Big)}{\eta  \mathcal{U}^{(\omega)}_{XY}} \Big \vert d\omega
\\& \le \sup_{\eta \in [0,1]} \Big \vert \frac{\eta T}{\flo{\eta T}}-1\Big \vert  \E \int_a^b \sup_{\eta \in [0,1]}\bigsnorm{ \frac{1}{ T_1} \sum_{s=1}^{\flo{\eta T_1}}\Big(\sum_{t=1}^{\flo{\eta T_1}}  \tilde{w}^{(\omega)}_{b_1,s,t} (X_s \otimes X_t) - \F_X^{(\omega)}\Big)}_2 \bigsnorm{\eta  \mathcal{U}^{(\omega)}_{XY}}_2d\omega
\end{align*}
%and from \autoref{cor:floor} and that $\sup_{\eta} \snorm{\eta \mathcal{U}^{(\omega)}_{XY}}_2 \le C$ for some constant $C$
%\begin{Corollary}\label{cor:floor}
%\[\sup_{\eta \in [0,1]} \Big \vert \frac{\eta T}{\flo{\eta T}}-1\Big \vert =O(1/T).\]
%\end{Corollary}
%\begin{proof}
Observe that for $\eta \le 1/T$, the summation is identically zero, while the approximation to the floor function is identically zero for $\eta=1/T$. Hence, we only we only have to consider the case where  $\eta > 1/T$. The approximation error of the floor function is in this case given by 
\[
\sup_{\eta \in (1/T,1]} \Big \vert \frac{\eta T}{\flo{\eta T}}-1\Big \vert  \le C\sup_{\eta \in (1/T,1]} \Big \vert \frac{1}{\flo{\eta T}}\Big \vert=O(1/T). \tageq \label{eq:floor} 
\] 
Then, using additionally that $\sup_{\eta} \snorm{\eta \mathcal{U}^{(\omega)}_{XY}}_2 \le C$ for some constant $C$, we find
%\end{proof}
\begin{align*}
& \le \big(\frac{C}{T_1}\big) \int_a^b \E \sup_{\eta \in [0,1]}\bigsnorm{ \frac{1}{ T_1} \sum_{s=1}^{\flo{\eta T_1}}\Big(\sum_{t=1}^{\flo{\eta T_1}}  \tilde{w}^{(\omega)}_{b_1,s,t} (X_s \otimes X_t) - \F_X^{(\omega)}\Big)}_2 d\omega
%\\& \le \big(\frac{C}{T_1}\big) \int_a^b\Big( \E \sup_{\eta \in [0,1]}\bigsnorm{ \frac{1}{T_1} \sum_{s=1}^{\flo{\eta T_1}}\Big(\sum_{t=1}^{\flo{\eta T_1}}  \tilde{w}^{(\omega)}_{b_1,s,t} (X_s \otimes X_t) - \F_X^{(\omega)}\Big)}^\gamma_2\Big)^{1/\gamma} d\omega
\\&{ \le \big(\frac{C}{T_1}\big) \int_a^b\Big( \E \big( \sup_{\eta \in [0,1]}\bigsnorm{ \frac{1}{T_1} \sum_{s=1}^{\flo{\eta T_1}}\Big(\sum_{t=1}^{\flo{\eta T_1}}  \tilde{w}^{(\omega)}_{b_1,s,t} (X_s \otimes X_t) - \F_X^{(\omega)}\Big)}_2\big)^{\gamma}\Big)^{1/\gamma} d\omega}
\\& 
=O\big(\frac{1}{T}\big) O\big(T_1^{-1} b_1^{-\gamma/2} T_1^{\gamma/2} +b_1^{2 \gamma \ell}\big)^{1/\gamma} 
%=O(T_1^{-1}) O(b_1^{-1/2} T_1^{-1/2} +b_1^{2 \ell}) = o\Big(\frac{1}{\sqrt{b_1 T_1+b_2 T_2}}\Big),
\end{align*}
where we used Jensen's inequality in the last inequality and where the last line follows from \eqref{eq:varB1} and \eqref{eq:biasB1}, and \autoref{as:ratiorates}.
\end{proof}

\begin{lemma}\label{lem:approx}
Suppose Assumptions \ref{as:depstruc}-\ref{as:bwrates} are satisfied. Denote $\hat{\F}^{(\omega)}_{k} = k^{-1}\sum_{s,t=1}^{k}  \tilde{w}^{(\omega)}_{b_T,s,t}(X_s \otimes X_t) $ and let $\mathcal{M}^{(\omega)}_{X,k,m} = k^{-1}\sum_{s=1}^{k} \sum_{t=0}^{s-1} \tilde{w}^{(\omega)}_{b_T,s,t}(\dmp{s}{X} \otimes \dmp{t}{X})  $ where $\{\dmp{s}{X}\}$ is given by \eqref{eq:Dm}. Then, for any $1 \le k \le T$
\begin{align*}
\lim_{m\to \infty} \limsup_{T \to \infty} & \frac{1}{\mathcal{W}^\gamma_{b_T}}%b^{\gamma/2}_T T^{-\gamma/2} \kappa^{-\gamma/2}
\E \Big( \max_{1 \le k \le T}\bigsnorm{k(\hat{\F}^{(\omega)}_{k}-\E\hat{\F}^{(\omega)}_{k}) -\ldm{}{k}-\ldm{\dagger}{k} }_{2}\Big)^{\gamma} =0,
% \label{eq:bilcheck}
\end{align*}
for $\gamma >2$.
%where $\kappa=\int_{-1}^1 w^2(x)dx<\infty$.
\end{lemma}
\begin{proof}[Proof of \autoref{lem:approx}]
From the proof of Lemma 3.3 of \cite{vD19}, applied to the sequential spectral density estimators, we have for some constant $C_{2\gamma}$
\begin{align*}
&\Big( \E\bigsnorm{k(\hat{\F}^{(\omega)}_{k}-\E\hat{\F}^{(\omega)}_{k}) -\ldm{}{k}-\ldm{\dagger}{k} }^\gamma_{2}\Big)^{1/\gamma} 
\\& \le  C_{2\gamma} \Big(k \sum_{h=1}^{k-1}\big\vert {w}(b_T h) \big\vert^2 \Big)^{1/2} \Upsilon_{2\gamma,m}\sum_{t=0}^{\infty}\nu_{\hi,2\gamma}(X_t) +|w(0)| \sqrt{k}\sum_{t=0}^{\infty}\nu_{\hi,2\gamma}(X_t)\\& +C_{2\gamma}\|X_0\|^2_{\hi,2\gamma} m^{3/2} \sqrt{k}\Big( \max_{t}|w(b_T t)|+\big(m\sum_{h=1}^{k}| w(b_Th)- w(b_T(h-1))|^2 \big)^{1/2}\Big),
% \label{eq:bilcheck}
\end{align*}
where $\Upsilon_{2p,m} = 2\sum_{t=0}^{\infty}\min\big(\nu_{\hi, 2p}(X_t),\big(\sum_{i=m}^{\infty}\nu_{\hi, 2p}^2(X_i))^{1/2}\big)$.  
Therefore, using a similar argument as in the proof of \autoref{lem:boundseqSDO}  shows that Theorem 1 of M{\'o}ricz, (1976) and the conditions on the weight function imply, for $\gamma >2$,
\begin{align*}
\E\Big(\max_{1\le k\le T}\bigsnorm{k(\hat{\F}^{(\omega)}_{k}-\E\hat{\F}^{(\omega)}_{k}) -\ldm{}{k}-\ldm{\dagger}{k} }_{2} \Big)^\gamma = O( T^{\gamma/2} b_T^{-\gamma/2} \Upsilon^{\gamma}_{2p,m})+m^{3/2 \gamma} T^{\gamma/2}\big(O(1)+ o(m^{1/2} b^{-1/2}_T) \big)^{\gamma}.
\end{align*}
Consequently, since $b_T \to 0$ as $T \to \infty$,
\begin{align*}
&\lim_{m \to \infty} \limsup_{T \to \infty} b^{\gamma/2}_T T^{-\gamma/2}\E\Big(\max_{1\le k\le T}\bigsnorm{k(\hat{\F}^{(\omega)}_{k}-\E\hat{\F}^{(\omega)}_{k}) -\ldm{}{k}-\ldm{\dagger}{k} }_{2}\Big)^\gamma  
%\\& = 
%\lim_{m \to \infty} \limsup_{T \to \infty}b^{\gamma/2}_T T^{-\gamma/2}O( T^{\gamma/2} b_T^{-\gamma/2} \Upsilon^{\gamma}_{2\gamma,m}+m^{2\gamma} T^{\gamma/2}\big( o(\frac{1}{b_T})(1+m) \big)^{\gamma/2})
\\& =\lim_{m \to \infty} \limsup_{T \to \infty} O( \Upsilon^{\gamma}_{2\gamma,m})
%\\& \phantom{=}
+\lim_{m \to \infty} \limsup_{T \to \infty}m^{3/2\gamma}b^{\gamma/2}_T \big(O(1)+ o(m^{1/2} b^{-1/2}_T) \big)^{\gamma}=0,
\end{align*}
where we used  that $\lim_{m \to \infty} \Upsilon^{\gamma}_{2\gamma,m}=0$.
\end{proof}

\begin{lemma}\label{lem:BrMbias}
Suppose Assumptions \ref{as:depstruc}-\ref{as:bwrates} hold. %and for bandwidths satisfying $b_T^{\ell+1/2} T^{1/2} \to 0$
Then, for any $1 \le k \le T_1$, 
\begin{align*}
&\lim_{T_1, T_2 \to \infty}\frac{\sqrt{b_{T_1} T_1+b_{T_2} T_2}}{\sqrt{b_{T_1} T_1}}  \sup_{\omega} \max_{1\le k\le T_1}  \biginprod{\frac{\sqrt{b_{T_1}}}{\sqrt{T_1}}\E\sum_{s=1}^{k} \Big(\sum_{t=1}^{k} \tilde{w}^{(\omega)}_{b_1,s,t}(X_s \otimes X_t)  - {\F}^{\omega}_X \Big)}{\mathcal{U}^{(\omega)}_{X,Y}}_S =0.
\end{align*}
\end{lemma}
\begin{proof}[Proof of \autoref{lem:BrMbias}]
For any $1 \le k \le T_1$, we obtain using the Cauchy-Schwarz inequality and Proposition \autoref{prop:consF}
\begin{align*} 
\sup_{\omega}  \max_{1\le k\le T_1} \bigsnorm{\frac{\sqrt{b_{T_1}}}{\sqrt{T_1}}\E\sum_{s=1}^{k} \Big(\sum_{t=1}^{k} \tilde{w}^{(\omega)}_{b_{T_1},s,t}(X_s \otimes X_t)   - {\F}^{\omega}_X \Big)}_2 \bigsnorm{\mathcal{U}^{(\omega)}_{X,Y}}_2 & =
  \max_{1\le k\le T}O(\frac{\sqrt{b_{T_1}}}{\sqrt{T_1}}) O(b_{T_1}^\ell k)
%\\& =  \max_{1\le k\le T} O(b_T^{\ell+1/2}k T^{-1/2})
\\& = O(b_{T_1}^{\ell+1/2} T_1^{1/2}),
\end{align*}
which goes to zero for bandwidths satisfying satisfying \autoref{as:bwrates} % $b_T^{\ell+1/2} T^{1/2} \to 0$ as $T \to \infty$. %Hence, if $b_T = T^{-\kappa}$, this implies that we can choose bandwidths such that $\kappa \ge 1/(1+2 \ell)$. 
The result now follows from \autoref{as:ratiorates}.
\end{proof}

\begin{lemma}\label{lem:diffboundEig}
Let $\F_X^{(\omega)}$ be the spectral density operator of a weakly stationary functional time series $\{X_t\}_{t\in \znum}$ with eigendecomposition $\sum_{k=1}^{\infty}\lambda^{(\omega)}_{X,k} \Pi^{(\omega)}_{X,k}$. Let $\Big\{  \hat{\lambda}^{(\omega)}_{X,k} (\eta) \Big\}_{k \ge 1}$ $\Big\{  \hat{\Pi}^{(\omega)}_{X,k} (\eta) \Big\}_{k \ge 1}$ be the sequence of eigenvalues and eigenprojectors, respectively, of the sequential estimators $\hat{\F}_X^{(\omega)}(\eta)$, $\eta \in [0,1]$ of $\F_X^{(\omega)}$ as defined in \eqref{eq:seqFQ}. Then, under Assumption \ref{as:depstruc}-\ref{as:bwrates}, for each $\omega \in [-\pi,\pi]$, 
\begin{align*}
\mathrm{(i)}\, & \bigsnorm{\hat{\Pi}^{(\omega)}_{X,k}(\eta) - \Pi^{(\omega)}_{X,k}}_2  \le  \frac{\sqrt{8}}{G_{\widetilde{\otimes},k} }\Big(
\snorm{\hat{\F}_X^{(\omega)}(\eta)-\F_X^{(\omega)}}^2_2+2 \snorm{\F_X^{(\omega)}}_2 \snorm{\hat{\F}_X^{(\omega)}(\eta)-\F_X^{(\omega)}}_2 \Big)
\\  \mathrm{(ii)}\,&  \Big\vert (\lamxT{k}(\eta))^2 \ -( \lamx{\omega}{k})^2 \Big\vert 
%\le 2 \lamx{\omega}{k}|\lamxT{k}(\eta)- \lamx{\omega}{k}| + |\lamxT{k}(\eta)- \lamx{\omega}{k}| | \lamx{\omega}{k}-\lamxT{k}(\eta)| 
%\\&
\le C \snorm{\hat{\F}_X^{(\omega)}(\eta)-\F_X^{(\omega)}}_2+ \snorm{\hat{\F}_X^{(\omega)}(\eta)-\F_X^{(\omega)}}^2_2,
\end{align*} 
where  $C\ge 0$ is a bounded constant and where $G_{\widetilde{\otimes},k} =  \inf_{j,l:\{j,  l = k\}^\complement}|\lambda^{(\omega)}_{X,l}\lambda^{(\omega)}_{X,j}-(\lambda^{(\omega)}_{X,k})^2|$. 
\end{lemma}

\begin{proof}
We drop the dependency on $\omega$ in the notation below. For (i), we will first show that
\begin{align*}
\bigsnorm{\hat{\Pi}_{kk}(\eta) - \Pi_{kk}}^2_2 \le 2\sum_{\substack{l,j=1,\\ \{j,  l = k\}^\complement}}^{\infty} \Big \vert \biginprod{\hat{\Pi}_{kk}(\eta)}{{\Pi}_{lj}}\Big\vert^2. \tageq \label{eq:PinormB5}
\end{align*}
By Parseval's identity and orthogonality of the projectors
\begin{align*}
\bigsnorm{\hat{\Pi}_{kk}(\eta) - \Pi_{kk}}^2_2& = \sum_{j,l=1}^{\infty} \Big \vert \biginprod{\hat{\Pi}_{kk}(\eta)- \Pi_{kk}}{{\Pi}_{lj}}\Big\vert^2
 =\Big \vert \biginprod{\hat{\Pi}_{kk}(\eta)}{{\Pi}_{kk}} -1\Big\vert^2+\sum_{\substack{l,j=1,\\ \{j,  l = k\}^\complement}}^{\infty}\Big \vert \biginprod{\hat{\Pi}_{kk}(\eta)}{{\Pi}_{lj}}\Big\vert^2.
\end{align*}
Then observe that, using again Parseval's identity, and that the (empirical) projectors have unit norm;
\begin{align*}
\Big \vert \biginprod{\hat{\Pi}_{kk}(\eta)}{{\Pi}_{kk}} -1\Big\vert^2 &= \Big \vert \biginprod{\hat{\Pi}_{kk}(\eta)}{{\Pi}_{kk}} \Big\vert^2+ \sum_{j,l=1}^{\infty} \Big \vert \biginprod{\hat{\Pi}_{kk}(\eta)}{{\Pi}_{lj}}\Big\vert^2  - 2\Re\biginprod{\hat{\Pi}_{kk}(\eta)}{{\Pi}_{kk}}
\\& = \sum_{\substack{l,j=1,\\ \{j,  l = k\}^\complement}}^{\infty} \Big \vert \biginprod{\hat{\Pi}_{kk}(\eta)}{{\Pi}_{lj}}\Big\vert^2 +2\Big \vert \biginprod{\hat{\phi}_{k}(\eta)}{\phi_{k}} \Big\vert^4 - 2\Big \vert \biginprod{\hat{\phi}_{k}(\eta)}{\phi_{k}} \Big\vert^2
\\& \le \sum_{\substack{l,j=1,\\ \{j,  l = k\}^\complement}}^{\infty} \Big \vert \biginprod{\hat{\Pi}_{kk}(\eta)}{{\Pi}_{lj}}\Big\vert^2 
\end{align*}
where the last equality follows from the fact that $\big \vert \inprod{\hat{\phi}_{k}(\eta)}{\phi_{k}} \big\vert \le \|\hat{\phi}_{k}(\eta)\|_\Hi\|\phi_{k}\|_\Hi \le 1$. This proves \eqref{eq:PinormB5}. Next, we observe that 
\begin{align*}
\bigsnorm{\F_{\widetilde{\otimes}}\big(\hat{\Pi}_{kk}(\eta)\big) - \lambda^2_{k}\hat{\Pi}_{kk}(\eta)}_2 & \le 
\bigsnorm{\big(\F_{\widetilde{\otimes}}-\hat{\F}_{\widetilde{\otimes}}(\eta)\big)\big(\hat{\Pi}_{kk}(\eta)\big)}_2+\bigsnorm{ \big(\hat{\lambda}^2_{k}(\eta)-\lambda^2_{k}\big)\hat{\Pi}_{kk}(\eta)}_2
\\&  \le 
2 \bigsnorm{\F_{\widetilde{\otimes}}-\hat{\F}_{\widetilde{\otimes}}(\eta)}_2. \tageq \label{eq:FnormB5}
\end{align*}
Furthermore 
\begin{align*}
\bigsnorm{\F_{\widetilde{\otimes}}(\hat{\Pi}_{kk}(\eta)) - \lambda^2_{k}(\hat{\Pi}_{kk}(\eta))}^2_2 & = \sum_{l,j=1}^{\infty}\Big \vert \biginprod{\F_{\widetilde{\otimes}}(\hat{\Pi}_{kk}(\eta))}{\Pi_{lj}} -\biginprod{\lambda^2_{k}(\hat{\Pi}_{kk}(\eta))}{\Pi_{lj}}\Big \vert^2
\\& = \sum_{\substack{l,j=1,\\ \{j,  l = k\}^\complement}}^{\infty}\Big \vert \biginprod{\hat{\Pi}_{kk}(\eta)}{(\lambda_{l}\lambda_{j}-\lambda^2_{k})\Pi_{lj}}\Big \vert^2
\\& \ge G^2_{\widetilde{\otimes},k}  \sum_{\substack{l,j=1,\\ \{j,  l = k\}^\complement}}^{\infty} \Big \vert \biginprod{\hat{\Pi}_{kk}(\eta)}{{\Pi}_{lj}}\Big\vert^2
\end{align*}
where $G_{\widetilde{\otimes},k} =  \inf_{j,l:\{j,  l = k\}^\complement}|\lambda_{l}\lambda_{j}-\lambda^2_{k}|$. 
This together with \eqref{eq:PinormB5} and 
\eqref{eq:FnormB5} yields
%\begin{align*}
%4\bigsnorm{\big(\F_{\widetilde{\otimes}}-\hat{\F}_{\widetilde{\otimes}}(\eta)\big)}^2_2
%& \ge G^2_{\widetilde{\otimes},k}  \frac{1}{2}\bigsnorm{\hat{\Pi}_{kk}(\eta) - \Pi_{kk}}^2_2. 
%\end{align*}
after rearranging
\begin{align*}
\bigsnorm{\hat{\Pi}_{kk}(\eta) - \Pi_{kk}}_2 &\le  \frac{\sqrt{8}}{G_{\widetilde{\otimes},k} }
\bigsnorm{\F_{\widetilde{\otimes}}-\hat{\F}_{\widetilde{\otimes}}(\eta)}_2
\le  \frac{\sqrt{8}}{G_{\widetilde{\otimes},k} }\Big(
\snorm{\F-\hat{\F}(\eta)}^2_2+2 \snorm{\F}_2 \snorm{\F-\hat{\F}(\eta)}_2 \Big).
\end{align*}
For (ii), we simply observe that \[
  \Big\vert \hat{\lambda}^2_k(\eta)\ -\lambda^2_k(\eta) \Big\vert
\le 2 \lambda_k(\eta) \Big\vert \hat{\lambda}_k(\eta) - \lambda_k(\eta) \Big\vert +  \Big\vert\hat{\lambda}_k(\eta)\ -\lambda_k(\eta) \Big\vert^2.
\] %where we used that the eigenvalues are real. Since $\hat{\F}_X^{(\omega)}(\eta)$ is a sequential consistent estimator of $\F_X^{(\omega)}$ under Assumption \ref{as:depstruc}-\ref{as:bwrates}, the results now follows.% (see also the proof of  Proposition 4.1 of \cite{vD19}).
\end{proof}
\begin{lemma}[Lemma A.1 of \citep{vD19}] \label{lem:Burkh}
Let ${\{M_k\}}_{k=1,\ldots,n} \in \op^p_H, p>1$, be a martingale with respect to $\G$ with $\{D_k\}$ denoting its difference sequence and let $\{A_k\}_{k=1,\ldots,n} \in S_\infty(H)$. Then, for $q=\min(2,p)$,
\[
\Bignorm{\sum_{k=1}^{n} A_k (D_k)}^q_{\hi,p} \le K^q_{p} \sum_{k=1}^{n} \snorm{A_k}_{\infty}^q \|D_k\|^q_{\hi,p}
\]
where $K^q_{p}= (p^\star-1)^q$ with $p^\star=\max(p,\frac{p}{p-1})$.
\end{lemma}
\begin{proof}[Proof of \autoref{lem:Burkh}]
By Burkholder's inequality,
\begin{align*}
\Bignorm{\sum_{k=1}^{n}A_k (D_k)}^q_{\hi,p} =\Big(\E\Bignorm{\sum_{k=1}^{n} A_k (D_k)}^p_H\Big)^{q/p}  \le 
(p^\star-1)^q \Big(\E  \Big| \Big(\sum_{k=1}^{n}\|A_k (D_k)\|^2_H\Big)^{1/2}\Big|^p\Big)^{q/p}.
\end{align*}
Let $p<2$. Then, applying the inequality $|\sum_k x_k|^{r} \le \sum_k |x_k|^r$ for $r<1$ to $x_k = \|A_k(D_k)\|^2_{H}$, we obtain
\begin{align*}
(p^\star-1)^q \Big(\E \Big| \Big(\sum_{k=1}^{n}\|A_k (D_k)\|^2_H\Big)^{1/2}\Big|^{p}\Big)^{q/p} 
& \le (p^\star-1)^q \Big( \E\sum_{k=1}^{n}  \|A_k (D_k)\|^{p}_H\Big)^{q/p}  
\\& \le (p^\star-1)^q \Big( \sum_{k=1}^{n}  \snorm{A_k}_{\infty}^p \E\| D_k\|^{p}_H\Big)^{q/p} 
\\&  \le (p^\star-1)^q \sum_{k=1}^{n}  \snorm{A_k}^q _{\infty}\|D_k\|^q_{\hi,p},
\end{align*}
where the one before last inequality follows from Holder's inequality for operators and from the fact that $p=q$. %and where the last inequality follows from subadditivity of the function $(\cdot)^{p/q}$ in case $q/p<1$.
For $p\ge2$, $q=2$. Therefore, an application of Minkowski's inequality to $\|\cdot\|_{\mathbb{C},{p/q}}$ and Holder's inequality yield in this case
\begin{align*}
(p^\star-1)^q \Big(\E \Big| \Big(\sum_{k=1}^{n}\|A_k (D_k)\|^q_H\Big)^{1/q}\Big|^p\Big)^{q/p} 
& \le 
%(p^\star-1) \Big(\E \Big| \sum_{k=1}^{n}\|A_k (D_k)\|^{2}_H\Big|^{p/2}\Big)^{2q/ 2p} 
(p^\star-1)^q \Big(\E \Big\vert\sum_{k=1}^{n}\|A_k (D_k)\|^q_H\Big\vert^{p/q}\Big)^{q/p} 
%  \\& \le (p^\star-1) \sum_{k=1}^{n}\big(\E\|A_k (D_k)\|^{p}_H\big)^{q/p} 
\\& \le (p^\star-1)^q \sum_{k=1}^{n} \snorm{A_k}^q _{\infty}( \E\|D_k\|^p_{H})^{q/p}
\\&  \le (p^\star-1)^q\sum_{k=1}^{n}  \snorm{A_k}^q _{\infty}\|D_k\|^q_{\hi,p}.
\end{align*}
%where the one before last inequality follows from Holder's inequality for operators and where the last inequality follows from subadditivity of the function $(\cdot)^{p/q}$ in case $q/p<1$.
\end{proof}

\begin{lemma}\label{lem:secM}
Let $\{\mathcal{D}^{(\omega)}_{XY,s,t}\}$ be defined as in \eqref{eq:ldm}. Then,  under \autoref{as:Weights}, for $\gamma >2$,
\begin{align*}
&\E\Big(\max_{1 \le k \le T}\bignorm{\sum_{t=2}^{k} \sum_{s=t-4m+1 \vee 1}^{t-1} \tilde{w}^{(\omega)}_{b_T,s,t} \mathcal{D}^{(\omega)}_{XY,s,t}}_{\Hs}\Big)^{\gamma} = o(T^{\gamma/2}/{b^{\gamma/2}_T}). 
\end{align*}
\end{lemma}

\begin{proof}%[Proof of \autoref{lem:secM}]
By definition of $\Hs$, we can write
\begin{align*}
\bignorm{\sum_{t=2}^{k} \sum_{s=t-4m+1 \vee 1}^{t-1} \tilde{w}^{(\omega)}_{b_T,s,t} \mathcal{D}^{(\omega)}_{XY,s,t}}^2_{\Hs,\gamma} & =\bignorm{\sum_{t=2}^{k} \sum_{s=t-4m+1 \vee 1}^{t-1} \tilde{w}^{(\omega)}_{b_T,s,t} \mathcal{D}^{(\omega)}_{XY,s,t}}^2_{\Hs,\gamma} 
\\& \le 2  \sum_{i=1}^{2} \bignorm{\sum_{t=2}^{k} \sum_{s=t-4m+1 \vee 1}^{t-1} \tilde{w}^{(\omega)}_{b_T,s,t} \big( \dmp{s}{X_i} \otimes  \dmp{t}{X_i} \big)}^2_{\Hi_i,\gamma} 
\end{align*}
Using \autoref{lem:Burkh} and Holder's inequality, we obtain for fixed $m$,
\begin{align*}
& \sum_{i=1}^{2}\bignorm{\sum_{t=2}^{k} \sum_{s=t-4m+1 \vee 1}^{t-1} \tilde{w}^{(\omega)}_{b_T,s,t} \big( \dmp{t}{X_i} \otimes  \dmp{s}{X_i} \big)}^2_{\Hi_i,\gamma} 
 \le  \sum_{i=1}^2 \sum_{t=2}^{k} \norm{  \dmp{t}{X_i}}^2_{\hi_i,2\gamma} \bignorm{\sum_{s=t-4m+1 \vee 1}^{t-1} \tilde{w}^{(\omega)}_{b_T,s,t} \dmp{s}{X_i} }^2_{\hi_i,2\gamma}
  \\& =2  K^4_{2\gamma} \sum_{i=1}^2  \|\dmp{0}{X_i}\|^4_{\hi_i,2\gamma} \Big( \sum_{t=2}^{4m} \sum_{s=1}^{t-1 }|\tilde{w}^{(\omega)}_{b_T,s,t}|^2  + \sum_{t=4m+1}^{k} \sum_{s=t-4m+1 }^{t-1}| \tilde{w}^{(\omega)}_{b_T,s,t}|^2 \Big)
%  \\& =2  K^4_4 \sum_{i=1}^2  \|\dmp{0}{X_i}\|^4_{\hi_i,2\gamma} \Big( \sum_{t=2}^{4m} \sum_{h=1}^{t-1 }|w(b_T h)|^2  + \sum_{t=4m+1}^{k} \sum_{h=1 }^{4m-1}|w(b_T h)|^2 \Big)
  \\& =2  K^4_{2\gamma} \sum_{i=1}^2  \|\dmp{0}{X_i}\|^4_{\hi_i,2\gamma} \Big( \sum_{h=1}^{4m-1} \sum_{t=h+1}^{4m }|w(b_T h)|^2  + \sum_{h=1}^{4m-1} \sum_{t=4m+1 }^{k}|w(b_T h)|^2 \Big).
% = O(m k). 
\end{align*}
The result now follows by Theorem 1 of \cite{Mor76}.
\end{proof}

\begin{lemma}\label{lem:inpsZt}
Let $\widetilde{J}^{(\lambda)}_{m,b_T,t}$ be defined as in \eqref{eq:J4m} and suppose %$\{\dm{t}\} \in \op^2_{\Hi}$ be a $\Hi$-valued martingale difference process and let 
\autoref{as:Weights} is satisfied. Furthermore, assume $\lambda_1 \pm\lambda_2 \ne 0 \mod 2\pi$. Then for any $u, v \in \Hi$,
\[\sum_{t=1+4m}^{T} |\E \widetilde{J}^{(\lambda_1)}_{m,b_T,t}(u) \widetilde{J}^{(\lambda_2)}_{m,b_T,t}(v)| =o(T/b_T).\]
\end{lemma}
\begin{proof}This follows from a slight adjustment of the proof of Lemma B.3 of \cite{vD19}.\end{proof}

% \begin{Remark}\textcolor{blue}{Provided some conditions on the lag window function, we remark that the estimator in \eqref{eq:relFQ} can alternatively be given by 
%\[
%\hat{\F}_X^{(\omega)}=\frac{1}{b_1 T_1} \sum_{j=1}^{ T_1} K\big(\frac{\omega-\omega_j}{b_1}\big) D_{X}^{\omega_j}\otimes D_{X}^{\omega_j} \quad \eta \in [0,1]
%\]
%where $\frac{1}{b_1}K(\frac{\omega-\omega_j}{b_1}) = \sum_{h \in \znum} w(bh)e^{-\im h \omega}$ and where $D_{X}^{\omega_j}$ is the functional discrete Fourier transform of $X$, i.e., $(2\pi T_1)^{-1/2}\sum_{t=1}^{T_1}X_t e^{-\im \omega_j t}$. A similar observation can be made for \eqref{eq:relFQy}.} \end{Remark}

%% if your bibliography is in bibtex format, uncomment commands:
%\bibliographystyle{imsart-number} % Style BST file
%\bibliography{bibliography}       % Bibliography file (usually '*.bib')

\begin{thebibliography}{}
\newcommand{\enquote}[1]{``#1''}
\expandafter\ifx\csname natexlab\endcsname\relax\def\natexlab#1{#1}\fi
\expandafter\ifx\csname url\endcsname\relax
 \def\url#1{{\tt #1}}\fi
\expandafter\ifx\csname urlprefix\endcsname\relax\def\urlprefix{URL }\fi

\bibitem[Antoniadis et~al.(2006)]{AnPapSap06}
Antoniadis, A., Paparoditis, E. and Sapatinas, T. (2006).
\newblock {A functional wavelet-kernel approach for time series prediction.}
\newblock {\em Journal of the Royal Statistical Society: Series B (Statistical
  Methodology)}, 68:837--857.

\bibitem[Aston and Kirch(2012)]{aston2012}
Aston, J.~A. and Kirch, C. (2012).
\newblock Detecting and estimating changes in dependent functional data.
\newblock {\em Journal of Multivariate Analysis}, 109:204--220.

\bibitem[Aue et~al.(2019)]{auedettrice2019}
Aue, A., Dette, H., and Rice, G. (2019).
\newblock Two-sample tests for relevant differences in the eigenfunctions of
  covariance operators.
\newblock {\em arXiv:1909.06098}.

\bibitem[Benko et~al.(2009)]{BHK2009}
Benko, M., H{\"a}rdle, W., and Kneip, A. (2009).
\newblock Common functional principal components.
\newblock {\em Annals of Statistics}, 37:1--34.

\bibitem{berkson1938}
Joseph Berkson.
\newblock Some difficulties of interpretation encountered in the application of
  the chi-square test.
\newblock {\em Journal of the American Statistical Association}, 33:526--536,
  1938.
  
%\bibitem[{Biswal et al.(2010)}]{Biswal10}
%Biswal, B.B. et al. (2010).
%\newblock Toward discovery science of human brain function.
%\newblock {\em Proc. Natl. Acad. Sci. USA\/}, 107:4734--4739.


%\bibitem[{Billingsley (1968)}]{bil68}
%Billingsley, P.
%\newblock {\em Convergence of Probability Measures\/}.
%\newblock Wiley, New York (1968).

%\bibitem[{Cerovecki and H{\"o}rmann (2017)}]{CevHor17}
%Cerovecki, C. \&\ H{\"o}rmann, S.  (2017).
%\newblock On the CLT for discrete Fourier transforms of functional time series.
%\newblock {\em Journal of multivariate analysis\/}, 154:282--295.


\bibitem[{van Delft (2020)}]{vD19}
van Delft, A. (2020).
\newblock{A note on quadratic forms of functional time series under mild conditions.}
\newblock{\em Stochastic Processes and their Applications\/}, 130(7):4206--4251.
%\newblock{\em ArXiv Preprint\/}, arXiv:1905.13186 (2019).
\bibitem[{van Delft and Eichler(2020)}]{vDE19}
van Delft, A. \&\ Eichler, M (2020).
\newblock{A note on Herglotz's theorem for time series on function spaces.}
\newblock{\em Stochastic Processes and their Applications\/},  130(6):3687--3710.

\bibitem[{van Delft and Dette (2020)}]{vDD18}
van Delft, A. \&\ Dette, H. (2020).
\newblock{A similarity measure for second order properties of non-stationary functional time series with applications to clustering and testing.}
\newblock{\em Bernoulli\/}, to appear.


\bibitem[{van Delft and Dette (2020)}]{vDD20}
van Delft, A. \&\ Dette, H. (2020).
{Supplement to ``Pivotal tests for relevant differences in the second order dynamics  of functional time series''.}


\bibitem[{van Delft and Eichler(2018)}]{vDE18}
van Delft, A. \&\ Eichler, M.
\newblock \enquote{Locally stationary functional time series.}
\newblock {\em Electronic Journal of Statistics\/}, 12:107--170 (2018).

\bibitem[Deheuvels and Martynov (2003)]{DeMa03}
Deheuvels, P. and Martynov, G. (2003). 
Karhunen-Lo{\`e}ve expansions for weighted Wiener processes and Brownian bridges via Bessel functions. 
\newblock{\em High dimensional probability, III (Sandjberg, 2002). Progr. Probab. 55 57--93. Birkhauser,
Basel.}

\bibitem[Dette et~al.(2020)]{DKV2018}
Dette, H., Kokot, K., and Volgushev, S. (2020).
\newblock Testing relevant hypotheses in functional time series via
  self-normalization. To appear in: \newblock{\em Journal of the Royal Statistical Society: Series B \/}
\newblock {\em arXiv:1809.06092}.

\bibitem[Fiecas and Ombao(2016)]{FiecOmb16}
Fiecas, M., and Ombao, H. (2016).
\newblock Modeling the Evolution of Dynamic Brain Processes during an Associative Learning Experiment.
\newblock {\em Journal of the American Statistical Association}, 111:1440--1453.

\bibitem[Fremdt et~al.(2013)]{fremdt2013}
Fremdt, S., Steinebach, J.~G., Horv\'ath, L., and Kokoszka, P. (2013).
\newblock Testing the equality of covariance operators in functional samples.
\newblock {\em Scandinavian Journal of Statistics}, 40(1):138--152.

\bibitem[{Fogarty \& Small(2014)}]{fogarty2014}
\textsc{Fogarty, C.~B.} \& \textsc{Small, D.~S.} (2014).
\newblock Equivalence testing for functional data with an application to
  comparing pulmonary function devices.
\newblock \textit{Ann. Appl. Stat.} \textbf{8}, 2002--2026.

\bibitem[Guo et~al., 2016]{Guo2016}
Guo, J., Zhou, B., and Zhang, J.-T. (2016).
\newblock A supremum-norm based test for the equality of several covariance
  functions.
\newblock {\em Computational Statistics \& Data Analysis}, 124:15--26.

\bibitem[{H{\"o}rmann and Kokoszka(2010)}]{HorKok10}
H{\"o}rmann, S. \&\  Kokoszka, P.  (2010).
\newblock  {Weakly dependent functional data.}
\newblock {\em The Annals of Statistics\/}, 38:1845--1884.


\bibitem[{H{\"o}rmann et~al.(2015)}]{Hormann2015}
H{\"o}rmann S., Kidzi{\'n}ski, L. \&\ Hallin, M. (2015). 
\newblock Dynamic functional principal components.
\newblock {\em The Royal Statistical Society: Series B\/} 77:319--348.


\bibitem[{H{\"o}rmann et~al.(2018)}]{HorKokNis17}
H{\"o}rmann, S., Kokoszka, P. \&\ Nisol, G. (2018). 
 \newblock  {Detection of periodicity in functional time series.}
 \newblock {\em The Annals of Statistics.\/} 46:2960--2984.



\bibitem[{Kadison and Ringrose(1997)}]{KadRing97}
Kadison, R.V. \&\ Ringrose, J. R.  (1997).
\newblock {\em Fundamentals of the Theory of Operator Algebras. Graduate Studies in Mathematics}.
\newblock Amer. Math. Soc., Providence, RI.



\bibitem[{Kowal et  al.(2019)}]{KMR2019}
Kowal, D.R., Matteson, D.S. and Ruppert, D. (2019).
\newblock {Functional autoregression for sparsely sampled data.}
\newblock {\em Journal of Business and Economics Statistics\/} 37:97--109. 


\bibitem[{Leucht et al.(2018)}]{LPapSap2018}
Leucht, A., Paporoditis, E. and Sapatinas, T. (2018).
\newblock {Testing equality of spectral density operators for functional linear processes.}
\newblock {\em arXiv:1804.03366\/}. 

\bibitem[{McLeish(1974)}]{McL74}
McLeish, D. L.  (1974).
\newblock {Dependent central limit theorems and invariance principles.}
\newblock {\em The Annals of Probability\/}, 2:620--628.


\bibitem[{M{\`o}ricz(1976)}]{Mor76}
M{\`o}ricz, F. (1976).
\newblock {Moment inequalities and the strong law of large numbers.}
\newblock {\em Z. Wahrscheinlichkeitstheorie verw. Gebiete \/}, 35:299--314.


\bibitem[Panaretos et~al.(2010)]{Panaretos2010}
Panaretos, V.~M., Kraus, D., and Maddocks, J.~H. (2010).
\newblock Second-order comparison of gaussian random functions and the geometry
  of dna minicircles.
\newblock {\em Journal of the American Statistical Association},
  105(490):670--682.


\bibitem[{Panaretos \harvardand\ Tavakoli(2013)}]{PanTav2013b}
Panaretos, V. \harvardand\ Tavakoli, S.   (2013).
\newblock {Cram{\'e}r--Karhunen--Lo{\`e}ve representation and harmonic principal component analysis of functional time series.}
\newblock {\em Stochastic Processes and their Applications\/} 123:2779--2807.

\bibitem[{Panaretos \harvardand\ Tavakoli(2013a)}]{PanTav2013a}
Panaretos, V. \harvardand\ Tavakoli, S.  (2013a).
\newblock {Fourier analysis of stationary time series in function space.}
\newblock {\em The Annals of Statistics\/} 41(2):568--603.


\bibitem[Paparoditis and Sapatinas(2016)]{Paparoditis2016}
Paparoditis, E. and Sapatinas, T. (2016).
\newblock {Bootstrap-based testing of equality of mean functions or equality of
  covariance operators for functional data}.
\newblock {\em Biometrika}, 103(3):727--733.


\bibitem[Pilavakis et~al.(2019)]{Pilavakis2019}
Pilavakis, D., Paparoditis, E., and Sapatinas, T. (2019).
\newblock Testing equality of autocovariance operators for functional time
  series.
\newblock {\em ArXiv e-print 1901.08535}.

\bibitem[Pomann et~al.(2016)]{PoStGh16}
Pomann, G.-M., Staicu, A.-M., and Ghosh, S. (2016)
\newblock A two-sample distribution-free test for functional data
with application to a diffusion tensor imaging study of multiple sclerosis.
\newblock {\em Journal of the Royal Statistical Society, Series C}, 65:395--414.

%\bibitem[Ruttimann et~al.(1998)]{Ruttetal1998}
%Ruttimann, U.E. et~al. (1998).
%\newblock {Statistical analysis of functional MRI data in the wavelet domain.}
%\newblock {\em IEEE Transactions on Medical Imaging\/}, 17(2):142--154.


\bibitem[Shao(2010)]{shao2010}
Shao, X. (2010).
\newblock A self-normalized approach to confidence interval construction in
  time series.
\newblock {\em Journal of the Royal Statistical Society: Series B (Statistical
  Methodology)}, 72(3):343--366.

\bibitem[Shao(2015)]{shao2015}
Shao, X. (2015).
\newblock Self-normalization for time series: A review of recent developments.
\newblock {\em Journal of the American Statistical Association},
  110(512):1797--1817.

\bibitem[Shao and Zhang(2010)]{shazha2010}
Shao, X. and Zhang, X. (2010).
\newblock Testing for change points in time series.
\newblock {\em Journal of the American Statistical Association},
  105(491):1228--1240.

%\bibitem[Stoehr et~al.(2019)]{StAsKi2019}
%Stoehr, C., Aston, J.~A. and Kirch, C. (2019).
%\newblock Detecting changes in the covariance structure of functional time series with application to fMRI data.
%\newblock {\em arXiv:1903.00288}.


\bibitem[{Tavakoli and Panaretos(2016)}]{TavPan2016}
{Tavakoli, S. and Panaretos, V.} (2016)
\newblock {Detecting and Localizing Differences in Functional Time Series Dynamics: A Case
Study in Molecular Biophysics}
\newblock {\em Journal of the American Statistical Association\/}, 111:1020--1035.

\bibitem{wellek2010}
Stefan Wellek.
\newblock {\em Testing Statistical Hypotheses of Equivalence and
  Noninferiority}.
\newblock CRC Press, Boca Raton, FL, second edition, 2010.

\bibitem[{Wu(2005)}]{Wu05}
Wu, W. B. (2005).
\newblock {Nonlinear system theory: Another look at dependence.}
\newblock {\em Proceedings of the National Academy of Sciences.\/}, 102:14150--14154.


%\bibitem[Worsley et al.(2002)]{Worsley02}
%Worsley, K.J, Liao, C.H., Aston, J., Petre, V., Duncan, G.H., Morales, F. and Evans, A.C. (2015).
%\newblock A general statistical analysis for fMRI data.
% \newblock{\em Neuroimage} 15:1--15.
  
%\bibitem[{Kato(1966)}]{Kato66}
%Kato, T.
%\newblock {\em Perturbation Theory for Linear Operators\/}.
%\newblock Springer, Berlin, Germany, 1966.

%\bibitem[{Lax(2002)}]{Lax02}
%Lax, P.D.
%\newblock {\em Functional Analysis, Pure and Applied Mathematics\/}.
%\newblock John Wiley \& Sons, New York, USA, 2002.

  \bibitem[Yuen et~al.(2019)]{yuen2019}
Yuen N.H., Osachoff N., and Chen J.J. (2019).
\newblock Intrinsic Frequencies of the Resting-State fMRI Signal: The Frequency Dependence of Functional Connectivity and the Effect of Mode Mixing.
\newblock {\em Frontiers in Neuroscience}, 13:900--917.

\bibitem[Zhang and Shao(2015)]{Zhang2015}
Zhang, X. and Shao, X. (2015).
\newblock Two sample inference for the second-order property of temporally
  dependent functional data.
 \newblock{\em Bernoulli} 21:90--929.
  
  
  \bibitem[Zhang et~al.(2011)]{zhang2011}
Zhang, X., Shao, X., Hayhoe, K., and Wuebbles, D.~J. (2011).
\newblock Testing the structural stability of temporally dependent functional
  observations and application to climate projections.
\newblock {\em Electron. J. Statist.}, 5:1765--1796.



\end{thebibliography}

\begin{thebibliography}{}
\newcommand{\enquote}[1]{``#1''}
\expandafter\ifx\csname natexlab\endcsname\relax\def\natexlab#1{#1}\fi
\expandafter\ifx\csname url\endcsname\relax
 \def\url#1{{\tt #1}}\fi
\expandafter\ifx\csname urlprefix\endcsname\relax\def\urlprefix{URL }\fi

\bibitem[Aston and Kirch(2012)]{aston2012}
Aston, J.~A. and Kirch, C. (2012).
\newblock Detecting and estimating changes in dependent functional data.
\newblock {\em Journal of Multivariate Analysis}, 109:204--220.


\bibitem[{Biswal et al.(2010)}]{Biswal10}
Biswal, B.B. et al. (2010).
\newblock Toward discovery science of human brain function.
\newblock {\em Proc. Natl. Acad. Sci. USA\/}, 107:4734--4739.



\bibitem[{van Delft (2020)}]{vD19}
van Delft, A. (2020).
\newblock{A note on quadratic forms of functional time series under mild conditions.}
\newblock{\em Stochastic Processes and their Applications\/}, 130(7):4206--4251. 
\newblock arXiv:1905.13186.

\bibitem[{McLeish(1974)}]{McL74}
McLeish, D. L.  (1974).
\newblock {Dependent central limit theorems and invariance principles.}
\newblock {\em The Annals of Probability\/}, 2:620--628.


\bibitem[Ruttimann et~al.(1998)]{Ruttetal1998}
Ruttimann, U.E. et~al. (1998).
\newblock {Statistical analysis of functional MRI data in the wavelet domain.}
\newblock {\em IEEE Transactions on Medical Imaging\/}, 17(2):142--154.


\bibitem[Stoehr et~al.(2019)]{StAsKi2019}
Stoehr, C., Aston, J.~A. and Kirch, C. (2019).
\newblock Detecting changes in the covariance structure of functional time series with application to fMRI data.
\newblock {\em arXiv:1903.00288}.


\bibitem[Worsley et al.(2002)]{Worsley02}
Worsley, K.J, Liao, C.H., Aston, J., Petre, V., Duncan, G.H., Morales, F. and Evans, A.C. (2015).
\newblock A general statistical analysis for fMRI data.
 \newblock{\em Neuroimage} 15:1--15.
  
\bibitem[{M{\`o}ricz(1976)}]{Mor76}
M{\`o}ricz, F. (1976).
\newblock {Moment inequalities and the strong law of large numbers.}
\newblock {\em Z. Wahrscheinlichkeitstheorie verw. Gebiete \/}, 35:299--314.





\end{thebibliography}

%% or include bibliography directly:
% \begin{thebibliography}{}
% \bibitem{b1}
% \end{thebibliography}

\section{Application to resting state fMRI} \label{sec:data}

In ths section, we illustrate the methodology developed in this paper by an application to resting state functional Magneting Resonance Imaging (fMRI) data. fMRI data allows to capture brain activity and consists of a sequence of three-dimensional images of the brain recorded every few seconds. Since the brain operates as a single unit of which we record the activity at a large number of spatial locations, it is natural to model the brain as a function and hence the time record as a functional time series, thereby taking into account the present temporal dependencies. The data we use are publicly available as part of the \textit{1000 connectome project} \citep[][]{Biswal10}. In order to avoid differences in scanner types and locations, we consider testing for relevant differences for 6 subjects of which the data was  measured at a single site (Beijing, China). For each subject, the resting state scans are comprised of 225 temporal scans, measured 2 seconds apart, where each temporal scan consists of three dimensional images of size 64$\times$64$\times$33 voxels. The fMRI data set for one of the subjects is depicted in \autoref{fig:fmriscan}. In order to correct for technical effects such as scanner drift, a polynomial trend of order 3 was removed from each voxel time series which are voxel-wise normalized \citep[see also][]{Worsley02}. % it is important to normalize the data in each voxel. This is because data may differ between different voxels for ‘uninteresting’ reasons, such as local blood-flow differences and signal amplitude differences due to the distance from the receive coil.

\begin{figure}[t!th]
\centering 
\includegraphics[scale=0.4]{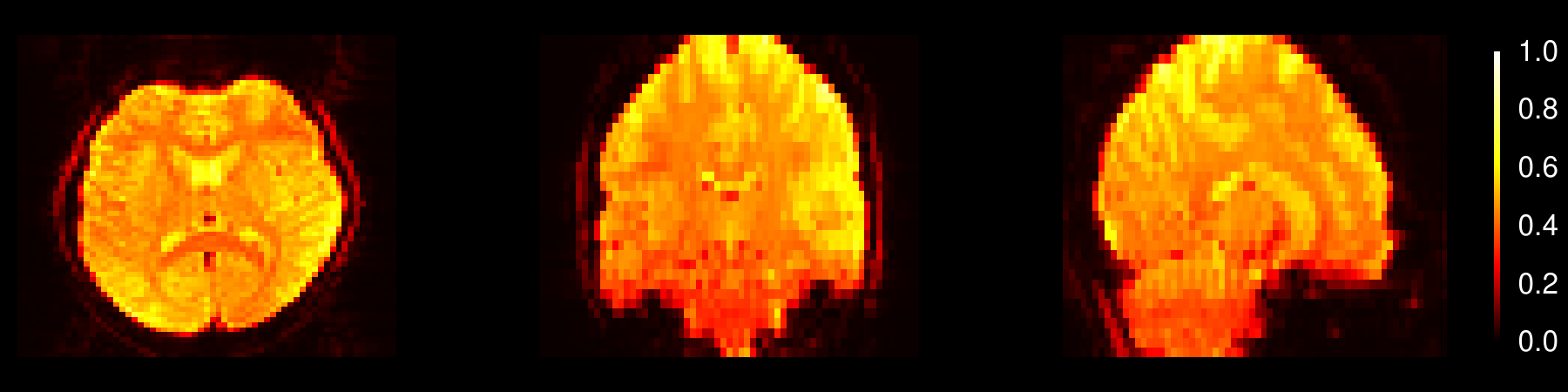}
\caption{Slices of the fMRI data set for one of the subjects.}
\label{fig:fmriscan}
\end{figure}

The high-dimensionality of fMRI data and hence of the corresponding second order dependence structure requires an efficient method to allow for the functional eigenanalysis. To make this computationally efficient and to avoid spurious identifiability issues that come with alternative discretized matrix approaches \citep[see][for a discussion]{aston2012}, we shall assume the functional component has a separable structure. The assumption of separable functions in the context of brain imaging is commonly applied as a method to deal with the high dimensionality of this type of data \citep[see e.g.,][]{Worsley02,Ruttetal1998,aston2012,StAsKi2019} %and has also found considerable attention in spatial-temporal statistics or imaging \HD{Fuentes2006,Dryden}.
To make this more precise in our set up, let $\prod_{i=1}^{3} \mathcal{S}_i$ be a product of compact sets. Then we model each fMRI data set as a functional time series $\{X_t(\tau_1,\tau_2,\tau_3) : \tau_i  \in \mathcal{S}_i, i =1,2, 3 \}_{t \in \znum}$ with well-defined  spectral density operator in $S_2(\prod_{i=1}^{3} \mathcal{S}_i)$ that satisfies
\begin{align*}
\F^{(\omega)}_{X} = \F^{(\omega)}_{X,1} \widetilde{\bigotimes}\F^{(\omega)}_{X,2} \widetilde{\bigotimes} \F^{(\omega)}_{X,3}, \tageq \label{eq:sepF0}
\end{align*}
where $\F^{(\omega)}_{X,i}$ denotes the $i$-th \textit{directional spectral density operator} 
which has a kernel function in $L^2_{\mathbb{C}}(\mathcal{S}_i\times  \mathcal{S}_i)$ that replicates across the other two directions. For example, the first directional component has kernel given by 
\begin{align*}
\F^{(\omega)}_{X,1}(\tau_1,\sigma_1) = \frac{1}{2\pi}\sum_{h \in \znum}\int_{\mathcal{S}_2} \int_{\mathcal{S}_3} \Cov\big(X_{h}(\tau_1,\tau_2,\tau_3), X_{0}(\sigma_1,\tau_2,\tau_3)\big) d \tau_2 d \tau_3 e^{-\im h \omega} \quad \tau_1, \sigma_1 \in \mathcal{S}_1,
\end{align*}
The kernels of $\F^{(\omega)}_{X,2}$ and $\F^{(\omega)}_{X,3}$ are similarly defined. Let the eigenelements of $\F^{(\omega)}_{X,i}$ be denoted as $\big\{\lambda^{(\omega)}_{i,j}, \phi^{(\omega)}_{i,j}\big\}_{j \ge 1}$ for $i=1,2,3$. Then the eigendecomposition of the  operator in \eqref{eq:sepF0} at frequency $\omega$ is given by
\begin{align*}
\F^{(\omega)}_{X} = \sum_{j,k,l=1}^{\infty} \lambda^{(\omega)}_{jkl} \Pi^{(\omega)}_{jkl} {\bigotimes} \Pi^{(\omega)}_{jkl}. \tageq \label{eq:sepF}
\end{align*}
where the eigenvalues $\lambda^{(\omega)}_{jkl}: = \lambda^{(\omega)}_{1,j} \lambda^{(\omega)}_{2,k} \lambda^{(\omega)}_{3,l}$ are ordered in a descending manner, i.e., $\{\lambda^{(\omega)}_{jkl}\}_{jkl}\searrow 0$ and where $\Pi^{(\omega)}_{jkl}= \phi^{(\omega)}_{1,j} \otimes  \phi^{(\omega)}_{2,k} \otimes  \phi^{(\omega)}_{3,l}$ is the eigenfunction belonging the $jkl$-th largest eigenvalue at frequency $\omega$  (see also \autoref{sec:Eigapprox}). The sequential eigenelements can now be efficiently estimated via consistent estimators of the sequential eigenelements of the directional operators as given in \eqref{eq:seqFQ}. For example, the estimator of the first directional sequential spectral density kernels is defined as 
\[
\hat{\F}_{X,1}^{(\omega)}(\eta)(\tau_1,\sigma_1) =\frac{1}{\flo{\eta T}}\sum_{s=1}^{\flo{\eta T}}\Bigg(\sum_{t=1}^{\flo{\eta T}}  \frac{\tilde{w}^{(\omega)}_{b,s,t}}{|\mathcal{S}_2| |\mathcal{S}_3|} \sum_{\substack{\tau_2 \in \mathcal{S}_2\\ \tau_3 \in \mathcal{S}_3}}(X_s(\tau_1, \tau_2, \tau_3) (X_t(\sigma_1, \tau_2, \tau_3)\Big), \quad \tau_1, \sigma_1 \in \mathcal{S}_1, \tageq  \label{eq:seqFQfmri}
\]
%\[
%\hat{\F}_{X,1}^{(\omega)}(\eta)(\tau_1,\sigma_1) =\frac{1}{\flo{\eta T}}\sum_{s=1}^{\flo{\eta T}}\Bigg(\sum_{t=1}^{\flo{\eta T}}  \frac{\tilde{w}^{(\omega)}_{b,s,t}}{|\mathcal{S}_2| |\mathcal{S}_3|} \sum_{\substack{\tau_2 \in \mathcal{S}_2\\ \tau_3 \in \mathcal{S}_3}} \Big((X_s(\tau_1, \tau_2, \tau_3)-\overline{X}_{\flo{\eta T}}(\tau_1, \tau_2, \tau_3)\Big)\Big( X_t(\sigma_1, \tau_2, \tau_3)-\overline{X}_{\flo{\eta T}}(\sigma_1, \tau_2, \tau_3)\Big) \Bigg), \tageq \label{eq:seqFQfmri}
%\]
where $|\mathcal{S}_i|$ is the set of the discrete observations of the function in the $i$-th direction. 

The raw data was converted into functional observations by using cubic b-spline functions on $[0,1]^3$.
It is worth mentioning that in the computation of the eigenfunctions we took into account that  the B-spline basis functions do not form an orthogonal basis. Estimators of the sequential sequence of directional spectral density kernels were obtained using the same parameters as in the previous section and evaluated at a $100^2$ equispaced grid of $[0,1]^ 2$. We then investigated whether we could find evidence of relevant differences in the second order structure of the different subjects by applying the three tests developed in Section \ref{sec31}-\ref{sec33} pairwise.  Since the sampling rate of the data is $0.5$ Hertz per second, we restrict our analysis to the interval $[0,\pi/2)$.

\begin{figure}[t!]
\centering 
\includegraphics[width=0.33\textwidth]{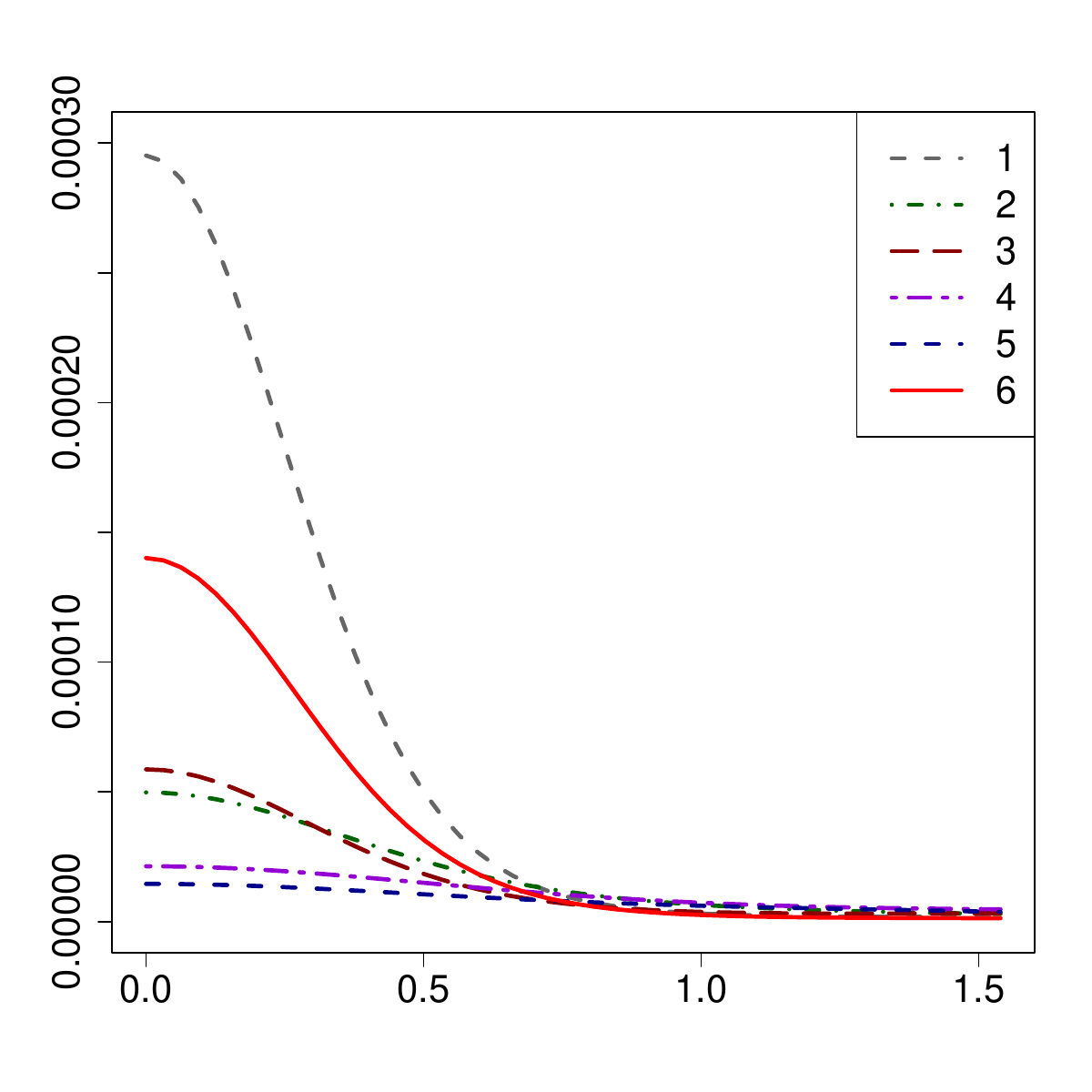}
\includegraphics[width=0.33\textwidth]{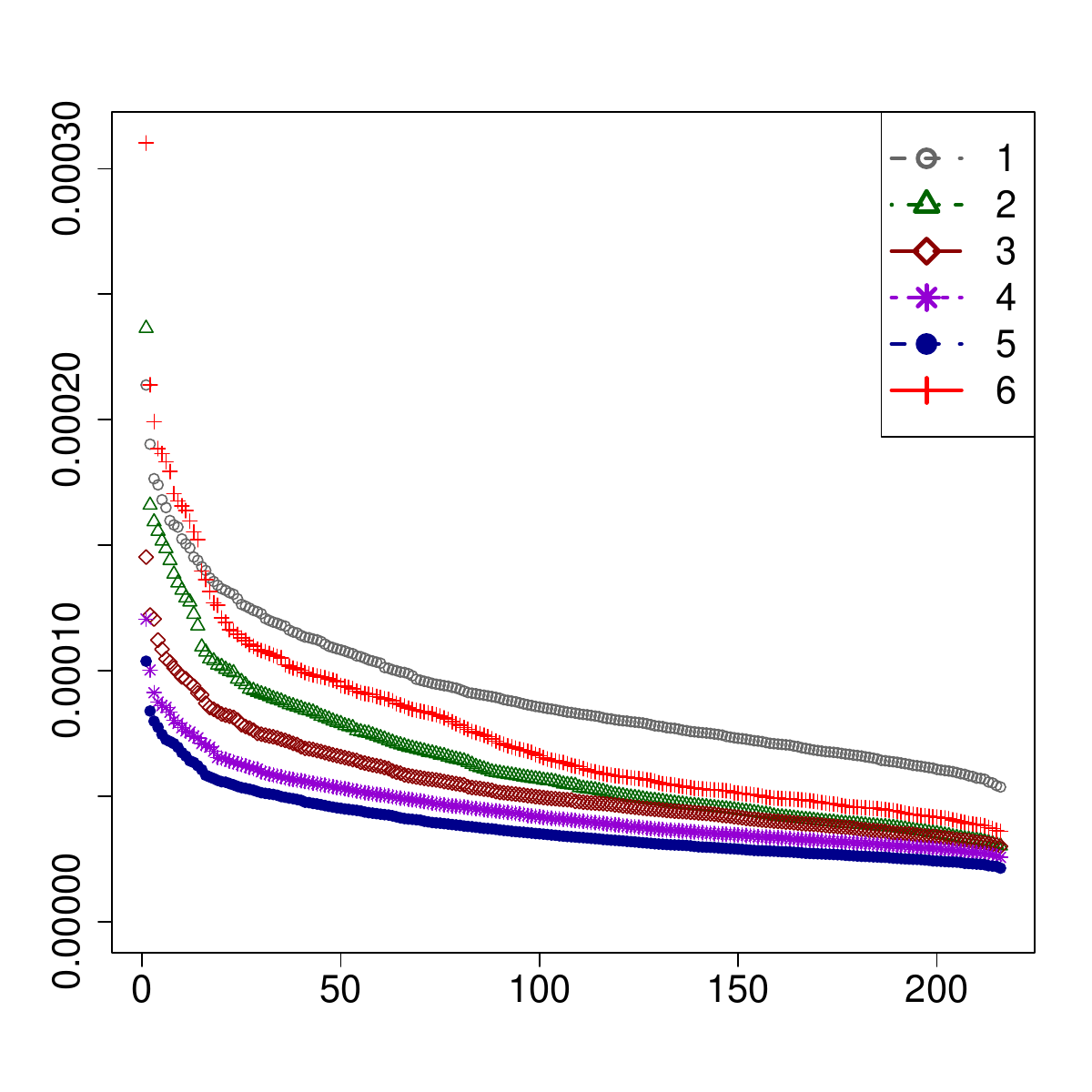}
\caption{\it  Squared Hilbert-Schmidt norm of the estimator of \eqref{eq:sepF} as a function of frequency for the 6 subjects (left) and the corresponding largest 216 eigenvalues averaged over $[0,\pi/2)$ (right).}
\label{fig:fmri_Flam}
\end{figure}

\begin{figure}[h!]
\centering 
\includegraphics[width=0.33\textwidth]{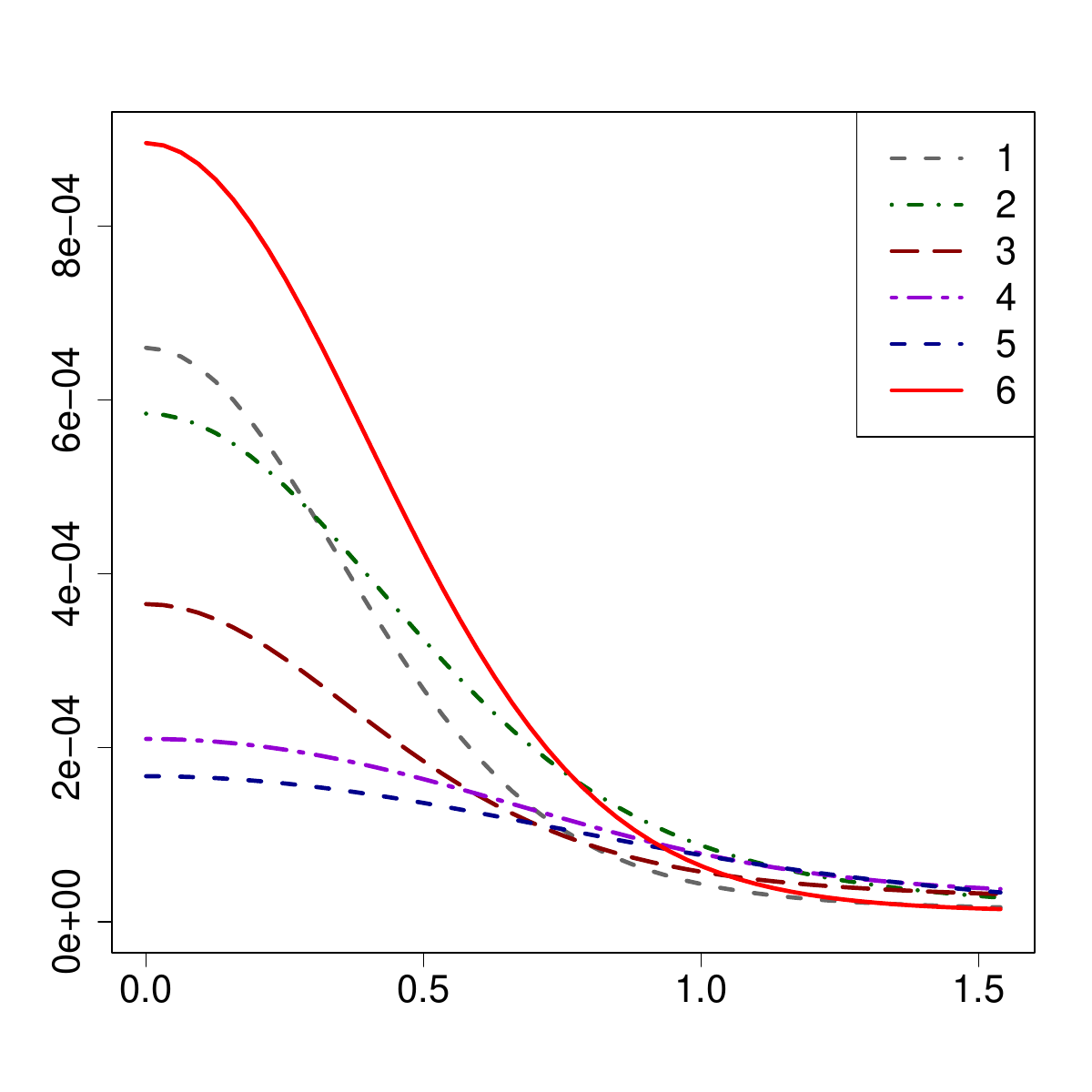}
\includegraphics[width=0.33\textwidth]{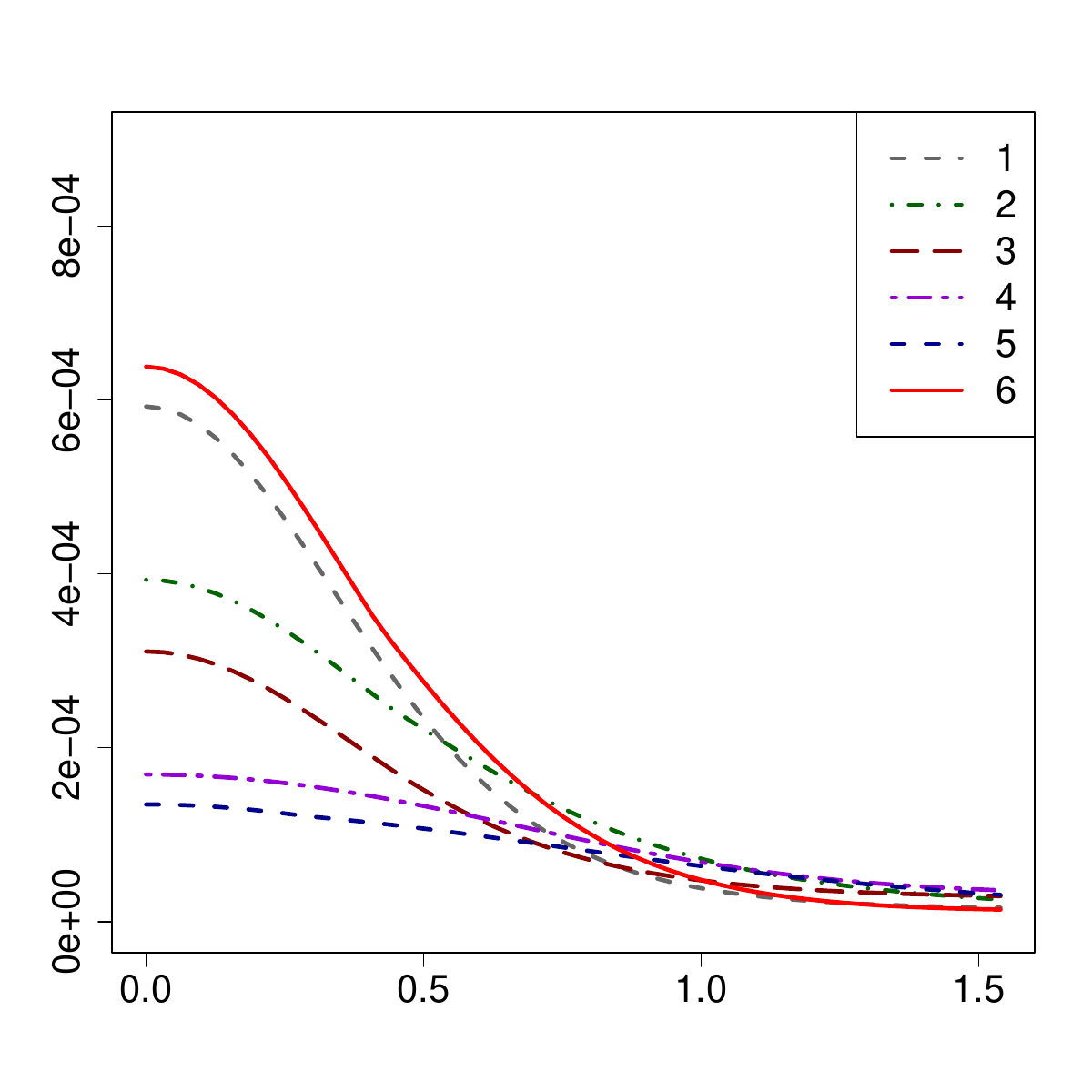}
\caption{\it  The first (left) and second (right) estimated eigenvalues of \eqref{eq:sepF} as a function of frequency for the six subjects.}
\label{fig:fmri_ev}
\end{figure}
In Figure \eqref{fig:fmri_Flam}, the squared Hilbert-Schmidt norm of the estimator of \eqref{eq:sepF} is plotted as a function of frequency for each of the six subjects. We observe most signal in the low-frequency band (<0.1 Hertz). %, which is in accordance with existing literature\HD{Biswal,Gohel}. %\HD{Some indication of oversmoothing is present in the frequencies around zero with the peak being slightly moved to frequency zero.}\todo{not sure to state this}  
The first 216 estimated eigenvalues averaged over frequences are given in the right graph. For fixed frequency, these were obtained by taking the  first $T^{1/3}={(225)}^{1/3}$ eigenvalues of each of the directional operators, taking the kronecker product of these eigenvalues and arranging the resulting values in descending order. For all subjects, a clear gap is present after approximately 2 to 3 largest averaged eigenvalues, which then taper off slowly. Tables~\ref{tab:opfmri}--\ref{tab:oplam} provides the p-values of the test statistics $\hat{\mathbb{D}}^{[a,b]}_{T_1,T_2}$, $\widehat{\mathbb{D}}^{[a,b],k}_{\Pi,T_1,T_2}$ and $\widehat{\mathbb{D}}^{[a,b],k}_{\lambda,T_1,T_2}$, respectively, for the specified relevant hypotheses. To clarify that the tests were conducted over the frequency band $[a,b]$, the hypotheses values are equipped with a superscript $[a,b]$, e.g., we write $\Delta^{[a,b]}$. For the testing frameworks of no relevant differences between the spectral density operator and between the  eigenvalues,  the thresholds $\Delta^{[a,b]}$ and $\Delta_{\lambda,i}^{[a,b]}$ are specified based upon 
%\HD{ the value of $\Delta$ and $\Delta_{\lambda,i}$ are specified based upon %an approximate 10\% deviation of the median level of signal over all directions and subjects. 
the overall signal present in the data. 

For the hypothesis of a relevant difference of at most $\Delta_{[0,\pi/2)}=1e-06$ between the spectral density operators, we observe only between subjects 3 and 5 approximately a rejection at the 6\% level, but not for $\Delta_{[0,\pi/2)}=2e-06$.  When we restrict to the frequency band where most signal is present (<0.1 Hertz), we find p-values less than $0.05$ for subjects 3 and 5 and we find some marginal evidence of relevant differences for between subjects 2 and 4 with a p-value of $0.11$ for the hypothesis $\Delta_{[0,\pi/5]}=1e-06$. For the first two eigenvalues, which are plotted in \autoref{fig:fmri_ev}, we do not find evidence of significant differences (\autoref{tab:oplam}), also not when we restrict to $[0,\pi/5]$ or when we change the value of $\Delta^{[0,\pi/2)}_{\lambda,i}$. This appears to be caused by a relatively large variance.

 More interestingly,  \autoref{tab:opPi} provides clear evidence of relevant differences between the eigenprojectors. The null of no relevant differences between the first eigenprojectors with $\Delta^{[0,\pi/2)}_{\Pi,1}=0.2$ can be rejected  for quite a few combinations at the 10\% level and in particular for all those combinations with  subject 1, for which the p-values are less than $0.07$.  The p-values corresponding to the hypotheses $\Delta^{[0,\pi/2)}_{\Pi,1}=0.3$ provides for some pairs less clear evidence of relevant differences in the first eigenprojector. In case of the second eigenprojectors, we do however reject the null of no relevant differences  for both the hypotheses  $\Delta^{[0,\pi/2)}_{\Pi,2}=0.2$ and  $\Delta^{[0,\pi/2)}_{\Pi,2}=0.3$ in most cases at a 1\% or 5\% level. Except for the tests between subject 2 and 6, we in fact reject all cases at 10\% level. This behavior did not change when we restricted to $[0,\pi/5]$.

This preliminary analysis would indicate that differences in brain activity between subjects might be driven by differences in shapes in their primary modes of variation. However, given the complexity of the data a more detailed analysis and a longer observation length might be of interest. This is however beyond the scope of this paper and is left for future research. 
\begin{table}[!h] 
\centering
\hspace*{-0pt}
\begin{tabular}{r| lllll}
  \hline
  \multicolumn{6}{c}{$H_0: \Delta^{[0,\pi/2)}=1e-06$}  \\ \hline
&  2 & 3 & 4 & 5 & 6 \\ 
 \hline
1& 0.336 & 0.328 & 0.296 & 0.269 & 0.422 \\ 
 2&  & 0.375 & 0.193 & 0.175 & 0.227 \\ 
 3&  &  & 0.175 & 0.053 & 0.176 \\ 
 4&   &  &  & 0.957 & 0.19 \\ 
 5 &  &  &  &  & 0.185 \\ 
%2.1 & 2.2 & 2.6 & 3.2 & 0.9 \\ 
 % & 1.5 & 4.9 & 5.4 & 4.1 \\ 
%  &  & 5.4 & 10.5 & 5.4 \\ 
% &  &  & -10.4 & 5 \\ 
 % &  &  &  & 5.1 \\ 
   \hline
\end{tabular}
\begin{tabular}{lllll}
  \hline
  \multicolumn{5}{c}{$H_0: \Delta^{[0,\pi/2)}=2e-06$}  \\ \hline
  2 & 3 & 4 & 5 & 6 \\ 
  \hline
   0.341 & 0.337 & 0.299 & 0.272 & 0.434 \\ 
    & 0.829 & 0.496 & 0.287 & 0.265 \\ 
    &  & 0.364 & 0.141 & 0.218 \\ 
    &  &  & 0.998 & 0.208 \\ 
    &  &  &  & 0.204 \\ 
%  2 & 2.1 & 2.6 & 3.1 & 0.8 \\ 
 %     & -4.8 & -0.1 & 2.9 & 3.3 \\ 
 %     &  & 1.7 & 5.9 & 4.2 \\ 
 %     &  &  & -28.2 & 4.5 \\ 
 %     &  &  &  & 4.7 \\ 
   \hline
\end{tabular}\\
\hspace*{-0pt}
%\begin{tabular}{rlllll}
 % \hline
%  \multicolumn{6}{c}{$H_0: \Delta^{[0,\frac{\pi}{5}]}=1e-07$}  \\ \hline
% & II & III & IV & V & VI \\ 
%  \hline
%I  & 2.1 & 2.3 & 2.7 & 3.3 & 1 \\ 
%II   &  & 6 & 9.2 & 7.5 & 4.6 \\ 
%III   &  &  & 7.8 & 14.3 & 6.3 \\ 
%IV   &  &  &  & 5.5 & 5.2 \\ 
%V &  &  &  &  & 5.3 \\ 
 %  \hline
%\end{tabular}
\begin{tabular}{r|lllll}
  \hline
  \multicolumn{6}{c}{$H_0: \Delta^{[0,\frac{\pi}{5}]}=1e-06$}  \\ \hline
 & 2 & 3 & 4 & 5 & 6 \\ 
  \hline
1   & 0.332 & 0.326 & 0.296 & 0.265 & 0.416 \\ 
2   &  & 0.242 & 0.107 & 0.127 & 0.214 \\ 
3   &  &  & 0.129 & 0.031 & 0.15 \\ 
4   &  &  &  & 0.615 & 0.187 \\ 
5   &  &  &  &  & 0.182 \\ 
% &2.1 & 2.2 & 2.7 & 3.3 & 1 \\ 
%  &   & 3.7 & \HD{7.3} & 6.5 & 4.3 \\ 
%   &  &  & 6.4 & \HD{12.6} & 5.8 \\ 
 %   &   &  &  & -1.7 & 5 \\ 
  % & &  &  &  & 5.2 \\ 
   \hline
\end{tabular}
\begin{tabular}{lllll}
  \hline
  \multicolumn{5}{c}{$H_0: \Delta^{[0,\frac{\pi}{5}]}=2e-06$}  \\ \hline
  2 & 3 & 4 & 5 & 6 \\ 
 \hline
  0.337 & 0.328 & 0.296 & 0.27 & 0.422 \\ 
    & 0.403 & 0.182 & 0.167 & 0.234 \\ 
     &  & 0.196 & 0.046 & 0.177 \\ 
     &  &  & 0.946 & 0.2 \\ 
    &  &  &  & 0.189 \\ 
   %& 2.1 & 2.2 & 2.6 & 3.2 & 0.9 \\ 
  % &  & 1.2 & 5.2 & 5.5 & 3.9 \\ 
   % &  &  & 4.9 & \HD{10.7} & 5.3 \\ 
   % &  &  &  & -9.7 & 4.8 \\ 
    %&  &  &  &  & 5 \\ 
   \hline
\end{tabular}
\caption{\it  $p$-values corresponding to the test in \autoref{thm:leveltestZF} for the specified null hypotheses and frequency bands.}
\label{tab:opfmri}
\end{table}
\begin{table}[h!] \label{tab:opPi}
\centering
%%%%%%%%%%%%%%%%%%%%%%%%%
\hspace*{-0pt}
\begin{tabular}{r | lllll }
  \hline
  \multicolumn{6}{c}{$H_0: \Delta^{[0,\pi/2)}_{\Pi,1}=0.2$ }  \\ \hline
 & 2 & 3 & 4 & 5 & 6 \\ 
  \hline
1 & 0.001 & 0.001 & 0.061 & 0.014 & 0.064 \\ 
  2 &  & 0.073 & 0.314 & 0.289 & 0.465 \\ 
  3 &  &  & 0.043 & 0.122 & 0.131 \\ 
  4 &  &  &  & 0.583 & 0.375 \\ 
  5 &  &  &  &  & 0.278 \\ 
   \hline
\end{tabular}
\hspace{-0pt}
\begin{tabular}{l  llll }
  \hline
  \multicolumn{5}{c}{$H_0: \Delta^{[0,\pi/2)}_{\Pi,2}=0.2$}  \\ \hline
  2 & 3 & 4 & 5 & 6 \\ 
  \hline
  0.004 & 0.031 & 0.057 & 0.034 & 0.024 \\ 
    & 0.006 & 0.001 & 0.007 & 0.195 \\ 
    &  & 0.004 & 0.007 & 0.029 \\ 
    &  &  & 0.019 & 0.004 \\ 
    &  &  &  & 0.056 \\ 
   \hline
\end{tabular}
\hspace{0.0pt}
%\begin{tabular}{r|l llll }
%  \hline
%  \multicolumn{6}{c}{$H_0: \Delta^{[0,\pi/2)}_{\Pi,3}=0.2$}  \\ \hline
% & 2 & 3 & 4 & 5 & 6 \\ 
 % \hline
% 1 & 0 & 0.001 & 0.039 & 0.007 & 0 \\ 
%  2 &  & 0.057 & 0.073 & 0.007 & 0.014 \\ 
 % 3 &  &  & 0 & 0.001 & 0.014 \\ 
 % 4 &  &  &  & 0 & 0.001 \\ 
 % 5 &  &  &  &  & 0.029 \\ 
 %  \hline
%\end{tabular}\\
%\hspace*{-0pt}
\begin{tabular}{r | lllll }
  \hline
  \multicolumn{6}{c}{$H_0: \Delta^{[0,\pi/2)}_{\Pi,1}=0.3$ }  \\ \hline
 & 2 & 3 & 4 & 5 & 6 \\ 
  \hline
1 & 0.002 & 0.001 & 0.132 & 0.036 & 0.17 \\ 
  2 &  & 0.205 & 0.482 & 0.414 & 0.627 \\ 
  3 &  &  & 0.117 & 0.282 & 0.242 \\ 
  4 &  &  &  & 0.887 & 0.508 \\ 
  5 &  &  &  &  & 0.42 \\ 
   \hline
\end{tabular}
\hspace{0.0pt}
\begin{tabular}{r  llll }
  \hline
  \multicolumn{5}{c}{$H_0: \Delta^{[0,\pi/2)}_{\Pi,2}=0.3$}  \\ \hline
  2 & 3 & 4 & 5 & 6 \\ 
  \hline
0.007 & 0.038 & 0.076 & 0.043 & 0.038 \\ 
   & 0.014 & 0.001 & 0.016 & 0.24 \\ 
   &  & 0.007 & 0.016 & 0.039 \\ 
   &  &  & 0.039 & 0.007 \\ 
   &  &  &  & 0.08 \\ 
   \hline
\end{tabular}
%\hspace{0.0pt}
%\begin{tabular}{r|l llll }
 % \hline
 % \multicolumn{6}{c}{$H_0: \Delta^{[0,\pi/2)}_{\Pi,3}=0.3$}  \\ \hline
%  &2 & 3 & 4 & 5 & 6 \\ 
 % \hline
%1 & 0 & 0.001 & 0.054 & 0.012 & 0 \\ 
 % 2 &  & 0.07 & 0.086 & 0.011 & 0.019 \\ 
 % 3 &  &  & 0 & 0.001 & 0.019 \\ 
 % 4 &  &  &  & 0 & 0.001 \\ 
 % 5 &  &  &  &  & 0.038 \\
 %  \hline
%\end{tabular}
\caption{\it $p$-values corresponding to the eigenprojector test (\autoref{thm:leveltestZPi}) applied to the first two  eigenprojectors for the stated hypotheses. }
\label{tab:opPi}
\end{table}

  \begin{table}[!h] 
\centering
%%%%%%%%%%%%%%%%%%%%%%%55
 \vspace*{-0pt}\begin{tabular}{l |lllll }
  \hline
  \multicolumn{6}{c}{$H_0: \Delta^{[0,\pi/2]}_{\lambda,1}=(2e-06)^2$ }  \\ \hline
 & 2 & 3 & 4 & 5 & 6 \\ 
  \hline
1 & 0.473 & 0.342 & 0.311 & 0.264 & 0.341 \\ 
  2 &  & 0.289 & 0.239 & 0.234 & 0.214 \\ 
  3 &  &  & 0.304 & 0.184 & 0.241 \\ 
  4 &  &  &  & 0.291 & 0.215 \\ 
  5 &  &  &  &  & 0.216 \\ 
   \hline
\end{tabular}
\begin{tabular}{  lllll }
  \hline
  \multicolumn{5}{c}{$H_0: \Delta^{[0,\pi/2]}_{\lambda,2}=(2e-06)^2$}  \\ \hline
 2 & 3 & 4 & 5 & 6 \\ 
  \hline
 0.446 & 0.32 & 0.302 & 0.255 & 0.487 \\ 
    & 0.295 & 0.204 & 0.206 & 0.204 \\ 
    &  & 0.319 & 0.2 & 0.207 \\ 
    &  &  & 0.255 & 0.187 \\ 
    &  &  &  & 0.196 \\ 
   \hline
\end{tabular}
%\begin{tabular}{r| lllll |}
%  \hline
%  \multicolumn{6}{c}{$H_0: \Delta^{[0,\pi/2]}_{\lambda,3}=(2e-06)^2$ }  \\ \hline
%& 2 & 3 & 4 & 5 & 6 \\  
%  \hline
% 1 & 0.452 & 0.344 & 0.293 & 0.256 & 0.485 \\ 
%  2 &  & 0.296 & 0.231 & 0.222 & 0.203 \\ 
%  3 &  &  & 0.252 & 0.165 & 0.228 \\ 
%  4 &  &  &  & 0.304 & 0.207 \\ 
 % 5 &  &  &  &  & 0.207 \\ 
 %    \hline
 %    \end{tabular}
\caption{\it  $p$-values corresponding to the eigenvalue test (\autoref{thm:leveltestZlam}) applied to the first two eigenvalues for the stated hypotheses. }
\label{tab:oplam}
\end{table}

\FloatBarrier

\small{}

\end{document}